\documentclass[a4paper,11pt]{article}
\usepackage{mdframed}
\usepackage{csquotes}
\input{Preamble.tex}

\usepackage{thmtools, thm-restate}
\usepackage{thm-autoref}
\usepackage{wrapfig}

\begin{document}

\title{Hitting Topological Minors is FPT}

\author{
Fedor V. Fomin\thanks{Department of Informatics, University of Bergen, Norway. \texttt{fomin@ii.uib.no}}
 \and Daniel Lokshtanov\thanks{University of California Santa Barbara, USA. \texttt{daniello@ucsb.edu}}
 \and Fahad Panolan\thanks{Department of Computer Science and Engineering, IIT Hyderabad, India. \texttt{fahad@iith.ac.in}}
 \and  Saket Saurabh\thanks{The Institute of Mathematical Sciences, HBNI, Chennai, India. \texttt{saket@imsc.res.in}}
 \and Meirav Zehavi\thanks{Ben-Gurion University, Beersheba, Israel. \texttt{meiravze@bgu.ac.il}}
}

\maketitle

\thispagestyle{empty}

\begin{abstract} 
In the \TMHlong\ (\TMH)  problem input consists of an undirected graph  $G$, a family of undirected graphs   ${\cal F}$ and an integer $k$. The task is to determine whether $G$ contains a set of vertices $S$ of size at most $k$, such that the graph $G\setminus S$ obtained from $G$ by removing the vertices of $S$, contains no graph from  ${\cal F}$  as a topological minor. We give an algorithm for \TMH{} with running time $f(h^\star,k)\cdot |V(G)|^{4}$. Here $h^\star$ is the maximum size of a graph in ${\cal F}$ and $f$ is a computable function of $h^\star$ and $k$. This is the first fixed parameter tractable algorithm (FPT) for the problem. In fact, even for the restricted case of planar inputs the first FPT algorithm was found only recently by Golovach et al.~[SODA 2020].  For this case we improve upon the algorithm of Golovach et al.~[SODA 2020] by designing an FPT algorithm with explicit dependence on $k$ and $h^\star$. 

\end{abstract}

\newpage
\pagestyle{plain}
\setcounter{page}{1}

 \newcommand{\fd}{{\sc $\mathcal{F}$-Minor Deletion}}
\newcommand{\isobound}{\beta}
\section{Introduction}\label{sec:intro} 

Many   important graph optimization problems can be phrased as $\Pi$-{\sc Deletion} problems. Here input is a graph $G$ on $n$ vertices and the task is to find a minimum size vertex subset $S$ such that the graph $G \setminus S$ obtained from $G$ by removing $S$ and incident edges has the property $\Pi$. 
By a well known result of Lewis and Yannakakis~\cite{LewisY80} almost all $\Pi$-{\sc Deletion} problems are \NP-complete. For this reason the study of such problems has mostly been from the perspective of methods for coping with computational intractability, such as approximation~\cite{BafnaBF99,Bar-YehudaE81,BarYGJ98,CaiDZ00,ENSS98,FominLMS12,
Fujito98,DBLP:conf/soda/GuptaLLM019,KawarabayashiS17,DBLP:journals/corr/abs-1809-08437,LundY93,
NemT74,vazirani2013approximation,williamson2011design,Yannakakis79,Yannakakis94}, exact~\cite{BliznetsFPV16,FominGK09,FominTV15,FominGLS19,FominK10,XiaoN17}, or parameterized algorithms~\cite{Cai1996,CaoM15,ChenLLOR08,FellowsL88,FlumGroheBook,FominLMS12,JansenLS14,KLPRRSS16,Nie06,ReedSV04,
DBLP:journals/corr/abs-1906-12298,paramalgoCFKLMPPS,ParameterizedComplexityBook}.
In this paper we focus on parameterized algorithms for $\Pi$-{\sc Deletion} problems: more concretely, for every property $\Pi$ the aim is to design an algorithm for $\Pi$-{\sc Deletion} that given a graph $G$ and integer $k$, determines in time $f(k)n^{\OO(1)}$ time whether a solution set $S$ of size at most $k$ exists. Such algorithms are called {\em fixed parameter tractable} (FPT). We refer to the textbooks~\cite{paramalgoCFKLMPPS,ParameterizedComplexityBook,FlumGroheBook,Nie06} for an introduction to parameterized algorithms and complexity. 

Some of the most attractive results in parameterized complexity, called {\em meta-theorems}, simultaneously establish fixed parameter tractability of entire classes of problems~\cite{ArnborgLS91,Courcelle90,CourcelleMR01,FellowsL88,FrickG01, 
GajarskyHLOORS15,GroheKS17,RobertsonS04,RobertsonS95b}. Most of these results are of the form: problems expressible in certain fragment  of logic are FPT on some restricted classes of graphs, such as graphs of bounded treewidth or cliquewidth or satisfy certain sparsity conditions. 
Unfortunately it appears unlikely that a meta-theorem on this form will apply to wide classes of vertex deletion problems on general graphs, since even very restricted logic (such as FO logic) already capture intractable problems  like {\sc Independent Set} or {\sc Dominating Set} (i.e. the problems that do not admit algorithms with running time of the form $f(k)n^{\OO(1)}$, unless an unlikely collapse of the W-hierarchy occurs).

%
%

On general graphs,   one of the  deepest and prominent  generic results about vertex deletion problems is that $\Pi$-{\sc Deletion} is (non-uniformly) FPT for every {\em minor-closed} property $\Pi$~\cite{FellowsL88,RobertsonS04}.  
 Here a graph $H$ is a \emph{minor} of a graph $G$ if $H$ can be obtained from $G$ by vertex and edge deletions and edge contractions. A property   is minor-closed if every minor of every graph that has the property also has the property. 
This meta-algorithm rests on two pillars: first, by the Graph Minors Theorem~\cite{RobertsonS04},  the set of undirected graphs is {\em well-quasi ordered} by the minor relation. Therefore, for every minor-closed family $\Pi$,  there exists a {\em finite} list of graphs ${\cal F}$ (called  {\em forbidden minors}) such that a graph $G$ is in $\Pi$ if and only if no graph $H \in {\cal F}$ is a minor of $G$. Second, by another celebrated result of Robertson and Seymour~\cite{RobertsonS95b}, there is an algorithm that given graphs $G$ and $H$ determines whether $H$ is a minor of $G$ in time $f(H) \cdot n^{3}$. Both of these results are cornerstones of the celebrated Graph Minors project of Robertson and Seymour.
Together the two results yield for every $k$ and minor closed family $\Pi$ an algorithm for $\Pi$-{\sc Deletion} with running time $f(\Pi,k)\cdot n^{3}$: the family of graphs having a solution set $S$ of size at most $k$ is itself a minor closed property. Hence it has a finite set of forbidden minors (that depends only on $\Pi$ and $k$), and whether $G$ contains any one of these forbidden minors can be checked in 
$f(\Pi,k) \cdot n^{3}$ time. This result, first observed by Fellows and Langston~\cite{FellowsL88}, together with the techniques developed for it~\cite{RobSeym91,RobertsonS86,RobertsonS-V,RobertsonS03,RobertsonS95b,RobertsonS04},  has been a driving force for a wealth of research within parameterized algorithms~\cite{FominLST10,DemaineHaj05,DemaineFHT05,GroheKR13,DBLP:conf/stoc/GroheKMW11,DBLP:journals/jct/KawarabayashiKR12,KakimuraK15,KawarabayashiK12,KawarabayashiRW11,KawarabayashiR10,Kawarabayashi09}. 

\medskip
Given the significance of the result for $\Pi$-{\sc Deletion} for minor-closed families, it is very natural to ask whether the success can be repeated. More concretely, our starting point is the following question. 
\begin{mdframed}[backgroundcolor=yellow!20] 
Can $\Pi$-{\sc Deletion} be shown to be \FPT{} for properties $\Pi$ that are not (necessarily) closed under taking minors, but instead closed under a different natural graph containment relation?
\end{mdframed}





In this paper we answer this question in the affirmative 
by studying deletion to families characterized by forbidding a finite family of  {\em topological minors}.
 A graph $H$ is a topological minor of $G$ if $H$ can be obtained from $G$ by deleting vertices or edges, and contracting edges {\em incident to at least one vertex of degree precisely $2$}. Closure under topological minors and forbidden topological minors are defined just as for minors. More precisely we study the following problem.

\medskip

\defparproblem{\TMHlong\ (\TMH)}{An  undirected graph $G$,  a family of undirected graphs   ${\cal F}$ such that every graph in  ${\cal F}$ has at most $h^\star$ vertices, and an integer $k$.}{$k+h^\star$}{Is there a vertex subset $S\subseteq V(G)$ of size $k$ such that $G \setminus S$ contains  no  graph from ${\cal F}$ as a topological minor?}

\medskip

Very recently Golovach et al.~\cite{GolovachST19} proved that the special case of \TMH\ where the input graph $G$ is required to be planar is \FPT{}. They conjecture that this is the case for also for general input graphs $G$. Our main result is a proof of this conjecture. 



\begin{theorem}\label{thm:mainIntro}
There exists an algorithm for \TMH{} with running time 
$f(k,h^\star)\cdot n^{4}$, for a computable function $f$.
\end{theorem}

\noindent
{\bf Remarks to Theorem~\ref{thm:mainIntro}.} Theorem~\ref{thm:mainIntro} is a {\em strict generalization} of the algorithm for $\Pi$-{\sc Deletion} for minor closed families $\Pi$. Indeed, every minor-closed family $\Pi$ has a {\em finite} list of forbidden minors, and for every fixed graph $H$ there is a finite list ${\cal H}$ of graphs such that $G$ contains $H$ as a minor if and only if $G$ contains some graph in ${\cal H}$ as a topological minor. To see that the generalization is {\em strict} observe that the family $\Pi$ of graphs of maximum degree $3$ has a finite list of topological minors (the star $K_{1,4}$ with four leaves), but it is not minor closed.

The restriction to topological-minor-closed families with a {\em finite} set ${\cal F}$ of forbidden topological minors is necessary. Indeed, there exist topological-minor-closed families $\Pi$ such that it is {\em undecidable} to determine whether an input graph $G$ belongs to the family\footnote{Let $\{G_i\}$ be any infinite set of graphs such that no graph $G_i$ is a topological minor of a graph $G_j$ for $i \neq j$. Then one can put property $\Pi$ be the set of all topological minors of graphs $G_i$ where the binary encoding of $i$ encodes a Turing machine that halts on the empty string.}. 

Theorem~\ref{thm:mainIntro} shows that \TMH{} is \FPT\ on arbitrary graphs. For the restricted case when the input graph $G$ is planar or embeddable in a surface of constant genus, we obtain a faster algorothm. 
\begin{restatable}{theorem}{hittingboundedgenus}
\label{thm:mainIntroGenus}
For every fixed integer $g \geq 0$, \TMH{} on graphs of Euler genus at most $g$ admits an algorithm with running time $2^{2^{k \cdot 2^{{\OO( (h^\star)^2)}}}}n^2$. 
\end{restatable}
The algorithm of Theorem~\ref{thm:mainIntroGenus} matches the quadratic dependence on $n$ of Golovach et al.~\cite{GolovachST19}, and simultaneously improves the dependence on $k$ and $h^\star$ from an un-specified computable function to triple exponential. 



The very special case of  \TMH\ with $k=0$ and ${\cal F}$ consisting of one graph $H$, is known as  the \TMClong\ (\TMC) problem. In \TMC\  input consists of two undirected graphs $G$ and $H$, and the task is to determine whether $G$ contains $H$ as a topological minor. 

\medskip
\defparproblem{\TMClong\ (\TMC)}{Two undirected graphs, $G$ and $H$.}{$h^\star=|V(H)|$}{Does $G$ contain $H$ as a topological minor?}

Since the cycle on $|V(G)|$ vertices is a topological minor of $G$ if and only if $G$ is Hamiltonian, the problem of deciding whether a graph $G$ contains a graph $H$ as a topological minor is \NP-complete. 
The complexity  study of  \TMC,  which is also known as the {\sc Subgraph Homeomorphism} problem, can be traced back to the 1970s \cite{LapaughR78}. \TMC\ admits an algorithm running in time $n^{\OO(|E(H)|)}$ by a reduction to  $n^{\OO(|E(H)|)}$ instances of the \DisjointPathsLong\ problem (given an undirected graph $G$ and a set of $k$ pairs, $\{s_i,t_i\}_{i=1}^k$,  the objective is to find $k$ vertex-disjoint paths connecting $s_i$ to $t_i$)~\cite{RobertsonS95b}). Each of these \DisjointPathsLong\ instances can be solved by the $f(k)n^3$ algorithm of Robertson and Seymour~\cite{RobertsonS95b}. One of the longstanding open questions in parameterized complexity was whether \TMC{} is \FPT{}, that is, solvable in time $f(H)\cdot n^{\OO(1)}$ time. In 2010, Grohe et al.~\cite{DBLP:conf/stoc/GroheKMW11} resolved this question in the affirmative by designing an algorithm with running time $f(H)\cdot n^{3}$ for a computable function $f$.

Because  the \TMH\ contains \TMC{} as a special case, Theorem~\ref{thm:mainIntro} generalizes the result of  Grohe et al.~\cite{DBLP:conf/stoc/GroheKMW11}. It appears difficult to obtain an \FPT{} algorithm for \TMH\ by invoking the results of Grohe et al.~\cite{DBLP:conf/stoc/GroheKMW11} in a black box fashion. Indeed, a part of our proof of  Theorem~\ref{thm:mainIntro} is a new \FPT{} algorithm for \TMC{}. This new algorithm has some appealing features, such as a ``only'' triple-exponential dependence on $h^\star$ when the input graph $G$ is planar or embedded on a surface of constant Euler genus.

\subsection{Related Work}






For containment relations $\preceq$ that are not well-quasi-ordered we cannot hope for algorithmic results for $\Pi$-{\sc Deletion} for all $\preceq$-closed families $\Pi$, because, just as for topological minors, there exist  $\preceq$-closed families $\Pi$ such that determining whether $G$ is in the family is undecidable. Hence one has to settle for results that handle only some $\preceq$-closed properties $\Pi$, e.g. ones that have a finite number of minimal elements that do not have the property.
Even results of this type are rare: to the best of our knowledge, prior to our work, the {\em subgraph} and {\em induced subgraph} relations were the only relations $\preceq$ that on the hand are not a well quasi order, and on the other hand admit an \FPT{} algorithm for $\Pi$-{\sc Deletion} for every $\preceq$-closed family $\Pi$ with a finite number of $\preceq$-forbidden graphs. For subgraphs and induced subgraphs, a simple and elegant \FPT-algorithm based on branching was obtained by Cai~\cite{Cai1996}. For most other natural containment relations (such as {\em induced minors} or {\em contractions}) an \FPT{} algorithm for $\Pi$-{\sc Deletion} for every $\preceq$-closed family $\Pi$ with a finite number of $\preceq$-forbidden graphs would imply that ${\sf P}=\NP$ (see the discussion in Golovach et al.~\cite{GolovachST19}).

For minor closed families $\Pi$ excluding at least one planar graph, algorithms for $\Pi$-{\sc Deletion} with improved running times~\cite{Bodlaender97,FellowsL88,FominLMS12,KLPRRSS16,GiannopoulouPRT17}  have been found. For certain restricted immersion-closed families 
$\Pi$, Giannopoulou et al.~\cite{GiannopoulouPRT17} obtained improved algorithms for the edge deletion variant of $\Pi$-{\sc Deletion}. A substantial body of work focuses on developing \FPT{} algorithms for $\Pi$-{\sc Deletion} for concrete instantiations of $\Pi$, and on optimizing the running times of these algorithms~\cite{Cai1996,ChenKX10,BeckerBRG00-Ra,CaoM15,ChenLLOR08,FellowsL88,FominLMS12,JansenLS14,KLPRRSS16,Nie06,ReedSV04,
DBLP:journals/corr/abs-1906-12298,DBLP:journals/corr/abs-1905-12233}.

\subsection{Our Methods}
Our algorithm is built on the template of Robertson and Seymour's algorithm for {\sc Disjoint Paths}~\cite{RobertsonS95b}. This approach has previously successfully been deployed for an impressive array of different problems~\cite{GroheKR13,JansenLS14,DBLP:journals/jct/KawarabayashiKR12,KakimuraK15,
KawarabayashiK12,KawarabayashiRW11,KawarabayashiR10,Kawarabayashi09}, including the algorithm for \TMC\ by Grohe et al.~\cite{DBLP:conf/stoc/GroheKMW11}. Algorithms using this scheme distinguish between the following three cases.

\begin{description}[noitemsep,topsep=0pt]
\item[Case 1.] The treewidth of the input graph $G$ is {\em small}. Here small means that it is upper bounded by a computable function of the parameters. 
\item[Case 2.] The input graph $G$ has a {\em large} clique minor. Here large means that the size of the clique minor is lower bounded by a computable function of the parameters. 
\item[Case  3.] Neither Case 1 nor Case 2 occurs, which means that treewidth of $G$ is ``large", while $G$ excludes a ``small" clique as a minor. In this case the ``weak structure theorem''~\cite{RobertsonS95b}  implies that the graph $G$ contains a ``large flat wall". This is (essentially) a large grid subgraph of $G$ such that only the outer face vertices of the grid subgraph have neighbors in $G$ that are not also in the subgraph. 
\end{description}

The algorithm for Case 1 is easy, as we can write a MSO (monadic second order) formula for \TMH\ and use the meta-theorems for graphs of bounded treewidth to obtain the algorithm ``for free"~\cite{ArnborgLS91,Courcelle90}.  The template for Cases $2$  and $3$ is to identify an {\em irrelevant vertex}, that is a vertex $v$ such that the answer to the problem under consideration is the same in $G$ and in $G\setminus v$. The crucial and problem-dependent piece of algorithms using this template is how they identify the irrelevant vertex, and this is where the novelty of such algorithms (including ours) lies. 




We now give a {\em very} high level outline of how our algorithm identifies irrelevant vertices. A more detailed overview is provided in the next section. We will need to distinguish between different types of irrelevance, so we will now introduce simplified versions of the technical notions of irrelevance that we use. 

Let us define {\em $0$-irrelevant}, to mean that the $\delta$-folio (the family of all graphs $H$ of size at most $\delta$ that are topological minors of $G$) is the same in $G$ and $G\setminus v$. This is the notion of irrelevance used for \TMC, since removing a $0$-irrelevant vertex will not change whether a small graph $H$ is a topological minor of $G$ or not. Grohe et al.~\cite{DBLP:conf/stoc/GroheKMW11} give efficent algorithms to compute a $0$-irrelevant vertex when $G$ contains a large clique minor or a large flat wall.   
The notion of $0$-irrlevance is not strong enough for \TMH, because even though the $\delta$-folio of $G$ and $G-v$ is the same there could exist a vertex set $S$ of size at most $k$ such that some small graph $H$ is in the $\delta$-folio of $G\setminus S $, but not in the $\delta$-folio of $G \setminus (S\cup \{v\})$. Thus $S$ is a solution for \TMH{} with ${\cal F} = \{H\}$ in $G\setminus (S \cup \{v\})$ but not in $G$. This motivates the definition of $k$-irrelevant vertices: a vertex $v$ is $k$-{\em irrelevant} if for every set $S$ of size at most $k$, the $\delta$-folio of $G\setminus S$ and 
$G\setminus (S\cup \{v\})$ is the same. Observe that removing a $k$-irrelevant vertex $v$ from $G$ will not change whether $G$ is a ``yes'' instance of \TMH{}. By varying $k$ we get a smooth transition between $0$-irrelevance, for which we already have an efficient algorithm from Grohe et al.~\cite{DBLP:conf/stoc/GroheKMW11}, and the notion of $k$-irrelevance that we need.





In the large clique minor case the algorithm to compute a $k$-irrelevant vertex is a branhcing algorithm with a twist. Branching algorithms are ubiquitous in parameterized algorithms (see e.g.~\cite{paramalgoCFKLMPPS}). Typically branching algorithms are employed to find some object (in our case we are searching for a $k$-irrelevant vertex $v$). The branching algorithm non-deterministically ``guesses" some features of the object we are looking for. In the branch that guesses the ``correct" features of $v$ the algorithm uses the guess to find $v$. Such a scheme cannot possibly work for us, because we need to iteratively find many (perhaps as many as $\Omega(n)$) irrelevant vertices, and we can not afford to guess the features of all of these vertices.

Our branching algorithm instead guesses features of the set $S$. In each branch the algorithm marks some vertices of the clique as ``relevant". The important properties of our branching algorithm are: (a) The number of different branches is a function of $k$ and $h$; (b) In each branch the number of vertices in the clique minor marked as ``relevant" is a function of $k$ and $h$; and (c) For every vertex $v$ and set $S$ of size at most $k$ such that the $\delta$-folio $G\setminus (S\cup \{v\})$ and $G\setminus S$ is {\em not} the same, $v$ is marked as relevant in the branch that correctly guesses the features of $S$. Observe that if we start with a sufficiently large clique minor then some vertex $v$ in the clique minor will not be marked as relevant in any branch. By property (c) of the branching algorithm this vertex is $k$-irrelevant and can be removed. As opposed to traditional branching algorithms, here the irrelevance of $v$ is not contingent on being in the correct branch, therefore we can run this algorithm again and again to find new irrelevant vertices without getting a combinatorial explosion of the number of non-deterministic guesses the algorithm needs to make. We believe that this trick of mixing branching and irrelevant vertex techniques is interesting in its own right and will find further applications.

For reasons that are too technical to go into in this brief overview the branching strategy employed for the large clique minor case does not quite work out for the large flat wall case. 
However, here a strengthening of the irrelevant vertex rule for \TMC{} comes to the rescue. We are able to show that in a sufficiently large flat wall $W$ in $G$, one can efficiently find a large subwall $W'$ such that the $\delta$-folio of $G$ and $G \setminus W'$ are the same. Furthermore, $W'$ depends only on $W$ and not on the rest of $G$. By $W'$ being large we mean that the size of $W'$ tends to infinity as the size of $W$ tends to infinity. Notice that this is stronger than finding just one $0$-irrelevant vertex in the flat wall case, because all vertices in $W'$ can be deleted simultaneously without changing the $\delta$-folio.

This strengthening gives a simple way of locating $k$-irrelevant vertices in a flat wall: assuming $W=W_0$ is sufficiently large, apply the irrelevant wall lemma and find a large subwall $W_1$ in it. Re-apply the lemma to $W_1$ to obtain $W_2$ and so on, until we have a chain of sub-walls $W_0, W_1, \ldots, W_{k+1}$. We claim that every vertex $v$ in $W_{k+1}$ is $k$-irrelevant.
To see this, consider an arbitrary set $S$ of size at most $k$.
There exists some $i$ such that $(W_{i} \setminus W_{i+1} )\cap S = \emptyset$. Set $S_{out} = S \setminus W_{i}$, we have that $S \cap W_{i} \subseteq W_{i+1}$. Therefore, $W_{i}$ is a flat wall in $G \setminus S_{out}$ and hence all of $W_{i+1}$ is irrelevant in $G \setminus  S_{out}$. In other words, $G \setminus S_{out}$ and $G \setminus (S_{out} \cup W_{i+1})$ have the same $\delta$-folio. But the $\delta$-folio of $G\setminus S$ is contained in the $\delta$-folio of $G \setminus S_{out}$ while the $\delta$-folio of $G \setminus(S_{out} \cup W_{i+1})$ contains the $\delta$-folio of $G \setminus (S \cup \{v\})$. So $G \setminus S$ and $G \setminus  (S \cup \{v\})$ must have the same $\delta$-folio, proving that $v$ is $k$-irrelevant.  We believe that the methods employed to find an irrelevant vertex/wall should be applicable to other problems as well.


%

\paragraph{Topological Minor Containment} 
Our algorithm for \TMH{} relies on several crucial subroutines for \TMC{}. First we need an \FPT{} algorithm for \TMH{} on bounded treewidth graphs - this automatically yields such an algorithm for \TMC{} on bounded treewidth graphs as well. 
Second, in the large clique minor case of \TMH{} we need an \FPT{} algorithm for \TMC{} (for this we can invoke Grohe et al.~\cite{DBLP:conf/stoc/GroheKMW11}), and an efficient algorithm that finds $0$-irrelevant vertices in the case of a large grid minor. For this step we can also directly invoke the results of Grohe et al.~\cite{DBLP:conf/stoc/GroheKMW11}. 
Finally, for the large flat wall case we need to find a large irrelevant flat subwall of a huge flat wall. At present we do not know how to directly invoke the results of Grohe et al.~\cite{DBLP:conf/stoc/GroheKMW11} to achieve this. For this reason we have included our own proof of the ``irrelevant wall" lemma. 
The irrelevant wall lemma and the bounded treewidth algorithm puts us quite close to having an alternative \FPT{} algorithm for \TMC{} - all that is missing is an algorithm for finding $0$-irrelevant vetices in the case of large clique minors. For this reason we include that as well even if the ideas here are essentially same as what is presented by Grohe et al.~\cite{DBLP:conf/stoc/GroheKMW11}. Our algorithm for TMC has some appealing features, making it interesting in its own right.




The proof of the irrelevant subwall lemma crucially relies on the Unique Linkage Theorem of Robertson and Seymour~\cite{RobertsonS12}. The Unique Linkage Theorem 
essentially states that there exists a computable function $h : \mathbb{N} \rightarrow \mathbb{N}$, such that for any instance $(G, \{s_i,t_i\}_{i=1}^k)$ of \DisjointPathsLong{}, if the graph $G$ can be partitioned in two pieces $A$ and $B$, such that $A$ can be drawn in the plane without edges crossing, at most a constant number of vertices of $B$ are adjacent to a vertex of $A$ that is not on the outer face of $A$, all terminals $\{s_i,t_i\}_{i=1}^k$ are in $B$, and and a vertex $v$ in $A$ is insulated from the outer face of $A$ by at least $h(k)$ pairwise disjoint concentric cycles, then $v$ is irrelevant (for the {\sc Disjoint Paths} problem, see (3.1) in \cite{RobertsonS12} or \autoref{lem:RoSe} for a formal statement). This means that $(G,  \{s_i,t_i\}_{i=1}^k)$ is a ``yes'' instance of \DisjointPathsLong{} if and only if $(G \setminus v, \{s_i,t_i\}_{i=1}^k)$ is.
Robertson and Seymour state at the end of~\cite{RobertsonS95b} that the function $h$ is computable, however, to the best of our knowledge no concrete upper bound for $h$ has been published. For the special cases of  planar graphs and graphs of bounded genus a single exponential bound is known~\cite{mazoit2013single}.


The size bounds on $W$ in our irrelevant subwall lemma depend on $\delta$, but also on the bound $\isobound(h^\star)$ for \DisjointPathsLong. As a consequence we get an algorithm for  \TMC{}  whose running time depends on the upper bound $h$ for the ``cycle insulation bound'' in the Unique Linkage Theorem for the \DisjointPathsLong{}. Specifically, we obtain the following theorem:


\begin{theorem}\label{thm:generalGraphsAlgorithmTMC}
\label{thm:maingen}
\TMC{} admits an algorithm with running time 
$f(h^\star)\cdot n^3$ on instances where 
$f(h^\star)=2^{2^{\OO(r_1^{3})}} 2^{2^{2^{(r_2 \cdot 2^{\OO((h^{\star})^4  )})} }}$, 
$r_2=\isobound (2^{c \cdot (h^{\star })^4})$, 
$r_1=\isobound(2^{2^{(c\cdot (h^{\star })^4)}r_2})$,
$\isobound(\cdot)$ is the cycle insulation bound for the Unique Linkage Theorem for \DisjointPathsLong{},
and $c$ is a constant.
\end{theorem} 

While \autoref{thm:generalGraphsAlgorithmTMC} does not yield a concrete bound on the running time  of our algorithm, it shows that it is sufficient to provide such a bound for the Unique Linkage Theorem. 
On graphs of bounded genus our algorithm performs better - indeed it is the first algorithm for \TMC{} on planar graphs, and more generally on graph of bounded genus, with explicit bounds on the running time dependence on $h^\star$. 

\begin{theorem}
\label{thm:main}
For every fixed integer $g \geq 0$, \TMC{} on graphs of Euler genus at most $g$ admits an algorithm with running time $2^{2^{2^{\OO((h^\star)^2)}}}n^2$.
\end{theorem}

 \paragraph*{Organization of the paper.} In Section~\ref{sec:overviewintro} we give a technical overview of our results and 
 methodologies. Section~\ref{sec:prelims} contains graph theoretic notions needed in the paper,  such as treewidth,  (topological) minors, grid, flat wall, etc. It also contains several known results. Section~\ref{sec:mainthms} gives a roadmap of statements  using which we derive our main results. In Section~\ref{sec:irr} we define the notion of  $\delta$-representative and prove some auxiliary lemmas which are used in later sections. In Section~\ref{sec:cliqueminor} we explain how to find a $(\delta,k)$-irrelevant vertex when  the input graph has a large clique minor.  Finding a $\delta$-representative when the input graph has bounded treewidth is explained in Section~\ref{sec:smalltw}. In Section~\ref{sec:recursive}, we explain how all the pieces can be put together using recursive understanding to solve \TMC. Sections~\ref{sec:disjirr} to \ref{sec:finallyend} are devoted to  explain the computation of $(\delta,k)$-irrelevant vertex when the input graph has a large flat wall.   Finally, we conclude with some future research avenues in Section~\ref{sec:conclusion}. 

\medskip 
\noindent 
\section{Overview of our Proof and Techniques}
\label{sec:overviewintro}
In this section we give a brief overview of our proof. 
Our algorithm is built on the template of Robertson and Seymour's algorithm for {\sc Disjoint Paths}~\cite{RobertsonS95b}. 
 Algorithms using this scheme distinguish between the following three cases: {\bf Case 1:} The treewidth of the input graph $G$ is {\em small}. Here small means that it is upper bounded by a computable function of the parameters. {\bf Case 2:} The input graph $G$ has a {\em large} clique minor. Here large means that the size of the clique minor is lower bounded by a computable function of the parameters. {\bf Case 3:} Neither Case 1 nor Case 2 occurs, which means that treewidth of $G$ is ``large", while $G$ excludes a ``small" clique as a minor. In this case the ``weak structure theorem''~\cite{RobertsonS95b}  implies that the graph $G$ contains a ``large flat wall". This is (essentially) a large grid subgraph of $G$ such that only the outer face vertices of the grid subgraph have neighbors in $G$ that are not also in the subgraph. 

These three cases are enveloped under the ``recursive understanding technique''~\cite{DBLP:conf/stoc/GroheKMW11}. This  top level structure of our algorithm is same as that of Grohe et al.~\cite{DBLP:conf/stoc/GroheKMW11}. The base of this recursive scheme falls into Cases $1$, $2$ and $3$, as mentioned above.  
Case $(1)$ can be handled easily using a standard dynamic programming algorithm~\cite{ArnborgLS91,Courcelle90}.  For case $(2)$ we design a new algorithm to find an irrelevant vertex for $\TMH$, and for $\TMC$ we make use of the algorithm by  Grohe et al.~\cite{DBLP:conf/stoc/GroheKMW11}. For Case $3$, as a subroutine for \TMH, we need an algorithm that finds {\em not} an irrelevant vertex but rather an irrelevant flat subwall. 
At present we do not know how to directly invoke the results of Grohe et al.~\cite{DBLP:conf/stoc/GroheKMW11} to achieve this. For this reason we have included our own proof of the ``irrelevant wall" lemma.  
Thus, to find this irrelevant flat subwall  
 we give a complete stand alone proof with black box calls to the unique linkage  theorem of Robertson and Seymour~\cite{RobertsonS12}. Instead of solving \TMC\ and \TMH\, we solve more general problems called ``finding folio'' and ``hitting folio'' in rooted graphs. These general problems are essentially required to be solved in the recursive calls  for \TMH and \TMC. We start with some essential notations and definitions in the next subsection and then we move to the technique of the proof.

\subsection{Notations and (Informal) Definitions}
\label{sec:notations}
For a graph $G$, we use $n=\vert V(G)\vert$ and $m=\vert E(G)\vert$. 
A {\em rooted graph} is a graph with a specified subset of vertices, called roots, denoted by $R(G)$, and with an injective map $\rho_G(R(G))\rightarrow {\mathbb N}$.  We say that two rooted graphs $G_1$ and $G_2$ are {\em compatible} if  $\rho_{G_2}(R(G_1))=\rho_{G_2}(R(G_2))$. In this case we also say that the roots of $G_1$ and $G_2$ are the same, where we consider $u\in R(G_1)$  and  $v\in R(G_2)$ as identical if $\rho_{G_1}(u)=\rho_{G_2}(v)$. For two (not necessarily vertex-disjoint) graphs $G_1$ and $G_2$, we use $G_1\cup G_2$ to denoted the graph $(V(G_1)\cup V(G_2),E(G_1)\cup E(G_2))$. A {\em separation} of a graph $G$ is a pair $(G_1,G_2)$ with the property that $G=G_1\cup G_2$ and $E(G_1)\cap E(G_2)=\emptyset$.  The order of the separation $(G_1,G_2)$ is $\vert V(G_1)\cap V(G_2)\vert$. For a separation $(G_1,G_2)$ of $G$ and a graph $G_1'$ with $V(G_1)\cap V(G'_1)=V(G_1)\cap V(G_2)$, replacing $G_1$ with $G_1'$ in the separation $(G_1,G_2)$ will lead to the graph $G'_1\cup G_2$.

A rooted graph $H$ is a minor in a rooted graph $G$ if $H$ is a minor in $G$ with the additional  restriction that a root vertex $v$ in $H$ is mapped to a vertex subset in $G$ that contains the identical vertex of $v$ in $G$. 
We say that $H$ has {\em detail} $\leq \delta$  if $\vert E(H)\vert \leq \delta$ and $\vert V(H)\setminus R(H)\vert \leq \delta$. The {\em $\delta$-\mfolio} of $G$ is the set of all minors of $G$ with detail $\leq \delta$.

A rooted graph $H$ is a topological minor in a rooted graph $G$ if $H$ is a topological minor in $G$ with the additional restriction that a root vertex in $H$ can only be mapped to its identical vertex in $G$. That is, there is a pair of injective functions 
$(\phi,\varphi)$, where $(a)$ $\phi$ is a map from $V(H)$ to $V(G)$ with the restriction that for any $u\in V(H)$, $\rho_{G}(\phi(u))=\rho_H(u)$, $(b)$ for any $\{x,y\}\in E(H)$, $\varphi(\{x,y\})$ is a path from $\phi(x)$ to $\phi(y)$ in $G$ such that $\{\varphi(e)\colon e\in E(H)\}$ is a set of pairwise internally vertex disjoint paths and no vertex in $\phi(V(H))$ is an internal vertex in any of these paths. We note that there is no restriction on where a non-root vertex of $H$ could be mapped to; in fact, it can be mapped to a root vertex in $G$. The pair  $(\phi,\varphi)$ corresponds to a subgraph $G'$ of $G$, which is a subdivision of $H$ where the roots of $H$ are mapped to their corresponding vertices in $G$. We call $G'$ a {\em realization} of $H$ in $G$ witnessed by $(\phi,\varphi)$.  Moreover, we call the vertices of $\phi(H)$ terminals of $(H,\phi,\varphi)$. For an integer $\delta\in {\mathbb N}$, the {\em $\delta$-folio} of a rooted graph $G$ is the set of topological minors $H$ in $G$ with $|E(H)|+\isolated(H)\leq \delta$, where $\isolated(H)$ is the number of degree zero vertices in $H$. 
Given a rooted graph $G$ and a graph $X$ such that $V(X)=R(G)$, 
The {\em $(X,\delta)$-folio} of $G$ is the $\delta$-folio of $(G\cup X)$ with $R(G)$ as the set of roots.
The {\em extended $\delta$-folio} of $G$ is the function $f$ whose domain is the set of all graphs on vertex set $R
(G)$, and for each graph $X$ on $R(G)$, $f(X)$ is equal to the $(X,\delta)$-folio of $G\cup X$. The general problem of finding the folio of a rooted graph is the following. 

\defparproblem{\FindFolio}{A rooted undirected graph $G$ with $\vert R(G)\vert \leq 16\delta^{2}$, and a non-negative integer $\delta$.}{$\delta$}{What is the extended $\delta$-folio of $G$?}

\noindent 
 The bound $16\delta^2$ in \FindFolio\ is required to handle the base case of the recursion when there is a {\em large} (depending on $\delta$) clique minor in the graph, by the algorithm of Grohe et al.~\cite{DBLP:conf/stoc/GroheKMW11}.  
 
 A vertex $v$ in the input rooted graph $G$ is  {\em irrelevant} for \FindFolio\ if the extended $\delta$-folio of $G$ is same as the extended $\delta$-folio of $G\setminus v$. For \TMH, we need a stronger notion of irrelevance. We say that a vertex $v$ in a rooted graph $G$ is 
 {\em $(\delta,k)$-irrelevant}  if for any graph $X$ on $R(G)$ and any vertex subset $S\subseteq V(G)$ of size at most $k$, the $\delta$-folio of $G'=(G\cup X)\setminus S$ (where $R(G')=R(G)\setminus S$) is equal to the the $\delta$-folio of $G' \setminus v$.  Notice that a vertex $v$ is irrelevant for \FindFolio\ if and only if it is $(\delta,0)$-irrelevant.

\subsection{(Nice) Flat Wall and (Nice) Flat Wall Theorem}
\label{subsec:flat}
Let $[n]$ denote the set $\{1,2,\ldots,n\}$. 
An {\em $a\times b$-grid} is a graph $G$ whose vertex set can be denoted as $\{v_{i,j}: i\in[a],j\in[b]\}$, so that the edge set of $G$ is exactly $\{\{v_{i,j},v_{i',j'}\}: i,i'\in[a], j,j'\in[b],|i-i'|+|j-j'|=1\}$. Let $J$ be an $h\times(2r)$ grid. For any column $C_j$ (the path $v_{1,j}-\ldots-v_{h,j}$) of $J$,  where $j\in[2r]$, let $e^j_1 , e^j_2,\ldots, e^j_{h-1}$ be the edges of $C_j$, in the order of their appearance on $C_ j$, where $e^j_1$ is incident with  $v_{1, j}$.  For any column $C_j$, if $j$ is odd, then delete from the graph $J$ all edges $e^j_i$ where $i$ is even. For any column $C_j$, if $j$ is even, then delete from the graph all edges $e^j_i$ where $i$ is odd. After this, delete all vertices of degree $1$. The resulting graph $\widehat{W}$ is the {\em elementary wall} of height $h$ and width $r$ (also called the $(h\times r)$-elementary wall) and the vertices of degree $2$ in $\widehat{W}$ are called the {\em pegs} of $\widehat{W}$. 
(See \autoref{fig:elementarywall} for an example). 
\begin{wrapfigure}{l}{0.5\textwidth}
\begin{center}
\begin{tikzpicture}[scale=0.45]
\node[red] (a1) at (0,0) {$\bullet$}; 
\node[red] (a1) at (1,0) {$\bullet$}; 
\node[] (a1) at (2,0) {$\bullet$}; 
\node[red] (a1) at (3,0) {$\bullet$}; 
\node[] (a1) at (4,0) {$\bullet$}; 
\node[red] (a1) at (5,0) {$\bullet$}; 
\node[] (a1) at (6,0) {$\bullet$}; 
\node[red] (a1) at (7,0) {$\bullet$}; 
\node[red] (a1) at (8,0) {$\bullet$}; 

\node[red] (a1) at (0,1) {$\bullet$}; 
\node[] (a1) at (1,1) {$\bullet$}; 
\node[] (a1) at (2,1) {$\bullet$}; 
\node[] (a1) at (3,1) {$\bullet$}; 
\node[] (a1) at (4,1) {$\bullet$}; 
\node[] (a1) at (5,1) {$\bullet$}; 
\node[] (a1) at (6,1) {$\bullet$}; 
\node[] (a1) at (7,1) {$\bullet$}; 
\node[] (a1) at (8,1) {$\bullet$}; 
\node[red] (a1) at (9,1) {$\bullet$}; 

\node[red] (a1) at (0,2) {$\bullet$}; 
\node[] (a1) at (1,2) {$\bullet$}; 
\node[] (a1) at (2,2) {$\bullet$}; 
\node[] (a1) at (3,2) {$\bullet$}; 
\node[] (a1) at (4,2) {$\bullet$}; 
\node[] (a1) at (5,2) {$\bullet$}; 
\node[] (a1) at (6,2) {$\bullet$}; 
\node[] (a1) at (7,2) {$\bullet$}; 
\node[] (a1) at (8,2) {$\bullet$}; 
\node[red] (a1) at (9,2) {$\bullet$}; 

\node[red] (a1) at (0,3) {$\bullet$}; 
\node[red] (a1) at (1,3) {$\bullet$}; 
\node[] (a1) at (2,3) {$\bullet$}; 
\node[red] (a1) at (3,3) {$\bullet$}; 
\node[] (a1) at (4,3) {$\bullet$}; 
\node[red] (a1) at (5,3) {$\bullet$}; 
\node[] (a1) at (6,3) {$\bullet$}; 
\node[red] (a1) at (7,3) {$\bullet$}; 
\node[red] (a1) at (8,3) {$\bullet$}; 

\draw (0,3)--(1,3)--(2,3)--(3,3)--(4,3)--(5,3)--(6,3)--(7,3)--(8,3);
\draw (0,2)--(1,2)--(2,2)--(3,2)--(4,2)--(5,2)--(6,2)--(7,2)--(8,2)--(9,2)--(9,1); 
\draw (0,1)--(1,1)--(2,1)--(3,1)--(4,1)--(5,1)--(6,1)--(7,1)--(8,1)--(9,1);
\draw (0,0)--(1,0)--(2,0)--(3,0)--(4,0)--(5,0)--(6,0)--(7,0)--(8,0);

\draw (0,0)--(0,1);
\draw (0,2)--(0,3);
\draw (2,0)--(2,1);
\draw (2,2)--(2,3);
\draw (4,0)--(4,1);
\draw (4,2)--(4,3);
\draw (6,0)--(6,1);
\draw (6,2)--(6,3);
\draw (8,0)--(8,1);
\draw (8,2)--(8,3);
\draw (1,2)--(1,1);
\draw (3,2)--(3,1);
\draw (5,2)--(5,1);
\draw (7,2)--(7,1);
\end{tikzpicture}
\end{center}
\caption{Example of a $(4\times 5)$-elementary wall $\widehat{W}$. The red colored vertices are the pegs of $\widehat{W}$.}
\label{fig:elementarywall}
\end{wrapfigure}
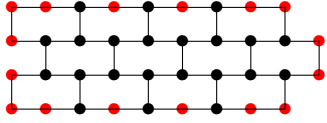
A {\em wall} $W$ of size $w$ (also called a $(w\times w)$-wall) is a subdivision of a $(w\times w)$-elementary wall $\widehat{W}$ and the pegs of $W$ are defined to be the vertices of $W$ that serve as the pegs of $\widehat{W}$. In other words,  $\widehat{W}$ is a topological minor in $W$ where $W$ itself is the realization of it. If there is an $(h\times 2r)$-grid minor in a graph $G$, then there is an $(h\times r)$-wall (as a subgraph) in $G$. On the other hand if there is an $(h\times r)$-wall in a graph $G$, then there is an $(h\times r)$-grid minor in $G$.

Informally, a wall $W$ in a graph $G$ is called a {\em flat wall} if there is a ``portion of  the graph $G$" that contains $W$ and this portion,  say, $G'$, can ``almost" be drawn in a disc in the plane where the pegs of $W$ appear on the boundary of the disc in the order in which they appear on the outer boundary of $W$. Moreover, the vertices in $V(G)\setminus V(G')$ do not have neighbors strictly inside the disc.  
More formally, a wall $W$ is a  flat wall in $G$  if there is a tuple $(A',B',C, \tilde{G},G_0,G_1,\ldots,G_k)$, called {\em flatness tuple}, 
with the following properties.  
\begin{enumerate}[noitemsep,topsep=0pt]
\item $(A',B')$ is a separation of  $G$ with $V(C)=V(A')\cap V(B')\subseteq V(D)$ and $V(W)\subseteq V(B')$, where $D$ is the outer boundary of $W$;
\item $C$ is a cycle graph, the pegs of $W$ are in $V(C)$, and the order in which the vertices of $V(C)$ appear in $C$ is same as the order in which they appear on the outer boundary $D$ of $W$;
\item $B'\cup C=G_0 \cup G_1 \cup\ldots \cup G_k$, and the graphs $G_0,G_1,\ldots,G_k$ are pairwise edge disjoint; 
\item $G_0$ and $C$ are subgraphs of $\tilde{G}$, with $V (\tilde{G}) = V(G_0)$. Moreover, $\tilde{G}$ is a plane graph, and the cycle $C$ bounds its outer face;
\item \label{condition4intro} For $i\in \{1,\ldots,k\}$, $|V(G_i)\cap V(G_0)|\leq 3$ and $V(G_i)\cap V(G_0)$ forms a clique in $\tilde{G}$. 
Moreover, if $V(G_i)\cap V(G_0) = \{u,v,w\}$, then some finite face of $\tilde{G}$ is incident with $u,v,w$ and no other vertex;
\item\label{condition5intro} For all distinct $i,j\in [k]$, $V(G_i)\cap V(G_j)\subseteq V(G_0)$.
\end{enumerate}

We say that a flat wall $W$ is an {\em $\ell$-nice flat wall}, where $\ell\in {\mathbb N}$, if the tree-width of $G_i$, denoted by $\tw(G_i)$, is at most $\ell$, for all $i\in \{1,\ldots,k\}$.  In this case we say that $(A',B',C, \tilde{G},G_0,G_1,\ldots,G_k)$ is an {\em $\ell$-nice flatness tuple}.

The flat wall theorem of Robertson and Seymour states that there are functions $f$ and $p$ such that for any two positive integers $t$ and $w$, any graph $G$ containing a wall of size $f(t,w)$ (equivalently, there is a lower bound on the treewidth of $G$ depending on $f(t,w)$), must contain either $(a)$ a $K_t$-minor, or $(b)$ a subset $A\subset V(G)$ of at most $p(t)$ vertices  and a flat wall of size $w$ in $G\setminus A$.  Kawarabayashi, Thomas and Wollan~\cite{newFlatWall} proved that there is an $\OO(t^{24}m+n)$ time algorithm to output $(a)$ or $(b)$ with parameters  $f(t,w)=\Theta(t^{24}(t^2+w))$ and $p(t)=\OO(t^{24})$ (see \autoref{prop:newFlatWall}).  Using this result and the fact that there is a linear time algorithm to output a $K_t$-minor in a graph with at least $2^{t-3}n$ edges~\cite{ReedWood2009}, we prove the following lemma. 

\begin{lemma}\label{lem:flatWallCombinedintro}
There is a constant $c\in\mathbb{N}$ and an algorithm that given a graph $G$, and integers $g,w,t\geq 1$ such that $g\geq ct^{48}(t^2+w)$, runs in time  $2^{\OO({g^{58}})}n\log^2n$ and outputs one of the following.
\begin{itemize}[noitemsep,topsep=0pt]
\item[$(a)$] A tree decomposition of $G$ of width at most $\tw = g^{\OO(1)}$.
\item[$(b)$] A function $\phi$ witnessing that $K_t$ is a minor of $G$.
\item[$(c)$] A subset $A\subseteq V(G)$ of at most $ct^{24}$ vertices and a $\tw$-nice flat wall $W$  of size at least $(w\times w)$ in $G\setminus A$, along with a $\tw$-nice flatness tuple $(A',B',C,\tilde{G},G_0,G_1,\ldots,G_k)$. 
\end{itemize}
\end{lemma}

We would like to mention that a result similar to \autoref{lem:flatWallCombinedintro} is not new. For example, see result~(9.9) in~\cite{RobertsonS95b} and Theorem~1.8 in~\cite{newFlatWall}.  To prove \autoref{lem:flatWallCombinedintro}, we use ``binary search" type of arguments on the output of the flat wall theorem.

\subsection{Case 2: Graph with Large Clique Minor}
In this section we outline an algorithm to find a $(\delta,k)$-irrelevant vertex (for the definition see Section~\ref{sec:notations}) in a rooted graph $G$, when $G$ has a large clique minor.   Our irrelevant vertex rule could be viewed as  $(\delta,k)$-irrelevant vertex version of a result of  Robertson and Seymour (\autoref{prop:RS13:6.1}). 
Our algorithm  exploits the notion of {\em important separators}, defined by Marx~\cite{Marx06i}. 
It is a branching algorithm that goes over different guesses for how the deletion set $S$ could  look. For each guess and all sets $S$ that match the guess, the algorithm marks all the vertices in the clique minor model that are relevant ({\em that is, not irrelevant}) if we delete such a set $S$. Since we go over all possible guesses we would have marked all relevant vertices for all $S$. By our choice of $t$ (the size of the clique minor), there will be one model set which is not marked and we can prove that any vertex in that model set is $(\delta,k)$-irrelevant.

Let us briefly recall important separators. 
 Given a graph $G$ and two vertex sets $X$ and $Y$, an $X$-$Y$-{\em separator} is a vertex set $Z \subseteq V(G)$ such that there is no path from $X \setminus Z$ to $Y \setminus Z$ in $G \setminus Z$. An $X$-$Y$-separator $Z$ is called a {\em minimal} $X$-$Y$-{\em separator} if no proper subset of $Z$ is also an $X$-$Y$-separator. Given a vertex set $Z$, we define the {\em reach of $X$ in $G\setminus Z$} as the set $R_G(X,Z)$ of vertices reachable from $A$ by a path in $G \setminus Z$. We can now define a partial order on the set of minimal $A$-$B$ separators as follows. Given two minimal $X$-$Y$ separators $Z_1$ and $Z_2$, we say that $Z_1$ is ``at least as good as''  $Z_2$ if $|Z_1| \leq |Z_2|$ and $R_G(X, Z_2) \subset R_G(Y, Z_1)$. In simple words, $Z_1$ \enquote{costs less} than  $Z_2$ in terms of the number of vertices deleted and $Z_1$ \enquote{is pushed further towards $Y$} than $Z_2$ is. A minimal $X$-$Y$-separator $Z$ is an {\em important} $X$-$Y$-{\em separator} if no minimal $X$-$Y$-separator other than $Z$ is at least as good as $S$. 
 The important properties  of these separators are the following. 
 There is a unique  important $X$-$Y$-separator of minimum size and  for every $k$, the number of important $X$-$Y$-separators of size at most $k$ is at most $4^k$~\cite{ChenLL09}. 

Let $G$ be a graph and $Z\subseteq V(G)$.  We say that the $\delta$-\mfolio\ of $G$ relative to $Z$ is {\em generic} if  the $\delta$-\mfolio\ of the rooted graph $G$ with $R(G)=Z$ contains every graph with $\vert Z\vert$ roots and with detail at most $\delta$ (see Section~\ref{sec:notations} for the definition of detail). 
Let $G_1,\ldots, G_t$ be the connected subgraphs in the input graph $G$, where $V(G_1),\ldots, V(G_t)$ are the model sets of $K_t$ in $G$. The result of Robertson and Seymour (\autoref{prop:RS13:6.1}) informally states the following. For a large $t$ (depending on $\delta$ and $\vert R(G)\vert$), there exists a separation $(A,B)$ of minimum order such that  $R(G)\subseteq V(A)$, $V(G_i)\subseteq V(B)\setminus V(A)$ for some $i\in \{1,\ldots,t\}$, and any vertex in $V(B)\setminus V(A)$ is 
irrelevant for the $\delta$-\mfolio\ of $G$. Moreover $B\setminus v$ relative to $V(A)\cap V(B)$ is generic.  
It turns out that $R'=V(A)\cap V(B)$ is the unique \imsep{R(G)}{N(V(G_i))}.

In our branching algorithm we first enrich the roots of $G$ as follows. For each graph $X$ on $R(G)$ and each $H$ 
in the $\delta$-folio of $G\cup X$ we find a realization $G_H$ of $H$ as a topological minor in $G\cup X$ (of course if it exists).  
Then we add all the branch vertices (i.e., the vertices in $G_H$ to which $V(H)$ is mapped) and its neighbors in $G_H$ to $R(G)$.  Let $R=R(G)$ be the new set of roots. Notice that the cardinality of $R$ is bounded by a function of $\delta$ and $\vert R\vert$ alone.   
For each such graph $H$ one can get a graph $H'$ of detail at most $4\delta$ (see \autoref{fig:exampleHprime}) such that any minor model of the rooted graph $H'$ (where all vertices are roots) in $G$ can be turned into a topological minor $H$ of $G$. By   \autoref{prop:RS13:6.1} any vertex in $B\setminus v$ 
is $(\delta,0)$-irrelevant. Notice that most of the vetices in the minor model are irrelevant. 
Towards obtaining $(\delta,k)$-irrelevance, we branch based on the behavior of $S$ as follows. From above recall $R$ and the unique \imsep{R(G)}{N(V(G_i))}
$R'=V(A)\cap V(B)$. 

\noindent 
{\bf Case 1:} $S$ contains some vertex ``{\em between''} $R$ and $R'$. That is $S\cap V(A)\neq \emptyset$. 
In this case we call the algorithm recursively to mark 
$(\delta,k-1)$-relevant vertices in $B'=B\setminus (S\cap R')$. Towards this we guess the intersection of $S$ with $R'$ and for each guess we 
recursively mark $(\delta,k-1)$-relevant vertices in $B'=B\setminus (S\cap R')$. This implies that the branching factor of the algorithm 
is bounded by a function of $\vert R'\vert$ and $k$ alone. Since the extended $\delta$-folio of $G\setminus S$ is a function of the extended $\delta$-folios of 
$A\setminus S$ and $B\setminus S$, and $\vert V(B')\cap S\vert <k$, any $(\delta,k)$-irrelevant vertex is 
also ``$(\delta,k)$-irrelevant for sets that have same behavior as $S$''.

\smallskip
\noindent 
{\bf Case 2:} All of $S$ is on the ``{\em right hand side}'' of $R'$. That is $S\subseteq V(B)\setminus V(A)$.  
Notice that $(A,B\setminus S)$ is a separation of $G\setminus S$. If $(A,B\setminus S)$ satisfies the conditions of 
\autoref{prop:RS13:6.1}, then any vertex $v\in V(B)\setminus V(A)$ is irrelevant for the $4\delta$-\mfolio\ in $G\setminus S$, and hence ``$(\delta,k)$-irrelevant for sets that have same behavior as $S$''. 
If $(A,B\setminus S)$ does not satisfy the conditions of \autoref{prop:RS13:6.1}, then one of the following happens: 
either there is a \sep{R'}{N(V(G_i))} of size strictly less than $\vert R'\vert$ or there is a \sep{R'}{N(V(G_i))} of size $\vert R'\vert$ which is at least as good as $R'$. In either case there is a  minimal \sep{R'}{N(V(G_i))} of size at most $\vert R'\vert+k$ containing at least one vertex from $S$. So if we guess an important separator $R''$ that is better than this separator,  then we know that $S$ contains one vertex ``{\em to the left}'' of $R''$ and so we are back to Case $1$ (for $R''$ this time). Here of course it is crucial that the number of important separators is small and that each important separator contains vertices from a small number of model sets of $K_t$.

\subsection{Irrelevant Vertex in a Large Flat Wall}
\label{subsec:kirrelevant}

Our objective is to find a $(\delta,k)$-irrelevant vertex. Here, we briefly explain how to find a $(\delta,k)$-irrelevant vertex under the assumption that we already have an algorithm that finds a  {\em large} $(\delta,0)$-irrelevant subwall within a given (``slightly'' larger) wall, which is $(\delta,0)$-irrelevant irrespective of what is the graph $G$ outside the given wall. That is, 
suppose we have an algorithm ${\cal A}$ that given a {\em large} flat wall $W$ outputs a large subwall $W'$ 
such that the $\delta$-folios of $G$ and $G \setminus W'$ are the same, and the subwall $W'$ depends only on $W$ 
in the following sense. For any $Y\subseteq V(G)$ disjoint from $W$, the $\delta$-folio of $G\setminus Y$ and $(G\setminus Y) \setminus W'$ are the same. The overview of such an algorithm is given in \autoref{subsec:irrelevant}. 
Here, the size of $W'$ depends on the size of $W$. That is, there is a function $g(\delta,t)$  (where $t$ is the size of clique minor used in the previous subsection) such that if the size of $W$ is at least $g(\delta,t)w'\times g(\delta,t) w'$, then the size of $W'$ is at least $w'\times w'$.  Notice that this is stronger than finding just one $(\delta,0)$-irrelevant vertex in the flat wall case, because all vertices in $W'$ can be deleted {\em simultaneously} without changing the $\delta$-folio (critically, without consideration of the part of the graph outside the input wall).

Now we explain how to use algorithm ${\cal A}$ to find a $(\delta,k)$-irrelevant vertex. We set $w=g(\delta,t)^{k+2}$. We start with a wall $W_0=W$ of size $w\times w$.
In step 1, we apply algorithm ${\cal A}$ to find a large subwall $W_1$ of order $g(\delta,t)^{k+1}$. Re-apply  ${\cal A}$ on $W_1$ to obtain $W_2$ and so on, until we have a chain of subwalls $W_0, W_1, \ldots, W_{k+1}$. 
That is, in step $i$, we apply ${\cal A}$ on $W_{i-1}$ and obtain  a subwall $W_i$ of size $g(\delta,t)^{k+2-i}$.  
We claim that every vertex $v$ in $W_{k+1}$ is $(\delta,k)$-irrelevant.
To see this, consider an arbitrary set $S$ of size at most $k$. By pigeonhole principle, 
there exists some $i\in \{0,\ldots,k\}$ such that $(W_{i} \setminus W_{i+1} )\cap S = \emptyset$. Set $S_{out} = S \setminus W_{i}$, we have that $S \cap W_{i} \subseteq W_{i+1}$. Therefore, $W_{i}$ is a flat wall in $G \setminus S_{out}$ and hence all of $W_{i+1}$ is irrelevant in $G \setminus  S_{out}$. In other words, $G \setminus S_{out}$ and $G \setminus (S_{out} \cup W_{i+1})$ have the same $\delta$-folio. But the $\delta$-folio of $G\setminus S$ is contained in the $\delta$-folio of $G \setminus S_{out}$ while the $\delta$-folio of $G \setminus(S_{out} \cup W_{i+1})$ contains the $\delta$-folio of $G \setminus (S \cup \{v\})$. So $G \setminus S$ and $G \setminus  (S \cup \{v\})$ must have the same $\delta$-folio, proving that $v$ is $(\delta,k)$-irrelevant.

\subsection{Irrelevant Subwall in a Large Flat Wall}
\label{subsec:irrelevant}

Here, the objective is to find an irrelevant subwall with respect to the {\em extended} $\delta$-folio of a rooted graph $G$. One can easily show that any irrelevant subwall 
for the problem of  finding the (non-extended) ``$\delta+\vert R(G)\vert$-folio" of $G$ (call this problem  \FindFolio$^{\star}$) is also an irrelevant subwall for \FindFolio. Let $\delta^{\star}=\delta+\vert R(G)\vert$. 
In this subsection we explain how to find an irrelevant subwall for the problem of finding the $\delta^{\star}$-folio of  a given graph $G$. 
For \FindFolio$^{\star}$, we define a {\em solution} as a set $${\cal S}=\{(H,\phi,\varphi)\colon H \mbox{ is in the }\delta^{\star}\mbox{-folio and witnessed by }(\phi,\varphi)\}.$$ For exposition purposes, we will suppose (of course, only in the overview) that for each rooted graph $H$, there is at most one tuple in ${\cal S}$ that contains $H$. For a solution ${\cal S}$, we say that $\bigcup_{(H,\phi,\varphi)\in {\cal S}} \phi(V(H))$ is the set of terminals with respect to ${\cal S}$.

Assume that we are given a $\tw$-nice flat wall of size $w\times w$ in $G\setminus A$. That is, in this case we have a flatness tuple $(A',B',C,\tilde{G},G_0,G_1,\ldots,G_{k'})$ in 
$G\setminus A$.  Also for exposition purposes, we assume that $A=\emptyset$ and $k'=0$. The presence of apex vertices $A$ and ``protrusions" $G_1,\ldots,G_{k'}$ will add technical difficulty to the material, but the conceptual explanation will remain the same at its core. That is, we have a flat wall $W$ in $G$ with a flatness tuple $(A',B',C,\tilde{G},G_0)$. Without loss of generality, we  assume that $\tilde{G}=G_0\cup C$. That is, $G=A'\cup B'$, $(A',B')$ is a separation of  $G$ and $B'\cup C=G_0\cup C$ is a planar graph. Moreover, $V(W)\subseteq V(B')$ and the vertices from the ``non-planar portion" $A'$ of $G$ are not adjacent to $G_0$ except on $C$. 
Notice that we can choose the size $w$ of the wall $W$ in \autoref{lem:flatWallCombinedintro} and the value we choose will be a function of $\delta$ and $\vert R(G)\vert$.    

At this point, let us briefly revisit the \DisjointPathsLong\ problem, see how we can find an irrelevant vertex with respect to it, and then figure out the difficulties in the case of \FindFolio$^{\star}$. 

 \defparproblem{\DisjointPathsLong\ (\DisjointPaths)}{An undirected graph $G$, and a set $T=\{\{s,t\}: s,t\in V(G)\}$.}{$|T|=k^{\star}$}{Does $G$ contain a set of internally vertex disjoint paths that for every $\{s,t\}\in T$, contains one path whose endpoints are $s$ and $t$?}

Robertson and Seymour~\cite{RobertsonS12} proved that a vertex that is sufficiently “insulated” in the aforementioned  planar piece $G_0$ of $G$ is an  irrelevant vertex for \DisjointPathsLong. In other words, there is a function $h$ such that if there is a set  of $h(k^{\star})$ many ``concentric cycles"  $\{C_0,\ldots, C_{h(k^{\star})}\}$ in $G_0$  (see \autoref{def:concentricCycles}), where no terminal vertex in $T$ is in the inner face of any of these cycles, then all  vertices in the inner face of $C_0$ are (in fact, simultaneously) irrelevant (see \autoref{lem:RoSe}).  The difficulty in the case of \FindFolio$^{\star}$ is that {\em we do not know which are the terminal vertices!} (Where, roughly speaking, the terminal vertices are those to belong to the image of $\phi$, and the sought internally vertex disjoint paths reflect the computation of $
\varphi$.) However, because the terminal vertices should belong {\em somewhere} (for every $H$ in the folio), the existence of an irrelevant subwall is guaranteed for a large enough $w$ by the aforementioned result in~\cite{RobertsonS12}.  However, as we are given no information as to where the terminal vertices are, we face major difficulty in the computation of where this irrelevant subwall is.

The starting point of our method is finding a noose (i.e., a closed curve in the plane) enclosed {\em workspace} that contains a ``large'' {\em grid of inner nooses} (clarified below), where the graph  induced on the vertices inside the enclosing (outer) noose has bounded treewidth. Here, both the size of the grid and the treewidth are a function of $\delta$ and $\vert R(G)\vert$.   
For  a noose $N$, we use $\inNoose(N)$ to denote the set of vertices and edges fully contained in the interior (including boundary) of the noose. Informally, an $(a\times b)$-noose grid ${\cal N}$ is a grid of disjoint nooses $\{N_{i,j}\colon i\in [a], j\in [b]\}$ such that $\tilde{G}[\inNoose(N_{i,j})\cap V(G)]$ is connected and for each edge in the grid there is an edge between two vertices in the corresponding nooses.   
A $(p,2q)$-workspace is a pair $(M,{\cal N})$, where ${\cal N}$ is a $(2q\times 2q)$-noose grid ${\cal N}$ and $M$ is a noose with the following properties.
\begin{itemize}
\item $\tilde{G}[\inNoose(M)\cap V(G)]$ is connected and $M$ is a minimum noose that encloses $\tilde{G}[\inNoose(M)\cap V(G)]$. 
\item The treewidth of $\tilde{G}[\inNoose(M)\cap V(G)]$ is at most $p$. 
\item Each vertex  in $\tilde{G}[\inNoose(M)\cap V(G)]$ is exactly in the interior of one noose $N\in {\cal N}$. 
\end{itemize}
The nooses in ${\cal N}$ can be partitioned into $q$ ``frames'': the nooses in the outer boundary of the grid ${\cal N}$ yield $\fr[q-1]$, the nooses in the outer boundary after deletion of $\fr[q-1]$ yield $\fr[q-2]$, and so on (see \autoref{fig:frame}).  Now, we come to a critical definition: a solution ${\cal S}$ of \FindFolio$^{\star}$ is {\em $(\ell,\eta)$-untangled}, for some $0<\ell-3<\ell+3<q-1$ and $\eta\in {\mathbb N}$, if the following holds (with the most critical property being the second one). 
\begin{itemize}
\item[$(a)$] There is no terminal with respect to ${\cal S}$ in the nooses in $\bigcup_{\ell-3\leq j\leq \ell+3}\fr[j]$. 
\item[$(b)$] For any $(H,\phi,\varphi)\in {\cal S}$, the realization $G_H$ witnessed by $(\phi,\varphi)$ uses at most $\eta$ nooses in $\fr[\ell]$, and each used noose in $\fr[\ell]$ is exactly in ``one crossing subpath'' of a path in $\{\varphi(e)\colon e\in E(H)\}$. That is, the restriction of $G_H$ (i.e., partial solution) is small (depending on $\eta$ and $\delta$) and it ``behaves nicely'' on the frame $\ell$.  
\item[$(b)$] No vertices in the ``horizontal nooses'' of $\fr[\ell]$ (i.e., $\{N_{i,j}\colon i\in \{q-\ell,q+1+\ell\}, q-\ell\leq j\leq q+1+\ell\}$) are used by ${\cal S}$. 
\end{itemize}

\begin{figure}[t]
\begin{center}
\begin{tikzpicture}[scale=0.8]
\draw[red] (-5,-5)--(9.5,-5)--(9.5,9.5)--(-5,9.5)--(-5,-5); 
\draw[red] (-4,-4)--(8.5,-4)--(8.5,8.5)--(-4,8.5)--(-4,-4); 
\draw[red] (-3,-3)--(7.5,-3)--(7.5,7.5)--(-3,7.5)--(-3,-3); 
\draw[red] (-2,-2)--(6.5,-2)--(6.5,6.5)--(-2,6.5)--(-2,-2); 
\draw[red] (-1,-1)--(5.5,-1)--(5.5,5.5)--(-1,5.5)--(-1,-1); 

\node[] (a1) at (2,10) {$\fr[q-1]$}; 

\node[] (a1) at (2,4.65) {$\fr[s_1]$}; 


\node[] (a1) at (2,7.65) {$\fr[s_{\mu}]$};


\node[] (a1) at (-5,9.5) {$\bullet$}; 
\node[] (a1) at (-4.5,9.5) {$\bullet$}; 
\node[] (a1) at (-4,9.5) {$\bullet$}; 
\node[] (a1) at (-3.5,9.5) {$\bullet$}; 
\node[] (a1) at (-3,9.5) {$\bullet$}; 
\node[] (a1) at (-2.5,9.5) {$\bullet$}; 
\node[] (a1) at (-2,9.5) {$\bullet$}; 
\node[] (a1) at (-1.5,9.5) {$\bullet$}; 
\node[] (a1) at (-1,9.5) {$\bullet$}; 
\node[] (a1) at (-0.5,9.5) {$\bullet$}; 
\node[] (a1) at (0,9.5) {$\bullet$}; 
\node[] (a1) at (0.5,9.5) {$\bullet$}; 
\node[] (a1) at (1,9.5) {$\bullet$}; 
\node[] (a1) at (1.5,9.5) {$\bullet$}; 
\node[] (a1) at (2,9.5) {$\bullet$}; 
\node[] (a1) at (2.5,9.5) {$\bullet$}; 
\node[] (a1) at (3,9.5) {$\bullet$};
\node[] (a1) at (3.5,9.5) {$\bullet$};  
\node[] (a1) at (4,9.5) {$\bullet$};
\node[] (a1) at (4.5,9.5) {$\bullet$}; 
\node[] (a1) at (5,9.5) {$\bullet$};
\node[] (a1) at (5.5,9.5) {$\bullet$};  
\node[] (a1) at (6,9.5) {$\bullet$};
\node[] (a1) at (6.5,9.5) {$\bullet$};  
\node[] (a1) at (7,9.5) {$\bullet$};
\node[] (a1) at (7.5,9.5) {$\bullet$};  
\node[] (a1) at (8,9.5) {$\bullet$};
\node[] (a1) at (8.5,9.5) {$\bullet$};  
\node[] (a1) at (9,9.5) {$\bullet$};
\node[] (a1) at (9.5,9.5) {$\bullet$};


\node[] (a1) at (-5,-5) {$\bullet$};
\node[] (a1) at (-5,-4.5) {$\bullet$};
\node[] (a1) at (-5,-4) {$\bullet$};
\node[] (a1) at (-5,-3.5) {$\bullet$};
\node[] (a1) at (-5,-3) {$\bullet$};
\node[] (a1) at (-5,-2.5) {$\bullet$};
\node[] (a1) at (-5,-2) {$\bullet$};
\node[] (a1) at (-5,-1.5) {$\bullet$};
\node[] (a1) at (-5,-1) {$\bullet$};
\node[] (a1) at (-5,-0.5) {$\bullet$};
\node[] (a1) at (-5,0) {$\bullet$};
\node[] (a1) at (-5,0.5) {$\bullet$};
\node[] (a1) at (-5,1) {$\bullet$};
\node[] (a1) at (-5,1.5) {$\bullet$};
\node[] (a1) at (-5,2) {$\bullet$};
\node[] (a1) at (-5,2.5) {$\bullet$};
\node[] (a1) at (-5,3) {$\bullet$};
\node[] (a1) at (-5,3.5) {$\bullet$};
\node[] (a1) at (-5,4) {$\bullet$};
\node[] (a1) at (-5,4.5) {$\bullet$};
\node[] (a1) at (-5,5) {$\bullet$};
\node[] (a1) at (-5,5.5) {$\bullet$};
\node[] (a1) at (-5,6) {$\bullet$};
\node[] (a1) at (-5,6.5) {$\bullet$};
\node[] (a1) at (-5,7) {$\bullet$};
\node[] (a1) at (-5,7.5) {$\bullet$};
\node[] (a1) at (-5,8) {$\bullet$};
\node[] (a1) at (-5,8.5) {$\bullet$};
\node[] (a1) at (-5,9) {$\bullet$};


\node[] (a1) at (9.5,-5) {$\bullet$};
\node[] (a1) at (9.5,-4.5) {$\bullet$};
\node[] (a1) at (9.5,-4) {$\bullet$};
\node[] (a1) at (9.5,-3.5) {$\bullet$};
\node[] (a1) at (9.5,-3) {$\bullet$};
\node[] (a1) at (9.5,-2.5) {$\bullet$};
\node[] (a1) at (9.5,-2) {$\bullet$};
\node[] (a1) at (9.5,-1.5) {$\bullet$};
\node[] (a1) at (9.5,-1) {$\bullet$};
\node[] (a1) at (9.5,-0.5) {$\bullet$};
\node[] (a1) at (9.5,0) {$\bullet$};
\node[] (a1) at (9.5,0.5) {$\bullet$};
\node[] (a1) at (9.5,1) {$\bullet$};
\node[] (a1) at (9.5,1.5) {$\bullet$};
\node[] (a1) at (9.5,2) {$\bullet$};
\node[] (a1) at (9.5,2.5) {$\bullet$};
\node[] (a1) at (9.5,3) {$\bullet$};
\node[] (a1) at (9.5,3.5) {$\bullet$};
\node[] (a1) at (9.5,4) {$\bullet$};
\node[] (a1) at (9.5,4.5) {$\bullet$};
\node[] (a1) at (9.5,5) {$\bullet$};
\node[] (a1) at (9.5,5.5) {$\bullet$};
\node[] (a1) at (9.5,6) {$\bullet$};
\node[] (a1) at (9.5,6.5) {$\bullet$};
\node[] (a1) at (9.5,7) {$\bullet$};
\node[] (a1) at (9.5,7.5) {$\bullet$};
\node[] (a1) at (9.5,8) {$\bullet$};
\node[] (a1) at (9.5,8.5) {$\bullet$};
\node[] (a1) at (9.5,9) {$\bullet$};


\node[] (a1) at (-4.5,-5) {$\bullet$}; 
\node[] (a1) at (-4,-5) {$\bullet$}; 
\node[] (a1) at (-3.5,-5) {$\bullet$}; 
\node[] (a1) at (-3,-5) {$\bullet$}; 
\node[] (a1) at (-2.5,-5) {$\bullet$}; 
\node[] (a1) at (-2,-5) {$\bullet$}; 
\node[] (a1) at (-1.5,-5) {$\bullet$}; 
\node[] (a1) at (-1,-5) {$\bullet$}; 
\node[] (a1) at (-0.5,-5) {$\bullet$}; 
\node[] (a1) at (0,-5) {$\bullet$}; 
\node[] (a1) at (0.5,-5) {$\bullet$}; 
\node[] (a1) at (1,-5) {$\bullet$}; 
\node[] (a1) at (1.5,-5) {$\bullet$}; 
\node[] (a1) at (2,-5) {$\bullet$}; 
\node[] (a1) at (2.5,-5) {$\bullet$}; 
\node[] (a1) at (3,-5) {$\bullet$};
\node[] (a1) at (3.5,-5) {$\bullet$};  
\node[] (a1) at (4,-5) {$\bullet$};
\node[] (a1) at (4.5,-5) {$\bullet$}; 
\node[] (a1) at (5,-5) {$\bullet$};
\node[] (a1) at (5.5,-5) {$\bullet$};  
\node[] (a1) at (6,-5) {$\bullet$};
\node[] (a1) at (6.5,-5) {$\bullet$};  
\node[] (a1) at (7,-5) {$\bullet$};
\node[] (a1) at (7.5,-5) {$\bullet$};  
\node[] (a1) at (8,-5) {$\bullet$};
\node[] (a1) at (8.5,-5) {$\bullet$};  
\node[] (a1) at (9,-5) {$\bullet$};

\draw[red] (4.5,0)--(4.5,4.5);
\draw[red] (0,4.5)--(0,0);
\draw[red] (4.5,0)--(0,0);



\node[] (a1) at (1,1) {$\bullet$}; 
\node[] (a1) at (1.5,1) {$\bullet$}; 
\node[] (a1) at (2,1) {$\bullet$}; 
\node[] (a1) at (2.5,1) {$\bullet$}; 
\node[] (a1) at (3,1) {$\bullet$};
\node[] (a1) at (3.5,1) {$\bullet$};  

\node[] (a1) at (1,1.5) {$\bullet$}; 
\node[] (a1) at (1.5,1.5) {$\bullet$}; 
\node[] (a1) at (2,1.5) {$\bullet$}; 
\node[] (a1) at (2.5,1.5) {$\bullet$}; 
\node[] (a1) at (3,1.5) {$\bullet$};
\node[] (a1) at (3.5,1.5) {$\bullet$};  

\node[] (a1) at (1,2) {$\bullet$}; 
\node[] (a1) at (1.5,2) {$\bullet$}; 
\node[] (a1) at (2,2) {$\bullet$}; 
\node[] (a1) at (2.5,2) {$\bullet$}; 
\node[] (a1) at (3,2) {$\bullet$};
\node[] (a1) at (3.5,2) {$\bullet$};  

\node[] (a1) at (1,2.5) {$\bullet$}; 
\node[] (a1) at (1.5,2.5) {$\bullet$}; 
\node[] (a1) at (2,2.5) {$\bullet$}; 
\node[] (a1) at (2.5,2.5) {$\bullet$}; 
\node[] (a1) at (3,2.5) {$\bullet$};
\node[] (a1) at (3.5,2.5) {$\bullet$};  

\node[] (a1) at (1,3) {$\bullet$}; 
\node[] (a1) at (1.5,3) {$\bullet$}; 
\node[] (a1) at (2,3) {$\bullet$}; 
\node[] (a1) at (2.5,3) {$\bullet$}; 
\node[] (a1) at (3,3) {$\bullet$};
\node[] (a1) at (3.5,3) {$\bullet$};  

\node[] (a1) at (1,3.5) {$\bullet$}; 
\node[] (a1) at (1.5,3.5) {$\bullet$}; 
\node[] (a1) at (2,3.5) {$\bullet$}; 
\node[] (a1) at (2.5,3.5) {$\bullet$}; 
\node[] (a1) at (3,3.5) {$\bullet$};
\node[] (a1) at (3.5,3.5) {$\bullet$};  


\draw[red] (0,4.5)--(4.5,4.5);

%

\end{tikzpicture}
\end{center}
\caption{Example of frames. The bullets represent nooses in the noose grid ${\cal N}$.}
\label{fig:frame}
\end{figure}
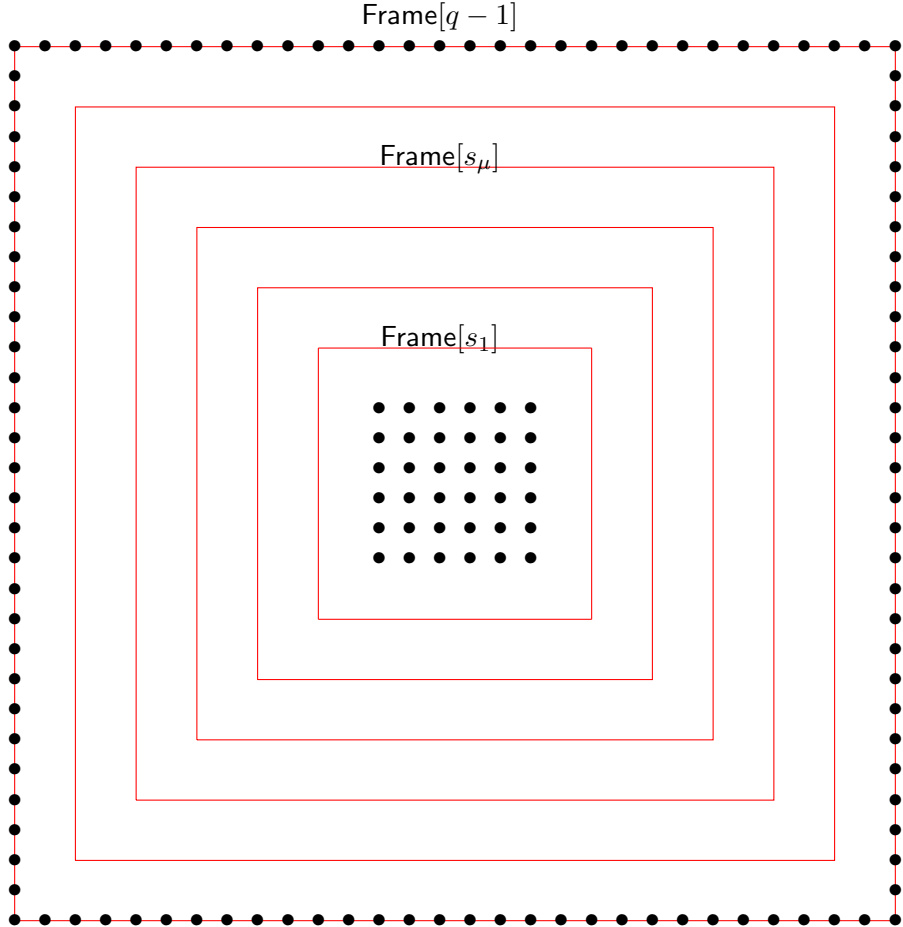

The proof that there exists an $(\ell,\eta)$-untangled solution (for certain $\ell$ and $
\eta$ where $\eta$ is ``small'') is very technical, and requires several steps (where in each step, we assert the existence of a ``nicer'' solution than before, until we eventually arrive at an $(\ell,\eta)$-untangled solution).

Afterwards, we prove that there exist $\eta,\lambda,\mu$ (depending on $\delta$) and $s\in [q/\lambda -\mu-1]$ with the following property. Let $s_i=s(i\lambda -1)$ for all $1\leq i \leq \mu$. Then, the set of ``partial solutions'' satisfying conditions $(b)$ and $(c)$ within $\tilde{G}[\bigcup_{j\leq s_1, N\in \fr[j]}\inNoose(N)\cap V(G)]$ and the set of partial solutions satisfying these conditions within $\ldots, \tilde{G}[\bigcup_{j\leq s_{\mu}, N\in \fr[j]}\inNoose(N)\cap V(G)]$ are ``identical''.\footnote{There are properties beyond isomorphism needed to be satisfied by partial solutions on different frames so that we will consider them as identical. Those technical details are omitted here.} We choose $\mu$ in such a way that there {\em exists} an integer $\ell\in \{s_1,\ldots,s_{\mu}\}$ such that there is a solution ${\cal S}$ that is $(\ell,\eta)$-untangled. Since the treewidth of $\tilde{G}[\inNoose(M)\cap V(G)]$ is at most $p$ (which is upper bounded by a function of $\delta$), using standard dynamic programming on tree decomposition we compute the set of partial solutions corresponding to each of the frames indexed $s_1,\ldots,s_{\mu}$. Even though we do {\em not} know $\ell$, we prove that there is a ``patch'' from the partial solution of the frame indexed $s_1$ to get a solution without using any vertex in any of the ``upper'' (top horizontal) nooses in $\fr[s_1],\fr[s_1+1],\ldots, \fr[s_1+w']$. Thus,  the subwall within the graph formed by these nooses of $\fr[s_1],\fr[s_1+1],\ldots, \fr[s_1+w']$ is irrelevant.

\subsection{Planar  and Bounded Genus Graphs}
\label{subsec:representative}

Lastly, we explain the reason for faster running time with respect to  planar graphs, or, more generally, graphs of bounded genus (\autoref{thm:mainIntroGenus} and \autoref{thm:main}). Recall that the starting point of our algorithm is  \autoref{lem:flatWallCombinedintro} with respect to general graphs. 
However, with respect to a  planar graph, or, more generally, a graph of bounded  genus graph $G$, using the linear dependence between treewidth and the size of the largest grid minor in  $G$, in time $\OO(w^2 n)$ we get either  a $2w\times 2w$ grid minor in $G$ (equivalently, a $w\times w$-flat wall without apex vertices and protrusions) or a tree decomposition of $G$ of width $\OO(w)$. 
(Clearly, we do get a large clique minor because graphs of bounded genus do not have them.) 
 In the case of bounded treewidth, we use standard technique to solve the problem. In the case of a flat wall, we find an irrelevant vertex as explained in Section~\ref{subsec:irrelevant}. However, here instead of 
\autoref{lem:RoSe}, 
we use its planar version proved by Adler et al.~\cite{DBLP:journals/jct/AdlerKKLST17} (see \autoref{prop:disPathIrrelevant}) or the bounded genus version proved by Mazoit~\cite{mazoit2013single}, where $h(k^{\star})\leq 2^{ck^{\star}}$ for some constant $c$ depending on $g$.  
These simplifications and tools yield the running times 
$2^{2^{k \cdot 2^{\OO((h^\star)^2)}}}n^2$ and  
$2^{2^{2^{\OO((h^\star)^2)}}}n^2$ in \autoref{thm:mainIntroGenus} and \autoref{thm:main}, respectively. 

\section{Preliminaries}\label{sec:prelims}

We use $[n]$ and $[n]_0$ as shorthands for $\{1,2,\ldots,n\}$ and $\{0,1,\ldots,n\}$, respectively. The domain and image of a function $f: A\rightarrow B$ are denoted by $\domain(f)$ and $\image(f)$, respectively. For a subset $A'\subseteq A$, $f(A')$ denotes the set of elements $b\in B$ for which there exists $a\in A'$ such that $b=f(a)$. Given two tuples of $m$ integers, $t=(n_1,n_2,\ldots,n_m)$ and $t'=(n_1',n_2',\ldots,n_m')$, we define the relation $<$ lexicographically, that is, we have $t<t'$ if and only if there exists $i\in[m]$ such that $n_i<n_i'$ and for all $j\in[i]$, $n_j\leq n_j'$.

\subsection{Minors, Plane Graphs and Treewidth}

Given a graph $G$, we denote its vertex set and its edge set by $V(G)$ and $E(G)$, respectively. When $G$ is clear from context, denote $n=|V(G)|$ and $m=|E(G)|$. In this paper we consider only graphs without multiple edges, self loops and labels. 
The set of isolated vertices of $G$ is denoted by $\isolated(G)$. Given a subset $U\subseteq V(G)$, $G[U]$ denotes the subgraph of $G$ induced by $U$, and $G\setminus U$ denotes the subgraph $G[V(G)\setminus U]$. Given a vertex $v\in V(G)$, $G\setminus v$ denotes the subgraph $G\setminus\{v\}$. Given a vertex $v\in V(G)$, $d_G(v)$ denotes the degree of $v$ in $G$, and $N_G(v)$ denotes the set of neighbors of $v$ in $G$. For a subset $U\subseteq V(G)$, $N_G(U)=(\bigcup_{v\in U}N_G(v))\setminus U$. Given (not necessarily vertex-disjoint) graphs $G_1$ and $G_2$, denote the graph $G'=(V(G_1)\cup V(G_2), E(G_1)\cup E(G_2))$ by $G_1\cup G_2$.

Given paths $P=v_1-v_2-\cdots-v_t$ and $P'=u_1-u_2-\cdots-u_\ell$ where $v_t=u_1$, $PP'$ denotes the path $v_1-v_2-\cdots-(v_t=u_1)-u_2-\cdots-u_\ell$. An {\em $a\times b$-grid} is a graph $G$ whose vertex set can be denoted as $\{v_{i,j}: i\in[a],j\in[b]\}$, so that the edge set of $G$ is exactly $\{\{v_{i,j},v_{i',j'}\}: i,i'\in[a], j,j'\in[b],|i-i'|+|j-j'|=1\}$. The {\em height} of the grid is $a$ and the {\em width} of the grid is $b$. For $i\in[a]$ and $j\in[b]$, the {\em row} $R_i$ of the grid is the path $v_{i,1}-v_{i,2}-\ldots-v_{i,b}$, and the {\em column} $C_j$ of the grid is $v_{1,j}-v_{2,j}-\ldots-v_{a,j}$. 
Given a vertex set $V$, $\allgraphs(V)$ denotes the set of all graphs on $V$. Given $t\in\mathbb{N}$, $K_t$ is the clique on $t$ vertices.

\begin{observation}\label{obs:allGraphsNumber}
For a vertex set $V$, the number of graphs in $\allgraphs(V)$ is $2^{{|V|\choose 2}}$.
\end{observation}

We proceed to present the notion of a separation of a graph.

\begin{definition}[{\bf Separation}]
A separation in graph $G$ is a pair $(G_1, G_2)$ of subgraphs of $G$, such that 
$G = G_1 \cup G_2$ and $E(G_1) \cap E(G_2) = \emptyset$. The order of the separation is 
$|V(G_1) \cap V (G_2)|$.
\end{definition}

\begin{definition}[{\bf $(s,t)$-Separator}]
For a graph $G$ and $s,t\in V(G)$, a subset $S\subseteq V(G)\setminus \{s,t\}$ is called an $(s,t)$-separator, if $s$ and $t$ are in two different connected components of $G\setminus S$.  
\end{definition}

\paragraph{Minor, Topological Minor.} Let $G$ and $H$ be two undirected graphs. We say that $H$ is a {\em minor} of $G$ if there exists a function $\phi: V(H)\rightarrow 2^{V(G)}$ such that for all $h\in V(H)$, $G[\phi(h)]$ is a connected graph, for all distinct $h,h'\in V(H)$, $\phi(h)\cap\phi(h')=\emptyset$, and for all $\{h,h'\}\in E(H)$, there exist $u\in \phi(h)$ and $v\in \phi(h')$ such that $\{u,v\}\in E(G)$. The sets $\phi(h_1),\ldots,\phi(h_{\ell})$ where 
$V(H)=\{h_1,\ldots,h_{\ell}\}$, are called the model sets. 
Let $\paths(G)$ be the set of all (simple) paths in $G$.
We say that $H$ is a {\em topological minor} of $G$ if there exist injective functions $\phi: V(H)\rightarrow V(G)$ and $\varphi: E(H)\rightarrow\paths(G)$ such that for all $e=\{h,h'\}\in E(H)$, the endpoints of $\varphi(e)$ are $\phi(h)$ and $\phi(h')$, for all distinct $e,e'\in E(H)$, the paths $\varphi(e)$ and $\varphi(e')$ are internally vertex-disjoint, and there do not exist a vertex $v$ in the image of $\phi$ and an edge $e\in E(H)$ such that $v$ is an internal vertex on $\varphi(e)$. The vertices and edges in $\phi(V(H))$ and $\varphi(E(H))$ form a subgraph $G'$ in $G$. That is $V(G')$ contains the vertices in $\phi(V(H))$ and the vertices in the paths in $\varphi(E(H))$.  The edge set of $G'$ is the edges in that paths in $\varphi(E(H))$. Then we call $G'$ is a {\em realization} of $H$ in $G$.  

%

\paragraph{Plane Graphs.} A graph $G$ is {\em planar} if there exists a mapping from every vertex in $V(G)$ to a point on the plane, and from every edge $e\in E(G)$ to a curve on the plane where the extreme points of the curve are the points mapped to the endpoints of $e$, and all curves are disjoint except on their extreme points. Such a mapping is called an {\em embedding in the plane}, or simply an {\em embedding}.
A {\em plane graph} is a planar graph having a fixed embedding. For a planar graph $G$, we can define its faces as follows: delete all the edges and vertices of $G$ from the plane. Then, the remaining part of the plane is a collection of disjoint areas. Each such area is called a {\em face}. The face whose area is unbounded is called the {\em outer-face}, and every other face is an {\em interior face}. When we say that a point is contained in a face, we mean that the point may also lie on the boundary of the face. 
Deletion of a cycle $C$ in a plane graph from the plane results in two disjoint areas. The area which contains the exterior-face of $C$, while the other area is called the interior-face of $C$. The union of $C$ and  interior(exterior)-face of $C$ is called the inner(outer)-face of $C$. A {\em noose} is a closed curve in the plane or surface. We say that a noose $N$ encloses a graph $G$ if all of the vertices as well as edges of $G$ lies in the interior (including the boundary) of $G$. The {\em minimum noose} that encloses $G$ is the noose that encloses minimum area among all nooses that enclose $G$.

\begin{observation}\label{obs:nooseOuterface}
Let $G$ be a plane graph. The minimum noose that encloses $G$ is the noose that encloses exactly those points in the plane that do not belong to the outer-face of $G$. 
\end{observation}

Given a plane graph $G$ and a noose $N$, $\inNoose_G(N)$ denotes the set of vertices and edges of $G$ that lie entirely in the interior (including the boundary) of $N$, and
$\outNoose_G(N)$ denotes the set of vertices and edges of $G$ that lie in the exterior (including the boundary) of $N$. 
In addition, $\inNoose^\star_G(N)$ denotes the set of vertices of $G$ that lie in the strict interior (excluding the boundary) of $N$ and the edges of $G$ that (excluding their endpoints) lie in the strict interior of $N$. 
Note that if a noose encloses a vertex set $U$, it does not imply that the noose encloses the graph induced by $U$ as some edges of that graph may not be enclosed by the noose. In other words, given a plane graph $G$ and a noose $N$, it might hold that $G[\inNoose(N)\cap V(G)]$ is not enclosed by $N$. We say that a plane graph $G$ is {\em nicely drawn} if every connected component of $G$ can be enclosed by a noose that does not contain vertices of other connected components. In other words, for every two connected components $C$ and $C'$ of $G$, it holds that $C$ lies in the outer-face of $C'$ and vice versa (see \autoref{fig:nicelydrawn} for an illustration). 

\begin{observation}
\label{obs:connnoose}
Let $G$ be a nicely drawn plane graph and $U\subseteq V(G)$ such that $G'=G[U]$ is a connected subgraph of $G$. Then, there is a noose 
$N$ such that $V(G')\cup E(G')=\inNoose^\star_G(N)$. 
\end{observation}

\begin{center}

\begin{figure}
\begin{subfigure}{.47\linewidth}
\begin{center}
\begin{tikzpicture}[scale=1]
\node[] (a1) at (0,0) {$\bullet$}; 
\node[] (a2) at (0,3) {$\bullet$}; 
\node[] (a3) at (3,0) {$\bullet$}; 
\node[] (a4) at (3,3) {$\bullet$};
\draw (0,0)--(0,3)--(3,3)--(3,0)--(0,0)--(3,3);
\node[] (a4) at (1,0.3) {$\bullet$};
\node[] (a4) at (2.5,1) {$\bullet$};

\node[] (a4) at (0.5,2) {$\bullet$};
\node[] (a4) at (2,2.5) {$\bullet$};

\draw (1,0.3)--(2.5,1);
\draw (0.5,2)--(2,2.5);
\end{tikzpicture}
\end{center}
\caption{This drawing is not nice}
\end{subfigure}
\hspace{0.2cm}
\begin{subfigure}{.47 \linewidth}
\begin{center}
\begin{tikzpicture}[scale=1]
\node[] (a1) at (0,0) {$\bullet$}; 
\node[] (a2) at (0,3) {$\bullet$}; 
\node[] (a3) at (3,0) {$\bullet$}; 
\node[] (a4) at (3,3) {$\bullet$};
\draw[blue] (0,0)--(0,3)--(3,3)--(3,0)--(0,0);
\draw (0,0)--(3,3);
\node[] (a4) at (3.5,1) {$\bullet$};
\node[] (a4) at (4,1) {$\bullet$};

\node[] (a4) at (3.5,2) {$\bullet$};
\node[] (a4) at (4,2) {$\bullet$};

\draw (3.5,1)--(3.5,2);
\draw (4,1)--(4,2);
\end{tikzpicture}
\end{center}
\caption{Nicely drawn embedding, where the boundary of outer face is blue edges}
\end{subfigure}
%
%
\caption{Graph $G$ with two embeddings}
\label{fig:nicelydrawn}
\end{figure}
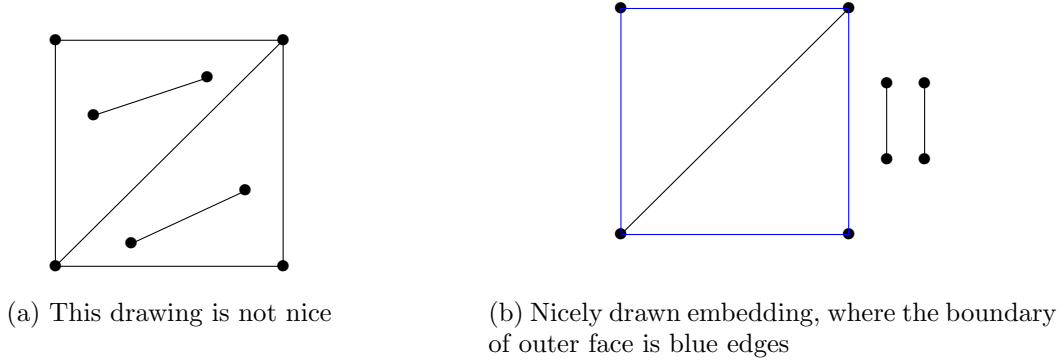
\end{center}

\paragraph{Treewidth.} Treewidth is a structural parameter indicating how much a graph resembles a tree. Formally, treewidth is defined as follows.

\begin{definition}[{\bf Tree Decomposition}]\label{def:treewidth}
A \emph{tree decomposition} of a graph $G$ is a pair $(T,\beta)$ of a rooted tree $T$
and $\beta:V(T) \rightarrow 2^{V(G)}$, such that
\vspace{-0.5em}
\begin{itemize}
\itemsep0em 
\item\label{item:twedge} for any edge $\{x,y\} \in E(G)$ there exists a node $v \in V(T)$ such that $x,y \in \beta(v)$, and
\item\label{item:twconnected} for any vertex $x \in V(G)$, the subgraph of $T$ induced by the set $T_x = \{v\in V(T): x\in\beta(v)\}$ is a non-empty tree.
\end{itemize}
The {\em width} of $(T,\beta)$ is $\max_{v\in V(T)}\{|\beta(v)|\}-1$. The {\em treewidth} of $G$, denoted by $\tw(G)$, is the minimum width over all tree decompositions of $G$.
\end{definition}

We also define a form of a tree decomposition that simplifies the design of DP algorithms.

\begin{definition}[{\bf Nice Tree Decomposition}]
A tree decomposition $(T,\beta)$ of a graph $G$ is {\em nice} if for the root $r$ of $T$, $\beta(r)=\emptyset$, and each node $v\in V(T)$ is of one of the following types.
\vspace{-0.5em}
\begin{itemize}
\itemsep0em 
\item {\bf Leaf}: $v$ is a leaf in $T$ and $\beta(v)=\emptyset$.
\item {\bf Forget}: $v$ has one child, $u$, and there is a vertex $x\in\beta(u)$ such that $\beta(v)=\beta(u)\setminus\{x\}$.
\item {\bf Introduce}: $v$ has one child, $u$, and there is a vertex $x\in\beta(v)$ such that $\beta(v)\setminus \{x\}=\beta(u)$.
\item {\bf Join}: $v$ has two children, $u$ and $w$, and $\beta(v)=\beta(u)=\beta(w)$.
\end{itemize}
\end{definition}

For $v \in V(T)$, we say that $\beta(v)$ is the \emph{bag} of $v$, and $\gamma(v)$ denotes the union of the bags of $v$ and the descendants of $v$ in $T$. Given a tree decomposition $(T,\beta)$, Bodlaender \cite{DBLP:journals/siamcomp/Bodlaender96} showed how to construct a {\em nice} tree decomposition of the same width as $(T,\beta)$ in linear time. As formally stated by the following proposition, a nice tree decomposition of small width can be computed by a parameterized algorithm whose dependency on $n$ is linear.

\begin{proposition}[\cite{DBLP:journals/siamcomp/Bodlaender96}]\label{prop:twExact}
Given a graph $G$ and $t\in\mathbb{N}$, in time $\OO(2^{\OO(t^3\log t)}n)$ we can decide whether $\tw(G)\leq t$, and if the answer is positive, compute a nice tree decomposition of $G$ of width at most $t$.
\end{proposition}

With respect to treewidth and \TMC, the following proposition is known.

\begin{proposition}[\cite{scheffler1994practical}]\label{prop:twDPNoDec}
\TMC\ is solvable in time $\OO(t^{\OO(t)} n)$ if every input instance $(G,H,k)$ is given along with a tree decomposition of $G$ of width $t$.
\end{proposition}

As a corollary to Propositions \ref{prop:twExact} and \ref{prop:twDPNoDec}, we have the following proposition.

\begin{proposition}\label{prop:twDP}
\TMC\ is solvable in time $\OO(2^{\OO(\tw^3\log \tw)} n)$, where $\tw=\tw(G)$.
\end{proposition}

\subsection{Flat Wall}

Here, we recall the definition of a flat wall and a result to compute 
a large flat wall in a graph $G$ if the treewidth of $G$ is large and it has no large clique minor~\cite{newFlatWall,DBLP:conf/soda/Chuzhoy15}.  We first present the definition of an elementary wall.

\begin{definition}[{\bf Elementary wall}]\label{def:elemWall}
Let $J$ be an $h\times(2r)$ grid. For 
any column $C_j$ of $J$,  where $j\in[2r]$, 
let $e^j_1 , e^j_2,\ldots, e^j_{h-1}$ be the edges of $C_j$, in the order of their
appearance on $C_ j$, where $e^j_1$ is incident on $v_{1, j}$.  
For any column $C_j$, if $j$ is odd, then delete from the graph $J$ all edges
$e^j_i$ where $i$ is even. For any column $C_j$, if $j$ is even, then delete from the graph all edges $e^j_i$ where $i$ is odd. After 
this, delete all vertices of degree $1$. The resulting graph $\widehat{W}$ is the {\em elementary wall} of height $h$ and width $r$. 
(See \autoref{fig:elementarywall} for an example).
\end{definition}

The definition of an elementary wall directly implies the following observation.

\begin{observation}\label{obs:gridToWall}
Every $h\times(2r)$ grid contains an $h\times r$ elementary wall as a subgraph.
\end{observation}

Let $\widehat{W}$ be an elementary wall. Let $E_1$ be the set of edges of $\widehat{W}$ that correspond to the horizontal edges of the original grid  and let $E_2$ be the set of the edges of $\widehat{W}$ that correspond to the vertical edges of the original grid. 
That is, $E_1\subseteq \{v_{i,j}v_{i,j+1}~:~i\in [h],j\in [2r-1]\}$ and $E_2\subseteq \{v_{i,j}v_{i+1,j}~:~i\in [h-1],j\in [2r]\}$. The subgraph of $\widehat{W}$ consisting of $E_1$ forms a collection of $h$  vertex disjoint paths, 
denoted  by $R_1, \ldots, R_h$, where for $i\in [h]$, $R_i$ is incident on $v_{i,1}$.  
Let $V_1$ and $V_h$ denote the set of all vertices of $R_1$ and $R_h$, respectively. There is a unique set ${\cal C}$ of $r$ vertex disjoint paths, where each path $C \in {\cal C}$ starts at a vertex of $V_1$, ends at a vertex of 
$V_h$, and is internally vertex disjoint from $V_1 \cup V_h$. We let $C_1$ and $C_r$ denote the leftmost and rightmost paths in ${\cal C}$. The subgraph $Z=R_1 \cup C_1 \cup R_h \cup C_r$ of $\widehat{W}$ is a simple cycle and it is called the {\em outer boundary} of $\widehat{W}$.
All the vertices of $Z$ that have degree $2$ in $\widehat{W}$ are called {\em pegs} of $\widehat{W}$. 

A graph $W$ is a {\em wall} of height $h$ and width $r$, or an {\em $h\times r$-wall}, if and only if it is a subdivision of the elementary wall $\widehat{W}$ of height $h$ and width $r$.  Notice that, in this case, $\widehat{W}$ is a 
topological minor of $W$.  Let $(\phi,\psi)$ be a pair of functions that witnesses that $\widehat{W}$ is a topological 
minor. The {\em pegs} of $W$ are the vertices to which the pegs of $\widehat{W}$ are mapped by the map $\phi$. The following observation is immediate.

\begin{observation}
\label{obs:wlalltogrid}
An $(h,r)$-wall contains the $h\times r$-grid as a minor.  
\end{observation}

A wall minor in a graph can be turned into a wall subgraph in the graph.     

\begin{observation}
\label{obs:wallminortosub}
If a graph $G$  contains a $(w\times w)$-wall as a minor, then there is a $(w\times w)$-wall as subgraph in $G$.
Moreover, there is a linear time algorithm which given a graph $G$ and a minor model of a $(w\times w)$-wall in $G$, outputs a $(w\times w)$-wall in $G$. 
\end{observation}


In the following definition, we consider the notion of a subwall.

\begin{definition}[\cite{DBLP:conf/soda/Chuzhoy15}]
Let $W$ and $W'$ be two walls, where $W'$ is a subgraph of $W$. We say that $W'$ is a sub-wall of $W$  if and only if 
 every row of $W'$ is a subpath of a row of $W$, and every column of $W'$ is a subpath of a column of $W$. 
\end{definition}

Towards the definition of a flat wall, we also need to recall the notion of $X$-reduction.

\begin{definition}[{\bf $X$-reduction}]\label{def:Xreduction}
Let $G$ be a graph, $X \subseteq V (G)$, and let $(L, R)$ be a separation of G of order at most $3$ with 
$X \subseteq V(R)$. Moreover, assume that the vertices of $V(L) \cap V(R)$ are connected in $L$. Let $H$ be the graph obtained from $R$ by adding edges connecting every pair of vertices in $V(L) \cap V(R)$. We say that $H$ is an {\em elementary 
$X$-reduction} in $G$, determined by $(L,R)$. We say that a graph $J$ is an {\em $X$-reduction} of $G$ if it can be obtained from $G$ by a series of elementary $X$-reductions.
\end{definition}

Now, we define $C$-flat graphs, which are closely related to flat walls (defined later).

\begin{definition}[{\bf $C$-flat graph}]
\label{def:Cflat1}
Let $G$ be a graph, and let $C$ be a cycle in $G$. We say that $G$ is {\em $C$-flat} if there exists a plane graph $J$ that is a $C$-reduction of $G$ where $C$ bounds the outer-face of $J$.
\end{definition}

The following is an equivalent definition of $C$-flat graphs, which is conceptually easier to work with. The terminology that we use to phrase this definition follows from \cite{DBLP:conf/soda/Chuzhoy15} (which also states the equivalence). The equivalence itself is proved in \cite{newFlatWall}, where a $C$-flat graph is phrased as an $\Omega$-rendition.

\begin{definition}[{\bf $C$-flat graph}]
\label{def:KWT}
Let $G$ be a graph, and let $C$ be a cycle in $G$. We say that $G$ is {\em $C$-flat} if there exist subgraphs 
$G_0,G_1,\ldots,G_k$ of $G$, and a plane graph $\tilde{G}$  such that:
\begin{enumerate}
\item $G=G_0 \cup G_1 \cup\ldots \cup G_k$, and the graphs $G_0,G_1,\ldots,G_k$ are pairwise edge disjoint; 
\item $C$ is a subgraph of $G_0$;
\item $G_0$ is a subgraph of $\tilde{G}$, with $V (\tilde{G}) = V(G_0)$. Moreover, $\tilde{G}$ is a plane graph, and the cycle $C$ bounds its outer face;
\item \label{condition4} For $i\in [k]$, $|V(G_i)\cap V(G_0)|\leq 3$;
\begin{itemize}
\item If $|V(G_i)\cap V(G_0)|=2$, then $u$ and $v$ are adjacent in $\tilde{G}$;
\item If $V(G_i)\cap V(G_0) = \{u,v,w\}$, then some finite face of $\tilde{G}$ is incident with $u,v,w$ and no other vertex;\footnote{That is, $\{u,v\},\{u,w\},\{v,w\}\in E(\tilde{G})$, and the curves corresponding to these edges form a face of $\tilde{G}$.}
\end{itemize}
\item\label{condition5} For all distinct $i,j\in [k]$, $V(G_i)\cap V(G_j)\subseteq V(G_0)$.
\end{enumerate}
\end{definition}

Now, we specify two additional properties that we would like $G_0,G_1,\ldots,G_k,\tilde{G}$ to satisfy.

\begin{observation}\label{obs:cflatmore}
There is an algorithm that, given a connected graph $G$, a cycle $C$ in $G$, subgraphs 
$G_0,G'_1,\ldots,G'_{k'}$ of $G$ and a plane graph $\tilde{G}'$ witnessing that $G$ is $C$-flat, outputs subgraphs 
$G_0,G_1,\ldots,G_k$ of $G$ and a plane graph $\tilde{G}$ witnessing that $G$ is $C$-flat and also satifying the following conditions:
\begin{enumerate}
\setcounter{enumi}{5}
\item\label{conditionsix} For any $i\in [k]$, $u\in V(G_i)\setminus V(G_0)$ and $v\in V(G_0)\cap V(G_i)$, 
there is a path from $u$ to $v$ with all internal vertices being in $V(G_i)\setminus V(G_0)$, and 
\item\label{conditionseven} For any edge $\{u,v\}\in E(\tilde{G})\setminus E(G_0)$, there exists $i\in [k]$ such that $\{u,v\}\subseteq V(G_i)\cap V(G_0)$. 
\end{enumerate} 
The running time of the algorithm is bounded by $\OO(n)$.
\end{observation}

\begin{proof}[Proof sketch]
To ensure that condition \ref{conditionseven} is satisfied, we can simply remove all edges $\{u,v\}\in E(\tilde{G}')\setminus E(G_0)$ from $\tilde{G}'$ if there is no $i\in [k']$ with $\{u,v\}\subseteq V(G'_i)\cap V(G_0)$.

To ensure that condition \ref{conditionsix} is satisfied, if there exist $i\in [k']$, $u\in V(G'_i)\setminus V(G_0)$ and $v\in V(G_0)\cap V(G'_i)$ such that there is no path from $u$ to $v$ with all internal vertices being in $V(G'_i)\setminus V(G_0)$, then we ``split'' $G_i'$ into 
smaller graphs. For each $\emptyset\neq S\subseteq V(G_i')\cap V(G_0)$, let $V_S$ be the set of vertices in $G_i'\setminus (V(G_i')\cap V(G_0))$ reachable from any vertex in $S$ with internal vertices from $V(G'_i)\setminus V(G_0)$ and not reachable from 
$(V(G_i')\cap V(G_0))\setminus S$ with internal vertices from $V(G'_i)\setminus V(G_0)$. Then we replace $G_i'$ 
with $G'_i(S)=G'_i[S\cup V_S]$. Clearly for any $S\neq S'$, there are no edges between vertices in $V_S$ and $V_{S'}$. Since $G$ is connected $\bigcup_{S\subseteq V(G_i')\cap V(G_0)} V_S=V(G_i')\setminus V(G_0)$. Moreover, for each $G'_i(S)$, $u\in V(G'_i(S))\setminus V(G_0)$ and $v\in V(G_0)\cap V(G'_i(S))$, there is a path from $u$ to $v$ with all internal vertices from $V(G'_i(S))\setminus V(G_0)$. 
Therefore after doing the above mentioned splitting operations, condition \ref{conditionsix} will be satisfied. 
 All required modifications can be done in time $\OO(n)$.
\end{proof}

Condition~\ref{conditionseven} in \autoref{obs:cflatmore} implies that $\tilde{G}$ is edge minimal (that is, no edge can be removed from $\tilde{G}$ such that all the above seven properties will still be satisfied).  
We are now ready to present the definition of a flat wall.

\begin{definition}[{\bf Flat wall}]\label{def:flatWall}
Let $G$ be a graph, and let $W$ be a wall in $G$ with outer boundary $D$. Suppose there is a separation $(A,B)$ of $G$, such that $V(A)\cap V(B) \subseteq V(D)$, $V(W) \subseteq V(B)$, and there is a choice of pegs of $W$, such that every peg belongs to $A$. If some $(V(A) \cap V(B))$-reduction $J$ of $B$ can be drawn in a disc with the vertices in 
$V(A) \cap V(B)$ drawn on the boundary of the disc in the order determined by $D$, then we say that the wall $W$ is {\em flat} in $G$.
\end{definition}

\paragraph{Simple Lemmata and Known Propositions.} We now state several results regarding flat walls. We first relate Definitions \ref{def:KWT} and \ref{def:flatWall}. This observation will be encompassed by \autoref{prop:newFlatWall}, but we explicitly state it for the sake of clarity of the connections between the notions introduced earlier.

\begin{observation}\label{obs:CFlat}
Let $G,W,D,A,B$ and $J$ be defined as in \autoref{def:flatWall}. Consider the disc $\mathscr{Q}$ where $J$ is drawn such that the vertices in
$V(A) \cap V(B)$ are drawn on the boundary of $\mathscr{Q}$ in the order determined by $D$.
Let $C$ be the cycle with vertex set $V(A)\cap V(B)$ and edge set defined by the boundary $\mathscr{Q}$, that is, every two consecutive vertices on the boundary of $\mathscr{Q}$ are connected by an edge in $E(C)$. Then, $G^\star=B+E(C)$ is $C$-flat.
\end{observation}

\begin{proof}
Note that $J$ is a $(V(A) \cap V(B))$-reduction of $B$ with the vertices in $V(A) \cap V(B)$ drawn on the boundary of $\mathscr{Q}$ in the order determined by $D$. Clearly, the definition of $C$ directly implies that $J^\star=J+E(C)$ is a plane graph whose outer-face is bounded by $C$.
By \autoref{def:Xreduction}, there is a series of elementary $(V(A) \cap V(B))$-reductions, denoted by $(L_1,R_1),(L_2,R_2),\ldots,(L_t,R_t)$, that when applied to $B$ eventually result in the graph $J$. Because $V(C)=V(A)\cap V(B)$, we have that $(L_1,R_1),(L_2,R_2),\ldots,(L_t,R_t)$ is also a series of elementary $V(C)$-reductions, that when applied to $G^\star$ eventually result in the graph $J^\star$. By \autoref{def:Cflat1}, we conclude that $G^\star$ is $C$-flat.
\end{proof}

In light of \autoref{obs:CFlat} and \autoref{def:KWT}, we explicitly define a type of tuple as follows for the sake of clarity of presentation.

\begin{definition}[{\bf Flatness tuple}]
Let $G$ be a graph, $w,t\in\mathbb{N}$, $A\subseteq V(G)$ be a set of at most $t$ vertices and $W$ be a wall of size at least $(w\times w)$, such that $V(W)\cap A=\emptyset$ and $W$ is a flat wall in $G\setminus A$. Let $D$ denote the outer boundary of $W$.

We say that $(A',B',C,\tilde{G}, G_0,G_1,\ldots,G_k)$ is a {\em flatness tuple} for $(G,w,t,A,W)$ if $(A',B')$ is a separation of $G\setminus A$ such that $V(A')\cap V(B') \subseteq V(D)$ and $V(W) \subseteq V(B')$, and $C,G_0,G_1,\ldots,G_k$ and $\tilde{G}$ witness that $B'+E(C)$ is $C$-flat as specified by \autoref{def:KWT}, where $V(A')\cap V(B')=V(C)\subseteq V(D)$ and the order in which vertices appear on $C$ is the same as the order in which they appear on $D$.
\end{definition}

%

In the above definition we have a subset $A\subseteq V(G)$ and the flatness tuple is defined for a flat wall in $G\setminus A$. The set $A$ is part of the definition to align with the flat wall theorem (\autoref{prop:newFlatWall}).  Specifically, the following proposition is the flat wall theorem as proved in~\cite{newFlatWall}. A theorem with better bounds regarding the relation between $w,t,R$ is proved in \cite{DBLP:conf/soda/Chuzhoy15}. However, the running time of the algorithm of the theorem in \cite{DBLP:conf/soda/Chuzhoy15} is stated to be polynomial (rather than linear) in $(n+m)$ and therefore we cite the theorem of~\cite{newFlatWall}.

\begin{proposition}[\cite{newFlatWall}]\label{prop:newFlatWall}
There exists an algorithm that given a graph $G$, integers $w,t\geq 1$, and an $(R\times R)$-wall $W$ in $G$ where $R=49152t^{24}(60t^2+w)$, outputs either a function $\phi$ witnessing that $K_t$ is a minor of $G$ or a subset $A\subseteq V(G)$ of at most $12288t^{24}$ vertices, and a sub-wall $W^*$ of $W$ of size at least $(w\times w)$, such that $V(W^*)\cap A=\emptyset$ and $W^*$ is a flat wall in $G\setminus A$. In addition, the algorithm outputs a flatness tuple $(A',B',C,\tilde{G},G_0,G_1,\ldots,G_k)$ for $(G,w,\vert A\vert ,A,W^\star)$.
The running time of the algorithm is bounded by $\OO(t^{24}m+n)$.
\end{proposition}

\autoref{prop:newFlatWall} give us a clique minor or a flatness tuple if the treewidth of the input graph is {\em large}. For our purpose we need a flatness tuple where the treewidth of $G_i$, $i\in [k]$ is bounded. 

\begin{definition}[{\bf Nice Flat Wall}]
Let $G$ be a graph, $w,t\in\mathbb{N}$, $A\subseteq V(G)$ be a set of at most $t$ vertices and $W$ be a wall of size at least $(w\times w)$, such that $V(W)\cap A=\emptyset$ and $W$ is a flat wall in $G\setminus A$. Let $D$ denote the outer boundary of $W$.
We say that the {flatness tuple} $(A',B',C,\tilde{G},G_0,G_1,\ldots,G_k)$   for $(G,w,t,A,W)$ is an {\em $\ell$-nice flatness tuple} if for all $i\in [k]$, $\tw(G_i)\leq \ell$. Then, we also call $W$ a {\em $\ell$-nice flat wall} in $G\setminus A$. 
\end{definition}

We will use \autoref{prop:newFlatWall} to get a nice flat wall. 
As the algorithm of \autoref{prop:newFlatWall} requires a wall, to be able to exploit it algorithmically, we need the two following propositions as well.

\begin{proposition}[\cite{DBLP:journals/jacm/ChekuriC16,DBLP:conf/swat/Chuzhoy16}]\label{prop:gridTWGeneral}
There exists $c\in\mathbb{N}$ such that for every $g\in\mathbb{N}$, every graph $G$ either contains a $(g\times g)$-grid as a minor or has treewidth at most $cg^{19}\log^cg$.
\end{proposition}

The following proposition can be proved using Propositions~\ref{prop:twExact} and \ref{prop:twDPNoDec}. 

\begin{proposition}\label{prop:gridTWGeneralalgo}
There exists $c\in\mathbb{N}$ and an algorithm that given a graph $G$ and integer $g\in {\mathbb N}$, 
runs in time $2^{\OO(g^{58})}n \log n$, and outputs either a $g\times g$-wall in $G$ or outputs a nice tree-decomposition of $G$ of width at most  $cg^{19}\log^cg$. 
\end{proposition}
\begin{proof}[Proof sketch]
Let $\widehat{c}$ be the constant mentioned in \autoref{prop:gridTWGeneral}. Let $c$ be a constant such that  $\widehat{c}(2g)^{19}\log^{\widehat{c}}(2g)\leq cg^{19}\log^c g$. 
First we run the algorithm of \autoref{prop:twExact} to test whether $\tw(G)\leq \widehat{c}(2g)^{19}\log^{\widehat{c}}(2g)$. If the answer is yes, then it is a valid output and this step takes time  $2^{g^{57} (\log g)^{\OO(1)}} n$.  

Otherwise, we use ``binary search type" of arguments to output an induced subgraph of $G$ with treewidth at least $\widehat{c} (2g)^{19} \log^{\widehat{c}}(2g)$, but at most  $2\widehat{c} (2g)^{19} \log^{\widehat{c}}(2g)$. Let  $t=\widehat{c} (2g)^{19} \log^{\widehat{c}}(2g)$. At step $i$ of the algorithm we will have two subsets of edges $F_i$ and $A_i$ with the following invariant: 
$\tw(G[F_i])\leq 2t$ and $G[A_i\cup F_i]> t$. The stopping condition of the algorithm is when $\tw(G[A_i\cup F_i])\leq 2t$. Initially we set $A_0=E(G)$ and $F_0=\emptyset$. Clearly the invariant holds initially. Now we explain how to perform step $i$, if step $i-1$ does not satisfy the stopping condition. Let $S_i$ be an arbitrary subset of $A_{i-1}$ of size $\lceil\frac{|A_i|}{2}\rceil$. If $\tw(G[F_{i-1}\cup S_i]) \leq t$, then we set $F_i=F_{i-1}\cup S_i$ and $A_i=A_{i-1}\setminus S_i$.  Clearly, $A_i$ and $F_i$ satisfy the invariant. If $\tw(G[F_{i-1}\cup S_i]) > t$, then we set $F_i=F_{i-1}$ and $A_i=S_i$. It is easy to see that in this case as well the invariant follows. We stop at step $i$ if $\tw(G[A_i\cup F_i])\leq 2t$. Let $A$ and $F$ be the edge sets of the last step of the above procedure. We have that $\tw(G[A\cup F])\leq 2t$ and $\tw(G[A\cup F])>t$. Then by \autoref{prop:gridTWGeneral}, there is a $(2g\times 2g)$-grid minor in $G[A\cup F]$. 
By Observations~\ref{obs:gridToWall} and \ref{obs:wallminortosub}, there is a $(g\times g)$-wall as subgraph in $G[A\cup F]$. Now we use \autoref{prop:twDP} to output a $(g\times g)$-wall in  $G[A\cup F]$, in time $t^{\OO(t)}n$.

Notice that the number of steps in the above procedure to output $A$ and $F$ is $\log m$. In each step we run the algorithm of 
\autoref{prop:twExact}, which runs in time $2^{\OO(t^3\log t)}n$. Therefore the total running time of the algorithm follows.  
\end{proof}

A linear relationship between the size of a grid minor and treewidth is known for planar graph, which is much better than the relationship in \autoref{prop:gridTWGeneral}.

\begin{proposition}[(6.2)~\cite{ROBERTSON1994323}]
\label{lem:planargridtw}
Let $g\geq 1$ be an integer. Any planar graph with treewidth larger than $6g-5$ has a $g\times g$ grid minor. 
\end{proposition}

The following proposition  provides a constant-factor linear-time algorithm to compute the treewidth of a planar graph, which also outputs a grid minor in case the treewidth is large.

\begin{proposition}[\cite{DBLP:journals/tcs/KammerT16}]\label{prop:surfaceMinor}
There exists a constant $c$ such that for any planar graph $G$ and integer $r\in\mathbb{N}$, in time $\OO(r^2n)$ one can compute either  an $r\times r$-grid as a minor of $G$ or a tree decomposition of width at most $cr$.
\end{proposition}


The following proposition follows from  \autoref{prop:surfaceMinor} and \autoref{obs:wallminortosub}. 


\begin{proposition}\label{prop:gridTWplanaralgof}
There exists a constant $c$ such that for any planar graph $G$ and integer $g\in\mathbb{N}$, in time $\OO(g^2n)$ one can compute either  a $g\times g$-wall in $G$ or a tree decomposition of width at most $cg$.
\end{proposition}


\autoref{prop:newFlatWall} is linear in $(n+m)$. Towards designing an $\OO(n\log n)$-time  algorithm to output a clique minor or a {\em nice flatness tuple} in a {\em large} treewidth graph, we require the following result. 

\begin{proposition}[\cite{ReedWood2009}]
\label{prop:almostlinearcliqueminor}
There is an algorithm which given a graph $G$ and $t \in {\mathbb N}$ with $m \geq 2^{t-3}\cdot n$, runs in time $\OO(t(n+m))$, and outputs a function $\phi$ witnessing that $K_t$ is a minor of~$G$. 
\end{proposition}

Because of \autoref{prop:almostlinearcliqueminor}, by just considering arbitrary $2^{t-3}\cdot n$ edges in a graph (if $m \geq 2^{t-3}\cdot n$) we get a $K_{t}$ minor model in time $\OO(2^{t}n)$.  

\begin{proposition}
\label{prop:linearcliqueminor}
There is an algorithm that given a graph $G$ and $t \in {\mathbb N}$ with $m \geq 2^{t-3}\cdot n$, runs in time $\OO(2^{t}n)$, and outputs a function $\phi$ witnessing that $K_{t}$ is a minor of $G$.  
\end{proposition}


Now we are ready to prove a lemma similar to the flat wall theorem, but here we will output a {\em nice} flat wall.

\begin{lemma}\label{lem:flatWallCombined}
There is a constant $c\in\mathbb{N}$ and an algorithm that given a graph $G$, and integers $g,w,t\geq 1$ such that $g\geq ct^{48}(t^2+w)$, outputs one of the following.
\begin{itemize}
\setlength\itemsep{0em}
\item A nice tree decomposition of $G$ of width at most $cg^{19}\log^cg$.
\item A function $\phi$ witnessing that $K_t$ is a minor of $G$.
\item A subset $A\subseteq V(G)$ of at most $ct^{24}$ vertices and a wall $W$ of size at least $(w\times w)$, such that $V(W)\cap A=\emptyset$ and $W$ is a $(cg^{19}\log^cg)$-nice flat wall in $G\setminus A$. In addition, it outputs a $(cg^{19}\log^cg)$-nice flatness tuple $(A',B',C,\tilde{G},G_0,G_1,\ldots,G_k)$ for $(G,w,\vert A\vert ,A,W)$.
\end{itemize}
The running time of the algorithm is upper bounded by $2^{\OO({g^{58}})}n\log^2n$. 
\end{lemma}

\begin{proof}
%
Recall that $n=\vert V(G)\vert$ and $m=\vert E(G)\vert$.  We choose the value of $c\geq 12288$, in such a way that 
$c$ is at least the constant mentioned in \autoref{prop:gridTWGeneralalgo} and $ct^{48}(t^2+w)\geq 49152t^{24}(60t^2+4(12288t^{24}+3) w)=49152t^{24}(60t^2+\widehat{w})$, where $\widehat{w}=4(12288t^{24}+3) w$. Let $\tw=cg^{19}\log^cg$.  
First we use the algorithm of \autoref{prop:twExact} to test whether the treewidth of $G$ is $\tw$ or not. If the answer is yes, then the algorithm will output a  nice tree decomposition of $G$ of width at most $\tw$, which is a valid output for our algorithm. 
This step takes time $2^{\OO(g^{57}(\log g)^{\OO(1)})}n \leq 2^{\OO(g^{58})}n$. If $m\geq 2^{t-3} n$, then we run the algorithm of \autoref{prop:linearcliqueminor} and it will output a function $\phi$ witnessing that $K_t$ is a minor in $G$. The time to output $\phi$, by the algorithm of \autoref{prop:linearcliqueminor}, is $\OO(2^t n)$. So, from now on we assume that $m <2^{t-3}n$ and $\tw(G)>\tw$. 

Now we design a recursive algorithm ${\cal A}$ as follows. At step $i$ we output either a valid output for our original task (i.e., either a function $\phi$ witnessing that $K_t$ is a minor of $G$ or a $(\tw)$-nice flat wall $W$ of size at least $(w\times w)$ in $G\setminus A$ for some $A\subseteq V(G)$, $\vert A\vert \leq 12288t^{24}$) or construct a separation $(H_i,H_i')$ of $G$ order at most $2(12288t^{24}+3)$ and  $\tw(H_i)>\tw$ (these are the invariants we maintain). Then, we will be searching for a $K_t$-minor or a nice flat wall in $H_i$ in the subsequent steps of the algorithm.


Initially w e set $(H_0,H_0')=(G,(\emptyset,\emptyset))$. Since $\tw(G)>\tw$, clearly the invariants hold initially. 
At step $i+1$, we do the following. 
Notice that $\tw(H_{i})>\tw$, because of the invariant. First we use \autoref{prop:gridTWGeneralalgo} to get a $(g\times g)$-wall $W_i$ in time $2^{\OO(g^{58})} n\log n$. Since $g\geq 49152t^{24}(60t^2+\widehat{w})$, by using the algorithm of \autoref{prop:newFlatWall} in time $2^{\OO(t)}n$ (because $m< 2^{t-3}n$), we get either a subset $A\subseteq V(H_i)$ and  a $(\widehat{w}\times \widehat{w})$ flat wall $W^{\star}_i$ in $H_i\setminus A$, or a function $\phi$ witnessing that $K_t$ is a minor in $H_{i}$ (and hence in $G$). In the latter case, the function $\phi$ is a valid output of our algorithm. In the former case, consider the flatness tuple $(A',B',C,\tilde{G}, G_0,G_1,\ldots,G_k)$ for $(H_i,\widehat{w},\vert A\vert, A,W^{\star}_i)$, outputted by the algorithm of \autoref{prop:newFlatWall}. Now we use the algorithm of \autoref{prop:twExact} to test whether the treewidth $G_i$ is at most $\tw$ or not, for all $i\in [k]$.  These computations of treewidth altogether take time $2^{\OO(g^{58})} n$. 
Let $S_{i}=V(H_{i})\cap V(H_{i}')$. Now we have two cases. 

\medskip
\noindent{\bf Case 1:  There exist $1\leq j_1<j_2<j_3\leq k$ such that  $\tw(G_{j_1}),\tw(G_{j_2}),\tw(G_{j_3})>\tw$.}
 Notice that $\vert S_{i}\vert \leq 2(12288t^{24}+3)$. There exit two distinct $i',j'\in \{j_1,j_2,j_3\}$ such that $\vert V(G_{i'})\vert , \vert V(G_{j'})\vert \leq \frac{\vert V(H_{i})\vert}{2}-\vert A\vert$, because $\tw(G_{j_1}),\tw(G_{j_2}),\tw(G_{j_3})>\tw$ and $\tw$ is much larger than $2\vert A\vert$. Then, since $V(G_{i'})\cap V(G_{j'})\subseteq V(G_0)$, there exists $r\in \{i',j'\}$ such that $\vert (V(G_r)\setminus V(G_0))\cap S_{i}\vert \leq 12288t^{24}+3$. In this case, we set $H_{i+1}=G_r\cup G[A]$ and let $H_{i+1}'$ be the minimal subgraph of $G$ such that $G=H_{i+1}\cup H_{i+1}'$. 
That is, $H_{i+1}'$ is a subgraph of $(G\setminus (V(G_r)\setminus V(G_0)))\cup G[A]$. 
Here, we have that $V(H_{i+1})\cap V(H_{i+1}') \subseteq (A\cup (V(G_r)\cap V(G_0))\cup ((V(G_r)\setminus V(G_0))\cap S_{i})$, 
because the only vertices in $H_{i+1}$ which have neighbours outside $V(H_{i+1})$ are from $A\cup S_i\cup (V(G_0)\cap V(G_r))$. Since $\vert (V(G_r)\setminus V(G_0))\cap S_{i}\vert \leq 12288t^{24}+3$,  $\vert A\vert \leq 12288t^{24}$ and $\vert V(G_0)\cap V(G_r)\vert \leq 3$, the order of the separation $(H_{i+1},H_{i+1}')$ is at most $2(12288t^{24}+3)$. Since $\tw(G_r)>\tw$ and $H_{i+1}=G_r+G[A]$, we have that $\tw(H_{i+1})>\tw$. Therefore, the separation $(H_{i+1},H_{i+1}')$ satisfies the invariants. 
Here, $\vert V(H_{i+1})\vert \leq \frac{\vert V(H_i)\vert}{2}$, because $\vert V(G_{r})\vert \leq \frac{\vert V(H_{i})\vert}{2}-\vert A\vert$. 
This completes the step $i+1$ of algorithm ${\cal A}$. 

\medskip
\noindent{\bf Case 2: Case 1 is false.} 
In this case, we directly output a $(\tw)$-nice flat wall in $G\setminus A$. Let $j_1,j_2\in [k]$ be such that $\tw(G_{j_1}),\tw(G_{j_2})>\tw$. That is, for 
all $j\in [k]\setminus \{i_1,i_2\}$, $\tw(G_j)\leq \tw$.  Let $T=(V(G_{j_1})\cup V(G_{j_2}))\cap V(G_0)$. Notice that $\vert T\vert \leq 6$.  
Since $W^{\star}_i$ is a flat wall of size $\widehat{w}\times \widehat{w}$,  $\widehat{w}=4(12288t^{24}+3)\cdot w$ and $\vert S_{i}\cup T \vert \leq 2(12288t^{24}+3)+6$, there exist $w$ distinct and consecutive rows and $w$ distinct and consecutive columns of $W^{\star}_i$ that are not hit by $S_{i}\cup T $. Let $W$ be the subwall of  $W^{\star}_i$ formed by these rows and columns. We claim that $W$ is a $(\tw)$-nice flat wall in $G\setminus A$.  Let $D'$ be the boundary of $W$. Let $X$ be the set of vertices in $\tilde{G}$ contained in the inner face of $D'$. Let $G_0'=G_0[X]$ and $\tilde{G}'=\tilde{G}[X]$. Let $I\subseteq [k]$ be such that  $j\in I$ if and only if $V(G_j)\cap V(G_0)\subseteq X$. Since $D'$ is a cycle in the plane graph $\tilde{G}$ and for any $j\in [k]$, $V(G_i)\cap V(G_0)$ forms a clique in $\tilde{G}$, we have that there is no $j\in [k]$ such that $V(G_j)\cap V(G_0)$ has a vertex in the interior face of $D'$ and another vertex in the exterior face of $D'$. Let $B^{\star}=G_0'\cup (\bigcup_{j\in I}G_j)$ and $A^{\star}$ be a minimal subgraph of $H_{i}$ such that $A^{\star}\cup B^{\star}=H_{i}$. Therefore, $(A^{\star},B^{\star},D',\tilde{G}',G_0',G_{i_1},\ldots, G_{i_{\ell}})$, where $I=\{i_1,\ldots, i_{\ell}\}$, is a flatness tuple for $(H_i,w,\vert A\vert,A,W)$.  Since  $S_{i}\cap V(B^{\star})=\emptyset$, $W$ is a flat wall in $G\setminus A$ and  $(A^{\star}\cup H_i',B^{\star},D',\tilde{G}',G_0',G_{i_1},\ldots, G_{i_{\ell}})$ is a flatness tuple for $(G,w,\vert A\vert,A,W)$. Moreover, since $\tw(G_{i_j})\leq \tw$
for all $j\in [\ell]$, $(A^{\star}\cup H_i',B^{\star},D',\tilde{G}',G_0',G_{i_1},\ldots, G_{i_{\ell}})$ is a $(\tw)$-nice flatness tuple. 

\medskip
\noindent Now we analyze the running time of algorithm ${\cal A}$. Since in step $i+1$, either we output a valid output or construct a separation $(H_{i+1},H_{i+1}')$, where $\vert H_{i+1}\vert \leq \frac{\vert V(H_i)\vert}{2}$, the number of steps of the algorithm is upper bounded by $\OO(\log n)$. In each step, first we run the algorithm 
of \autoref{prop:gridTWGeneralalgo}, 
which takes time $2^{\OO({g^{58}})}n \log n$. Then we run the algorithm of \autoref{prop:newFlatWall},which takes time $2^{\OO(t)}n$. Afterwards, we use \autoref{prop:twExact} to test the treewidth of each of the graphs $G_1,\ldots,G_k$ and these computations together take time  $2^{\OO({g^{58}})}n$, because any pair of these graphs intersect on at most three vertices. The execution of the above mentioned two cases takes time $\OO(n)$. Therefore, the total running time of ${\cal A}$ is upper bounded by $2^{\OO({g^{58}}+t)}n\log^2n=2^{\OO({g^{58}})}n\log^2n$. As the initial computation (before the execution of algorithm ${\cal A}$) takes time $2^{\OO({g^{58}})}n$, the total running time of the algorithm follows. 
\end{proof}

Finally, we show that the plane graph $\tilde{G}$ in the flatness tuple, has a large grid as a minor.
 

\begin{lemma}\label{lem:largeWall}
Let $G$ be a graph,  $W$ be a $(w\times w)$-wall in $G$, and $C$ be a cycle in $G$ such that there is a choice of pegs of $W$ where every peg belongs to $V(C)$. In addition, let $G_0,G_1,\ldots,G_k$ be subgraphs of $G$ with a plane graph $\tilde{G}$ that witness that $G$ is $C$-flat.
Then, $\tilde{G}$ has a $(w-2 \times w-2)$-grid as a minor.
\end{lemma}

\begin{proof}[Proof sketch]
We first observe that \autoref{def:KWT} directly implies the following claim.

\begin{claim}\label{claim:inLargeWall}
Let $P$ be a path in $G$. Define $\tilde{P}$ as the sequence of vertices in $V(P)\cap V(\tilde{G})$ where vertices appear in the same order in which they appear on $P$. Then, $\tilde{P}$ is a path in $\tilde{G}$.
\end{claim}

We define a collection of sets that will exhibit the existence of a $(w \times w)$-grid. Let $R_1,R_2,\ldots,R_w$ denote the rows of the wall $W$. For all $i,j\in\{2,\ldots,w-1\}$, we define a set $S_{i,j}$ as follows: 
$S_{i,j}$ is the set that contains the vertices of the subpath of $R_i$ between the $(2(j-1)-1)^{th}$ and $(2(j-1)+1)^{th}$ degree-3 vertices on $R_i$.

We now claim that for all $i,j\in\{2,\ldots w-1\}$, $V(\tilde{G})\cap S_{i,j}\neq \emptyset$. 
Suppose, by way of contradiction, that $V(\tilde{G})\cap S_{i,j}=\emptyset$. This means that there exists $G_p$, $p\in[k]$, such that $S_{i,j}\subseteq U$ where we denote $U=V(G_p)\setminus V(\tilde{G})$. By \autoref{def:KWT}, we have that $(G_p,G[V(G)\setminus U])$ is a separation of $G$ of order at most $3$. Recall that all the pegs of $W$ belong to $V(G)\setminus U$. Since the flow between $S_{i,j}$ (which contains two vertices of degree-3 in $W$) and these pegs is at least 4, we have reached a contradiction. For all $i,j\in\{2,\ldots,w-1\}$, denote $X_{i,j}=V(\tilde{G})\cap S_{i,j}$. 
Observe that for all $i,j\in\{2,\ldots,w-2\}$, $\tilde{G}[X_{i,j}]$ is a connected graph. Indeed, by noting that for all $i,j\in\{2,\ldots,w-1\}$, $S_{i,j}$ is a path, the correctness of this claim directly follows from \autoref{claim:inLargeWall}.

Finally, for all $2\leq i\leq w-2$ and $2\leq j\leq w-1$, we define the path $P_{i,j}$ in $G$ as follows. Because $W$ is a wall, there exists a subpath of $C_{j-1}$ from a vertex in $S_{i,j}$ to a vertex in $S_{i+1,j}$, that uses only vertices of degree 2 in $W$ (and these internal vertices are not from any $R_{i'}$). Let $\tilde{P}_{i,j}$ be the path specified by \autoref{claim:inLargeWall} for the path $P_{i,j}$. 
Similarly, we define the path $P'_{i,j}$ in $G$ for any $2\leq j\leq w-2$ and $2\leq i\leq w-1$, where now we take a subpath of $R_i$ rather than a subpath of $C_{j-1}$. 
Let $\tilde{P}'_{i,j}$ be the path specified by \autoref{claim:inLargeWall} for the path $P'_{i,j}$. 
Observe that the paths $\tilde{P}_{i,j}$ and  $\tilde{P}'_{i,j}$  are internally vertex-disjoint.
Furthermore, the sets $X_{i,j}$, $i,j\in\{2,\ldots, w-1\}$, are also pairwise disjoint.
Then, because  $V(\tilde{P}_{i,j})\subseteq V(P_{i,j})$, and as $\tilde{G}[X_{i,j}]$ is a connected graph, we have exhibited a $(w-2)\times (w-2)$-grid in $\widetilde{G}$ (which consists of the paths $\tilde{P}_{i,j}$ and the sets $X_{i,j}$). 
\end{proof}

\section{Statements of Main Theorems}
\label{sec:mainthms}

In this section we state our main theorems and prove them in the subsequent sections. 
To state the theorems, we first need to give several definitions and then define problems that are more general than \TMC\ and \TMH.  

\subsection{Rooted Graph, Folio, and Extended Folio}\label{sec:prelimsRooted}

We first present the definition of a rooted graph.

\begin{definition}[{\bf Rooted graph}~\cite{DBLP:conf/stoc/GroheKMW11}]
A {\em rooted graph} is an undirected graph $G$ with a set $R(G) \subseteq V(G)$ of vertices specified as {\em roots} and an injective
mapping $\rho_G : R(G) \rightarrow \mathbb{N}$ assigning a distinct positive integer label to each root vertex. We say that two rooted graphs $G_1$ and $G_2$ are {\em compatible} if $\rho_{G_1}(R(G_1))=\rho_{G_2}(R(G_2))$.  We also say that two graphs $G_1$ and $G_2$ have the same set of roots when they are compatible. 
\end{definition}

\begin{definition}[{\bf Replacement}~\cite{DBLP:conf/stoc/GroheKMW11}]
Let $G$ be a rooted graph and $(G_1,G_2)$ be a separation such that $S=V(G_1)\cap V(G_2)\subseteq R(G)$.  Let $G_1'$ be a graph compatible with $G_1$. Replacing $G_1$ with $G_1'$ in the separation $(G_1,G_2)$ gives the following graph $G'$ (which is denoted by $G_1'\cup G_2$). 
\begin{itemize}
\item $V(G')=V(G'_1)\cup (V(G_2)\setminus V(G_1))$, and 
\item $E(G')=E(G_1')\cup E(G_2-S)\cup \{\{u',v\}\colon \{u,v\}\in E(G_2),u\in S, v\notin S, \rho_{G_1}(u)=\rho_{G_1'}(u')\}$. 
\end{itemize}
\end{definition}
That is, we replace  $G_1$ with $G_1'$ in $G$ such that the role of $S$  is taken by $\rho_{G_1'}^{-1}(\rho(S))$. 
We now adapt the notion of topological minors to the presence of roots.

\begin{definition}[{\bf Topological minor of rooted graph}~\cite{DBLP:conf/stoc/GroheKMW11}]
Let $G$ and $H$ be two undirected rooted graphs. We say that $H$ is a {\em topological minor} of $G$ if there exist injective functions $\phi: V(H)\rightarrow V(G)$ and $\varphi: E(H)\rightarrow\paths(G)$ such that 
\begin{itemize}
\setlength\itemsep{0em}
\item for all $e=\{h,h'\}\in E(H)$, the endpoints of $\varphi(e)$ are $\phi(h)$ and $\phi(h')$,
\item for all distinct $e,e'\in E(H)$, the paths $\varphi(e)$ and $\varphi(e')$ are internally vertex-disjoint,
\item there do not exist a vertex $v$ in the image of $\phi$ and an edge $e\in E(H)$ such that $v$ is an internal vertex on $\varphi(e)$, and
\item for all $v\in R(H)$, $\rho_H(v) = \rho_G(\phi(v))$.
\end{itemize}
\end{definition}

Note that if $R(G)=\emptyset$, then the definition above coincides with the standard definition of a topological minor. The {\em folio} of a rooted graph $G$ is the collection of all topological minors of $G$. We are only interested in topological minors of bounded size, which gives rise to the following definition.

\begin{definition}[{\bf $\delta$-folio}]
Let $\delta\in\mathbb{N}$. The {\em $\delta$-folio} of a rooted graph $G$ is the collection of all topological minors $H$ of $G$ such that $|E(H)|+\isolated(H)\leq \delta$.
\end{definition}

Clearly, every graph in a $\delta$-folio has at most $2\delta$ vertices. We say that a vertex $v$ is {\em irrelevant} 
to the $\delta$-folio of $G$ (w.r.t roots $R(G)$), if the $\delta$-folio of $G$  is same as the $\delta$-folio of $G\setminus v$. 

\begin{definition}[{\bf Extended $\delta$-folio}]
\label{def:extfolio}
Let $\delta\in\mathbb{N}$, and let $G$ be a rooted graph. Given a graph $X$ such that $V(X)=R(G)$, 
the {\em $(X,\delta)$-folio} of $G$ is the $\delta$-folio of $(G\cup X)$ with $R(G)$ as the set of roots.
The {\em extended $\delta$-folio} of $G$ is the function $f$ whose domain is $\allgraphs(R(G))$, and for each graph $X$ in $\allgraphs(R(G))$, $f(X)$ is equal to the $(X,\delta)$-folio of $G$.
\end{definition}

We remark that an extended $\delta$-folio can also be thought of as a tuple where each entry is uniquely identified with a graph on $R(G)$. 

\begin{definition}
We say that a $U\subseteq V(G)$ is {\em irrelevant} 
to the extended $\delta$-folio of $G$ (w.r.t.~roots $R(G)$), if the extended $\delta$-folio of $G$ is same as the extended $\delta$-folio of $G\setminus U$. When $U=\{v\}$, we say that $v$ is an irrelevant vertex to   the extended $\delta$-folio of $G$. 
\end{definition}

We use the following simple observation later in the paper (see Proposition~2.3~\cite{DBLP:conf/stoc/GroheKMW11}). 
\begin{observation}
\label{obs:flatfolioonly}
Let $G$ be a graph and $\delta\in {\mathbb N}$. Let $Q\subseteq R \subseteq V(G)$. 
\begin{itemize}
\item[(a)] The extended $\delta$-folio of $G$ with respect to roots $R$ can be computed from the $(\delta+\vert R\vert)$-folio of $G$ w.r.t roots $R$. Moreover, if a vertex $v$ is irrelevant to the $(\delta+\vert R\vert)$-folio of $G$ w.r.t.~roots $R$, then $v$ is irrelevant to the extended $\delta$-folio of $G$ w.r.t.~roots $R$.
\item[(b)] The extended $\delta$-folio of $G$ with respect to roots $Q$ can be obtained from the 
extended $\delta$-folio of $G$ with respect to roots $R$. Moreover, if a vertex $v$ is irrelevant to the extended $\delta$-folio of $G$ w.r.t.~roots $R$, then $v$ is irrelevant to the extended $\delta$-folio of $G$ w.r.t.~roots $Q$.
\end{itemize}
\end{observation}


\begin{definition}
Let $G$ be a rooted graph and $\delta,k\in {\mathbb N}$. 
We say that a vertex $v\in V(G)$ is a $(\delta,k)$-irrelevant for $G$, if for a graph $X$ on $R(G)$ and a vertex subset $S\subseteq V(G)$ of size at most $k$, the $\delta$-folio of $G'=(G\cup X)\setminus S$ (where $R(G')=R(G)\setminus S$), is equal to the the $\delta$-folio of $G' \setminus v$. 
In other words, a vertex $v$ is $(\delta,k)$-irrelevant if for any vertex subset $S\subseteq V(G)$ of size at most $k$, the extended $\delta$-folio of $G\setminus S$ is same as the extended $\delta$-folio of $(G\setminus S)\setminus v$. 
\end{definition}

%
%

In what follows, we define a generalization of \TMC\ where the pattern $H$ is not specified.

\defparproblem{\FindFolio}{A rooted undirected graph $G$, and a non-negative integer $\delta$.}{$\delta$}{What is the extended $\delta$-folio of $G$?}

\smallskip

The special case of \FindFolio\ where the input graph $G$ is planar is called \pFindFolio.
In addition, we define the problem \fFindFolio\ similarly to \FindFolio, but where the input, along with an instance $(G,\delta)$ of \FindFolio, also consists of $t,w,s\in\mathbb{N}$, a subset $A\subseteq V(G)$ of size at most $t$, a flat wall $W$ of size $(w\times w)$ in $G\setminus A$, a separation $(A',B')$ of $G\setminus A$ such that $V(A')\cap V(B') \subseteq V(D)$, and $C,G_0,G_1,\ldots,G_k$ and $\tilde{G}$ witnessing that $B'+E(C)$ is $C$-flat as specified by \autoref{def:KWT} and with the properties in  \autoref{obs:cflatmore} and \autoref{lem:largeWall}, where $V(C)\subseteq V(D)$ for the outer-boundary $D$ of $W$, and the order in which vertices appear on $C$ is the same as the order in which they appear on $D$. Furthermore, we also demand that for all $i\in[k]$, $\tw(V(G_i))\leq s$.

\subsection{Main Theorems}

Finally, let us state the main theorems proved in this paper, which will lead us to the proofs of Theorems \ref{thm:mainIntro}, \ref{thm:mainIntroGenus}, \ref{thm:maingen} and \ref{thm:main}. Throughout the paper, we use $h$ to denote the function mentioned for 
{\sc Disjoint Paths} \cite[result (3.1)]{RobertsonS12} (see \autoref{lem:RoSe}). For planar graphs it is known that $h(k)=2^{ck}$ for a constant $c$ (see \autoref{prop:disPathIrrelevant}). Moreover, we assume that $h(k)\geq k$.

\paragraph{Results for \TMH.}  
As explained before, \autoref{thm:mainIntro} first 
applies \autoref{lem:flatWallCombinedintro} and obtains one of the following structure: $(i)$ a tree decomposition of $G$ of width   bounded by a function of $k$ and $h^{\star}$, $(ii)$ a large clique minor in the input graph, and $(iii)$ a large flat wall in $G$.    As \TMH\ is expressible in MSO (Monadic Second Order logic), Courcelle's theorem, we get the following theorem. 

\begin{theorem}[\cite{ArnborgLS91,Courcelle90,GolovachST19}]
\label{thm:golovach}
There is an algorithm for \TMH\ running in time $f(k,h^{\star},\tw)n$, where $\tw$ is the tree-width of the input graph. 
\end{theorem}

When the input graph contains a large clique minor, we use the following theorem to find a $(\delta,k)$-irrelevant vertex, where $\delta\leq (h^{\star})^2$.


\begin{restatable}{theorem}{dellargeclique}
\label{thm:dellargeclique}
There are two computable functions $f,g \colon {\mathbb N}\times {\mathbb N} \times {\mathbb N}\rightarrow {\mathbb N}$  and an algorithm  that given a rooted graph $G$, two integers $\delta,k\in {\mathbb N}$, and mutually vertex-disjoint connected subgraphs $G_1,\ldots,G_t$ such that for $1\leq i< j\leq t$ there is an edge of $G$ between $G_i$ and $G_j$ (i.e., $(G_1,\ldots,G_t)$ is a minor model of $K_t$ in $G$), where $t> f(\vert R(G)\vert,\delta,k)$, runs in time $g(\vert R(G)\vert,\delta,k)\cdot n^3$, and outputs a $(\delta,k)$-irrelevant vertex $v$ in $G$. 
\end{restatable}

When the input graph contains a large flat wall, the following theorem gives us a  $(\delta,k)$-irrelevant vertex, where $\delta\leq (h^{\star})^2$.


\begin{restatable}{theorem}{dellargewall}
\label{thm:finalFlatWall}
There is a computable function $\widehat{g}$ and an algorithm that, given $k\in\mathbb{N}$ and an instance $(G,\delta,t,w',s')$ of \fFindFolio\ such that  $|R(G)|\leq \boundary$ and $w'\geq (\widehat{g}(\delta^{\star},t))^{k+2}$, finds a $(\delta,k)$-irrelevant vertex in time 
$ 2^{2^{\OO(k((t+\delta^{\star})^2+r)\log (t+\delta^{\star}+r))}}(s')^{\OO(s')}n$
where  $\delta^{\star}=\delta+\boundary$, $\widehat{g}(\delta^{\star},t)=(t+\delta^{\star}+r)^{\OO((t+\delta^{\star})^2+r)}$ and $r=h(\delta^{\star}+t)$.
\end{restatable}

\autoref{thm:mainIntro} follows immediately from  \autoref{lem:flatWallCombinedintro}, \autoref{thm:golovach}, 
\autoref{thm:dellargeclique}, and \autoref{thm:finalFlatWall}. 
Next, using \autoref{thm:finalFlatWall}, we prove the following result for \TMH\ on planar graphs.

\begin{theorem}
\label{thm:mainTMHplanar}
\TMH{}   is solvable in $\OO({2^{2^{k \cdot 2^{{\sf poly(h^\star)}}}}}n^2)$-time on planar graphs. 
\end{theorem}


\begin{proof}[Proof sketch]


If the tree-width of $G$ is ``unbounded'', then we apply \autoref{thm:finalFlatWall} to get a $(\delta,k)$-irrelevant vertex and delete it, and we apply the same procedure until we get a bounded tree-width graph.  
For our problem,  it is enough to find a $(\delta,k)$-irrelevant  vertex for $G$  where $\delta=(h^{\star})^2$ and $R(G)=\emptyset$. 
Notice that $G$ is a plane graph and  we set $\delta=(h^{\star})^2$, $t=0$, $s'=0$, and $\boundary=0$. 
Therefore we have that $\widehat{g}(\delta,t)=2^{\OO(r^2)}$, where $r=h((h^{\star})^2)=2^{c(h^{\star})^2}$, because $G$ is a plane graph 
(see \autoref{prop:disPathIrrelevant} later in the paper), and $c$ is a constant. Hence, $\widehat{g}(\delta,t)=2^{2^{p(h^{\star})}}$ for a polynomial function $p$ in $h^{\star}$. Then, by \autoref{prop:gridTWplanaralgof}, there is a constant $c$ such that if $\tw(G)\geq c \cdot  (2^{2^{\cdot p(h^{\star})}})^{k+2}$, then the algorithm of \autoref{prop:gridTWplanaralgof} outputs $(w'\times w')$-wall, where $w'\geq \widehat{g}(\delta,t)^{k+2}$. Now we apply \autoref{thm:finalFlatWall} to find an irrelevant vertex. The running time of finding one $(\delta,k)$-irrelevant vertex is upper bounded 
by $2^{2^{k \cdot 2^{{\sf poly(h^\star)}}}}n$. 

When $\tw(G)\leq c\cdot (2^{2^{ p(h^{\star})}})^{k+2} \leq c \cdot 2^{(k+2)\cdot 2^{ p(h^{\star})}}$, by standard dynamic programming we solve the problem in time 
$2^{2^{k\cdot 2^{{\sf poly}(h^{\star})}}}n$. Thus, the overall running time follows. 
\end{proof}

To generlize  \autoref{thm:mainTMHplanar} to  graphs of bounded genus, we need to use the known single exponential bound on $h$~\cite{mazoit2013single}. That is, the proof for graphs of bounded genus is identical to  \autoref{thm:mainTMHplanar}, except that we need to use the known single exponential bound on $h$ for graphs of Euler genus $g$~\cite{mazoit2013single}. This leads to \autoref{thm:mainIntroGenus}.



%

\paragraph{The Case of Planar Graphs for \TMC.} In the context of \TMC\ (or \pTMC), for an instance $(G,H,k)$, we say that a vertex $v$ is {\em irrelevant} if $(G,H,k)$ is a \yes-instance if and only if $(G\setminus v,H,k)$ is a \yes-instance.  In order to prove \autoref{thm:main} for $g=0$ (i.e., planar graphs), it is sufficient to prove the following theorem.

\begin{theorem}\label{thm:findIrrelevant}
There exist a constant $c\in\mathbb{N}$ and an algorithm that, given an instance $(G,H,h^{\star})$ of \pTMC\ such that $\tw(G)\geq 2^{2^{c (h^{\star})^2}}$, finds an irrelevant vertex in time \irrTime.
\end{theorem}

Indeed, having \autoref{thm:findIrrelevant} at hand, it is straightforward to prove \autoref{thm:main} as follows.
First, by \autoref{prop:twExact}, we can test whether $\tw(G)\leq 2^{2^{c (h^{\star})^2}}$. If the answer is positive, we call the algorithm given by \autoref{prop:twDP} to solve the problem in time \mainTime. Otherwise, we call the algorithm given by \autoref{thm:findIrrelevant} to find an irrelevant vertex in time \irrTime, remove the outputted vertex from the graph, and return to the first step. As the first step can be applied at most $n$ times, the total running time is bounded as stated in \autoref{thm:main}. Thus, from now on, we focus only on the proof of \autoref{thm:findIrrelevant} for planar graph. 
A theorem for graphs of bounded genus is identical to  \autoref{thm:findIrrelevant}, except that we need to use the known single exponential bound on $h$ for graphs of Euler genus $g$~\cite{mazoit2013single}. This leads to \autoref{thm:main}. 

\paragraph{The Case of General Graphs for \TMC.} The function $\boundary$ to which we refer in the statements below is defined to be $16\delta^2$ in the proof of \autoref{thm:main1General} in Section \ref{sec:recursive}. In the context of \FindFolio\ (or \pFindFolio), for an instance $(G,\delta)$, we say that a vertex $v$ is {\em irrelevant} if the extended $\delta$-folio of $G$ is equal to the extended $\delta$-folio of $G\setminus v$. Towards the proof of \autoref{thm:maingen}, we will need to prove a stronger result as stated below.

\begin{theorem}\label{thm:main1General}
\FindFolio\ is solvable in time $f(\delta)\cdot n^3$ on instances $(G,\delta)$ such that $|R(G)|\leq \alpha(\delta)$, 
where 
$f(\delta)=2^{2^{\OO(r_1^{3})}} 2^{2^{2^{(r_2 \cdot 2^{\OO( \delta^{4})})} }}$, $r_2=h(2^{c(\delta)^4})$ and $r_1=h(2^{2^{(c\cdot \delta^{4})}r_2})$ for some constant $c$. Moreover, the algorithm outputs the following: for each graph $X$ on $R(G)$ and $H$ in the $\delta$-folio of $G\cup X$, a realization of $H$ as the topological minor in $G\cup X$.  
\end{theorem} 

Clearly, \autoref{thm:maingen} is a corollary of \autoref{thm:main1General}.
%
%
The main component in our proof of  \autoref{thm:main1General} is the following theorem.


\begin{theorem}
\label{thm:main1Flat}
There is a computable function $g$ and an algorithm that, given an instance $(G,\delta,t,w',s')$ of \fFindFolio\ such that  $|R(G)|\leq \boundary$ and $w'\geq g(\delta^{\star},t)$, and $w''\in\mathbb{N}$, finds a $w''\times w''$ flat wall within the input $w'\times w'$ flat wall such that the set of all vertices of the output inner flat wall is irrelevant. The algorithm 
runs in 
$ 2^{2^{\OO(((t+\delta^{\star})^2+r)\log (t+\delta^{\star}+r))}}(s')^{\OO(s')}(w'')^{\OO(w'')} n$ time 
where  $\delta^{\star}=\delta+\boundary$, $g(\delta^{\star},t)=(t+\delta^{\star}+r)^{\OO((t+\delta^{\star})^2+r)}w''$  and $r=h(\delta^{\star}+t)$. 
\end{theorem}

The following lemma says that the existence of a smaller flat wall is already good enough to prove that there {\em exists} an  irrelevant vertex.

\begin{lemma}\label{lem:mainexrele}
There is a computable function $g$ such that for  any  instance $(G,\delta,t,w,s)$ of \fFindFolio\ with  $|R(G)|\leq \boundary$ and $w\geq g(\delta,t)$, there is an irrelevant vertex in $G$, where 
$g(\delta,t)=2^{\OO((\delta+\boundary)^2)} t \cdot r$, 
and $r=h(\delta+\boundary+t)$. 
\end{lemma}

\section{Finding a Smaller Representative and Auxiliary Lemmas}\label{sec:irr}
%

Our algorithms use the technique of irrelevant vertices as well as recursive understanding (for \FindFolio). 
That is, at each stage  we will separate our graphs into two subgraphs, ``understand'' the behavior of one of these subgraphs, and then replace it with a so called representative that has the same ``behavior'' and is smaller.
To this end, let us first formally define the notion of a representative.


\begin{definition}[{\bf Representative}]
\label{def:rep}
Let $G$ and $F$ be two rooted graphs such that $R(G)=R(F)$ and $\delta\in {\mathbb N}$. 
We say that $F$ is a {\em $\delta$-representative of $G$} if $F$ and $G$ have the same extended $\delta$-folio.
\end{definition}


\begin{proposition}
\label{prop:no.ofgraphsinfolios}
Let $G$ be a  rooted graph and $\delta\in {\mathbb N}$. Then the number of distinct graphs (up to isomorphism) in the $\delta$-folio of $G$ is upper bounded by  $2^{\OO(\delta^2)}\cdot  \vert R(G)\vert^{\OO(\delta)}$. 
\end{proposition}

\begin{proposition}[Lemma 2.4~\cite{DBLP:conf/stoc/GroheKMW11}]
\label{prop:equi_graph_replacement}
Let $(G_1,G_2)$ be a separation of a rooted graph $G$ and let $V(G_1)\cap V(G_2)\subseteq R(G)$. Let $G_1'$ be a rooted graph compatible with $G_1$ such that $G_1$ and $G_1'$ have the same extended $\delta$-folio. Let $G'$ be the graph obtained by replacing $G_1$ with $G_1'$ in the separation $(G_1,G_2)$. Then both $G$ and $G'$ have the same extended $\delta$-folio.
\end{proposition} 


%
%

By \autoref{prop:equi_graph_replacement}, we easily relate the irrelevancy of a vertex with respect to a subgraph of the input graph to the input graph itself.

\begin{corollary}\label{lem:irrelInSubgraph}
Let $(G,\delta)$ be an instance of \FindFolio. Let $(G_1,G_2)$ be a separation of $G$ such that $V(G_1)\cap V(G_2)\subseteq R(G)$ and let $v$ be a vertex irrelevant for $(G_1,\delta)$. Then, $v$ is also irrelevant for $(G,\delta)$.
\end{corollary}

\begin{proof}
Since $v$ is irrelevant for $(G_1,\delta)$, we have that $G_1$ and $G_1\setminus v$ have the same extended $\delta$-folio. 
By \autoref{prop:equi_graph_replacement}, $G$ and $(G_1\setminus v)\cup G_2$ have the same extended $\delta$-folio. Since $G\setminus v = (G_1\setminus v)\cup G_2$, the proof is complete.
\end{proof}

We now define \emph{important separators}, which have played  
an important role in resolving some of the long standing open problems in the field of parameterized complexity. 

\begin{definition}[Important Separators, \cite{Marx06i}]
Let $G$ be a graph. For subsets $X, Y, S \subseteq V(G)$, the set of vertices reachable from $X\setminus S$ in $G-S$ is denoted
by $R_G(X, S)$. An \sep{X}{Y} $S$ {\em dominates} an \sep{X}{Y} $S'$ if  $\vert S \vert \leq \vert S'\vert$ 
and $R_G(X, S')\subset R_G(X, S)$. A subset $S$ is an {\em \imsep{X}{Y}} if it is minimal, and there is no \sep{X}{Y} $S'$ that dominates $S$. 
\end{definition}

\begin{proposition}
\label{prop:unimp}
Let $G$ be a graph and $X, Y \subseteq V (G)$.  Let $\lambda$ be the minimum size of an \sep{X}{Y} in $G$.  Then there is exactly one  \imsep{X}{Y} of size $\lambda$ in $G$. 
\end{proposition}

\begin{proposition}[\cite{ChenLL09,Marx06i,paramalgoCFKLMPPS}] 
\label{prop:impsepcout}
Let $G$ be a graph, $X, Y \subseteq V (G)$, and  $k\in {\mathbb N}$.  
Let ${\cal S}_k$ be the 
set of all \imseps{X}{Y} of size at most $k$.
Then $\vert {\cal S}_k\vert$ is upper bounded by $4^k$, 
and these separators can be enumerated in time $\OO(\vert {\cal S}_k\vert \cdot k^2 \cdot (n+m))$.
\end{proposition}

\begin{lemma}
\label{lem:supreach}
Let $G$ be a graph and $Z,V'\subseteq V(G)$ such that $G[V']$ is connected. Let $U=N_G(V')$. 
Let $X$ and $Y$ be two (inclusion) minimal \seps{Z}{U} such that $Y$ dominates $X$. Let $C_X$ and $C_Y$ be the 
unique components containing $V'$ in $G\setminus X$ and $G\setminus Y$, respectively. 
Let $A_X=G\setminus V(C_X)$, $B_X=G[X\cup V(C_X)]$, $A_Y=G\setminus V(C_Y)$, $B_Y=G[Y\cup V(C_Y)]$. 
Then, $V(A_X)\subset V(A_Y)$  (and hence $X\subseteq V(A_Y)$). 
\end{lemma}

\begin{proof}
Notice that $(A_X,B_X)$ and $(A_Y,B_Y)$ are two separations of $G$ such that $A_X\cap B_X=X$ and 
$A_Y\cap B_Y=Y$. 
 Fix a vertex $u\in V(A_X)$. We need to show that $u\in A_Y$. Suppose not. Then,
$u\in V(B_Y)\setminus V(A_Y)=V(C_Y)$. 
Since $C_Y$ is connected there is a path $P$ from $u$ to $V'$ in $C_Y$. 
Since $(A_X,B_X)$ is a separation, $u\in V(A_X)$, $V'\subseteq V(C_X)$, path $P$  contains a vertex $x\in X$. Since $X$ is a minimal separator there is a vertex $y\in R_{G}(Z,X)$ such that 
$x$ is adjacent to $y$.  Since $P$ is a path in $C_Y$ and $x\in V(P)$, $y\notin R_{G}(Z,Y)$. This contradicts the assumption that  $Y$ dominates $X$. Thus, we have proved that $V(A_X)\subseteq V(A_Y)$.  If $V(A_X)=V(A_Y)$ then $X=Y$, because $X$ and $Y$ are minimal separators.  Therefore, since $Y$ dominates $X$, we have that $V(A_X)\subset(A_Y)$. 
\end{proof}

We also need a result (\autoref{prop:RS13:6.1}) by Robertson and Seymour~\cite{RobertsonS95b}. Towards that we first need to recall definitions related to minors in a rooted graph from~\cite{RobertsonS95b}. 

\begin{definition}[\cite{RobertsonS95b}]
\label{def:minorfoilo}
Let $G$ and $H$ be two rooted graphs. A minor model  of $H$ in $G$ is function $\phi: V(H)\rightarrow 2^{V(G)}$ such that $(i)$ for all $h\in V(H)$, $G[\phi(h)]$ is a connected graph, $(ii)$ for all distinct $h,h'\in V(H)$, $\phi(h)\cap\phi(h')=\emptyset$, $(iii)$ for all $\{h,h'\}\in E(H)$, there exist $u\in \phi(h)$ and $v\in \phi(h')$ such that $\{u,v\}\in E(G)$ and $(iv)$ for all $u\in R(H)$, $u'\in \phi(u)$, where $u'\in R(G)$ and $\rho_G(u')=\rho_H(u)$. Then, we say that $H$ is a minor of $G$, and we further say that $H$ has {\em detail} $\leq \delta$ (where $\delta\in{\mathbb N}$) if $\vert E(H)\vert \leq \delta$ and $\vert V(H)\setminus R(H)\vert \leq \delta$. The {\em $\delta$-\mfolio} of $G$ is the set of all minors of $G$ with detail $\leq \delta$. 
\end{definition}

We remark that in~\cite{RobertsonS95b}, the term $\delta$-folio is used in place of $\delta$-\mfolio. We adopted this term as we used $\delta$-folio with respect to topological minors.

\begin{definition}
Let $G$ be a graph and $Z\subseteq V(G)$.  We say that the $\delta$-\mfolio\ of $G$ relative to $Z$ is generic if  the $\delta$-\mfolio\ of the rooted graph $G$ with $R(G)=Z$ contains every graph with $\vert Z\vert$ roots and with detail at most $\delta$. 
\end{definition}

We say that a vertex $v$ in a rooted graph $G$ is {\em irrelevant}  with respect to the $\delta$-\mfolio\ of $G$ if the $\delta$-\mfolio\ of $G$ is same as the $\delta$-\mfolio\ of $G\setminus v$.

\begin{proposition}[6.1~\cite{RobertsonS95b}]
\label{prop:RS13:6.1}
Let $G$ be a rooted graph and let $Z=R(G)$ with $\vert Z\vert \leq \xi$. Let $\delta' \in {\mathbb N}$, $t=\lfloor\frac{5}{2}\xi \rfloor+3\delta'+1$ and let $G_1,\ldots,G_t$ be mutually vertex-disjoint connected subgraphs such that for $1\leq i< j\leq t$ there is an edge of $G$ between $G_i$ and $G_j$. Let $(A,B)$ be a separation of $G$ such that 
\begin{itemize}
\setlength{\itemsep}{-2pt}
\item[$(i)$] $V(A)\cap V(G_i)=\emptyset$ for some $i\in [t]$,
\item[$(ii)$] $Z\subseteq V(A)$, 
\item[$(iii)$] subject to $(i)$ and $(ii)$, $(A,B)$ has minimum order, and 
\item[$(iv)$] subject to $(i)$, $(ii)$ and $(iii)$, $A$ is maximal.
\end{itemize}
Let $v\in V(B)\setminus V(A)$. Then $v$ is irrelevant to the $\delta'$-\mfolio\ of $G$. Moreover, the $\delta'$-\mfolio\ of $B\setminus v$ relative to $V(A)\cap V(B)$ is generic. 
\end{proposition}

\autoref{lem:supreach} and \autoref{prop:unimp} imply the following observation. 

\begin{observation}
Let $G$ be a rooted graph and let $Z=R(G)$ and $t\in {\mathbb N}$. Let $G_1,\ldots,G_t$ be mutually vertex-disjoint connected subgraphs such that for $1\leq i< j\leq t$ there is an edge of $G$ between $G_i$ and $G_j$. Let $(A,B)$ be a separation of $G$ satisfying the conditions $(i)-(iv)$ of \autoref{prop:RS13:6.1}. Then, $V(A)\cap V(B)$ is the unique \imsep{Z}{U_i} of minimum size, where $U_i=N_G(V(G_i))$. 
\end{observation}


\begin{proposition}{\rm \cite[2.3]{RobertsonS95b}}
\label{prop:RS13:2.3}
Let $(A,B)$ be a separation of a graph $G$, $Z\subseteq V(G)$, and $\delta\in {\mathbb N}$. If $v\in V(B)$ is irrelevant to the $\delta$-\mfolio\ of $B$ relative to $V(A)\cap V(B)$, then $v$ is irrelevant to the $\delta$-\mfolio\ of $G$ relative to $Z$. 
\end{proposition}

\begin{lemma}
\label{lem:pointtovalidsep}
Let $G$ be a rooted graph and let $Z=R(G)$ with $\vert Z\vert \leq \xi$. Let $\delta' \in {\mathbb N}$, $t=\lfloor\frac{5}{2}\xi \rfloor+3\delta'+1$ and let $G_1,\ldots,G_t$ be mutually vertex-disjoint connected subgraphs such that for $1\leq i< j\leq t$ there is an edge of $G$ between $G_i$ and $G_j$. For each $i\in [t]$, let $V_i=V(G_i)$,  $U_i=N_G(V_i)$, and $Z_i$ be the unique \imsep{Z}{U_i} of minimum size. For each $i\in [t]$, let $(A_i,B_i)$ be the separation of $G$ such that $Z_i=A_i\cap B_i$, $A_i=G\setminus V(C)$, and $B_i=G[Z_i\cup V(C)]$, where 
$C$ is the connected component containing $V_i$ in $G\setminus Z_i$. Then, for any $I\subseteq [t]$ of cardinality at least $\xi+1$, at least one separation among $\{(A_r,B_r) \colon r\in I\}$ satisfies the conditions $(i)-(iv)$ of \autoref{prop:RS13:6.1}. 
\end{lemma}

\begin{proof}
Let $(A,B)$ be a separation of $(G)$ satisfying  conditions $(i)-(iv)$ of \autoref{prop:RS13:6.1}. Let $i\in [t]$ be the index such that condition $(i)$ is true. That is $V(A)\cap V(G_i)=\emptyset$. Since $(Z,V(G))$ is a separation of $G$ satisfying conditions $(i)$ and $(ii)$, we have that $\vert V(A)\cap V(B)\vert \leq \xi$. Let $S=V(A)\cap V(B)$.  Notice that $S$ is a \sep{Z}{U_i}, where $U_i=N_G(V(G_i))$. Since $t>\vert I\vert$, there exists $j\in I$ 
such that $S\cap V(G_j)=\emptyset$. Also, since $G_i$ and $G_j$ are connected and there is an edge between $V(G_i)$ and $V(G_j)$, $S$  is a \sep{Z}{U_j} as well. Moreover, $S$ is a minimum size \sep{Z}{U_j}. Otherwise, we will contradict condition $(iii)$ of \autoref{prop:RS13:6.1} for the separation $(A,B)$. 
We claim that $S$ is the unique \imsep{Z}{U_j} of the minimum size. If $S$ is not an  \imsep{Z}{U_j}, then the unique  \sep{Z}{U_j} $Z_j$ of the minimum size  dominates $S$. Then,by  \autoref{lem:supreach}, we get a contradiction 
to the fact that $(A,B)$ satisfies condition $(iv)$ of \autoref{prop:RS13:6.1}. 
This implies that $(A_j,B_j)$ satisfies the conditions $(i)-(iv)$ of \autoref{prop:RS13:6.1}. This completes the proof of the lemma. 
\end{proof}

\begin{lemma}
\label{lem:folioafterdel}
Let $G$ be a rooted graph, $(A,B)$ be a separation of $G$. Let $H$ be a rooted graph in the $\delta$-folio of $G$ and $G_H$ be a realization of $H$ in $G$ witnessed by a pair of functions $(\phi,\psi)$. Let $Z=N_{G_H}[\phi(V(H))]$ and $R(G)\cup Z\subseteq V(A)$. Let $S\subseteq V(B)\setminus V(A)$ be such that $\delta'$-\mfolio\ of $B\setminus S$ relative to $V(A)\cap V(B)$ is generic where $\delta'=4\delta$. Then, $H$ belongs to the $\delta$-folio of $G\setminus S$. 
\end{lemma}

%

\begin{proof}
We begin by constructing a new rooted graph $G'$. The underlying graph of $G'$ is same as $G$, except that the set of roots 
will change. That is, the graph $G'$ is same as $G$, with $R(G')= R(G)\cup Z$. 
Let $\rho_{G'}$ be an arbitrary injective map from $R(G')$ to ${\mathbb N}$ such that $\rho_{G'}(v)=\rho_{G}(v)$ for all $v\in R(G)$. 
%
Now we prove that $H$ belongs to the $\delta$-folio of $G'\setminus S$ and this will imply that 
$H$ belongs the $\delta$-folio of $G\setminus S$. 
%
%
Towards that  we construct a new graph $H'$ on the vertex set $Z$.  There is an edge $\{u,v\}\in E(H')$ if and only if there is a path from $u$ to $v$ in $G_{H}$, with internal vertices from $V(G)\setminus Z$. That is, $H'$ is obtained by subdividing the edges $e$ of $H$ once, if $\varphi(e)$ is a path of length $2$ and by subdividing the edges $e$ of $H$ twice, if $\varphi(e)$ is a path of length strictly more than $2$. As a result we also consider that $V(H)\subseteq V(H')$. 
Since $|E(H)|+\isolated(H)\leq \delta$, the number of vertices as well as the number of edges in $H'$ is at most $4\delta=\delta'$. 
See \autoref{fig:exampleHprime} for an illustration. Notice that $Z\subseteq R(G')$. Moreover we set $R(H')=V(H')$ with a natural map $\rho_{H'}$, which is $\rho_{G'}$ restricted to the domain ${R(H')}$. That is $\rho_{H'}(v)=\rho_{G'}(v)$ for all $v\in Z$.  
For an ease of presentation, for any $u\in V(H')=R(H')$, we also use $u$ to denote the vertex $v$ in $G'\cup X$ such that $\rho_{G'}(v)=\rho_{H'}(u)$.
In other words, the graph $H'$ is  obtained from $G_{H}$ by contracting edges of $G_{H}$ which are incident with at least one vertex in $V(G_{H})\setminus Z$. This implies that the rooted graph $H'$ is a minor of $G'$ and it has detail $\leq \delta'$. Since the $\delta'$-\mfolio\ of $B\setminus S$ relative to $V(A)\cap V(B)$ is generic, by \autoref{prop:RS13:2.3} we have that $H'$ is a minor of $G'\setminus S$. That is,  there is minor model $\phi'$ of $H'$ in the rooted graph $G^{\star}=G'\setminus S$. Here $R(G^{\star})=R(G')$ and $\rho_{G^{\star}}=\rho_{G'}$. 
 That is, $\phi' \colon V(H')\rightarrow 2^{V(G^{\star})}$ such that $(i)$ for all $y\in V(H')$, $G^{\star}[\phi'(y)]$ is a connected graph, $(ii)$ for all distinct $y,y'\in V(H')$, $\phi'(y)\cap\phi'(y')=\emptyset$, $(iii)$ for all $\{y,y'\}\in E(H')$, there exist $u\in \phi'(y)$ and $v\in \phi'(y')$ such that $\{u,v\}\in E(G^{\star})$ and $(iv)$ for all $u\in R(H')=V(H')$, 
$\rho_{G^{\star}}^{-1}(\rho_{H'}(u))\in \phi'(u)$ (That is, the vertex $u$ in $Z$ will be in $\phi'(u)$). 
%

Since $H$ is a topological minor in $H'$, to prove that $H$ is also a topological 
minor in $G^{\star}$, it is enough to prove that $H'$ is a topological minor in $G^{\star}$. Towards that   we construct a pair of functions $(\phi_{H'},\varphi_{H'})$ as follows. For any $u\in V(H')$, $\phi_{H'}(u)=\rho_{G'}^{-1}(\rho_{H'}(u))=u$. (Notice that, here $u$ also belongs  to $ Z\subseteq V(G^{\star})$). For any $\{u,v\}\in E(H')$ with $u\in V(H)$, $\{u,v\}$ is also present in $G^{\star}$. So we set $\varphi_{H'}(\{u,v\})=u-v$ whenever $u\in V(H)$. Now consider an edge of the form $\{u',v'\}\in E(H')$, where $u',v'\in V(H')\setminus V(H)$. In this case both $u'$ 
 and $v'$ have degree exactly $2$ in $H'$ and there is a path $u-u'-v'-v$ in $H'$ with $u,v\in V(H)$. Notice that $\{u,u'\},\{v,v'\}\in E(G^{\star})$ and 
there is an edge $\{w,w'\}\in E(G^{\star})$ where $w\in \phi'(u')$ and $w'\in \phi'(v')$ (because $\phi'$ is a minor model of $H'$ in $G^{\star}$). 
This implies that there a path $P_{u'v'}$ from $u'$ to $v'$ in $G^{\star}$ using internal vertices from $\phi'(u')\cup \phi'(v')$. We set $\varphi_{H'}(\{u',v'\})=P_{u'v'}$. This completes the definition  $(\phi_{H'},\varphi_{H'})$ and it witnesses that $H'$ is a topological minor in $G^{\star}$. In turn, this completes the proof of the lemma. 
\end{proof}

\section{Graphs with Large Clique Minor}\label{sec:cliqueminor}


In this section we prove that if there is a large clique minor in the input graph of \FindFolio, then there is a small $\delta$-representative. After that, we prove that if there is a large clique minor, then we can actually find a $(\delta,k)$-irrelevant vertex by employing \autoref{thm:main1General}. We remark that the lemma about the $(\delta,k)$-irrelevant vertex is not used for \autoref{thm:main1General}. 
To handle the case considered in this section we need to define an additional variant of \FindFolio.
Specifically, we define the problem \cFindFolio\ analogous to \FindFolio. Here the input, along with an instance $(G,\delta)$ of \FindFolio, also consists of an integer $t\in\mathbb{N}$ and a function $\phi: V(K_t)\rightarrow 2^{V(G)}$ that witnesses that $K_t$ is a clique minor of $G$.

\begin{lemma}\label{lem:representativeClique}
There exists a computable function $h'$ such that for any instance $(G,\delta,t,\phi)$ of \cFindFolio\ where $|R(G)|\leq \boundary$ and $t\geq h'(\delta)$, there is a vertex $v\in V(G)$ such that $G\setminus v$ is a  $\delta$-representative of $G$ (with the same set of roots), where $h'(\delta)=2^{\OO((\delta+\boundary)^2)}$.
\end{lemma}

\begin{proof}
We begin by constructing a new rooted graph $G'$. The graph $G'$ is same as the graph $G$, with $R(G')\supseteq R(G)$. First, we add $R(G)$ to $R(G')$. Then, for each graph $X$ on $R(G)$ and  $H$ in the $\delta$-folio of $G\cup X$, we add a set $Z_{H,X}$ to $R(G')$, which is defined as follows. 
Let $G_{H,X}$ be an arbitrary realization of $H$ in $G\cup X$ and  $(\phi_H,\varphi_H)$ be the corresponding pair of functions  which witnesses that $H$ is a topological minor of $G$. The set $Z_{H,X}$ is equal to $N_{G_H}[\phi_H(V(H))]$. Let $\rho_{G'}$ be an arbitrary injective map from $R(G')$ to ${\mathbb N}$ such that $\rho_{G'}(v)=\rho_{G}(v)$ for all $v\in R(G)$. By \autoref{prop:no.ofgraphsinfolios}, and the fact that $\vert V(X)\vert \leq \vert R(G)\vert \leq \boundary$, we have that $\vert R(G')\vert$ is upper bounded by $2^{\OO((\delta+\boundary)^2)}\cdot  \vert R(G)\vert^{\OO(\delta)}=\alpha'(\delta)$, for some computable function $\alpha'$. 
Let $\xi=\alpha'(\delta)$, $\delta'=4\delta$ and $k'=\lfloor\frac{5}{2}\xi \rfloor+3\delta'+1=2^{\OO((\delta+\boundary)^2)}\cdot  (\boundary)^{\OO(\delta)}=2^{\OO((\delta+\boundary)^2)}$. We choose the function $h'$ such that $h'(\delta)=k'$.  We are given a minor model $\phi$ of $K_t$ in $G$. 
Notice that since $\vert R(G')\vert < t$, $(G'[R(G')], G'\setminus E(R(G')))$ is a separation of $G'$  satisfying conditions $(i)$ and $(ii)$ of \autoref{prop:RS13:6.1}. This implies that there exists a separation $(A,B)$ of $G'$ such that conditions $(i)$--$(iv)$ of \autoref{prop:RS13:6.1} hold. Moreover, $(A\cup X,B)$ is a separation of $G'\cup X$, for any graph $X$ on $R(G)$ such that  conditions $(i)$--$(iv)$ of \autoref{prop:RS13:6.1} hold (because $R(G')\subseteq A$). Then by \autoref{prop:RS13:6.1}, we know that there is a vertex $v\in V(B)\setminus V(A)$ such that $(a)$ the $\delta'$-minor folio of $B\setminus v$  relative to $V(A)\cap V(B)$ 
is generic (for any graph $X$ on $R(G)$). 

We claim that $G\setminus v$ is a $\delta$-representative of $G$ (with the same set of roots). Clearly $v\notin R(G')$ and hence $R(G)\subseteq V(G\setminus v)$. Towards proving that $G\setminus v$ is a $\delta$-representative of $G$, it is enough to prove that for each graph $X$ on $R(G)$  and $H$ in the $\delta$-folio of $G\cup X$, $H$ is also a topological minor in $(G\cup X)\setminus v$. Fix an  arbitrary graph $X$ and another arbitrary graph $H$ in the $\delta$-folio of $G\cup X$. Consider the function $(\phi_H,\varphi_H)$ and the subgraph $G_{H,X}$ which is the (above mentioned fixed) realization of $H$ in $G\cup X$. We know that $(b)$ $Z_{H,X}\subseteq R(G')\subseteq V(A)$.  Therefore, by \autoref{lem:folioafterdel} and statements $(a)$ and $(b)$, we conclude that $H$ belongs to the $\delta$-folio of $(G\cup X)\setminus v$. This completes the proof of the lemma. \end{proof} 

\begin{center}

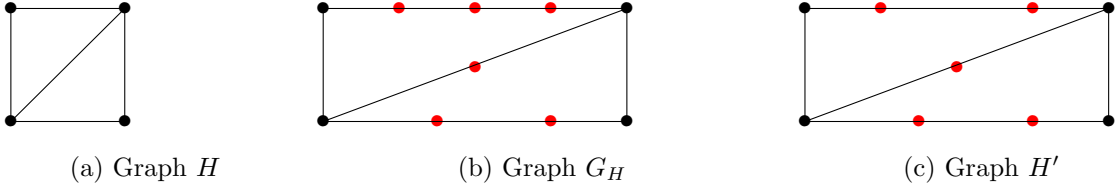
\begin{figure}
\begin{subfigure}{.25\linewidth}
\begin{tikzpicture}[scale=1]
\node[] (a1) at (0,0) {$\bullet$}; 
\node[] (a2) at (0,1.5) {$\bullet$}; 
\node[] (a3) at (1.5,0) {$\bullet$}; 
\node[] (a4) at (1.5,1.5) {$\bullet$};
\draw (0,0)--(0,1.5)--(1.5,1.5)--(1.5,0)--(0,0)--(1.5,1.5);
\end{tikzpicture}
\caption{Graph $H$}
\end{subfigure}
\begin{subfigure}{.39 \linewidth}
\begin{tikzpicture}[scale=1]
\node[] (a) at (0,0) {$\bullet$}; 
\node[] (a) at (0,1.5) {$\bullet$}; 
\node[] (a) at (4,0) {$\bullet$}; 
\node[] (a) at (4,1.5) {$\bullet$};
\node[red] (a) at (1,1.5) {$\bullet$};
\node[red] (a) at (2,1.5) {$\bullet$};
\node[red] (a) at (3,1.5) {$\bullet$};
\node[red] (a) at (1.5,0) {$\bullet$};
\node[red] (a) at (3,0) {$\bullet$};
\node[red] (a) at (2,0.72) {$\bullet$};
\draw (0,0)--(0,1.5)--(4,1.5)--(4,0)--(0,0)--(4,1.5);

\end{tikzpicture}
\caption{Graph $G_{H}$}
\end{subfigure}
\begin{subfigure}{.32 \linewidth}
\begin{tikzpicture}[scale=1]
\node[] (a) at (0,0) {$\bullet$}; 
\node[] (a) at (0,1.5) {$\bullet$}; 
\node[] (a) at (4,0) {$\bullet$}; 
\node[] (a) at (4,1.5) {$\bullet$};
\node[red] (a) at (1,1.5) {$\bullet$};
\node[red] (a) at (3,1.5) {$\bullet$};
\node[red] (a) at (1.5,0) {$\bullet$};
\node[red] (a) at (3,0) {$\bullet$};
\node[red] (a) at (2,0.72) {$\bullet$};
\draw (0,0)--(0,1.5)--(4,1.5)--(4,0)--(0,0)--(4,1.5);

\end{tikzpicture}
\caption{Graph $H'$}
\end{subfigure}
\caption{Illustration of the construction of $H'$ from $H$ and $G_{H}$.}
\label{fig:exampleHprime}
\end{figure}
\end{center}

Now we explain how to find a $(\delta,k)$-irrelevant vertex in the presence of a large clique minor. That is, we prove  
\autoref{thm:dellargeclique}. For convenience, we restate the theorem.

\dellargeclique*


\begin{proof}
Let $\xi=\vert R(G)\vert$. 
Towards the proof of the lemma, we need to design an algorithm to find a vertex $v$ such that for any graph $X$ on $R(G)$ and a vertex subset $S\subseteq V(G)$ of size at most $k$, the $\delta$-folio of $G'=(G\cup X)\setminus S$ (where $R(G')=R(G)\setminus S$), is equal to the the $\delta$-folio of $G' \setminus v$. Notice that the number of choices for $X$ is bounded by $2^{{\xi}^2}$, but the number of choices for $S$ is $\Omega(n^k)$. 
Towards the proof of the lemma, we  design a recursive algorithm ${\cal A}$ that marks at most ${f(\xi,\delta,k)}$ ``relevant'' graphs from $\{G_1,\ldots,G_t\}$. Since $t > {f(\xi,\delta,k)}$, there will be at least one graph $G_i$ that is not marked by the algorithm.
Then, we prove that any vertex in $V(G_i)$ is $(\delta,k)$-irrelevant. 


Our recursive algorithm ${\cal A}$ has the following specifications. 
\begin{itemize}
\setlength{\itemsep}{-2pt}
\item Input is a rooted graph $G^{\star}$, two integers $k,\delta\in {\mathbb N}$ and mutually vertex-disjoint connected subgraphs $G_1,\ldots,G_t$ such that for $1\leq i< j\leq t$ there is an edge of $G$ between $G_i$ and $G_j$ (i.e., $(G_1,\ldots,G_t)$ is a minor model of $K_t$ in $G$), where $t> f(\xi,\delta,k)$ and $\xi=\vert R(G^{\star})\vert$.   
\item Running time of ${\cal A}$ is $g(\xi,\delta,k)\cdot n^3$ for some computable function $g$. 
\item It marks at most ${f(\xi,\delta,k)}$ graphs from $\{G_1,\ldots, G_t\}$ with the following property: for any vertex $v$ in an unmarked graph $G_i$, for any graph $X$ on $R(G^{\star})$, and for any vertex subset $S\subseteq V(G^{\star})$ of size at most $k$, the $\delta$-folio of 
$G'=(G^{\star}\cup X)\setminus S$ (where $R(G')=R(G^{\star})\setminus S$) is same as the $\delta$-folio of 
$G'\setminus v$. That is, the extended $\delta$-folio of $G'$ and $G'\setminus v$ are same. 
\end{itemize}

Finally, 
we use algorithm ${\cal A}$ on  $(G,k,\delta,(G_j)_{j\in [t]})$ to find a $(\delta,k)$-irrelevant vertex.

\paragraph*{Description of ${\cal A}$: } Let $c$ be a non-negative integer constant such that the number of distinct graphs in the $\delta$-folio of $G^{\star}$ is at  most $2^{c\delta^2}\cdot \vert R(G^{\star})\vert ^{c\delta}$ (see \autoref{prop:no.ofgraphsinfolios}). Let $\beta \colon {\mathbb N}\times {\mathbb N} \mapsto {\mathbb N}$ be a function defined as $\beta (\delta,\xi)=2^{\xi^2}(2^{c\delta^2}\cdot \xi ^{c\delta})4\delta$. Let $\eta=\beta(\delta,\xi)$. 
%
We choose a monotonically increasing function $f$ satisfying the following recurrence relation. 
\begin{eqnarray*}  
f(\xi,\delta,k)&=& \eta(2\eta+k+2)+ 2\cdot 4^{\eta+k} (2\eta+k+1)^{k+3}f(\eta+k,\delta,k-1)\\
f(\xi,\delta,0)&=& \eta(2\eta+k+2)
\end{eqnarray*}
For each $i\in [t]$, let $V_i=V(G_i)$ and $U_i=N_G(V_i)$. The algorithm has following steps. 
\begin{enumerate}[label= {\bf Step \arabic*}, align=left]
\item 
\label{step111}
Using \autoref{thm:main1General} we computes the extended $\delta$-folio of $G^{\star}$\footnote{Even though in \autoref{thm:main1General}, $\vert R(G^{\star})\vert \leq 16\delta^2$, we can use it by choosing $\delta^{\star}\geq \delta$ such that $\vert R(G)\vert \leq 16 (\delta^{\star})^2$ and get the extended $\delta^{\star}$-folio which is a super set of the extended $\delta$-folio of $G^{\star}$.}. Using \autoref{thm:main1General}, we also get the following: for each graph $X$ on $R(G^{\star})$ and $H$ in the $\delta$-folio of $G^{\star}\cup X$, a realization $G_H$ of $H$ as a topological minor in $G^{\star}\cup X$, witnessed by $(\phi_{X,H},\psi_{X,H})$. Let 
$Z_0=\bigcup_{X,H}N_{G_H}[\phi_{X,H}(V(H))]$ where the union is over over graphs $X$ in $\allgraphs(R(G^{\star}))$  
and $H$ in the $\delta\mbox{-folio of }G^{\star}\cup X$. 
Observe that $R(G^{\star})\subseteq Z_0$. 
By \autoref{prop:no.ofgraphsinfolios} 
and the fact that $\vert \allgraphs(R(G^{\star}))\vert \leq 2^{\xi^2}$, we have that  $\vert Z_0\vert$ is upper bounded by $2^{\xi^2}(2^{c\delta^2}\cdot \xi ^{c\delta})4\delta=\beta(\delta,\xi)=\eta$. 
\item For each $i\in [t]$, if $V_i\cap Z_0\neq \emptyset$, then we mark $G_i$. 
\label{stepfirstmark}
\item  For each $i\in [2\eta+k+1]$ do the following. Using \autoref{prop:unimp} compute the unique  \imsep{Z_0}{U_i}  $Z_i$ of minimum size in $G$. 
 Since $\vert Z_0\vert \leq \eta$, $\vert Z_i\vert \leq \eta$ for all $i\in [2\eta+k+1]$. 
For each $j\in [t]$, if $V_j\cap Z_i \neq \emptyset$, then we mark $G_j$. 
Let $I_i=\{ j \in [t] \colon V_j\cap Z_i=\emptyset\}$.  
\label{stepuniqueimpsepmark}
\item  If  $k=0$, we stop. Otherwise go to the next step. \label{step:stop}
\item For each $i\in [2\eta+k+1]$, let $(A_i,B_i)$ be the separation of $G^{\star}$ such that $Z_i=V(A_i)\cap V(B_i)$, 
$A_i=G^{\star}\setminus V(C)$, 
and $B_i=G^{\star}[Z_i\cup V(C)]$, where  $C$ is the connected component  containing $V_i$ in $G^{\star}\setminus Z_i$.  Now, for any $i\in [2\eta+k+1]$, any subset $D\subseteq Z_i$ of size at most $k$,  and any non-negative integer $k'<k$, recursively run algorithm ${\cal A}$ on $(B_i\setminus D,k', \delta ,(G_j)_{j\in I_i})$, where $R(B_i\setminus D)=Z_i\setminus D$. 
\label{step:rec1}
\item For each $i,j\in [2\eta+k+1]$, using \autoref{prop:impsepcout}, we compute the set ${\cal Q}_{i,j}$ of all  \imseps{Z_i}{U_j} of size at most $\vert Z_i\vert+k$. 
\label{step:imp}
\item For each $i,j\in [2\eta+k+1]$ and each $Q\in {\cal Q}_{i,j}$ we do the following. Let $(A_Q,B_Q)$ 
be the separation of $G^{\star}$ such that $Q=V(A_Q)\cap V(B_Q)$, $A_Q=G^{\star}\setminus V(C)$, and $B_Q=G^{\star}[Q\cup V(C)]$, where $C$ is the connected component containing $V_j$ in $G^{\star}\setminus Q$.  
Let $J_Q=\{q\in [t] \colon V_q\cap Q=\emptyset\}$. 
For  any subset $D\subseteq Q$ of size at most $k$,  and any non-negative integer $k'<k$, recursively run algorithm ${\cal A}$ on $(B_Q\setminus D,k',\delta,(G_q)_{q\in J_Q})$, where $R(B_Q\setminus D)=Q\setminus D$.
\label{step:rec2}
\end{enumerate}

\paragraph*{Correctness of ${\cal A}$: }
Let ${\cal G}=\{G_1,\ldots,G_t\}$.
Notice that when $k=0$, algorithm ${\cal A}$ will not make any recursive calls. Moreover, when $k>0$,  in each recursive call the value of $k$ drops by at least one. This implies that algorithm ${\cal A}$ will terminate. Next we prove that algorithm ${\cal A}$ make recursive calls with valid inputs. That is, for each recursive call on input 
$(G^{\star}_1,k',\delta, (G_j')_{j\in [t']})$, $t'>f(\xi',\delta,k')$ where $\xi'=R(G^{\star}_1)$.  First, consider the recursive calls made in \ref{step:rec1}. Fix an index $i\in [2\eta+k+1]$, a subset $D\subseteq Z_i$, and a non-negative integer $k'<k$. Since $\vert Z_i\vert \leq \eta$ and graphs in ${\cal G}$ are pairwise vertex-disjoint, $\vert I_i\vert \geq t-\eta$. Since $t>f(\xi,\delta,k)=\eta(2\eta+k+2)+ 2\cdot 4^{\eta+k} (2\eta+k+1)^{k+3}f(\eta+k,\delta,k-1)$, we have that 
$\vert I_i\vert > f(\eta,\delta,k')$. Moreover, $\vert R(B_i\setminus D)\vert \leq \eta$. Since $G^{\star}[\bigcup_{j\in I_i} V(G_j)]$ is a connected graph, $i\in I_i$, and $V(G_j)\cap Z_i =\emptyset$, $G_j$ is a subgraph of $B_i\setminus D$ for all $j\in I_i$. This implies that $(B_i\setminus D, k',\delta,(G_j)_{j\in I_i})$ is a valid input for algorithm ${\cal A}$. 
Hence, all recursive calls in \ref{step:rec1} are on valid inputs. 

Now consider the recursive calls in \ref{step:rec2}.  Fix indices $i,j\in [2\eta+k+1]$, a separator $Q\in {\cal Q}_{i,j}$, a subset $D\subseteq Q$, and a non-negative integer $k'<k$. We know that $\vert Q\vert \leq \vert Z_i\vert +k \leq \eta+k$. This implies that $\vert J_Q\vert \geq t-(\eta+k)> f(\eta+k,\delta,k')$. By arguments similar to that in the above paragraph, we can show that $G_q$ is a subgraph of $B_Q\setminus D$ for any $q\in J_Q$. Moreover, 
$\vert R(B_Q\setminus D)\vert \leq \eta+k$. 
This implies that $(B_Q\setminus D, k',\delta,(G_q)_{q\in J_Q})$ is a valid input for algorithm ${\cal A}$. 
Hence, all recursive calls in \ref{step:rec2} are on valid inputs.

\begin{claim}
Algorithm ${\cal A}$ marks at most ${f(\xi,\delta,k)}$ graphs from ${\cal G}$.
\end{claim}
\begin{proof}
We prove by induction on $k$ that algorithm ${\cal A}$ marks at most ${f(\xi,\delta,k)}$ graphs from ${\cal G}$.  The base case is when $k=0$. 
Since $\vert Z_0\vert \leq \eta$ and the graphs $G_1,\ldots,G_t$ are vertex-disjoint, in \ref{stepfirstmark}, we mark at most $\eta$ graphs from ${\cal G}$. 
Since $\vert Z_i\vert \leq \eta$ for all $i\in [2\eta+k+1]$ and the graphs $G_1,\ldots,G_t$ are vertex-disjoint, in \ref{stepuniqueimpsepmark}, we mark at most $\eta(2\eta+k+1)$ graphs from ${\cal G}$. Thus, when $k=0$, 
the number of marked graphs is at most $\eta(2\eta+k+2) = {f(\xi,\delta,0)}$. 

Now consider the inductive step. That is, when $k>0$. By arguments similar to that in the base case, we have that the number of graphs marked in \ref{stepfirstmark} and \ref{stepuniqueimpsepmark} is at most $\eta(2\eta+k+2)$. 
Since $\vert Z_i\vert \leq \eta$, the number of times we recursively call algorithm ${\cal A}$ in \ref{step:rec1} is  
at most $(2\eta+k+1)\eta^k k$. By induction hypothesis, each of these recursive calls marks at most $f(\eta,\delta, k-1)$ graphs in ${\cal G}$. Thus in \ref{step:rec1}, the number of graphs  marked is at most 
$(2\eta+k+1)\eta^k k \cdot f(\eta,\delta,k-1)$ which is upper bounded by $(2\eta+k+1)^{k+2}f(\eta+k,\delta,k-1)$. 
%
%
Next, we consider the number of graphs marked in \ref{step:rec2}.  For each $i,j\in [2\eta+k+1]$ and each $Q\in {\cal Q}_{i,j}$, $\vert Q\vert \leq  \vert Z_i\vert +k \leq \eta+k$. For any $i,j\in [2\eta+k+1]$,  any $Q\in {\cal Q}_{i,j}$ and any non-negative integer $k'<k$, we recursively call ${\cal A}$, where the number of roots in the input graph of the recursive call (which is an induced subgraph of $G^{\star}$) is at most $\eta+k$ and the model sets of clique minor is a subset of $\{V(G_1),\ldots, V(G_t)\}$. Since $k'<k$, by induction hypothesis, the number of graphs in ${\cal G}$ marked by each of these recursive calls is at most $f(\eta+k,\delta,k-1)$. 
By \autoref{prop:impsepcout}, 
we have that $\vert {\cal Q}_{i,j}\vert\leq 4^{\eta+k}$ for all $i,j\in [2\eta+k+1]$.
 This implies that the number of  recursive calls made in \ref{step:rec2} is at most $(2\eta+k+1)^2\cdot 4^{\eta+k} \cdot k\cdot (\eta+k)^k$ which is upper bounded by $4^{\eta+k} (2\eta+k+1)^{k+3}$. Thus, the number of graphs marked in 
 \ref{step:rec2}
  is at most $4^{\eta+k} (2\eta+k+1)^{k+3}f(\eta+k,\delta,k-1)$. By summing the number of marked graphs in all steps, we get that the total number of graphs marked in ${\cal G}$ is at most $f(\xi,\delta,k)$. 
\end{proof}

Next we prove that any vertex in an unmarked graph in ${\cal G}$ is $(\delta,k)$-irrelevant. 
Towards that it is enough to prove the following statement. For any unmarked graph $G_r\in {\cal G}$, any vertex $v\in V(G_r)$, and any vertex subset $S\subseteq V(G^{\star})$ of size at most $k$, the extended $\delta$-folio of $G^{\star}\setminus S$ is same as the extended $\delta$-folio of $(G^{\star}\setminus S) \setminus v$. We prove this statement by induction on $k$.

Before proving the statement, we first prove some auxiliary claims. 
Recall the separators $Z_1,\ldots, Z_{2\eta+k+1}$ constructed 
in \ref{stepuniqueimpsepmark} of the algorithm. Here, $Z_i$ is the unique \imsep{Z_0}{U_i} of minimum size for all $i\in [2\eta+k+1]$ 
and $\vert Z_0\vert \leq \eta$. 
Also, consider the separations $(A_i,B_i)$, $i\in [2\eta+k+1]$, constructed in \ref{step:rec1}. 
Notice that $Z_0 \subseteq V(A_i)$, $Z_i=V(A_i)\cap V(B_i)$, and $V(G_i) \subseteq V(B_i)\setminus V(A_i)$ for all  $i\in [2\eta+k+1]$. 

\begin{claim}
\label{claim:onegoodsep}
Let $I\subseteq [2\eta+k+1]$ be subset of size at least $\eta+1$. Then, 
there exists $i\in I$ such that $(A_i,B_i)$ is a separation of $G^{\star}$ satisfying the conditions $(i)-(iv)$ of \autoref{prop:RS13:6.1}.
\end{claim}
\begin{proof}
Notice that 
$t\geq \lfloor \frac{5}{2}\rfloor \eta + 3\delta'+1$, where $\delta'=4\delta$. Then, by \autoref{lem:pointtovalidsep}, there exists $i\in I$ such that  $(A_i,B_i)$ is a separation of $G^{\star}$ satisfying the conditions $(i)-(iv)$ of \autoref{prop:RS13:6.1}. 
\end{proof}

\begin{claim}
\label{claim:HinGminusS}
Let $i\in [2\eta+k+1]$ and $S$ be a subset of $V(B_i)\setminus V(A_i)$  such that  the $\delta'$-\mfolio\ of $B_i\setminus S$  relative to $V(A_i)\cap V(B_i)$ is generic. Then, for any graph $X$ on $R(G^{\star})$ and $H$ in the $\delta$-folio of $G'=G^{\star}\cup X$, $H$ also belongs to the $\delta$-folio of $G' \setminus S$. 
\end{claim}

\begin{proof}
Fix a graph $X$ on $R(G^{\star})$ and $H$ in the $\delta$-folio of $G'=G^{\star}\cup X$. 
In \ref{step111} we computed a pair of functions $(\phi_{X,H},\psi_{X,H})$ that witnesses the fact that $H$ belongs to the $\delta$-folio of $G'$ and it is realized by the subgraph $G_H$ of $G'$. Moreover, by the construction of $Z_0$, we have that $(i)$ $N_{G_H}(\phi_{X,H}(V(H))) \subseteq Z_0$. 
Let  $A_i'=G'[V(A_i)]$ and $B_i'=G'[V(B_i)]$.  
Since $R(G^{\star})\subseteq V(A_i)$, $(A_i',B_i')$ is a separation of $G'$.  Since $B_i'\setminus S$ is a supergraph 
of $B_i\setminus S$,  the $\delta'$-\mfolio\ of $B_i\setminus S$  relative to $V(A_i)\cap V(B_i)$ is generic, and 
$V(A_i')\cap V(B_i')=V(A_i)\cap V(B_i)$, we have that the $(ii)$ $\delta'$-\mfolio\ of $B_i'\setminus S$  relative to $V(A'_i)\cap V(B'_i)$ is generic. Since $S\subseteq V(B_i)\setminus V(A_i)$, $(iii)$ $S\subseteq V(B'_i)\setminus V(A'_i)$. By statements $(i)-(iii)$ and \autoref{lem:folioafterdel} we conclude that $H$ belongs to the $\delta$-folio 
of $G'\setminus S$. 
\end{proof}

Now, using induction on $k$ we prove that for any unmarked graph $G_r\in {\cal G}$ and any vertex  $v\in V(G_r)$ and  any vertex subset $S\subseteq V(G^{\star})$ of size at most $k$, 
the extended $\delta$-folio of $G^{\star}\setminus S$ is same as the extended $\delta$-folio of $(G^{\star}\setminus S) \setminus v$. 
Let $\widehat{\cal G}$ be the set of all graphs from ${\cal G}$ that are unmarked by 
${\cal A}$  on input $(G^{\star},k,\delta,(G_j)_{j\in [t]})$. 

Consider the base case, i.e., when $k=0$. As mentioned before, in this case algorithm ${\cal A}$ will not make any recursive calls. That is, algorithm ${\cal A}$ will stop in \ref{step:stop}. 
Here, we need to show that for any unmarked graph $G_r\in {\cal G}$ and a vertex  $v\in V(G_r)$, 
the extended $\delta$-folio of $G^{\star}$ is same as the extended $\delta$-folio of $G^{\star} \setminus v$. 
By \autoref{claim:onegoodsep}, there exists $i\in [\eta+1]$ such that $(A_i,B_i)$ is a separation of $G^{\star}$ satisfying the conditions $(i)-(iv)$ of \autoref{prop:RS13:6.1}.  
Fix an integer $i\in [\eta+1]$ 
such that $(A_i,B_i)$ satisfies conditions $(i)-(iv)$ of \autoref{prop:RS13:6.1}. 
Since any graph $G_r$ in $\widehat{G}$ is not marked, $V(G_r)\subseteq V(B_i)\setminus V(A_i)$. 
Then, by \autoref{prop:RS13:6.1}, for any $G_r\in \widehat{G}$ and $v\in V(G_r)\subseteq V(B_i)\setminus V(A_i)$
the $\delta'$-\mfolio\ of $B_i\setminus v$ relative to $Z_i$ is generic, where $\delta'=4\delta$. 
Then, by \autoref{claim:HinGminusS} the extended $\delta$-folio of $G^{\star}$ is 
same as the the extended $\delta$-folio of $G^{\star}\setminus v$ for any unmarked $G_r\in {\cal G}$ and $v\in V(G_r)$.



Next consider the induction step for $k>0$. Fix an arbitrary vertex subset $S\subseteq V(G^{\star})$ of size at most $k$. We have the following three cases.

\paragraph*{Case 1: $V(A_i)\cap S\neq \emptyset$ for some $i\in [2\eta+k+1]$.} 
Recall that $Z_i=V(A_i)\cap V(B_i)$. Let $S_0=Z_i\cap S$, $S_1=S\setminus V(B_i) $, and $S_2=S\setminus V(A_i)$. 
Since $V(A_i)\cap S\neq \emptyset$, we have that $\vert S_2\vert<k$. 
Now consider the the recursive call made by the algorithm in \ref{step:rec1} for $D=S_0$ and $k'=\vert S_2\vert$.  
Let ${\cal G}'$ be the set of graphs from $\{G_j \colon j\in I_i\}$ that are unmarked  by the 
recursive call ${\cal} (B_i\setminus D,k',\delta,(G_j)_{j\in I_i})$.  Notice that $\widehat{\cal G}\subseteq {\cal G}'$ because a graph is unmarked if it is unmarked in every recursive calls. 
By induction hypothesis we know that for any  graph $G_r\in {\cal G}'$,  
any vertex  $v\in V(G_r)$, and  any vertex subset $S'\subseteq V(B_i\setminus S_0)$ of size at most $k'$, 
the extended $\delta$-folio of $B_i\setminus (S_0\cup S')$ is same as the extended $\delta$-folio of $(B_i\setminus (S_0\cup S')) \setminus v$. 
This implies that $(i)$ for any  graph $G_r\in {\cal G}'$ and  any vertex  $v\in V(G_r)$, 
 the extended $\delta$-folio of $B_i\setminus (S_0\cup S_2)$ is same as the extended $\delta$-folio of $(B_i\setminus (S_0\cup S_2)) \setminus v$. Let $A=A_i\setminus S$ and $B=B_i\setminus S$. Notice that 
$B=B_i\setminus (S_0\cup S_2)$ and $(A,B)$ is a separation of $G^{\star}\setminus S$.  Because of statement $(i)$ 
and \autoref{prop:equi_graph_replacement}, the extended $\delta$-folio of $G^{\star}\setminus S$ is 
same as the the extended $\delta$-folio of $(G^{\star}\setminus S)\setminus v$ for any $G_r\in {\cal G}'$ and any $v\in V(G_r)$. Therefore, since $\widehat{\cal G}\subseteq {\cal G}'$, the extended $\delta$-folio of $(G^{\star}\setminus S)\setminus v$ for any $G_r\in \widehat{\cal G}$ and any $v\in V(G_r)$ is same as the extended $\delta$-folio of $(G^{\star}\setminus S$.


\paragraph*{Case 2: There exist $i,j\in [2\eta+k+1]$ and $Q\in {\cal Q}_{i,j}$ such that $V(A_Q)\cap S\neq \emptyset$.}
The argument for this case is identical to the argument for Case 1 and hence we omit it.

\paragraph*{Case 3: Case 1 and 2 are not applicable.} 

Let $I=\{i\in [\eta+k+1] \colon S\cap V_i=\emptyset\}$. 
Since $\vert S\vert \leq k$, $\vert I\vert \geq 1+\eta$. 
By \autoref{claim:onegoodsep} there exists $i\in I$ such that $(A_i,B_i)$ is a separation of $G^{\star}$ satisfying the conditions $(i)-(iv)$ of \autoref{prop:RS13:6.1}. Fix an integer $i\in I$ 
such that $(A_i,B_i)$ satisfies conditions $(i)-(iv)$ of \autoref{prop:RS13:6.1}. Since Case 1 is not applicable 
we have that $S\subseteq V(B_i)\setminus V(A_i)$. 
Let $J=\{j\in [2\eta+k+1] \colon V_j\cap Z_i=\emptyset \mbox{ and } V_j\cap S=\emptyset\}$. 
Since $\vert Z_i\vert \leq \eta$, $\vert S\vert \leq k$ and $\{V_1,\ldots,V_{2\eta+k+1}\}$ are pairwise disjoint, $\vert J\vert \geq [\eta+1]$. 
Since $G^{\star}[V_i\cup V_j]$ is connected and $(V_i\cup V_j)\cap Z_i=\emptyset$ for all $j\in J$, 
$Z_i$ is a \sep{Z_0}{U_j}  in $G^{\star}$. In fact,  $(a)$
$Z_i$ is a \sep{Z_0}{U_j}  in $G^{\star}$ of minimum cardinality for any $j\in J$, because 
$(A_i,B_i)$ satisfies conditions $(i)-(iv)$ of \autoref{prop:RS13:6.1}.
Since $Z_i$ is a \sep{Z_0}{U_i}  in $G^{\star}$ we have that $Z_i\cap V_i=\emptyset$, and hence $i\in J$. 


Now we have two subcases as follows. In the first subcase there exists $j\in J$  such that the size of a minimum 
\sep{Z_i}{U_j} is strictly less than $\vert Z_i\vert$ in $G^{\star}\setminus S$. Let $Y$ be a minimum size 
\sep{Z_i}{U_j} in $G^{\star}\setminus S$. Then, $(b)$ $Y'=Y\cup S$ is a \sep{Z_i}{U_j} in $G^{\star}$. 
Since $(A_i,B_i)$ satisfies conditions $(i)-(iv)$ of \autoref{prop:RS13:6.1}, $(c)$
$Y$ is a not a  \sep{Z_i}{U_j} in $G^{\star}$. Statements $(b)$ and $(c)$ implies that 
there is a inclusion minimal \sep{Z_i}{U_j}  $W\subseteq Y'$ in $G^{\star}$  such that   $S\cap W\neq \emptyset$. 
Moreover $\vert W\vert < \vert Z_i\vert+k$. Then, there is an  \imsep{Z_i}{U_j} $Q\in {\cal Q}_{i,j}$ that dominates 
$W$.  Thus, by \autoref{lem:supreach}, $W\subseteq V(A_Q)$ (recall the construction $(A_Q,B_Q)$ in \ref{step:rec2}). This implies that $V(A_Q)\cap S\neq \emptyset$.   
This contradicts the assumption that Case 2 is not applicable. 

Next we consider the subcase that for all $j\in J$, the size of a minimum 
\sep{Z_i}{U_j} in $G^{\star}\setminus S$ is  $\vert Z_i\vert$. Then, notice that $Z_i$ 
is a minimum size \sep{Z_i}{U_j} in $G^{\star}\setminus S$ for all $j\in J$. 
Consider the separation $(A_i,B_i\setminus S)$ of $G^{\star}\setminus S$. 
Let $L=\{j'\in [t] \colon V_{j'}\cap S=\emptyset\}$. 
Notice that since $t$ is large enough, $\vert L\vert \geq \lfloor \frac{5}{2}\eta \rfloor +\delta'+1$ and 
$(G_{j'})_{j'\in L}$ form a minor model of $K_{\vert L\vert}$ in $G^{\star}\setminus S$. 
Notice that $J\subseteq L$. 
In the following claim we prove that $(A_i,B_i\setminus S)$ satisfies conditions $(i)-(iv)$ of \autoref{prop:RS13:6.1}. 

\begin{claim}
\label{claim:genafterdel}
Let $G'=G^{\star}\setminus S$, $R(G')=Z_0$. Then, 
$(A_i,B_i\setminus S)$ satisfies conditions $(i)-(iv)$ of \autoref{prop:RS13:6.1}.
\end{claim}

\begin{proof}
First we prove that for any $j\in J$, there are $\vert Z_i\vert$ vertex-disjoint paths from $Z_0$ to $U_j$ in $G'$. By the assumption of second subcase, there are $\vert Z_i\vert$ vertex-disjoint paths from $Z_i$ to $U_j$ 
in $G'$ for any $j\in J$. This follows from Menger's  theorem and the fact that the size of a minimum  \sep{Z_i}{U_j} in $G'$ is  $\vert Z_i\vert$ for all $j\in J$. 
Thus, since $(A_i,B_i\setminus S)$ is a separation of $G'$ with $V(A_i)\cap V(B_i\setminus S)=Z_i$, 
and $U_j\in B_i\setminus S$, there are $\vert Z_i\vert$ vertex-disjoint paths from $Z_i$ to $U_j$ in $B_i\setminus S$.  
Since $Z_i$ is a   minimum  \sep{Z_0}{U_j} in $G^{\star}$ (because $(A_i,B_i)$ is a separation satisfying conditions $(i)-(iv)$ of \autoref{prop:RS13:6.1} and $V_j\subseteq V(B_i)\setminus V(A_i)$), and $S\subseteq V(B_i\setminus V(A_i)$, 
there are $\vert Z_i\vert$ vertex-disjoint paths from $Z_0$ to $Z_i$ in $A_i$. Combining both the arguments, we get that  
for any $j\in J$, there are $\vert Z_i\vert$ vertex-disjoint paths from $Z_0$ to $U_j$ in $G'$ for all $j\in J$.

Suppose there is a separation $(A,B)$  of $G'$ such that $\vert A\cap B\vert< \vert Z_i\vert$ and satisfies 
conditions $(i)-(iv)$ of \autoref{prop:RS13:6.1}. 
Then, by \autoref{lem:pointtovalidsep}, there is a separation $(A_j',B_j')$ of $G'$ such that $V(A_j')\cap V(B_j')$ is the unique \imsep{Z_0}{U_j} in $G'$ for some $j\in J$ and $\vert V(A_j')\cap V(B_j')\vert <\vert Z_i\vert$.   
This contradicts the fact that there are $\vert Z_i\vert$ vertex-disjoint paths from $Z_0$ to $U_j$ in $G'$. Therefore $(A_i,B_i\setminus S)$ satisfies conditions $(i)-(iii)$ of \autoref{prop:RS13:6.1}. Now we prove that $(A_i,B_i\setminus S)$ satisfies condition $(iv)$ as well. For the sake of contradiction suppose $(A_i,B_i\setminus S)$ does not satisfy condition $(iv)$. This implies that there is a \sep{Z_0}{U_{\ell}} $Z$ that dominates $Z_i=V(A_i)\cap V(B_i\setminus S)$ and $\vert Z\vert=\vert Z_i\vert$ for some $\ell\in L$. 
Also, since $\vert Z\vert \leq \eta$ and $\vert J\vert \geq 1+\eta$, there exists $j\in J$ such that $Z$ is also 
a  \sep{Z_i}{U_{j}} in $G'$. This implies that $Z'=Z\cup S$ is a \sep{Z_i}{U_{j}} in $G^{\star}$.  
We claim that $Z$ is not a \sep{Z_i}{U_{j}} in $G^{\star}$. 
For the sake of contradiction, suppose $Z$ is a \sep{Z_i}{U_{j}} in $G^{\star}$. 
Let $A_{Z}=G^{\star}\setminus V(C)$ and $B_{Z}=G^{\star} [Z\cup V(C)]$, where $C$ is the connected component in $G^{\star}\setminus Z$ containing $V_j$.  
Notice that $(A_{Z},B_{Z})$ is a \sep{Z_0}{U_j} 
in $G^{\star}$ and $V(A_i)\subset V(A_{Z})$, because $Z$ dominates $Z_i$ (see \autoref{lem:supreach}).  Thus, if $Z$ is a \sep{Z_i}{U_{j}} in $G^{\star}$, then it contradicts the assumption that $(A_i,B_i)$ satisfies condition $(iv)$ of \autoref{prop:RS13:6.1}. 
Therefore, there is an inclusion minimal \sep{Z_i}{U_j} $W$ such that $W\subseteq Z'$ and $W\cap S\neq \emptyset$. This implies that there is a separator $Q\in {\cal Q}_{i,j}$ such that $S\cap V(A_Q)\neq \emptyset$. This is a contradiction to the assumption that Case $2$ is not applicable. This completes the proof of the claim. 
\end{proof}

For any graph $G_r\in {\cal G}$, if $V_r\cap V(A_i)\neq \emptyset$, then $V_r\cap Z_i\neq \emptyset$ (because $G[V_i\cup V_r]$ is connected and $V_i\subseteq V(B_i)\setminus V(A_i)$), and we mark 
$G_r$ in \ref{stepuniqueimpsepmark}. That is, for any unmarked graph $G_r$, $V_r\subseteq V(B_i)\setminus V(A_i)$. 
Therefore, by \autoref{claim:genafterdel} and \autoref{prop:RS13:6.1}, for any unmarked graph $G_r\in {\cal G}$ 
and any vertex $v\in V(G_r)$, the $\delta'$-\mfolio\ of $(B_i\setminus S)\setminus v$ relative to $Z_i$ is generic. 
Then, by \autoref{claim:HinGminusS} the extended $\delta$-folio of $G^{\star}\setminus S$ is 
same as the the extended $\delta$-folio of $(G^{\star}\setminus S)\setminus v$ for any unmarked $G_r\in {\cal G}$ and $v\in V(G_r)$.
%

\paragraph*{Running time analysis.}
Let $T(n,\xi, \delta,k)$ be the running time of algorithm ${\cal A}$, where $n=\vert V(G^{\star})\vert$ and $\xi=\vert R(G^{\star})\vert$.  By \autoref{thm:main1General},  \ref{step111} takes $f_1(\xi, \delta)n^3$ time for some computable function $f_1$. \ref{stepfirstmark} takes $\OO(n)$ time. By \autoref{prop:unimp}, the time required to execute \ref{stepuniqueimpsepmark} is $f_2(\xi,\delta,k)n^2$ for some computable function $f_2$.  \ref{step:rec1} takes time $f_3(\xi,\delta,k) T(n,\eta,\delta,k-1)$ for some computable function $f_3$, where $\eta=\beta(\delta,\xi)$ (mentioned in the algorithm).  By \autoref{prop:impsepcout}, \ref{step:imp} takes time $f_4(\xi,\delta,k)n^2$ for some computable function $f_4$. \ref{step:rec2} takes $f_5(\xi,\delta,k) T(n,\eta+k,\delta,k-1)$ time for some computable function $f_3$. Therefore, there exist computable functions $g_1$ and $g_2$ such that $T(n,\xi, \delta,k)$ satisfies the following recurrence relation. 
\begin{eqnarray*}
T(n,\xi, \delta,k)&=&g_1(\xi,\delta,k) T(n,\beta(\xi,\delta)+k, \delta,k-1)+ g_2(\xi,\delta,k) n^3\\
T(n,\xi, \delta,0)&=&g_2(\xi,\delta,k) n^3
\end{eqnarray*}
Above recurrence relation implies that there is a computable function $g$ such that $T(n,\xi, \delta,k)$ is  
at most $g(\xi,\delta,k) n^3$. This completes the proof of the lemma. 
\end{proof}

\section{Graphs with small treewidth}\label{sec:smalltw}

Recall the definition of $\delta$-representative (see \autoref{def:rep}). 
We prove that in the case of a graph of small treewidth, we can find a small $\delta$-representative.

\begin{lemma}\label{lem:representativeTW}
There exists a computable function $g$ and an algorithm that, given a rooted graph $G$, a nice tree  decomposition of $G$ of width $\tw$,  and $\delta\in\mathbb{N}$ such that $|R(G)|\leq \boundary$, computes a $\delta$-representative $F$ of $G$ (with the same set of roots) such that $|V(F)|\leq g(\delta,\tw)$ in time 
$2^{\OO((\delta+\boundary)^2)} \cdot (\vert R(G)\vert+\tw)^{\OO(\delta+\vert R(G)\vert+\tw)}\cdot n^2$, 
where $g(\delta,\tw)=2^{2^{\OO((\delta+\boundary)^2)} \cdot (\boundary+\tw)^{\OO(\delta+\boundary)}}$. 
\end{lemma}

\begin{proof}[Proof sketch of \autoref{lem:representativeTW}]
Let $\delta^{\star}=\delta+\vert R(G)\vert$. 
We present an algorithm \alg{Alg} that will output a {\em small} graph with the same $\delta^{\star}$-folio as $G$. Then by \autoref{obs:flatfolioonly}(a) the correctness follows. 


The input of \alg{Alg} is a rooted graph $G$, a nice tree decomposition  $(T',\beta')$ of width $\tw$ and $\delta\in\mathbb{N}$ such that $|R(G)|\leq \boundary$. 
First we add $R(G)$ to each bag of the decomposition $(T',\beta')$. Let $(T,\beta)$ be the resulting tree decomposition of $G$ of width $\tw+\vert R(G)\vert$. Notice that for any $v\in V(T)$, $R(G)\subseteq \beta(v)$ and the number of children of $v$ is at most $2$. Without loss of generality assume that $\rho_{G}(R(G))=[\vert R(G)\vert]$. For any $v\in V(T)$ define a rooted graph $G_v$ as follows: $G_v=G[\gamma(v)]$ and $R(G_v)=\beta(v)$. The function $\rho_{G_v}$ is an arbitrary injective map from $R(G_v)$ to $\{1,\ldots, \vert R(G)\vert + \tw+1\}$ such that $\rho_{G_v}(x)=\rho_{G}(x)$ for all $x\in R(G)$ and $\rho_{G_v}(x)\notin [\vert R(G)\vert]$ for all $x\in \beta(v)\setminus R(G)$.

Next \alg{Alg} compute the  $\delta^{\star}$-folio of $G_v$ for all $v\in V(T)$ as follows. 
For each rooted graph $H$ where $|E(H)|+\isolated(H)\leq \delta^{\star}$ with roots mapped 
to $\rho_{G_v}(R(G_v))$,  \alg{Alg} computes whether $H$ is a topological minor in $G_v$ or not using \autoref{prop:twDPNoDec}. This step of the algorithm takes time $(\vert R(G)\vert+\tw)^{\OO(\vert R(G)\vert+\tw)}n$. 
By \autoref{prop:no.ofgraphsinfolios}, $(i)$ number of graphs in the $\delta^{\star}$-folio of $G_v$, for any $v\in V(T)$, is upper bounded by  $2^{\OO((\delta^{\star})^2)} \cdot (\vert R(G)\vert+\tw)^{\OO(\delta^{\star})}$. 
Therefore the $\delta^{\star}$-folio of $G_v$ for all $v\in V(T)$ together can be computed in time 
$2^{\OO((\delta^{\star})^2)} \cdot (\vert R(G)\vert+\tw)^{\OO(\delta+\vert R(G)\vert+\tw)}\cdot n^2$.

Because of  $(i)$, the number of distinct $\delta^{\star}$-folios 
of $G_v,v\in V(T)$, is upper bounded by  $2^{\OO((\delta^{\star})^2)}\cdot (\vert R(G)\vert+\tw)^{\OO(\delta^{\star})}=q(\delta,\tw)$ (for some function $q$). 
%
Now if the depth of the tree-decomposition $(T,\beta)$ is strictly more than $q(\delta,\tw)$, then there exist two nodes $u,v\in V(T)$, $u$ is an ancestor of $v$ and the $\delta^{\star}$-folios of $G_u$ and $G_v$ are same. 
Thus the graph $G'$ obtained by replacing $G_u$ with $G_v$ in the separation $(G_u,(G\setminus (\gamma(u)\setminus \beta(u)))\setminus E(\beta(u)))$,  has the same $\delta^{\star}$-folio of $G$ with respect to roots $R(G)$ (by \autoref{prop:equi_graph_replacement} and \autoref{obs:flatfolioonly}(b)). 
Moreover, a  tree decomposition $(T^{\star},\beta^{\star})$  of $G'$ of width $\tw+\vert R(G)\vert$ and with the number of  nodes strictly less than that of $(T,\beta)$ can be derived from $(T,\beta)$ as follows. Let $T_u$ and $T_v$ be the subtrees of $T$, rooted at $u$ and $v$, respectively. Then $T^{\star}$ is obtained by replacing $T_u$ with $T_v$ in the separation $(T_u,T\setminus(V(T_u)\setminus \{u\}))$. 
There is a natural injective mapping $f_{T^{\star}}$ from $V(T^{\star})$ to $V(T)$. Then $\beta^{\star}(w)=\beta(f_{T^{\star}}(w))$ for any $w\in V(T^{\star})$.  For any $w\in V(T^{\star})$, the number of children of $w$ in $T^{\star}$ is at most $2$ and $R(G)\subseteq \beta^{\star}(w)$. Also the $\delta^{\star}$-folios of $G^{\star}_w$ for all $w\in V(T^{\star})$ together  can be obtained from the previously computed $\delta^{\star}$-folios of $G_v$ (for all $v\in V(T)$) in time $\OO(n)$ as follows. For any $w\in V(T^{\star})$ $(a)$ the  $\delta^{\star}$-folio of $G^{
\star}_w$ is equal to the  $\delta^{\star}$-folio of $G_{f_{T^{\star}}(w)}$.   
When $G^{\star}_w=G_{f_{T^{\star}}(w)}$, statement $(a)$ is clearly true. Otherwise the correctness follows from \autoref{prop:equi_graph_replacement} and \autoref{obs:flatfolioonly}(b). 

So,  \alg{Alg} continues to execute this procedure until the depth of the tree  is bounded by $q(\delta,\tw)$. Since the degree of each node in the tree decompositions constructed in each step is at most $3$,  
when the procedure stops we will have a graph $F$ that is a $\delta$-representative of $G$  and $|V(F)|\leq 2^{q(\delta,\tw)}$.  
The computations of the $\delta^{\star}$-folio of $G_v$ for all $v\in V(T)$ altogether take time 
$2^{\OO((\delta^{\star})^2)} \cdot (\vert R(G)\vert+\tw)^{\OO(\delta+\vert R(G)\vert+\tw)}\cdot n^2$.  
The subsequent computations of the $\delta^{\star}$-folios in each step take time $\OO(n)$. 
Thus, the overall running time of the algorithm is upper bounded by 
$2^{\OO((\delta^{\star})^2)} \cdot (\vert R(G)\vert+\tw)^{\OO(\delta+\vert R(G)\vert+\tw)}\cdot n^2$.  
\end{proof}

We would like to mention that 
we did not optimize the running time of the algorithm for \autoref{lem:representativeTW} 
as the  the running time of our main algorithm  is $\Omega(n^2)$. 
As a corollary to \autoref{lem:representativeTW}, we have the following result.

\begin{corollary}\label{cor:dp}
There exists an algorithm that, given a rooted graph $G$ of treewidth at most $\tw$  and $\delta\in\mathbb{N}$ such that $|R(G)|\leq \boundary$, solves the instance $(G,\delta)$ of \FindFolio\ in time $f(\delta,\tw)n^2$ 
where $f(\delta,\tw)=2^{\OO((\delta+\boundary)^2)} \cdot (\vert R(G)\vert+\tw)^{\OO(\delta+\vert R(G)\vert+\tw)}$.
\end{corollary}

\section{Recursive Understanding for Topological Minor Containment}\label{sec:recursive}

In this section, we prove \autoref{thm:main1General} under the assumption that \autoref{thm:main1Flat} holds. The proof of  \autoref{thm:main1Flat} is the crux of this paper, and it is given in the following sections. 
For our current purpose, we apply the method of recursive understanding, which means that at each stage (besides basis cases, among which is \autoref{thm:main1Flat}) we will separate our graphs into two subgraphs, ``understand'' the behavior of one of these subgraphs, and then replace it with a so called representative that has the same ``behavior'' and is smaller.
Combining Lemmata \ref{lem:mainexrele}, \ref{lem:representativeTW} 
 and \ref{lem:representativeClique},  
 we have the following lemma, which says that a $\delta$-representative of small size exists.

\begin{lemma}\label{lem:representative}
There exists a computable function $g$ such that for any instance $(G,\delta)$ of \FindFolio\ where $|R(G)|\leq \boundary$, there is a $\delta$-representative $F$ of $G$ (with the same set of roots) such that $|V(F)|\leq g(\delta)$, where $g(\delta)=2^{2^{c'(\delta+\boundary)^2}r}$ and $r=h(2^{c(\delta+\boundary)^2})$ for some constants $c$ and $c'$. 
\end{lemma}

\begin{proof}
Let $F$ be a rooted graph with minimum number of vertices such that $F$ is a $\delta$-representative of $G$ and $R(F)=R(G)$. 
Let $h'$ be the function defined in \autoref{lem:representativeClique} and $t=h'(\delta) \in 2^{\OO((\delta+\boundary)^2)}$. 
Let $g_0$ be the function $g$ defined in \autoref{lem:representativeTW} and $g_1$ be the function $g$ defined in  \autoref{lem:mainexrele}. Let $c$ be the constant mentioned in \autoref{lem:flatWallCombined} and let $t_1=ct^{24}$. 
Let $w=g_1(\delta,t_1)\in 2^{\OO((\delta+\boundary)^2)}t_1 \cdot r=2^{\OO((\delta+\boundary)^2)}\cdot r$, where 
$r=h(\delta+\boundary+t_1)$. 
Let $g_2$ be the constant $ct^{48}(t^2+w)$.  That is, $g_2=2^{c_1(\delta+\boundary)^2}r$, where $c_1$ is a constant.

Now we define the function $g$ to be as follows: $g(\delta)=g_0(\delta,\tw)$, where $\tw=cg_2^{19}\log^c g_2$. 
That is, $g(\delta)=g_0(\delta, \tw)=2^{2^{c'(\delta+\boundary)^2}r}$, for some constant $c'$.  
We claim that $\vert V(F)\vert \leq g(\delta)$. By applying \autoref{lem:flatWallCombined} on $F$ and integers $g_2,w$ and $t$, we know that one of the following is true. 
\begin{itemize}
\item[$(i)$] The treewidth of $F$ is at most $\tw=cg_2^{19}\log^c g_2$.
\item[$(ii)$] $K_t$ is a minor of $F$.
\item[$(iii)$] There is a subset $A\subseteq V(G)$ of size at most $t_1=ct^{24}$ and a $w\times w$ flat wall $W$ in $G\setminus A$. 
\end{itemize}
First notice that because of the minimality of $V(F)$ and by Lemmata~\ref{lem:mainexrele} and \ref{lem:representativeClique}, the statements $(ii)$ and $(iii)$ above are false. So statement $(i)$ above is true. 
That is, the treewidth of $F$ is at most $t^{\star}$. Then, by the minimality of $V(F)$ and by \autoref{lem:representativeTW}, we conclude that $\vert V(F)\vert \leq g_0(\delta,t^{\star})=g(\delta)$. 
\end{proof}

From  \autoref{lem:representative} we further have the following corollary that is more useful for our purposes.


\begin{corollary}\label{cor:representativeAlg}
There are constants $c$ and $c'$, and an algorithm that, given the extended $\delta$-folio of some graph $G$ where $|R(G)|\leq \boundary$, outputs a graph $F$ with the same extended $\delta$-folio (and the same set of roots) such that $|V(F)|\leq 2^{2^{c'(\delta+\boundary)^2}r}$ in time $2^{2^{2^{c'(\delta+\boundary)^2}r}}$, where  $r=h(2^{c(\delta+\boundary)^2})$.   
\end{corollary}

\begin{proof}
By \autoref{lem:representative}, there is a $\delta$-representative $F$ of $G$ (with the same set of roots) such that 
$|V(F)|\leq 2^{2^{c''(\delta+\boundary)^2}r}$.
Thus, to compute such a representative, we can simply consider all rooted graphs on at most 
$2^{2^{c''(\delta+\boundary)^2}r}$ vertices, compute the extended $\delta$-folio of each of them, and output one of the graphs among them that will have the same extended $\delta$-folio as $G$.
\end{proof}

Now, we state a lemma which we will invoke to resolve our problem once we discover a large clique minor in a graph.   
The correctness of the lemma is proved in \cite{DBLP:conf/stoc/GroheKMW11} and it uses the algorithm of  \autoref{thm:main1General} for finding the $(\delta-1)$-folio in a graph.


\begin{lemma}\label{lem:clique}
Let $\delta\in {\mathbb N}$ and $\delta'=\delta-1$. 
Suppose there is an algorithm ${\cal A}$ for \FindFolio\ (for finding the extended $\delta'$-folio of a graph) running in time $T(n,\delta')$, on $n$-vertex graphs with $\vert R(G)\vert\leq \alpha(\delta')$, where $\alpha(\delta')=16 (\delta')^2$.  
Then, there exists an algorithm that, given an instance $(G,\delta,t,\phi)$ of \cFindFolio\  where $|R(G)|\leq \boundary$ and $L\in\mathbb{N}$ such that 
$t\geq \max \{10\delta, L\}+\vert R(G)\vert$, uses algorithm ${\cal A}$  and  does one of the following. 
\begin{itemize}
\setlength\itemsep{0em}
\item Runs in time $2^{2^{(L+\delta+\alpha(\delta))^{\OO(1)}}}n^2+ 2^{\OO(\boundary)}T(n,\delta-1)$ and  solves $(G,\delta)$. 
\item Runs in time $2^{2^{(L+\delta+\alpha(\delta))^{\OO(1)}}}n^2$ and outputs a separation $(G_1,G_2)$ of $G$ of order $\leq 4\delta^2$ such that $\min\{|V(G_1)|,|V(G_2)|\}\geq L$.
\end{itemize}
\end{lemma}

Now we are ready to prove \autoref{thm:main1General} assuming \autoref{thm:main1Flat}.

\begin{proof}[Proof of \autoref{thm:main1General}]
Suppose that \autoref{thm:main1Flat} is correct.
 Accordingly, let  $g_{wall}$  and \alg{Alg-Wall} be the function and algorithm as specified by \autoref{thm:main1Flat}. Let $\delta^{\star}=\delta+\boundary$.  
 Let $f_{wall}$ be a function such that the running time of 
\alg{Alg-Wall} is $f_{wall}(\delta,t,s)n=2^{2^{\OO(((t+\delta^{\star})^2+r_1)\log (t+\delta^{\star}+r_1))}} s^{\OO(s)} n$   where  
 $r_1=h(\delta+\boundary+t)$.  Also, $g_{wall}(\delta,t)=2^{\OO(((t+\delta^{\star})^2+r_1)\log (t+\delta^{\star}+r_1))}$. 
Let  \alg{Alg-Rep} be the  algorithm as specified by \autoref{cor:representativeAlg} 
and $f_{rep},g_{rep}$  be the functions such that $f_{rep}(\delta)=2^{2^{2^{c_2(\delta^{\star})^2}r_2}}$ and 
$g_{rep}(\delta)={2^{2^{c_2(\delta^{\star})^2}r_2}}$, where $r_2=h(2^{c_2(\delta^{\star})^2})$ for some constant $c_2$. (Here, $c_2$ is the maximum of $c'$ and $c$ in \autoref{cor:representativeAlg}). Let $f_{tw}$ and \alg{Alg-TW} be the function and algorithm as specified by \autoref{cor:dp}. That is, $f_{tw}(\delta,\tw)=2^{\OO((\delta+\boundary)^2)} \cdot (\boundary+\tw)^{\OO(\delta+\boundary+\tw)}$.
In addition, let $c'$ and \alg{Alg-Decomp} be the constant and algorithm, respectively, specified by \autoref{lem:flatWallCombined}. 
Moreover, let $f_{dec}$ be the function such that the running time of the algorithm of \autoref{lem:flatWallCombined} is $f(g)n\log^2 n=2^{\OO(g^{58})}n\log^2n$. 

The proof is by induction on $\delta$. Clearly, if $\delta=1$, then the lemma is correct as the extended $\delta$-folio simply consists of a collection of isolated vertices and edges.
 Now, for some $\delta'\geq 2$, let us suppose that the lemma is correct for all instances where $\delta\leq \delta'-1$, and let us now prove it for instances where $\delta=\delta'$. 
 Let 
 \alg{Alg-Clique} be the  algorithm  specified by \autoref{lem:clique} when restricted to instances where $\delta\leq\delta'$.

\paragraph{Algorithm.} We now describe a recursive algorithm \alg{Alg-Main} that, given an instance $(G,\delta)$ of \FindFolio\ where $|R(G)|\leq \boundary$ and $\delta\leq\delta'$ solves it in time $f(\delta)n^3$,
where $f$  is the function defined in the statement of the theorem. We use $T(n,\delta)$ to denote  the running time of \alg{Alg-Main}. 
Later we will prove that $T(n,\delta)\leq f(\delta)n^3$. 
%
%
We fix some constants as follows. 
\begin{eqnarray*}
\boundary &=& 16\delta^2\\
L&=&\max\{16\delta^2+1,g_{rep}(\delta)+1\}\\
\widehat{t}&=&\max \{10\delta, L\}+\vert R(G)\vert\\
t&=&c'\widehat{t}^{24}\\
\widehat{w}&=&g_{wall}(\delta,t)\\
\widehat{g}&=&c'\widehat{t}^{48}(\widehat{t}^2+\widehat{w})\\
\tw&=&c'\widehat{g}^{19}\log^{c'}\widehat{g}
\end{eqnarray*}

First, if $n\leq g_{rep}+2$, then by \alg{Alg-TW} of \autoref{cor:dp}, \alg{Alg-Main} solves $(G,\delta)$ in time $f_{tw}(\delta, g_{rep}+2)n^2$. Now suppose that $n>g_{rep}+2$. 
By algorithm \alg{Alg-Decomp} of \autoref{lem:flatWallCombined} on $(G, \widehat{g},\widehat{w},\widehat{t})$, in time $f_{dec}(\widehat{g})n\log^2 n =2^{\OO(\widehat{g}^{58})}n \log^2n$, \alg{Alg-Main} obtains one of the following objects:
\begin{enumerate}
\item\label{recursive:case1} A nice tree decomposition of $G$ of width at most $\tw$.
\item\label{recursive:case2} A function $\phi$ witnessing that $K_{\widehat{t}}$ is a minor of $G$.
\item\label{recursive:case3} A subset $A\subseteq V(G)$ of at most $t$ vertices and a $(\tw)$-nice flat wall $W$ of size at least $(\widehat{w}\times \widehat{w})$
 in $G\setminus A$. In addition, it obtains a flatness tuple $(A',B',C,\tilde{G},G_0,G_1,\ldots,G_k)$ for $(G,\widehat{w},t,A,W)$. 
\end{enumerate}

Besed on the outcome above we have three cases. 

\medskip
\noindent
{\bf Case 1:}
\alg{Alg-Main} uses \autoref{cor:dp} to solve $(G,\delta)$ in time $f_{tw}(\delta,\tw)n^2$. 


\medskip
\noindent
{\bf Case 2:}
In this case, 
\alg{Alg-Main} calls the algorithm \alg{Alg-Clique}  with the instance $(G,\delta,\widehat{t},\phi)$ of \cFindFolio\  and $L$. \alg{Alg-Clique} does one of the following. 

\begin{itemize}
\setlength\itemsep{0em}
\item Runs in time $2^{2^{(L+\delta+\alpha(\delta))^{\OO(1)}}}n^2+ 2^{\OO(\boundary)} T(n,\delta-1)$ and  solves $(G,\delta)$. 
\item Runs in time $2^{2^{(L+\delta+\alpha(\delta))^{\OO(1)}}}n^2$ and outputs a separation $(G'_1,G'_2)$ of $G$ of order $\leq 4\delta^2$ such that $\min\{|V(G'_1)|,|V(G'_2)|\}\geq L$.
\end{itemize}

\medskip
\noindent
{\bf Case 2(a):}
In the first subcase, the execution is already complete. 

\medskip
\noindent
{\bf Case 2(b):}
We now describe how the execution proceeds when the second subcase occurs. Let $S=V(G'_1)\cap V(G'_2)$. Since $|R(G)|\leq \boundary$, there exists $i\in[2]$ such that $|V(G'_i)\cap R(G)|\leq |R(G)|/2\leq (\boundary)/2$. Without loss of generality, suppose that $i=1$. Now we set $R(G'_1)=(V(G'_1)\cap R(G))\cup S$. Then, since $\vert S\vert \leq 4\delta^2$ and $\boundary=16\delta^2$, we have that $|R(G'_1)|\leq (\boundary)/2 + 4\delta^2\leq \boundary$. Notice that $R(G'_1)\subseteq R(G)\cup S$. Furthermore, since $|V(G'_2)|\geq L\geq 16\delta^2+1$ and $\vert S\vert \leq 4\delta^2$, we have that 
$|V(G'_1)|\leq n-1$. Now, \alg{Alg-Main} calls itself recursively with $(G'_1,\delta)$ as input. The output of this recursive call is  the extended $\delta$-folio of $G'_1$.
Next, by algorithm \alg{Alg-Rep} of \autoref{cor:representativeAlg}, \alg{Alg-Main} obtains in time $f_{rep}(\delta)$ a $\delta$-representative $G_1''$ of $G_1'$ such that $|V(G_1'')|\leq g_{rep}(\delta)$. Let $G'$ be the graph obtained by replacing $G_1'$ with $G_1''$ in the separation $(G_1',G_2')$. By \autoref{prop:equi_graph_replacement}, we know that the extended $\delta$-folio of $G$ with respect to roots $R(G)\cup S$  is same as the extended $\delta$-folio of  $G'$ with respect to roots $R(G)\cup S$.  
Therefore, by \autoref{obs:flatfolioonly}(b),  
the extended $\delta$-folio of $G$ with respect to roots $R(G)$ is the same as the extended $\delta$-folio of  $G'$ with respect to roots $R(G)$. So now we set $R(G')=R(G)$.  That is, $(G,\delta)$ and $(G',\delta)$ are equivalent instances of \FindFolio.
Since $\vert V(G''_1)\vert<\vert V(G'_1)\vert$, we have that $\vert V(G')\vert \leq n-1$. 
%
%
Next, \alg{Alg-Main} calls itself recursively with $(G',\delta)$ as input. The output of this recursive call is the extended $\delta$-folio of $G'$ as well as $G$. 

\medskip
\noindent
{\bf Case 3:}
Finally, consider Case \ref{recursive:case3}. 
Here, \alg{Alg-Main} has an instance of \fFindFolio\ with arguments 
$w=\widehat{w}$ and $s=\tw$. 
Then, since $|A|\leq t$ and $\widehat{w}\geq g_{wall}(\delta,t)$, \alg{Alg-Main} can simply call algorithm \alg{Alg-Wall} of  \autoref{thm:main1Flat} to 
find  an irrelevant vertex $v$ in time $f_{wall}(\delta,t,s)n$. Then, the algorithm further performs a recursive call with $(G\setminus v,\delta)$ as input, and returns the output of this call. 



\paragraph{Analysis.} We prove that the execution of \alg{Alg-Main}, given an instance $(G,\delta)$ of \FindFolio\ where $|R(G)|\leq \boundary$ and $\delta\leq\delta'$, is correct 
and that \alg{Alg-Main} runs in time $f(\delta)n^3$. The proof is done by induction on $n$. It is clear that when $n\leq g_{rep}(\delta)+2$, the correctness of the execution is guaranteed, and also if we assume that the correctness of the execution is satisfied for all $n'\leq n-1$, then the correctness of the execution is also guaranteed for $n$. Thus, in what follows, we only analyze running time. Let $f(\delta)=2^{2^{(\beta\cdot r_1^{3})}} 2^{2^{2^{(r_2 \cdot 2^{\beta\cdot \delta^{4}})} }}$, where $\beta$ is a large constant we choose to satify the runing time bound in all the cases of the analysis.  
We prove by induction on the number of vertices that \alg{Alg-Main} on instance $(G,\delta)$ runs in time $T(n,\delta) \leq f(\delta) (n-g_{rep}(\delta)-2)n^2$, where $n=\vert V(G)\vert$.  


Before moving to the analysis, let us simplify $f_{wall}(\delta,t,\tw)$ and $g_{wall}(\delta,t)$, and prove a claim which we will use later. Recall that $r_1=h(\delta+\boundary+t)$, which is equal to $h(2^{2^{(c_1\cdot \delta^{4})}r_2})$ for some constant $c_1$ and  $r_2=h(2^{c_2(\delta^{\star})^2})$.  
We know that 
\begin{eqnarray*}
g_{wall}(\delta,t)&=&2^{\OO(((t+\delta^{\star})^2+r_1)\log (t+\delta^{\star}+r_1))}=2^{\OO(r_1^3)}   \qquad\qquad \mbox{Because $h(k)\geq k$}\\
f_{wall}(\delta,t,\tw)&=&2^{2^{\OO(((t+\delta^{\star})^2+r_1)\log (t+\delta^{\star}+r_1))}} \tw^{\OO(\tw)}\\
&=&2^{2^{\OO(r_1^3)}} 2^{\OO((\widehat{g})^{20})} \qquad\qquad\qquad\qquad\qquad\qquad \mbox{Because $h(k)\geq k$}\\
&=&2^{2^{\OO(r_1^3)}} 2^{(t\cdot \widehat{w})^{\OO(1)}}  \\
&=&2^{2^{\OO(r_1^3)}} 2^{(g_{wall}(\delta,t))^{\OO(1)}} \qquad\qquad\qquad\qquad\qquad \mbox{Because $g_{wall}(\delta,t)\geq t$} \\
&=&2^{2^{\OO(r_1^3)}} 
\end{eqnarray*}
That is, there is a constant $c$ such that $g_{wall}(\delta,t)\leq 2^{cr_1^3}$ and $f_{wall}(\delta,t,\tw)\leq 2^{2^{cr_1^3}}$.  Similarly, there is a constant $\beta_1$ such that $f_{dec}(\widehat{g})\leq 2^{2^{\beta_1 \cdot r_1^3}}$. Now we upper bound $f_{tw}(\delta,\tw)$. 

\begin{eqnarray*}
f_{tw}(\delta,\tw)&=&{2^{\OO((\delta^{\star})^2)}} (\boundary +\tw)^{\OO(\delta^{\star}+\tw)}\\
&=&{2^{\OO((\delta^{\star})^2)}} (\tw)^{\OO(\tw)}\\
&=& (\tw)^{\OO(\tw)}\\
&=&2^{2^{\OO(r_1^3)}} 
\end{eqnarray*}
Therefore, there is a constant $\beta_2$ such that $f_{tw}(\delta,\tw)\leq 2^{2^{\beta_2 \cdot r_1^3}}$. There is a constant $\beta_3$ 
such that $g_{rep}(\delta)\leq {2^{2^{\beta_3(\delta)^4}r_2}}$ and $f_{rep}(\delta)\leq 2^{2^{2^{\beta_3(\delta)^4}r_2}}$. The following claim follows from the above simplification.

\begin{claim}
\label{claim:gamma}
There is a constant $\gamma$ such that $$f_{dec}(\widehat{g})+f_{tw}(\delta,\tw)+f_{rep}(\delta)+f_{wall}(\delta,t,\tw) \leq 2^{2^{(\gamma \cdot r_1^3)}} {2^{2^{(r_2\cdot 2^{(\gamma\cdot \delta^4)})}}}.$$
\end{claim}

%

Now we are ready to analyse the running time. 
When $n\leq g_{rep}+2$, the running time is upper bounded by $f_{tw}(\delta,g_{rep}+2)n^2 \leq f_{tw}(\delta, \tw)n^2$.
Because of \autoref{claim:gamma}, we choose $\beta\geq \gamma$ (where $\gamma$ is the constant mentioned in \autoref{claim:gamma}) and the running time will be upper bounded by $f(\delta)n^2$. 
We $n>g_{rep}+2$, we upper bound the running time for each of the above cases. 

\medskip
\noindent
{\bf Case 1:}
In this case the running time is $f_{dec}(\widehat{g})n\log^2 n+f_{tw}(\delta,\tw)n^2$, which is upper bounded by $f(\delta) n^2$, by \autoref{claim:gamma}, when $\beta\geq \gamma$ (where $\gamma$ is the constant mentioned in \autoref{claim:gamma}). 


\medskip
\noindent
{\bf Case 2:}
Since $L=g_{rep}(\delta)+1$ and $g_{rep}(\delta)\leq {2^{2^{\beta_3(\delta)^4}r_2}}$, there is a constant $\beta_4$ such that 
the running time  of \alg{Alg-Clique}  is bounded as follows. 

\begin{itemize}
\setlength\itemsep{0em}
\item Runs in time $2^{2^{2^{(r_2\cdot 2^{(\beta_4(\delta)^4)})}}}n^2+ 2^{(\beta_4\cdot \boundary)}T(n-1,\delta-1)$ and  solves $(G,\delta)$. 
\item Runs in time $2^{2^{2^{(r_2\cdot 2^{(\beta_4(\delta)^4)})}}}n^2$ and outputs a separation $(G'_1,G'_2)$ of $G$ of order $\leq 4\delta^2$ such that $\min\{|V(G'_1)|,|V(G'_2)|\}\geq L$.
\end{itemize}

\medskip
\noindent
{\bf Case 2(a):} 
There is a constant $\beta^{\star}$ such that if $\beta>\beta^{\star}$, then $2^{\beta_4\boundary} f(\delta-1)\leq f(\delta)$. 
In this case the running time is upper bounded as follows (where we will choose $\beta>\beta_4,\beta^{\star},\gamma$). 
\begin{eqnarray*}
f_{dec}(\widehat{g})n\log^2 n &+& 2^{2^{2^{(r_2\cdot 2^{(\beta_4(\delta)^4)})}}}n^2+ 2^{(\beta_4\cdot \boundary)} f(\delta-1) (n-1-g_{rep}(\delta-1)-2)(n-1)^2\\
&\leq&f(\delta)n\log^2 n+ f(\delta)n^2 + 2^{(\beta_4\boundary)}f(\delta-1) (n-1-g_{rep}(\delta-1)-2)(n-1)^2\\
&\leq& 2f(\delta)n^2 + f(\delta)(n-1-g_{rep}(\delta-1)-2)(n-1)^2\\
&\leq&  f(\delta)(n-1-g_{rep}(\delta-1))n^2\\
&\leq&  f(\delta)(n-g_{rep}(\delta)-2)n^2\qquad\qquad\qquad \mbox{(Because $g_{rep}(\delta-1)\leq g_{rep}(\delta)-1$)}
\end{eqnarray*}



\medskip
\noindent
{\bf Case 2(b):}
We know that $(G_1',G_2')$ is a separation of $G$ of order $p\leq 4\delta^2$ and $L\leq \vert V(G_1')\vert,\vert V(G_2')\vert$ $\leq n-1$. Let $n_1=\vert V(G_1')\vert$ and $n_2=\vert V(G_2')\vert$. That is $n_1+n_2-p=n$. Also notice that $\vert V(G_1'')\vert\leq g_{rep}(\delta)$. Then $n'=\vert V(G')\vert \leq n_2+\vert V(G_1'')\vert -p \leq n_2-p+g_{rep}(\delta)$. Therefore, in this case the running time can be upper bounded as follows (where $\beta'$ is a constant). 

\begin{eqnarray*}
f_{dec}(\delta)n\log^2 n &+& 2^{2^{2^{(r_2\cdot 2^{(\beta_4(\delta)^4)})}}}n^2+ f_{rep}(\delta)+f(\delta) (n_1-g_{rep}-2)n_1^2 + f(\delta) (n_2-p-2) (n')^2\\
&\leq &2f(\delta)n^2+f(\delta) (n_1-g_{rep}-2)n_1^2 + f(\delta) (n_2-p-2) (n')^2\\
&\leq &f(\delta) (n_1-g_{rep}+n_2-p-2)n^2 \\
&\leq &f(\delta) (n-g_{rep}-2)n^2 
\end{eqnarray*}
The first inequality follows from \autoref{claim:gamma} and by choosing $\beta>\gamma,\beta_4$.

\medskip
\noindent
{\bf Case 3:}
In this case the running time is upper bounded as follows. 
\begin{eqnarray*}
f_{dec}(\delta)n\log^2 n &+& f_{wall}(\delta,t,\tw)n+f(\delta) (n-1-g_{rep}-2)(n-1)^2 \\
&\leq& f(\delta) n\log^2 n +f(\delta) (n-1-g_{rep}-2)n^2\\
&\leq &f(\delta) (n-g_{rep}-2)n^2 
\end{eqnarray*}
Here the first inequality follows from \autoref{claim:gamma}. This completes the proof of the theorem. 
\end{proof}

\begin{figure}[t]
\begin{center}
\begin{tikzpicture}[scale=1]

\node (rect) at (4,2) [draw,line width=1mm,minimum width=5cm,minimum height=2cm] {\makecell[l]{\Large{\FindFolio~$(G,\delta, \alpha(\delta))$}\\\\ Solves $(G,\delta)$}};


\node (rect) at (4,7.1) [draw,line width=0.5mm,minimum width=5cm,minimum height=1.5cm] {\Large{\cFindFolio~$(G,\delta, k)$}};

\node (a) at (-1.4,7.1) [] {Good separation};

\node (a) at (9,7.1) [] {Solves $(G,\delta)$}; 

\draw[->, line width=0.5mm,red] (0.7,7.1) -- (0,7.1);
\draw[->, line width=0.5mm,red] (7.3,7.1) -- (7.8,7.1);

\draw[->, line width=0.5mm] (4,3) -- (4,6.35);

\draw[<-, line width=0.5mm] (6,3) -- (8,6.8);

\node (a) at (9, 5.6) [] {$2^{\OO(\boundary)}$ times}; 


\node (a) at (8.75, 4.5) [] {$(n,\delta-1,\alpha(\delta-1))$}; 

\node (a) at (9.25, 2.5) [] {$(n_1,\delta,\alpha(\delta))$}; 

\node (a) at (9.25, 2) [] {$(n',\delta,\alpha(\delta))$};

\node (a) at (9.25, 1.3) [] {(where $n_1+n'\leq n+g_{rep}(\delta)$)};

\node (a) at (3.2, 4.5) [] {$(n,\delta,\widehat{t})$};

\node (rect) at (1,-4.1) [draw,line width=0.5mm,minimum width=5cm,minimum height=1.5cm] {\makecell[l]{\Large{\fFindFolio~$(G,\delta,t, w,s)$}\\\\ Finds an irrelevant vertex}};

\draw[->, line width=0.5mm] (2,1) -- (2,-3.25);
\draw[->, line width=0.5mm] (6,1) -- (6,-3.5);

\draw [->,line width=0.5mm] (6.7,2.8) -- (8,2.8)--(8,1.6)--(6.7,1.6);

\node (a) at (-0.1, -1) [] {$\left(n,\delta,{t}, g_{wall}(\delta,t),\tw\right)$}; 


\node (rect) at (8.2,-4.35) [draw,line width=0.5mm,minimum width=5cm,minimum height=1cm] {\makecell[l]{\Large{${\sf tw}$-\FindFolio~$(G,\delta, {\sf tw})$}\\\\ Solves $(G,\delta)$}};

\node (a) at (7, -1) [] {$\left(n,\delta, \tw\right)$}; 


\end{tikzpicture}
\end{center}
\caption{A schematic diagram of how different algorithms are used to solve \FindFolio. Output of each algorithm is specified in the respective box except for \cFindFolio. There are two kinds of output for \cFindFolio: it either solves $(G,\delta)$ or outputs a {\em good separation}. \cFindFolio\  calls \FindFolio\  ($2^{\OO(\boundary)}$ times) with parameter $\delta-1$ only when it outputs the extended $\delta$-folio of the given input.  The dependence of each value in the arguments on $\delta$ is mentioned in the proof  of \autoref{thm:main1General}.
}
\label{fig:overview}
\end{figure}

A schematic diagram of how each procedure calls others in the proof of  \autoref{thm:main1General}  is given in \autoref{fig:overview}. 
Next we prove \autoref{thm:findIrrelevant} assuming \autoref{thm:main1Flat}.

\begin{proof}[Proof of \autoref{thm:findIrrelevant}]
Towards the proof it is enough to find an irrelevant  vertex to the $\delta$-folio of $G$, with $R(G)=\emptyset$ where $\delta=(h^{\star})^2$. 
We will use \autoref{thm:main1Flat} to find an irrelevant vertex. 
Notice that $G$ is a plane graph and  we set $\delta=(h^{\star})^2$, $t=0$, $s=0$, and $\boundary=0$. 
Therefore we have that $g_{wall}(\delta,t)=2^{\OO(r^2)}$, where $r=h(k)=2^{ck}$, because $G$ is a plane graph 
(see \autoref{prop:disPathIrrelevant} later in the paper), where $c$ is a constant. Hence, $g_{wall}(\delta,t)=2^{2^{c_1k}}$ for some constant $c_1$. Then, by \autoref{prop:gridTWplanaralgof}, there is a constant $c$ such that if $\tw(G)\geq 2^{2^{ck}}$, then the algorithm of \autoref{prop:gridTWplanaralgof} outputs $(w\times w)$-wall, where $w\geq 2^{2^{c_1k}}$. Now we apply \autoref{thm:main1Flat} (where $w''=1$) to find an irrelevant vertex. The running time of the algorithm follows by substituting 
$\delta=(h^{\star})^2$, $t=0$, $s=0$, $g=2^{2^{ck}}$ and $\boundary=0$ in the runtime of \autoref{thm:main1Flat} and \autoref{prop:gridTWplanaralgof}. 
\end{proof}

Now the only thing left is the proof of \autoref{thm:main1Flat}. 
Let $(G,\delta,t,w',s')$ be an instance of \fFindFolio. Recall that a vertex $v\in V(G)$ is {\em irrelevant} if 
the extended $\delta$-folios of $G$ and $G\setminus v$ are same. 
%
Because of \autoref{obs:flatfolioonly}, we define \fFindFoliostar\ where the input instance is same as \fFindFolio, but the objective is to find the $(\delta+\vert R(G)\vert)$-folio of $G$. \autoref{thm:main1Flat} follows from the following theorem and  in the rest of the paper mainly focuses on its proof. 



\begin{restatable}{theorem}{mainflatstar}
\label{thm:main1Flatstar}
There is a computable function $g$ and an algorithm that, given an instance $(G,\delta,t,w',s')$ of \fFindFoliostar\ such that  $|R(G)|\leq \boundary$ and $w'\geq g(\delta^{\star},t)$, and $w''\in\mathbb{N}$, finds a $w''\times w''$ flat wall within the input $w'\times w'$ flat wall such that the set of all vertices of the output inner flat wall is irrelevant. The algorithm 
runs in 
$ 2^{2^{\OO(((t+\delta^{\star})^2+r)\log (t+\delta^{\star}+r))}}(s')^{\OO(s')}(w'')^{\OO(w'')} n$ time 
where  $g(\delta^{\star},t)=(t+\delta^{\star}+r)^{\OO((t+\delta^{\star})^2+r)}w''$  and $r=h(\delta^{\star}+t)$.

\end{restatable}

Also, to prove \autoref{lem:mainexrele}, it is enough to prove the following lemma.

\begin{lemma}\label{lem:mainexrelestar}
There is a computable function $g$ such that for  any  instance $(G,\delta,t,w,s)$ of \fFindFoliostar\ with  $|R(G)|\leq \boundary$ and $w\geq g(\delta,t)$, there is an irrelevant vertex in $G$, where 
$g(\delta,t)=2^{\OO((\delta^{\star})^2)} t \cdot r$, 
and $r=h(\delta+\boundary+t)$. 
\end{lemma}

Now the instance  $(G,\delta,t,w',s')$ of \fFindFoliostar\ is given along with 
a subset $A\subseteq V(G)$, a separation $(A',B')$
 and  
a flat wall $W$ in $G\setminus A$.  We are also given subgraphs $G_0,\ldots,G_k$ of $B'$, 
a cycle $C$ (not necessarily in $G$) containing all the pegs of $W$ in the order determined by the boundary of $W$  
 and an edge minimal supergraph $\tilde{G}$ of 
$G_0$ satisfying all the seven conditions mentioned in \autoref{def:KWT} and \autoref{obs:cflatmore} 
for $B'+E(C)$ being a  $C$-flat. 
That is, we have an embedding of $\tilde{G}$ in a disc with $C$ being the boundary of  the plane graph $\tilde{G}$.  
Without loss of generality we assume that $\tilde{G}$ is a nicely drawn plane graph. 
That is, if $G_0$ has many components, then we prune $G_0$ to be the component containing $W$. 
We fix $\delta^{\star}$ to denote the number $\delta+\vert R(G)\vert$. 
We refer to a set 
${\cal S}=\{(H,\phi_{H},\varphi_{H}) : H \mbox{ in the $\delta^{\star}$-folio of  $G$ and $(\phi_{H},\varphi_{H})$ is a witness for it}\}$  
as a {\em solution} of $(G,\delta,t,w',s')$.  We would like to mention that for any $H$ in the $\delta^{\star}$-folio of  $G$, there is only one tuple $(H,\phi,\varphi)$ in ${\cal S}$. 
We call $T=\bigcup_{(H,\phi_H,\varphi_H)\in {\cal S}} \phi_H(V(H))$ 
the set of {\em terminals} (with respect to ${\cal S}$).   

{\bf Whenever we refer to an instance $(G,\delta,t,w',s')$, we assume that we are given $A\subseteq V(G)$,  
a separator $(A',B')$ of $G\setminus A$, a cycle $C$, subgraphs $G_0,G_1,\ldots,G_k$ of $B'$ and 
a nicely drawn plane graph $\tilde{G}$. Moreover, these notations are fixed and we will not state it explicitly in the statement of results, but will use in their proofs. 
}



\section{Irrelevant Vertices for Disjoint Paths}\label{sec:disjirr}


In this section we state the results about finding an irrelevant vertex for \DisjointPathsLong\ on planar graphs as well as on general graphs. Then we use these results to show the 
existence of  a {\em region of irrelevant vertices} in the plane graph $\tilde{G}$ for \fFindFoliostar.

 \defparproblem{\DisjointPathsLong\ (\DisjointPaths)}{An undirected graph $G$, and a set $T=\{\{s,t\}: s,t\in V(G)\}$.}{$|T|=k$}{Does $G$ contain a set of internally vertex disjoint paths that for every $\{s,t\}\in T$, contain one path whose endpoints are $s$ and $t$?}

The special case of \DisjointPaths\ where the input graph $G$ is a planar graph is known as \pDP. Moreover, a vertex $v\in V(G)$ such that there exists a pair $P\in T$  and $v\in P$, is called a {\em terminal}. Let $(G,H,k)$ be an instance of \DisjointPaths. We say that a vertex $v\in V(G)$ is {\em irrelevant} if $(G,T,k)$ is a \yes-instance, then it holds that $(G\setminus v,T,k)$ is a \yes-instance. 
Notice that if $(G\setminus v,T,k)$ is a \yes-instance, then $(G,T,k)$ is a \yes-instance.  
The main ingredient of the known algorithms for \DisjointPaths\ both on planar as well as general graphs
is the proof that a vertex which is {\em sufficiently insulated} in planar region having no terminals is irrelevant. Towards that we 
first start with the definition of concentric cycles in a planar graph and state a result about 
irrelevant vertices for \DisjointPaths\ on planar graphs, and then we consider a similar result for general graphs.

\begin{definition}[{\bf Concentric Cycles}, \cite{DBLP:journals/jct/AdlerKKLST17}]\label{def:concentricCycles}
Let $G$ be a plane graph, and let ${\cal C}=(C_0,C_1,\ldots,C_s)$ be a sequence  in $G$. We say that $\cal C$ is {\em concentric} if the cycles in $\cal C$ are pairwise vertex-disjoint and for all $i\in[s-1]_0$, $C_i$ is contained in the inner face of $C_{i+1}$. Furthermore, we say that $\cal C$ is {\em tight} if it is concentric and the following conditions are satisfied.
\vspace{-0.5em}
\begin{itemize}
\itemsep0em 
\item There does not exist a path between two vertices in $V(C_0)$ that is contained in the interior face of $C_0$.
\item For all $i\in[s]$, there does not exist a cycle $C$ that is contained in the inner face of $C_i$ but is not equal to $C_i$, such that $C_{i-1}$ is contained in the inner face of $C$ and does not intersect $C_{i-1}$.
\end{itemize}
\end{definition}

\begin{proposition}[\cite{DBLP:journals/jct/AdlerKKLST17}]\label{prop:disPathIrrelevant}
There exists $c\in\mathbb{N}$ such that for every instance $(G,T,k)$ of \pDP\ and tight sequence ${\cal C}=(C_0,C_1,\ldots,C_{2^{ck}})$ of cycles in $G$ with $T$ being in the exterior face of $C_{2^{ck}}$, any vertex 
contained in inner face of $C_0$ is irrelevant.
\end{proposition}

Because of the work of Robertson and Seymour~\cite{RobertsonS12}, in fact if there is a planar portion in a general graph with a {\em large} sequence of 
concentric cycles ${\cal C}=(C_0,\ldots,C_h)$ with $T$ being in the exterior face of $C_h$, then any vertex contained in the inner  face of 
$C_0$ is irrelevant for \DisjointPaths\ on general graphs. Next, we explain the result  (\autoref{lem:RoSe}) of Robertson and Seymour 
and deduce a corollary for general graphs similar to the one in \autoref{prop:disPathIrrelevant}, which is 
a special case of \autoref{lem:RoSe}. 
Then we will prove a  similar result when the graph has a large flat wall.

We start with some definitions and notations needed to state the result of Robertson and Seymour~\cite{RobertsonS12}.
A {\em surface} is a connected compact $2$-manifold with boundary.    
For a surface $\Sigma$, $bd(\Sigma)$ denotes the boundary of $\Sigma$.  A {\em drawing}  in $\Sigma$ is a pair $(U,V)$, 
where $U\subseteq \Sigma$ is closed, $V\subseteq U$ is finite, $U\cap bd(\Sigma)\subseteq V$, $U\setminus V$ 
has only finitely many arc-wise connected components called {\em edges}, and for each edge $e$, either $\overline{e}$ 
is homeomorphic to a circle and $\vert \overline{e}\cap V\vert=1$, or $\overline{e}$ is homeomorphic to the unit 
interval $[0,1]$ and $\overline{e}\cap V$ is the set of ends of $\overline{e}$, where $\overline{e}$ denotes the 
topological closure of $e$.  For a drawing $\Gamma=(U,V)$ in $\Sigma$, we write $U(\Gamma)=U$ and $V(\Gamma)=V$. 
For a plane graph $\Gamma$, we also use $\Gamma$ to represent its drawing in a surface $\Sigma$.  
For a subgraph $L$ of a graph $G$ and $Z\subseteq V(G)$, the {\em effect of $L$ on $Z$} is the partition of 
$V(L)\cap Z$ in which two vertices belonging to the same block of the partition if and only if they belong to the 
same component of $L$.

\begin{definition}[\cite{RobertsonS12}]
\label{def:h-insulated}
Let $h$ be a positive integer, $\Gamma$  be a drawing on a surface $\Sigma$ and $Y\subseteq \Sigma$. We say that a point $x\in \Sigma$ is {\em $h$-insulated (in $\Sigma$) from $Y$ (by $\Gamma$)} if there are $h$ disjoint circuits of $\Gamma$, all bounding discs (of circuits) in $\Sigma$ disjoint from $Y$ and containing $x$ in their interiors. More formally, there are $h$ closed discs $\Delta_1,\ldots,\Delta_h\subseteq \Sigma$ with the following properties. 
\begin{itemize}
\item $x\in \Delta_1 \setminus bd(\Delta_1)$ and $Y\cap \Delta_h=\emptyset$,
\item  for $i\in [h-1]$, $\Delta_{i} \subseteq \Delta_{i+1}\setminus bd(\Delta_{i+1})$,
\item for $i\in [h]$, $bd(\Delta_i)\subseteq U(\Gamma)$. 
\end{itemize}
\end{definition}

In the above definition, we have that $\Delta_1\subseteq \Delta_2\subseteq \ldots \subseteq \Delta_h$ unlike 
the definition in ~\cite{RobertsonS12}. In the definition given in~\cite{RobertsonS12}, $\Delta_h\subseteq \Delta_{h-1}\subseteq \ldots \subseteq \Delta_1$. 
This reversal of order is to match  our notations for the planar case. 

Our proof is based on the following variant of the unique linkage theorem of Robertson and Seymour from  \cite{RobertsonS12}.

\begin{proposition}[(3.1) \cite{RobertsonS12}]
\label{lem:RoSe}
For every non-negative integer $p$, there exists $h(p)\geq 1$ with the following property. Let $\Gamma,K$ be subgraphs 
of a graph $G$ with $G=\Gamma \cup K$, and let $\Gamma$ be a drawing in a surface $\Sigma$. 
 Let $v\in V(\Gamma)$ 
($v$ is a vertex in $\Gamma$) be $h(p)$-insulated from $V(\Gamma\cap K)$ by $\Gamma$, let $Z\subseteq V(K)$ with 
$\vert Z\vert\leq p$, and let $L$ be a subgraph of $G$. Then there is a subgraph $L'$ of $G\setminus v$ with the same effect on 
$Z$ as $L$ and $L'\cap K$ is a subgraph of $L$.  
\end{proposition}

Now we are ready to restate this result for our need.

\begin{corollary}
\label{corr:RoSe1}
Let $(G,T,k)$ be an instance of \DisjointPaths. 
Let $\Gamma,K$ be subgraphs 
of a graph $G$ with $G=\Gamma \cup K$, $\Gamma$ is a plane graph, and $V(K)\cap V(\Gamma)$  is on the exterior face of $\Gamma$.  
Then, there exists an integer $h(k)\geq 1$ with the following property. 
If there is a sequence ${\cal C}=(C_0,\ldots,C_{h(k)})$ of concentric cycles in $\Gamma$ 
such that no terminal vertex in $T$ is contained in the inner face of $C_{h(k)}$,  
then any vertex  contained in the inner face of $C_0$ is irrelevant. 
\end{corollary}

\begin{proof}
$\Gamma$ is a plane graph and hence we also denote $\Gamma$ as drawing in a surface $\Sigma$. 
We define disks $\Delta_1,\ldots,\Delta_h$ as follows. For any $i\in [h]$, $\Delta_i$ is the inner 
face of $C_i$. Let $v\in V(C_0)$.  Since $V(K\cap \Gamma)$ is on the exterior face of $\Gamma$, by \autoref{def:h-insulated}, 
we have that $v$ is $h$-insulated from $V(\Gamma\cap K)$. Suppose $(G,T,k)$ be an \yes-instance  of  \DisjointPaths\  and ${\cal P}$ be a subgraph of $G$, which realizes that  $(G,T,k)$ is an \yes-instance. 
That is, ${\cal P}$ is a collection of internally vertex disjoint paths 
such that for each $\{s,t\}\in T$, there is a path from $s$ to $t$ in ${\cal P}$. 
Without loss of generality we assume that for any $\{u,v\}\in T$, $e=\{u,v\}\notin E(G)$. Otherwise, we modify our instance 
to $(G\setminus e, T\setminus \{\{u,v\}\}, k-1)$ and apply the corollary. 
Now we substitute $p=k$, $Z=\bigcup_{S\in T} N_{\cal P}(S)$ and $L={\cal P}\setminus (\bigcup_{S\in T} S)$ in \autoref{lem:RoSe}. 
That is, $L$ is a collection of $k$ vertex disjoint paths in $G$.  Since no vertex in $T$ is contained in the inner face of $C_h$, no vertex in $Z$ is contained in the inner face of $C_{h-1}$. 
By \autoref{lem:RoSe} (where $h'=h-1$),
there is a subgraph $L'$ of $G\setminus v$ with the same effect on $Z$ as $L$ and no terminal vertex 
in $T$ belongs to $L'$. That is, $L'$ is a collection of vertex disjoint paths such that there is a path $P$ with end vertices $u$ and $v$ in $L'$ if and only if there is a paths $P'$ in $L'$ with end vertices $u$ and $v$. Moreover, $V(L')\cap (\bigcup_{S\in T} S)=\emptyset$. 
Let $L^{\star}$ be the graph obtained from $L'$ by adding vertices $\bigcup_{S\in T} S$ and edges $\{\{u,v\}\colon \{u,v\}\in {\cal P},u\in (\bigcup_{S\in T} S)\}$. 
This implies that for each $\{s,t\}=S\in T$, there is a path $P_S$ from $s$ to $t$ in  $L^{\star}$ and these paths are  internally vertex disjoint. 
This implies that $(G\setminus v,T,k)$ is a \yes-instance of  \DisjointPaths.   
\end{proof}

Now we are ready to use \autoref{corr:RoSe1}, to prove a result similar to \autoref{prop:disPathIrrelevant},   
for general graphs in the presence of a large flat wall.

\begin{lemma}
\label{lemma:irrelevantflat}
Let $(G,T,k)$ be an instance of \DisjointPaths. 
Let $W$ be a flat wall in $G$, witnessed by $((A',B'),C,(G_0,G_1,\ldots,G_t),\tilde{G})$ 
Then, there exists $h(k)\geq 1$ with the following property. 
Suppose there is a sequence ${\cal C}=(C_0,\ldots,C_{h(k)})$ of concentric cycles in $\tilde{G}$. 
Let $U$ be the vertices of $\tilde{G}$ in the inner face of $C_{h(k)}$. 
Let $T^{\star}=\bigcup_{S\in T}S$. 
Suppose $U\cap T^{\star}=\emptyset$,  and for any $i\in [t]$, $V(G_i)\cap T^{\star}=\emptyset$ if $V(G_i)\cap U\neq \emptyset$.
Then any vertex  of $\tilde{G}$ contained in the inner face of $C_0$ is irrelevant. 
\end{lemma}
\begin{proof}
We first prove the following claim. 
\begin{claim}
\label{claim:irr1}
Let $i\in [t]$ be such that $V(G_i)\cap T^{\star}=\emptyset$. 
Let $H_i$ be the graph obtained from 
$G$, by deleting $V(G_i)\setminus V(G_0)$ and adding edges between every pair of vertices in $V(G_0)\cap V(G_i)$. 
Then, $(G,T,k)$ is a \yes-instance if and only if $(H_i,T,k)$ is a \yes-instance.  
\end{claim}
\begin{proof}
Suppose $(G,T,k)$ is \yes-instance. Let ${\cal P}=\{P_1,\ldots,P_k\}$ be a solution for the instance $(G,T,k)$. 
Because of Conditions~\ref{condition4}  and \ref{condition5} of  \autoref{def:KWT}, and  $V(G_i)\cap T^{\star}=\emptyset$, at most one path in ${\cal P}$ 
can intersect with the vertices in $V(G_i)\setminus V(G_0)$. 
If $V({\cal P})\cap (V(G_i)\setminus V(G_0))=\emptyset$, then the set of paths ${\cal P}$ is present in $H_i$ and hence 
$(H_i,T,k)$ is a \yes-instance. Otherwise  let $P\in {\cal P}$ be the path such that  $V({P})\cap (V(G_i)\setminus V(G_0))\neq \emptyset$. 
Let $P'$ be a maximal subpath of $P$ such that $V(P)\subseteq V(G_i)$. 
This implies that the end vertices $u$ and $v$ of $P'$ are in $V(G_i)\cap V(G_0)$, because $T^{\star}\cap V(G_i)=\emptyset$. 
Because of Conditions~\ref{condition4}, and \ref{condition5} of \autoref{def:KWT}, and  $V(G_i)\cap T^{\star}=\emptyset$, we have that $V(P)\setminus V(P')$ 
does not intersect with $V(G_i)\setminus V(G_0)$. Now by replacing the subpath $P'$, with $u-v$ (notice that $\{u,v\}\in E(H_i)$), we get a path $P^{\star}$ 
in $H_i$. This implies that $({\cal P}\setminus \{P\})\cup \{P^{\star}\}$ is a solution for $(H_i,T,k)$.   

Now we prove the reverse direction. 
Suppose $(H_i,T,k)$ is a \yes-instance. Let ${\cal P}'=\{P'_1,\ldots,P'_k\}$ be a solution to $(H_i,T,k)$.
Since  $E(H_i)\setminus E(G)$ forms a clique on at most $3$ vertices at most 
one edge in  $E(H_i)\setminus E(G)$ is used by any path in ${\cal P}'$.
Suppose no edge in $E(H_i)\setminus E(G)$ is used by any path in ${\cal P}'$, then ${\cal P}'$ is a solution 
to $(G,T,k)$ as well. Otherwise, let $P$ be the path ${\cal P}'$ such that 
$P$ has exactly one edge $\{u,v\}$ from $E(H_i)\setminus E(G)$. Then by replacing the edge $\{u,v\}$ of $P$
with a path $P'$ from $u$ to $v$ in $G_i$ using internal vertices from $V(G_i)\setminus V(G_0)$ 
(see Condition~\ref{conditionsix} of \autoref{obs:cflatmore}), we get 
a path  $P^{\star}$ in $G$ connecting the same end vertices of $P$. This implies that $({\cal P}'\setminus \{P\})\cup \{P^{\star}\}$ 
is a solution  for $(G,T,k)$. 
\end{proof}
Suppose for all $i\in [t]$, $V(G_i)\cap V(\tilde{G})$ is in the exterior face 
of $C_{h(k)}$. Then, let $\Gamma$ be the subgraph of $G$ induced by the vertices in the inner face of 
$C_{h(k)}$. Let $K$ be the minimal subgraph of $G$ such that $G=\Gamma\cup K$. Then, by \autoref{corr:RoSe1}, 
any vertex  of $\tilde{G}$ contained in the inner face of $C_0$ is irrelevant.  Otherwise, without loss 
of generality assume that there exists $j\in [t]$ such that for each $i\in [j]$, $V(G_i)\cap V(\tilde{G})$ contains 
a vertex from the inner face of $C_h$ and for each $i\in [t]\setminus [j]$, $V(G_i)\cap V(\tilde{G})$ does not contain
a vertex from the inner face of $C_h$. 
By our assumption we know that for any $i\in [j]$, $V(G_i)\cap T^{\star}=\emptyset$. 
Let $H_0=G$ and for any $i\in [j]$, $H_i$ is obtained from 
$H_{i-1}$ by deleting $V(G_i)\setminus V(G_0)$ (which is equal to $V(G_i)\setminus V(H_{i-1})$) and adding  edges between every pair of vertices in $V(G_0)\cap V(G_i)$ (which is equal to $V(H_{i-1})\cap V(G_{i})$). 
Then, by repeatedly applying \autoref{claim:irr1} on $(H_0,T,k),\ldots, (H_j,T,k)$ we have 
that these instances are equivalent. That is, $(a)$ for any $0\leq i<i'\leq j$, $(H_i,T,k)$ is a \yes-instance if and only if 
$(H_{i'},T,k)$ is a \yes-instance. Now consider the instance $(H_j,T,k)$. Let $\Gamma'$ be the subgraph of $H_j$ induced by the vertices in the inner face of 
$C_h$ and this graph is a plane graph with $C_h$ being the boundary. Let $K'$ be the subgraph of $H_j$ such that $H_j=\Gamma'\cup K'$. Then, by \autoref{corr:RoSe1},  
any vertex  of $H_j$ contained in the inner face of $C_0$ is irrelevant for the instance $(H_j,T,k)$. Hence, 
by statement $(a)$,  any vertex  of $\tilde{G}$ contained in the inner face of $C_0$ is irrelevant for $(H_0,T,k)=(G,T,k)$. 
This completes the proof of the lemma. 
\end{proof}

Now we would like to discuss a scenario where we will be able to say that many elements in a graph 
are irrelevant, first for \DisjointPaths\ and  then for \fFindFoliostar,  by being able to find many long sequences of terminal free concentric cycles.  
Towards that, we first define the notion of a wrapped noose.

\begin{definition}[{\bf Wrapped Noose}]\label{def:wrapped} 
Let $G$ be a graph and  $W$ be a flat wall in $G$ witnessed by $((A',B'),C,(G_0,G_1,\ldots,G_t),\tilde{G})$ 
where $\tilde{G}$ is a plane graph
 with $C$ being the boundary (if $G$ is planar, then $G=B'=G_0=\tilde{G}$). 
We say that a noose $F$ is {\em wrapped} in $\tilde{G}$ if there exists a sequence ${\cal C}=(C_0,C_1,\ldots,C_{r})$ of cycles in $\tilde{G}$, called a {\em wrap}, such that $\inNoose_{\tilde{G}}(F)$ is contained in the inner face of $C_0$. 
(Here,  
$r=2^{ck}$ with $c$ being the constant in \autoref{prop:disPathIrrelevant} if $G$ is planar, 
and otherwise,  $r=h(k)$, the constant mentioned in \autoref{lemma:irrelevantflat}.)

Furthermore, we say that a sequence of noose-element pairs ${\cal F}=((F_1,Z_1),\ldots,(F_s,Z_s))$ for some $s\in\mathbb{N}$, 
where $F_i$ is a noose and $Z_i\subseteq \inNoose_G^\star(F_i)$ is a set of vertices for all $i\in [s]$, is {\em wrapped} if for all $i\in[s]$, $F_i$ is wrapped in $G\setminus(\bigcup_{j\in[i-1]}Z_j)$.
\end{definition}

Throughout the section, the constant $r$ is fixed and it is the one mentioned in \autoref{def:wrapped}.

\begin{lemma}\label{lem:wrapped}
Let $(G,T,k)$ be an instance of \DisjointPaths. 
Let $W$ be a flat wall in $G$ witnessed by $((A',B'),C,(G_0,G_1,\ldots,G_t),\tilde{G})$ 
where $\tilde{G}$ is a plane graph 
with $C$ being the boundary (if $G$ is planar, then $G=B'=G_0=\tilde{G}$). 
Let ${\cal F}=((F_1,Z_1),\ldots,(F_s,Z_s))$ for some $s\in\mathbb{N}$ be a wrapped sequence of
 nooses-element pairs in $\tilde{G}$, and for all $i\in[s]$, let ${\cal C}^i=(C^i_0,C^i_1,\ldots,C^i_{r})$ be a wrap of $F_i$ in $G\setminus(\bigcup_{j\in[i-1]}Z_j)$. 
Let $U$ be the union of vertices of $\tilde{G}$ in the inner face of $C^i_r$, $i\in [s]$. 
Let $T^{\star}=\bigcup_{S\in T}S$. 
Suppose $U\cap T^{\star}=\emptyset$  and for any $i\in [t]$, $V(G_i)\cap T^{\star}=\emptyset$ if $V(G_i)\cap U\neq \emptyset$.
Then, $(G,T,k)$ is a \yes-instance if and only if $(G\setminus \bigcup_{i\in[s]}Z_i,T,k)$ is a \yes-instance.
\end{lemma}

\begin{proof}[Proof sketch]
Let $G'=G\setminus \bigcup_{i\in[s]}Z_i$. 
First consider the reverse direction of the proof. 
Suppose $(G,T,k)$ is a \no-instance. Then, since $G'$ is a subgraph of $G$,  
$(G',T,k)$ is a \no-instance. 

Now we consider the forward direction of the proof. 
That is, we assume that $(G,T,k)$ is a \yes-instance. 
To prove $(G',T,k)$ is a \yes-instance, we first prove the following claim.   
\begin{claim}
\label{claim:conc1}
Let $F$  be a wrapped noose in $\tilde{G}$ and $Z\subseteq \inNoose_{\tilde{G}}^\star(F)$. Let ${\cal C}=(C_0,C_1,\ldots,C_{r})$ be a wrap of $F$ in $G$ and $U'$ be the set of vetices of $\tilde{G}$ in  the inner face of $C_{r}$.  
Suppose $U'\cap T^{\star}=\emptyset$  and for any $i\in [t]$, $V(G_i)\cap T^{\star}=\emptyset$ if $V(G_i)\cap U'\neq \emptyset$.
If $(G,T,k)$ is a \yes-instance, then $(G\setminus Z,T,k)$ is a \yes-instance.
\end{claim}
\begin{proof}
Suppose $(G,T,k)$ is a \yes-instance. 
Notice that ${\cal C}=(C_0,C_1,\ldots,C_{r})$ is a set of concentric cycles in $G\setminus Z$. 
%
%
Since $F$ is contained in the inner face of $C_0$ 
and $Z$ is strictly contained in the interior of $F$, the sequence of cycles ${\cal C}=(C_0,C_1,\ldots,C_{r})$ is present in $G\setminus Z$. This implies that for any subset $Z'\subseteq Z$, there is a sequence of tight concentric cycles $(C_0',\ldots,C_r')$ in $G\setminus Z'$ such that at least one element from $Z\setminus Z'$ is in the 
inner face of $C_0'$.  Hence, there is a permutation  $z_1,\ldots,z_{\ell}$ of elements in $Z$ 
such that by repeatedly applying  
\autoref{lemma:irrelevantflat} (\autoref{prop:disPathIrrelevant}  if $G$ is planar) on $(G,T,k),(G\setminus \{z_1\},T,k), \ldots, (G\setminus \{z_1,\ldots,z_{\ell}\},T,k)$, we get that $(G\setminus Z,T,k)$ is a \yes-instance, because 
$(G,T,k)$ is a \yes-instance by assumption.  
\end{proof}

Notice that $(a)$ for any $i\in [s]$, $F_i$ is wrapped in $G\setminus(\bigcup_{j\in[i-1]}Z_j)$.
By using statement $(a)$ and \autoref{claim:conc1}, with the method of  induction on $i$, 
one can show   
the following statement:  
For any $i\in [s]$, $(G\setminus \bigcup_{j\in[i-1]}Z_j,T,k)$ is a \yes-instance. 
This completes the proof of the lemma. 
\end{proof}

Now we prove  the main corollary of the subsection which we will be using in the later sections.

\begin{corollary}\label{cor:wrapped}
Let $(G,\delta,t,w',s')$ be an instance of \fFindFoliostar. Recall that $\delta^{\star}=\delta+\vert R(G)\vert$. 
Let $k'=\delta^{\star}+t$. 
%
Let ${\cal F}=((F_1,Z_1),\ldots,(F_s,Z_s))$ for some $s\in\mathbb{N}$ be a wrapped sequence of nooses-element pairs 
in $\tilde{G}$, and for all $i\in[s]$, let ${\cal C}^i=(C^i_0,C^i_1,\ldots,C^i_{r})$ be a wrap of $F_i$ in $G\setminus \bigcup_{j\in[i-1]}Z_j$ (recall that $r=h(k')$).
Let $U$ be the union of vertices of $\tilde{G}$ in the inner face of $C^i_r$, $i\in [s]$.
Let ${\cal S}=\{(H,\phi_H,\varphi_H) : H \mbox{ in the $\delta^{\star}$-folio of } G\}$ be a  solution with the following properties. 
\begin{itemize}
\item There does not exist a terminal (with respect to ${\cal S}$) in $U$,
\item  there does not exist a terminal (with respect to ${\cal S}$) in $V(G_i)$, where $V(G_i)\cap U\neq \emptyset$,  and
\item there is no edge $\{u,v\}$  
in the image of $\varphi_H$ for any $H \mbox{ in the $\delta^{\star}$-folio of } G$
such that $u\in A$ and $v\in U \cup \{V(G_i) : i\in [k], V(G_i)\cap U\neq \emptyset\}$. 
\end{itemize}
Then the  $\delta^{\star}$-folios for $G$ and $G\setminus (\bigcup_{i\in[s]} Z_i)$ are the same. 
Moreover, 
$(G,\delta,t,w',s')$ has a solution ${\cal S}^{\star}=\{(H,\phi'_H,\varphi'_H) : H \mbox{ in the $\delta^{\star}$-folio of } G\}$ with the following properties.
\begin{itemize}
\item[(i)] No element from $\bigcup_{i\in[s]} Z_i$ is in the image of $\varphi'_H$ for any $H \mbox{ in the $\delta^{\star}$-folio of } G$, and  
\item[(ii)] the set of terminals with respect to ${\cal S}$ and ${\cal S}^{\star}$  are same. 
\end{itemize}
\end{corollary}

\begin{proof}
Let $G^{\star}=G\setminus A$. Recall that $\vert A\vert \leq t$. Let $H$ be a graph in the $\delta^{\star}$-folio of $G$ and consider the tuple  $(H,\phi_H,\varphi_H)$ in ${\cal S}$. Let $J$ be the realization of $H$ in $G$ obtained through $(\phi_H,\varphi_H)$. Let $J^{\star}=J\setminus A$. Since $\vert A\vert \leq t$ and $H$ is a $\delta^{\star}$-folio of $G$. The graph $J^{\star}$ can be decomposed into a set ${\cal P}$  of at most  $\delta^{\star}+t$ internally vertex disjoint paths with end vertices being in $\phi_H(H)\cup N_J(A\cap V(J))$.  That is, there exists a set $T$ of terminal pairs 
such that ${\cal P}$ is a solution of the instance $(G^{\star},T,\vert T\vert)$ of  \DisjointPaths\ such that 
$\vert T\vert \leq \delta^{\star}+t=k$ and $(\bigcup_{S\in T}S)\subseteq \phi_H(H)\cup N_J(A\cap V(J))$. 
\begin{observation}
\label{obs:jointotopo}
Let $Q$ be the minimal subgraph of $J$ such that $J={\cal P}\cup Q$. Let ${\cal P}'$ be a solution to  the instance $(G^{\star},T,\vert T\vert)$ of  \DisjointPaths. 
Then ${\cal P}'\cup Q$ is a realization of the topological minor $H$ in $ G$. 
\end{observation}

Let $Z=\bigcup_{i\in[s]} Z_i$
and $T^{\star}=\bigcup_{S\in T} S$. By the assumption  any terminal of $(\phi_H,\varphi_H)$ and any vertex in $N_J(A\cap V(J))$, does not belong to  $U \cup \{V(G_i) : i\in [k], V(G_i)\cap U\neq \emptyset\}$. This implies that 
$U\cap T^{\star}=\emptyset$  and for any $i\in [t]$, $V(G_i)\cap T^{\star}=\emptyset$ if $V(G_i)\cap U\neq \emptyset$.
Hence, by \autoref{lem:wrapped}, $(G^{\star}\setminus Z,T,\vert T \vert)$ is a \yes-instance of \DisjointPaths, and 
let ${\cal P}'$ be a solution. Then by \autoref{obs:jointotopo}, $J'={\cal P}'\cup Q$ is a realization of the topological minor $H$ in $G$ (and hence in $G\setminus Z$, because $Z\cap V(J')=\emptyset$). This completes the proof of the corollary. 
\end{proof}

\newcommand{\findfolio}{{\sc FindFolio}}


\section{Workspace}
\label{sec:workspace}

We want to find an irrelevant vertex for the given instance of \fFindFoliostar. Towards that, we 
have to use \autoref{cor:wrapped}. To apply this corollary, we need to prune down to a large ``terminal free portion'' of $\tilde{G}$, which is a plane graph. The difficulty to get such a terminal free portion is because of the fact that we do not know which are the terminal vertices.  As a first step towards that, we seek a noose-enclosed ``workspace'' that would contain an irrelevant vertex in the planar portion $\tilde{G}$, but which would also have low treewidth. 
 We define the {\em treewidth of a noose} $N$ as the treewidth of the graph $\tilde{G}[\inNoose(N)]$. In addition, we need the following definition.

\begin{definition}[{\bf Noose Grid}]\label{def:nooseGrid}
Let $G$ be a plane graph, and let $\cal N$ be a set of $a\cdot b$ nooses ordered as $N_{i,j}$ for all $i\in[a]$ and $j\in[b]$.
We say that $\cal N$ is an {\em $a\times b$-noose grid} with respect to $G$ if the following properties are satisfied.
\begin{itemize}
\item[(a)] For all $N,N'\in{\cal N}$, $\inNoose(N)\cap\inNoose(N')=\emptyset$.
\item[(b)] $\bigcup_{N\in{\cal N}}\inNoose(N) \cap V(G) = V(G)$.
\item[(c)] For all $N\in{\cal N}$, $G[\inNoose(N)\cap V(G)]$ is a connected graph such that $N$ is the minimum noose that encloses it.
\item[(d)] For all $i,i'\in[a]$ and $j,j'\in[b]$ such that $|i-i'|+|j-j'|=1$, there exist $u\in\inNoose(N_{i,j})\cap V(G)$ and $v\in\inNoose(N_{i',j'})\cap V(G)$ such that $\{u,v\}\in E(G)$.
\end{itemize}
\end{definition}

When the graph $G$ is clear from context, we omit the reference ``with respect to $G$'' in the definition above. We now formally define our notion of a workspace.

\begin{definition}[{\bf Workspace}]\label{def:workspace}
Let $G$ be a nicely drawn plane graph. 
A pair $(M,{\cal N})$ is a {\em $q$-workspace} if $G[\inNoose(M)\cap V(G)]$ is a connected graph, $M$ is the minimum noose that encloses $G[\inNoose(M)\cap V(G)]$ and $\cal N$ is a $q\times q$-noose grid with respect to $G[\inNoose(M)\cap V(G)]$. If in addition the treewidth of $M$ is a most $p$, then $(M,{\cal N})$ is a {\em $(p,q)$-workspace}.
\end{definition}

Notice that if $(M,{\cal N})$ is a $q$-workspace, then all the nooses in ${\cal N}$ are enclosed by $M$. 
Let us explicitly state the following observation concerning workspaces, which follows directly from their definition.

\begin{observation}\label{obs:workAdj}
Let $G$ be a plane graph with a $q$-workspace $(M,{\cal N})$. For every edge $\{u,v\}\in E(G)$ such that $u\in\inNoose(N_{i,j})$ for some $i,j\in[q]$, one of the following conditions holds:
\vspace{-0.5em}
\begin{itemize}
\itemsep0em 
\item $v\in N_{i',j'}$ for some $i'\in\{i-1,i,i+1\}$ and $j'\in\{j-1,j,j+1\}$.
\item $i\in \{1,q\}$ or $j\in \{1,q\}$, and $v\notin\inNoose(M)$.
\end{itemize}
\end{observation}

The rest of the section is devoted to prove the following lemma. 

\begin{lemma}\label{lem:workspace}
There exists a constant $c$ such that for every nicely drawn plane graph $G$ and integers $p,q$ such that $2<q\leq \frac{p}{2c}$, in time $p^{\OO(1)}n$ one can either compute a $(p,q)$-workspace or a tree decomposition of $G$ of width at most $cp$.
\end{lemma}

Using \autoref{prop:surfaceMinor}, we derive the following lemma.

\begin{lemma}\label{lem:surfaceMinor}
There exists a constant $c$ such that for any connected plane graph $G$ and integer $r\in\mathbb{N}$, in time $\OO(r^2n)$ one can compute either an $r$-workspace with respect to $G$ or a tree decomposition of $G$ of width at most $cr$.
\end{lemma}


\begin{proof}
Let $M$ be the minimum noose that encloses $G$. Moreover, fix $c$ to be equal to $2\widehat{c}$, where $\widehat{c}$ denotes the constant in \autoref{prop:surfaceMinor}. First, we apply the algorithm given by \autoref{prop:surfaceMinor} to obtain in time $\OO(r^2n)$ either a $2r\times 2r$ grid as a minor of $G$ or a tree decomposition of width at most $\widehat{c}(2r)=cr$. In the latter case, we are done. Thus, we next suppose that we have a $2r\times 2r$ grid $H$ whose vertex set is accordingly denoted as $\{h_{i,j}: i\in[a],j\in[b]\}$, where $\phi$ is a function witnessing that $H$ is a minor of $G$.

Initialize $V_h=\phi(h)$ for all $h\in V(H)$. Now, since $G$ is a connected graph, as long as $\bigcup_{h\in V(H)}V_h\neq V(G)$, there exist vertices $u\notin \bigcup_{h\in V(H)}V_h$ and $v\in V_h$ for some $h\in V(H)$ such that $\{u,v\}\in E(G)$. Then, insert $u$ into $V_h$. Clearly, this process terminates in less than $n$ iterations. At the end of this process, the following properties are satisfied.
\begin{enumerate}
\item\label{item:surfaceMinor1} For all $h,h'\in V(H)$, $V_h\cap V_{h'}=\emptyset$.
\item $\bigcup_{h\in V(H)} V_h = V(G)$.
\item For all $h\in V(H)$, $G[V_h]$ is a connected graph.
\item\label{item:surfaceMinor4} For all $i,i',j,j'\in[r']$ such that $|i-i'|+|j-j'|=1$, there exist $u\in V_{h_{i,j}}$ and $v\in V_{h_{i',j'}}$ such that $\{u,v\}\in E(G)$.
\end{enumerate}


For all $i,j\in[2r]$, we define $N^\star_{i,j}$ as the minimum noose that encloses $G[V_{h_{i,j}}]$. We thus defined a set ${\cal N}^\star$ of $(2r)^2$ nooses ordered as $N^\star_{i,j}$ for all $i,j\in[2r]$. If for all $i,j\in[2r]$, $\inNoose(N^\star_{i,j})\cap V(G)= V_{h_{i,j}}$,\footnote{We remark that this condition may be false as given an arbitrarily noose $N$, it is not necessarily true that $N$ encloses $G[\inNoose(N)]$.} then the proof is complete as by properties \ref{item:surfaceMinor1} to \ref{item:surfaceMinor4} above, we have that $(M,{\cal N}^\star)$ is an $2r$-workspace (from which we can clearly derive an $r$-workspace).
In what follows, we therefore suppose that there exist $i,j\in[2r]$ such that $\inNoose(N^\star_{i,j})\cap V(G)\neq V_{h_{i,j}}$. Clearly, this means that $\inNoose(N^\star_{i,j})\cap V(G)\supset V_{h_{i,j}}$. We assume w.l.o.g.~that $i,j\in[r]$.  Let $U$ be the union of all sets $V_{h_{i,j}}$ such that $i,j>r$. 

For the next argument, recall that $G$ is a plane graph, and observe that for all $i',j'\in[2r]$, 
both $G[V_{h_{i',j'}}]$ and $G\setminus V_{h_{i',j'}}$ are connected graphs. We thus have that for all $i',j'\in[2r]$, $G[V_{h_{i',j'}}]$ is contained in a face (possibly the outer face) of $G\setminus V_{h_{i',j'}}$ as well as $G\setminus V_{h_{i',j'}}$ is contained in a face of $G[V_{h_{i',j'}}]$. In particular, for all $i',j',\widehat{i},\widehat{j}\in[2r]$ such that $(i',j')\neq(\widehat{i},\widehat{j})$, the areas enclosed by $N^\star_{i',j'}$ and $N^\star_{\widehat{i},\widehat{j}}$ are either disjoint or one of them is contained in the other. However, as $\inNoose(N^\star_{i,j})\cap V(G)\neq V_{h_{i,j}}$ and 
$N^\star_{i,j}$ is the minimum noose that encloses $G[V_{h_{i,j}}]$, we have that $G\setminus V_{h_{i,j}}$ is contained in an interior face of $G[V_{h_{i,j}}]$. This implies that for all $i',j'\in[2r]$ such that $(i',j')\neq(i,j)$, $G[V_{h_{i',j'}}]$ is contained in an interior face of $G\setminus V_{h_{i',j'}}$. Indeed, if for some $i',j'\in[2r]$, $G[V_{h_{i',j'}}]$ were contained in the outer face of $G\setminus V_{h_{i',j'}}$, then it would not have been possible for $G\setminus V_{h_{i,j}}$ to be contained in an interior face of $G[V_{h_{i,j}}]$ (because $G[V_{h_{i',j'}}]$ is a subgraph of $G\setminus V_{h_{i,j}}$). In particular, we have that for all $i',j',\widehat{i},\widehat{j}\in[2r]$ such that $(i',j'), (\widehat{i},\widehat{j})$ and $(i,j)$ are three distinct pairs, it holds that the areas enclosed by $N^\star_{i',j'}$ and $N^\star_{\widehat{i},\widehat{j}}$ are disjoint. Indeed, if for some such $i',j',\widehat{i},\widehat{j}\in[2r]$, it were true that, say, the area enclosed by $N^\star_{i',j'}$ is contained in the area enclosed by $N^\star_{\widehat{i},\widehat{j}}$, then $G[V_{h_{\widehat{i},\widehat{j}}}]$ would have not been contained in an interior face of $G\setminus V_{h_{\widehat{i},\widehat{j}}}$ (because $G[V_{h_{i',j'}}]$ is a subgraph of $G\setminus V_{h_{\widehat{i},\widehat{j}}}$).

Now, observe that $G[U]$ and $G\setminus U$ are connected graphs. Thus, since for all $i',j'\in[2r]$ such that $(i',j')\neq(i,j)$, $G[V_{h_{i',j'}}]$ is contained in an interior face of $G\setminus V_{h_{i',j'}}$, we have that $G[U]$ is contained in a face of $G\setminus U$. Therefore, the minimum noose $M^\star$ that encloses $G\setminus U$ satisfies $\inNoose(M^\star)=U$. Thus, the desired $r$-workspace is $(M^\star,\{N^\star_{i,j}: i,j>r\})$.
\end{proof}

As a corollary of \autoref{lem:surfaceMinor}, we have the following result.

\begin{corollary}\label{cor:prelimWork}
There exists a constant $c$ such that for any plane graph $G$, connected subgraph $G'$ of $G$ such that the minimum noose $M$ enclosing $G'$ satisfies $\inNoose(M)=V(G')$, and integer $r\in\mathbb{N}$, in time $\OO(r^2n)$ one can compute either an $r$-workspace with respect to $G'$ or a tree decomposition of $G'$ of width at most $cr$.
\end{corollary}

Having \autoref{cor:prelimWork} at hand, we can now prove \autoref{lem:workspace}.

\begin{proof}[Proof of \autoref{lem:workspace}]
Let $\widehat{c}$ denote the constant in \autoref{cor:prelimWork}. Let us fix $c=2\widehat{c}$. Let $C_1,C_2,\ldots,C_t$ denote the connected components of $G$. Since $G$ is nicely drawn, for every component $C_i$ of $G$, the minimum noose $M$ enclosing $C_i$ satisfies $\inNoose(M)=V(C_i)$.
Hence, we can apply \autoref{cor:prelimWork} to every connected component of $G$, and thus either find a connected component $C_i$ of $G$ with a $2q$-workspace $(M_i,{\cal N}_i)$, or find a tree decomposition $(T_i,\beta_i)$ for every connected component $C_i$ of $G$ of width at most $2\widehat{c}q$. In the latter case, for all $i\in[t-1]$, we add an edge between some vertex of $T_i$ and some vertex of $T_{i+1}$ to obtain a tree decomposition of $G$ of width at most $2\widehat{c}q=cq\leq cp$, which completes the proof. Thus, we next suppose that we face the former case. The total time spent so far is $\OO(q^2n)$.

We initialize $G^\star$ to be the connected component $C_i$, and denote $(M^\star,{\cal N}^\star)=(M_i,{\cal N}_i)$.
Since $G$ is nicely drawn, the minimum noose $M$ enclosing $G^\star$ satisfies $\inNoose(M)\cap V(G)=V(G^\star)$. Denote ${\cal N}^1=\{N^\star_{i,j}\in{\cal N}^\star: i,j\leq q\}$ and $U^1=V(G)\cap \bigcup_{N^1_{i,j}\in{\cal N}^1}\inNoose(N^1_{i,j})$, ${\cal N}^2=\{N^\star_{i,j}\in{\cal N}^\star: i\leq q,j>q\}$ and $U^2=V(G)\cap\bigcup_{N_{i,j}^2\in{\cal N}^2}\inNoose(N^2_{i,j})$, ${\cal N}^3=\{N^\star_{i,j}\in{\cal N}^\star: i>q,j\leq q\}$ and $U^3=V(G)\cap\bigcup_{N_{i,j}^3\in{\cal N}^3}\inNoose(N^3_{i,j})$, and ${\cal N}^4=\{N^\star_{i,j}\in{\cal N}^\star: i,j> q\}$ and $U^4=V(G)\cap\bigcup_{N^4_{i,j}\in{\cal N}^4}\inNoose(N^4_{i,j})$. Observe that for all distinct $i,j\in[4]$, $U^i\cap U^j=\emptyset$. Thus, there exists $i\in[4]$ such that $|U^i|\leq |V(G^\star)|/4$. 
We arbitrarily pick such $i\in[4]$, and denote $({\cal N},U)=({\cal N}^i,U^i)$. Let $M$ be the minimum noose enclosing $G[U]$. We now argue that $M$ satisfies $\inNoose(M)\cap V(G)=U$. First, the condition $\inNoose(M)\cap (V(G)\setminus V(G^\star))=\emptyset$ follows from the fact that $M^\star$ encloses $G^\star$ and satisfies $\inNoose(M^\star)\cap V(G)=V(G^\star)$. Now, consider the condition $\inNoose(M)\cap (V(G^\star)\setminus U)=\emptyset$. Since $G[U]$ and $G^\star\setminus U$ are connected graphs, each one of these two subgraphs is contained in a face (possibly the outer face) of the other. However, due to the grid structure of these two subgraphs, each of their {\em inner} faces contains on its boundary vertices that belong to $\inNoose(N^\star_{i,j})$ for at most four nooses $N^\star_{i,j}\in{\cal N}^\star$. In addition, as $q>2$, each of these two subgraphs contains a set of vertices $W$ such that the set of vertices in $N_G(W)$ that belong to the other subgraph cannot be enclosed by at most four nooses in ${\cal N}^\star$. Thus, each subgraph among $G[U]$ and $G^\star\setminus U$ is contained in the outer face of the other subgraph. Hence, we derive that $\inNoose(M)\cap (V(G^\star)\setminus U)=\emptyset$.

We have thus shown that $(M,{\cal N})$ is a $q$-workspace with respect to $G[U]$.
By using the algorithm provided by 
\autoref{cor:prelimWork}, 
since $p/c\geq 2q$, in time $\OO(p^2|U|)$ we compute either a $p/c$-workspace $(M',{\cal N}')$ with respect to $G[U]$ or a tree decomposition of $G[U]$ of width at most $p$. In the latter case, we are done, as then it holds that $(M,{\cal N})$ is a $(p,q)$-workspace with respect to $G$. In addition, if $|U|\leq p+1$, we are also done as then again $\tw(G)\leq p$. In the former case, supposing $|U|>p+1$, we update $G^\star$ to $G[U]$ and ${\cal N}^\star$ to ${\cal N}$, and then repeat the computation starting with the identification of $({\cal N}^i,U^i)$ for all $i\in[4]$. Clearly, the process terminates, at the latest once $|U|\leq p+1$.

Finally, let us analyze the running time of our algorithm. At each iteration, the size of the current set $U$ is at most $1/4$ of the size of the former set $U$, and the running time to execute the iteration is $\OO(p^2|U|)$. Thus, the running time of all the iterations together is $\OO(p^2n/4^0 + p^2n/4^1 + p^2n/4^2 + \cdots + p^2n/4^x)=\OO(p^2n)$, where $x$ is smallest integer such that $n/4^x\leq p+1$. Thus, the total running time is $\OO(q^2n+p^2n)=\OO(p^2n)$. This completes the proof.
\end{proof}

%


\section{Fitting Solutions into Frames}\label{sec:frame}

In this section we define the frames of a $2q$-workspace of the plane graph $\tilde{G}$ and prove that 
some solutions behave {\em nicely} in some frames. That is, we prove that  there are ``$(\ell,\eta)$-untangled" solutions in {\em any} large terminal free region. The main result used towards proving the existence of such nice solutions is \autoref{cor:wrapped} (derived from the unique linkage theorem). 
%
In the later sections these nice solutions allow us to define {\em nice partial solutions}
and thereby  find an irrelevant vertex.

\subsection{Frames, Terminal Free Frames, Vacant Frames and Few Crossings Frames}

In this subsection we start with the definition of {\em frames} and various properties of it with respect to a solution 
of \fFindFoliostar. Solution satisfying these properties will be  very useful, because the partial solutions
arising from these solutions are {\em small}. 

\begin{definition}[{\bf Frame}]\label{def:frame}
Let $\tilde{G}$ be a plane graph with a $2q$-workspace $(M,{\cal N})$, and $\ell\in[q-1]_0$. The {\em $\ell$-frame} of $(M,{\cal N})$ is the set of nooses $\fr[{\cal N},\ell] \triangleq \{N_{i,q-\ell}\in{\cal N}: i\in \{q-\ell,\ldots,q+\ell+1\}\}\cup \{N_{i,q+\ell+1}\in{\cal N}: i\in \{q-\ell,\ldots,q+\ell+1\}\}\cup \{N_{q-\ell,j}\in{\cal N}: j\in \{q-\ell,\ldots,q+\ell+1\}\}\cup \{N_{q+\ell+1,j}\in{\cal N}: j\in \{q-\ell,\ldots,q+\ell+1\}\}$. When ${\cal N}$ is clear from context, denote $\fr[\ell]=\fr[{\cal N},\ell]$. (See \autoref{fig:frame}.)
\end{definition}

Next we define and classify the nooses in a frame into four groups and define an order among these nooses. 

\begin{definition}[{\bf Ordered Frame}]\label{def:orderedFrame}
Let $(G,\delta,t,w',s')$ be an instance of \fFindFoliostar.
Let $(M,{\cal N})$ be a $2q$-workspace in $\tilde{G}$ and $\ell\in[q-1]_0$.
The {\em order $<$} on $\fr[\ell]$ is defined as follows. First, a noose $N_{i,j}\in\fr[\ell]$ is categorized as {\bf (i)} {\em up-noose} if $i=q-\ell$, {\bf (ii)} {\em right-noose} if $q-\ell<i<q+\ell+1$ and $j=q+\ell+1$, {\bf (iii)} {\em down-noose} if $i=q+\ell+1$, and {\bf (iv)} {\em left-noose} if $q-\ell<i<q+\ell+1$ and $j=q-\ell$.
For any two nooses $N_{i,j},N_{i',j'}\in\fr[\ell]$, $N_{i,j}<N_{i',j'}$ if and only if one of the following conditions hold.
\vspace{-0.5em}
\begin{itemize}
\itemsep0em
\item The number of the type of $N_{i,j}$ is strictly smaller than the number of the type of $N_{i',j'}$.
\item $N_{i,j}$ and $N_{i',j'}$ are up-nooses such that $j<j'$.
\item $N_{i,j}$ and $N_{i',j'}$ are right-nooses such that $i<i'$.
\item $N_{i,j}$ and $N_{i',j'}$ are down-nooses such that $j>j'$.
\item $N_{i,j}$ and $N_{i',j'}$ are left-nooses such that $i>i'$.
\end{itemize}
\end{definition}

We would like to have small representations for solutions induced on the vertices of nooses in $\bigcup_{i\leq \ell}\fr[i]$, 
for some $\ell$, which we could use algorithmically. To achieve this goal, some property which we would like to have is as follows: if there is a solution, then  there is a solution whose ``interaction'' with $\fr[\ell]$  (which in some sense act as a separator) is bounded and moreover this interaction behaves ``nicely'' with respect to $\fr[\ell]$.    This will allow us to get small representations for solutions.  As we have seen in Section~\ref{sec:disjirr}, one of the important requirements for finding an  irrelevant vertex is to have a region free from terminals.

\begin{definition}[{\bf Terminal-Free Frame}]\label{def:termFreeFrame}
Let 
$(G,\delta,t,w',s')$ be an instance of \fFindFoliostar.
Let $(M,{\cal N})$ be a $2q$-workspace in $\tilde{G}$, $\ell\in[q-1]_0$ and $d\in[\ell]$ where $\ell+d\leq q-1$. A solution ${\cal S}=\{(H,\phi_H,\varphi_H) : H \mbox{ in the $\delta^{\star}$-folio of } G\}$ is {\em $(\ell,d)$-terminal free} if the 
following conditions hold, where $Q=\bigcup_{\ell-d\leq i \leq \ell+d}\fr[i]$ and $U= \bigcup_{N\in Q}\inNoose_{\tilde{G}}(N)\cap V(\tilde{G})$.  
\begin{itemize}
\item There does not exist a terminal (with respect to ${\cal S}$) in $U$,
\item  there does not exist a terminal (with respect to ${\cal S}$) in $V(G_i)$, where $V(G_i)\cap U\neq \emptyset$, $i\in [k]$,  and
\item there is no edge $\{u,v\}$  
in the image of $\varphi_H$ for any $H \mbox{ in the $\delta^{\star}$-folio of } G$
such that $u\in A$ and $v\in U \cup \{V(G_i) : i\in [k], V(G_i)\cap U\neq \emptyset\}$. 
\end{itemize}
\end{definition}

If we are given a flat wall $W$ in $G\setminus A$, witnessed by $(A',B',C ,\tilde{G}, G_0,G_1,\ldots,G_k)$ and if there is a $G_i$ for some $i\in [k]$ and a subpath of $\varphi_H(e)$ (for some edge $e$ of a graph $H$ in the $\delta^{\star}$-folio of $G$)  between two vertices in $V(G_i)\cap V(G_0)$ with internal vertices in $V(G_i)\setminus V(G_0)$, 
 then in fact we could use the path $u-v$ in $\tilde{G}$ to encode it. Consequently, we define a {\em representation of a solution} of \fFindFoliostar, which 
is simply the replacement of paths in $G_i$ with edges in $\tilde{G}$. 

\begin{definition}
Let $(G,\delta,t,w',s')$ be an instance of \fFindFoliostar\ 
and let $H$ be a  $\delta^{\star}$-folio of $G$ witnessed by $(\phi,\varphi)$. 
The representation of  $(H,\phi,\varphi)$ in $\tilde{G}\cup G$ is the tuple $(H,\tilde{\phi},\tilde{\varphi})$, which is defined as follows. For each $e\in E(H)$ and $\{u,v\}\in E(\tilde{G})\setminus E(G)$, if there is a subpath $P$ of $\varphi(e)$, from $u$ to $v$ with internal vertices from $V(G_i)\setminus V(G_0)$ for some $i\in [k]$, then replace  $P$ with $u-v$. The resulting graph is a subgraph of $\tilde{G}\cup G$ which witnesses that $H$ is topological minor of   $\tilde{G}\cup G$ and  $(\tilde{\phi},\tilde{\varphi})$ denotes  this witness. 
\end{definition}

\begin{definition}
Let 
$(G,\delta,t,w',s')$ be an instance of \fFindFoliostar\ 
and let 
${\cal S}=\{(H,\phi_H,\varphi_H) : H \mbox{ in the $\delta^{\star}$-folio of } G\}$ be a solution. 
The representation of  ${\cal S}$ is defined as 
$\tilde{\cal S}=\{(H,\tilde{\phi}_H,\tilde{\varphi}_H) :  (H,\tilde{\phi}_H,\tilde{\varphi}_H) \mbox{ is the representation of } (H,\phi_H,\varphi_H) \mbox{ in $\tilde{G}\cup G$ and } (H,\phi_H,\varphi_H)  \in {\cal S}\}$.
\end{definition}

The following easy observation can be proved from  the definitions of an $(\ell,d)$-terminal free solution and its 
representation, and \autoref{obs:cflatmore}.

\begin{observation}
\label{obs:allpresent}
Let 
$(G,\delta,t,w',s')$ 
be an instance of \fFindFoliostar. 
Let $(M,{\cal N})$ be a $2q$-workspace in $\tilde{G}$, $\ell\in[q-1]_0$ and $d\in[\ell]$ where $\ell+d\leq q-1$. 
Let 
${\cal S}$
be an $(\ell,d)$-terminal free solution to $(G,\delta,t,w',s')$ and $\tilde{S}$ be its representation.  
Let 
 $Q=\bigcup_{\ell-d\leq i \leq \ell+d}\fr[i]$ and $U=\bigcup_{N\in Q}\inNoose_{\tilde{G}}(N)\cap V(\tilde{G})$. 
Let $G'=\tilde{G}\cup G$ and $U'=  U \cup \{V(G_i)~:~i\in [k], V(G_i)\cap U\neq \emptyset\}$. 
Then, no edge in $E(G[U'])\setminus E(\tilde{G})$ and $E_G(U',A)$ is used by $\tilde{\cal S}$. That is, 
if an edge $e\in E(G[U'])$ is used by $\tilde{\cal S}$, then $e\in E(\tilde{G})$.
\end{observation}



Now we define the notion of vacant frames, which, intuitively, says that a solution 
will not use large portions from frames. Later, we prove the existence of a solution with this property 
using  \autoref{cor:wrapped}. Afterwards, we will modify such solutions using large unused portions to make solutions that cross frames ``nicely''.

\begin{definition}[{\bf Vacant Frame}]\label{def:vacantFrame}
Let $(G,\delta,t,w',s')$  be an instance of \fFindFoliostar.
 Let ${\cal S}$
 be a solution 
and $\{(H,\tilde{\phi}_H,\tilde{\varphi}_H) : H \mbox{ in the $\delta^{\star}$-folio of } G\}$  
be the representation of it.  
Let $(M,{\cal N})$ be a $2q$-workspace in $\tilde{G}$, $\ell\in[q-1]_0$, $d,\alpha \in[\ell]$ where $\ell+d< q$ and $d+\alpha <\ell$. 
Then, ${\cal S}$ 
 is {\em $(\ell,d,\alpha)$-vacant} if for any vertex $v$ in the image of $\tilde{\varphi}_H$, 
 $H$ in the $\delta^{\star}$-folio of  $G$, 
 belonging to $\inNoose(N_{i,j})$ 
for some  $N_{i,j}\in \bigcup_{\ell-d\leq t \leq \ell+d} \fr[t]$,  
it holds that $i\in\{q-\ell+d+1,\ldots,q-\ell+d+\alpha\}\cup\{q+1+\ell-d-\alpha,\ldots,q+\ell-d\}$.
\end{definition}

See \autoref{fig:frametypes} for an illustration of $(\ell,d)$-terminal free and $(\ell,d,\alpha)$-vacant solutions. 
The following observation follows from \autoref{def:vacantFrame}.

\begin{observation}
\label{obs:irr:vacant}
Let $(G,\delta,t,w',s')$  be an instance of \fFindFoliostar. Let $(M,{\cal N})$ be a $2q$-workspace in $\tilde{G}$, $\ell\in[q-1]_0$, $d,\alpha \in[\ell]$ where $\ell+d< q$ and $d+\alpha <\ell$. If there is an  $(\ell,d,\alpha)$-vacant solution, then  the vertices in $\bigcup_{q-\ell\leq j\leq q+1+\ell}\inNoose(N_{q-\ell,j})\cap V(\tilde{G})$ are irrelevant. 
\end{observation}

\begin{figure}[t]
\begin{center}
\begin{tikzpicture}[scale=0.8]
\draw [fill=green, green] (-4,0.25) rectangle (0,1.5);
\draw [fill=green, green] (-4,4.25) rectangle (0,3);

\draw [fill=green, green] (4.5,0.25) rectangle (8.5,1.5);
\draw [fill=green, green] (4.5,4.25) rectangle (8.5,3);

\draw[red] (-5,-5)--(9.5,-5)--(9.5,9.5)--(-5,9.5)--(-5,-5); 
\draw[blue] (-4,-4)--(8.5,-4)--(8.5,8.5)--(-4,8.5)--(-4,-4); 
\draw[red] (-3,-3)--(7.5,-3)--(7.5,7.5)--(-3,7.5)--(-3,-3); 
\draw[red] (-2,-2)--(6.5,-2)--(6.5,6.5)--(-2,6.5)--(-2,-2); 

\draw[blue] (0,4.5)--(0,0)--(4.5,0)--(4.5,4.5)--(0,4.5);

\node[] (a1) at (2,10) {$\fr[q-1]$}; 
\node[] (a1) at (2,4.65) {$\fr[\ell-d]$}; 
\node[] (a1) at (2,8.65) {$\fr[\ell+d]$}; 
\node[] (a1) at (2,6.65) {$\fr[\ell]$}; 

\node[] (a1) at (2,7.65) {$\fr[\ell+d_2]$};

\draw[<->] (1.5,0)--(1.5,-2); 
\draw[<->] (2.5,-2)--(2.5,-4); 

\draw[<->] (3.9,-2)--(3.9,-3);

\node[] (a1) at (1.8,-1) {$d$}; 
\node[] (a1) at (2.8,-3) {$d$}; 

\node[] (a1) at (4.2,-2.5) {$d_2$};

\draw[<->] (5.5,3)--(5.5,4.25); 
\node[] (a1) at (5.8,3.65) {$\alpha$};

\node[] (a1) at (-5,9.5) {$\bullet$}; 
\node[] (a1) at (-4.5,9.5) {$\bullet$}; 
\node[] (a1) at (-4,9.5) {$\bullet$}; 
\node[] (a1) at (-3.5,9.5) {$\bullet$}; 
\node[] (a1) at (-3,9.5) {$\bullet$}; 
\node[] (a1) at (-2.5,9.5) {$\bullet$}; 
\node[] (a1) at (-2,9.5) {$\bullet$}; 
\node[] (a1) at (-1.5,9.5) {$\bullet$}; 
\node[] (a1) at (-1,9.5) {$\bullet$}; 
\node[] (a1) at (-0.5,9.5) {$\bullet$}; 
\node[] (a1) at (0,9.5) {$\bullet$}; 
\node[] (a1) at (0.5,9.5) {$\bullet$}; 
\node[] (a1) at (1,9.5) {$\bullet$}; 
\node[] (a1) at (1.5,9.5) {$\bullet$}; 
\node[] (a1) at (2,9.5) {$\bullet$}; 
\node[] (a1) at (2.5,9.5) {$\bullet$}; 
\node[] (a1) at (3,9.5) {$\bullet$};
\node[] (a1) at (3.5,9.5) {$\bullet$};  
\node[] (a1) at (4,9.5) {$\bullet$};
\node[] (a1) at (4.5,9.5) {$\bullet$}; 
\node[] (a1) at (5,9.5) {$\bullet$};
\node[] (a1) at (5.5,9.5) {$\bullet$};  
\node[] (a1) at (6,9.5) {$\bullet$};
\node[] (a1) at (6.5,9.5) {$\bullet$};  
\node[] (a1) at (7,9.5) {$\bullet$};
\node[] (a1) at (7.5,9.5) {$\bullet$};  
\node[] (a1) at (8,9.5) {$\bullet$};
\node[] (a1) at (8.5,9.5) {$\bullet$};  
\node[] (a1) at (9,9.5) {$\bullet$};
\node[] (a1) at (9.5,9.5) {$\bullet$};

\node[] (a1) at (-5,-5) {$\bullet$};
\node[] (a1) at (-5,-4.5) {$\bullet$};
\node[] (a1) at (-5,-4) {$\bullet$};
\node[] (a1) at (-5,-3.5) {$\bullet$};
\node[] (a1) at (-5,-3) {$\bullet$};
\node[] (a1) at (-5,-2.5) {$\bullet$};
\node[] (a1) at (-5,-2) {$\bullet$};
\node[] (a1) at (-5,-1.5) {$\bullet$};
\node[] (a1) at (-5,-1) {$\bullet$};
\node[] (a1) at (-5,-0.5) {$\bullet$};
\node[] (a1) at (-5,0) {$\bullet$};
\node[] (a1) at (-5,0.5) {$\bullet$};
\node[] (a1) at (-5,1) {$\bullet$};
\node[] (a1) at (-5,1.5) {$\bullet$};
\node[] (a1) at (-5,2) {$\bullet$};
\node[] (a1) at (-5,2.5) {$\bullet$};
\node[] (a1) at (-5,3) {$\bullet$};
\node[] (a1) at (-5,3.5) {$\bullet$};
\node[] (a1) at (-5,4) {$\bullet$};
\node[] (a1) at (-5,4.5) {$\bullet$};
\node[] (a1) at (-5,5) {$\bullet$};
\node[] (a1) at (-5,5.5) {$\bullet$};
\node[] (a1) at (-5,6) {$\bullet$};
\node[] (a1) at (-5,6.5) {$\bullet$};
\node[] (a1) at (-5,7) {$\bullet$};
\node[] (a1) at (-5,7.5) {$\bullet$};
\node[] (a1) at (-5,8) {$\bullet$};
\node[] (a1) at (-5,8.5) {$\bullet$};
\node[] (a1) at (-5,9) {$\bullet$};

\node[] (a1) at (9.5,-5) {$\bullet$};
\node[] (a1) at (9.5,-4.5) {$\bullet$};
\node[] (a1) at (9.5,-4) {$\bullet$};
\node[] (a1) at (9.5,-3.5) {$\bullet$};
\node[] (a1) at (9.5,-3) {$\bullet$};
\node[] (a1) at (9.5,-2.5) {$\bullet$};
\node[] (a1) at (9.5,-2) {$\bullet$};
\node[] (a1) at (9.5,-1.5) {$\bullet$};
\node[] (a1) at (9.5,-1) {$\bullet$};
\node[] (a1) at (9.5,-0.5) {$\bullet$};
\node[] (a1) at (9.5,0) {$\bullet$};
\node[] (a1) at (9.5,0.5) {$\bullet$};
\node[] (a1) at (9.5,1) {$\bullet$};
\node[] (a1) at (9.5,1.5) {$\bullet$};
\node[] (a1) at (9.5,2) {$\bullet$};
\node[] (a1) at (9.5,2.5) {$\bullet$};
\node[] (a1) at (9.5,3) {$\bullet$};
\node[] (a1) at (9.5,3.5) {$\bullet$};
\node[] (a1) at (9.5,4) {$\bullet$};
\node[] (a1) at (9.5,4.5) {$\bullet$};
\node[] (a1) at (9.5,5) {$\bullet$};
\node[] (a1) at (9.5,5.5) {$\bullet$};
\node[] (a1) at (9.5,6) {$\bullet$};
\node[] (a1) at (9.5,6.5) {$\bullet$};
\node[] (a1) at (9.5,7) {$\bullet$};
\node[] (a1) at (9.5,7.5) {$\bullet$};
\node[] (a1) at (9.5,8) {$\bullet$};
\node[] (a1) at (9.5,8.5) {$\bullet$};
\node[] (a1) at (9.5,9) {$\bullet$};

\node[] (a1) at (-4.5,-5) {$\bullet$}; 
\node[] (a1) at (-4,-5) {$\bullet$}; 
\node[] (a1) at (-3.5,-5) {$\bullet$}; 
\node[] (a1) at (-3,-5) {$\bullet$}; 
\node[] (a1) at (-2.5,-5) {$\bullet$}; 
\node[] (a1) at (-2,-5) {$\bullet$}; 
\node[] (a1) at (-1.5,-5) {$\bullet$}; 
\node[] (a1) at (-1,-5) {$\bullet$}; 
\node[] (a1) at (-0.5,-5) {$\bullet$}; 
\node[] (a1) at (0,-5) {$\bullet$}; 
\node[] (a1) at (0.5,-5) {$\bullet$}; 
\node[] (a1) at (1,-5) {$\bullet$}; 
\node[] (a1) at (1.5,-5) {$\bullet$}; 
\node[] (a1) at (2,-5) {$\bullet$}; 
\node[] (a1) at (2.5,-5) {$\bullet$}; 
\node[] (a1) at (3,-5) {$\bullet$};
\node[] (a1) at (3.5,-5) {$\bullet$};  
\node[] (a1) at (4,-5) {$\bullet$};
\node[] (a1) at (4.5,-5) {$\bullet$}; 
\node[] (a1) at (5,-5) {$\bullet$};
\node[] (a1) at (5.5,-5) {$\bullet$};  
\node[] (a1) at (6,-5) {$\bullet$};
\node[] (a1) at (6.5,-5) {$\bullet$};  
\node[] (a1) at (7,-5) {$\bullet$};
\node[] (a1) at (7.5,-5) {$\bullet$};  
\node[] (a1) at (8,-5) {$\bullet$};
\node[] (a1) at (8.5,-5) {$\bullet$};  
\node[] (a1) at (9,-5) {$\bullet$};

\node[] (a1) at (1,1) {$\bullet$}; 
\node[] (a1) at (1.5,1) {$\bullet$}; 
\node[] (a1) at (2,1) {$\bullet$}; 
\node[] (a1) at (2.5,1) {$\bullet$}; 
\node[] (a1) at (3,1) {$\bullet$};
\node[] (a1) at (3.5,1) {$\bullet$};  

\node[] (a1) at (1,1.5) {$\bullet$}; 
\node[] (a1) at (1.5,1.5) {$\bullet$}; 
\node[] (a1) at (2,1.5) {$\bullet$}; 
\node[] (a1) at (2.5,1.5) {$\bullet$}; 
\node[] (a1) at (3,1.5) {$\bullet$};
\node[] (a1) at (3.5,1.5) {$\bullet$};  

\node[] (a1) at (1,2) {$\bullet$}; 
\node[] (a1) at (1.5,2) {$\bullet$}; 
\node[] (a1) at (2,2) {$\bullet$}; 
\node[] (a1) at (2.5,2) {$\bullet$}; 
\node[] (a1) at (3,2) {$\bullet$};
\node[] (a1) at (3.5,2) {$\bullet$};  

\node[] (a1) at (1,2.5) {$\bullet$}; 
\node[] (a1) at (1.5,2.5) {$\bullet$}; 
\node[] (a1) at (2,2.5) {$\bullet$}; 
\node[] (a1) at (2.5,2.5) {$\bullet$}; 
\node[] (a1) at (3,2.5) {$\bullet$};
\node[] (a1) at (3.5,2.5) {$\bullet$};  

\node[] (a1) at (1,3) {$\bullet$}; 
\node[] (a1) at (1.5,3) {$\bullet$}; 
\node[] (a1) at (2,3) {$\bullet$}; 
\node[] (a1) at (2.5,3) {$\bullet$}; 
\node[] (a1) at (3,3) {$\bullet$};
\node[] (a1) at (3.5,3) {$\bullet$};  

\node[] (a1) at (1,3.5) {$\bullet$}; 
\node[] (a1) at (1.5,3.5) {$\bullet$}; 
\node[] (a1) at (2,3.5) {$\bullet$}; 
\node[] (a1) at (2.5,3.5) {$\bullet$}; 
\node[] (a1) at (3,3.5) {$\bullet$};
\node[] (a1) at (3.5,3.5) {$\bullet$};  

\draw plot [smooth] coordinates {(-4,4) (-2.3,3.2) (-2.3,3.6) (-3,3.8)(0,4)};
\node[] (a1) at (-3.5,4) {$P_1$};  
\draw plot [smooth] coordinates {(-4,1) (-3.2,1.2) (-4,0.35)};
\node[] (a1) at (-3.5,0.3) {$P_2$};  

\draw plot [smooth] coordinates {(4.5,1) (7.3,1.2) (4.5,0.35)};
\node[] (a1) at (5,0.3) {$P_3$};

\end{tikzpicture}
\end{center}
\caption{The frames between two blue colored frames is the union of $\bigcup_{\ell-d\leq j\leq \ell+d}\fr[j]$. If a solution ${\cal S}$ is $(\ell,d)$-terminal free then any vertex in this region is not a terminal w.r.t.~${\cal S}$. If ${\cal S}$ is $(\ell,d,\alpha)$-vacant, then vertices used by ${\cal S}$ from this region are only from the green colored portion of it. In a solution ${\cal S}$, $(i)$ $P_1$ is an example of an  $(\ell,d)$-up left segment which is both interior and exterior, $(ii)$ $P_2$ is an example of an $(\ell,d)$-exterior down left segment, and $(iii)$ $P_3$ is an example of an $(\ell,d)$-interior down right segment. Here, $P_1$ is $(\ell,d,d_2)$-crossing because $P_1$ hits $\fr[\ell+d_1]$. However, $P_2$ is not an $(\ell,d,d_2)$-crossing segment.}
\label{fig:frametypes}
\end{figure}

Even for solutions which are $(\ell,d)$-terminal free and $(\ell,d,\alpha)$-vacant, there are solutions which will 
interact with $\fr[\ell]$ {\em many} times. To overcome this, we define the notion of {\em few crossings} 
by solutions on $\fr[\ell]$. Towards that we need to classify the {\em segments} of paths in the solution.

\begin{definition}[{\bf Annulus Segment}]\label{def:annSegment}
Let $(G,\delta,t,w',s')$ be an instance of \fFindFoliostar.
Let ${\cal S}$
 be a solution 
and $\tilde{\cal S}=\{(H,\tilde{\phi}_H,\tilde{\varphi}_H) : H \mbox{ in the $\delta^{\star}$-folio of } G\}$  be the representation of it.  
Let $(M,{\cal N})$ be a $2q$-workspace of $\tilde{G}$, $\ell\in[q-1]_0$ and $d\in[\ell]$ where $\ell+d< q$.
A path $P$ in $G$ is an {\em $(\ell,d)$-exterior segment} (resp.~{\em $(\ell,d)$-interior segment}) of ${\cal S}$ if the following conditions are satisfied.
\vspace{-0.5em}
\begin{itemize}
\itemsep0em 
\item $P$ is  a subpath of $\tilde{\varphi}_H(e)$ for some $H \mbox{ in the  $\delta^{\star}$-folio of } G$ and $e\in E(H)$.  
\item $V(P)\subseteq V(\tilde{G})$.  
\item Each endpoint of $P$ belongs to $\inNoose_{\tilde{G}}(N)$ for some $N\in \fr[\ell-d]\cup\fr[\ell+d]$, and at least one endpoint of $P$ belongs to $\inNoose_{\tilde{G}}(N)$ for some $N\in \fr[\ell+d]$ (resp.~$N\in \fr[\ell-d]$).
\item Each internal vertex on $P$ belongs to $\inNoose_{\tilde{G}}(N)$ for some $N\in \fr[i]$ where $i\in\{\ell-d+1,\ldots,\ell+d-1\}$.\footnote{Different vertices may be enclosed by different nooses.}
\end{itemize}
We also say that $P$ is an {\em $(\ell,d)$-exterior segment} (resp.~{\em $(\ell,d)$-interior segment}) of $(H,\tilde{\phi},\tilde{\varphi})$. 
Furthermore, 
\vspace{-0.5em}
\begin{enumerate}
\itemsep0em 
\item {\bf Up-Left.} $P$ is called up-left segment if  $V(P)\subseteq \bigcup_{i,j\leq q}\inNoose_{\tilde{G}}(N_{i,j})$,
\item {\bf Up-Right.}  $P$ is called up-right segment if  $V(P)\subseteq \bigcup_{i\leq q, j>q}\inNoose_{\tilde{G}}(N_{i,j})$,
\item {\bf Down-Left.} $P$ is called down-left segment if  $V(P)\subseteq \bigcup_{i>q,j\leq q}\inNoose_{\tilde{G}}(N_{i,j})$, and
\item {\bf Down-Right.} $P$ is called down-right segment if  $V(P)\subseteq \bigcup_{i,j> q}\inNoose_{\tilde{G}}(N_{i,j})$.
\end{enumerate}
\end{definition}

See \autoref{fig:frametypes}, for an illustration of $(\ell,d)$-segments. Clearly in an $(\ell,d,\alpha)$-vacant solution, any $(\ell,d)$-$x$ segment, where $x\in \{$interior, exterior$\}$, is up-left or up-right or down-left or down-right segment, as formalized in the following observation. 

%
\begin{observation}
\label{obs:ulrb}
Let $(G,\delta,t,w',s')$ be an instance of \fFindFoliostar.
Let $(M,{\cal N})$ be a $2q$-workspace in $\tilde{G}$, $\ell\in[q-1]_0$, $d,\alpha \in[\ell]$ where $\ell+d< q$ and $d+\alpha <\ell$. 
%
Let ${\cal S}$ be an $(\ell,d,\alpha)$-vacant solution and $P$ be an $(\ell,d)$-$x$ segment, where $x\in \{$interior, exterior$\}$. Then 
$P$ is either up-left or up-right or down-left or down-right segment. 
\end{observation}

\begin{proof}
Let $U=\{q-\ell+d+1,\ldots,q-\ell+d+\alpha\}$ and $L=\{q+1+\ell-d-\alpha,\ldots,q+\ell-d\}$.
Let $P$ be an $(\ell,d)$-$x$ segment, where $x\in \{\mbox{exterior, interior}\}$. 
We know that for any vertex $v\in V(P)$,  $v$ belongs to $\inNoose_{\tilde{G}}(N_{i,j})$ for some $N_{i,j}\in \fr[i']$ where $i'\in\{\ell-d,\ldots,\ell+d\}$ and $i\in U\cup L$ (because ${\cal S}$ is {\em $(\ell,d,\alpha)$-vacant}). 
Since $d+\alpha<\ell$ (by assumption), we have that the largest index $z$ in $U$ is strictly less than $q$, 
the smallest index $z'$ in $L$ is strictly more than $q$.   Also, since $P$ is a path, which is connected, 
we get that $(a)$ either $V(P)\subseteq \bigcup_{i<q,j\in [2q]}\inNoose_{\tilde{G}}(N_ij)$ or  
$V(P)\subseteq \bigcup_{i>q,j\in [2q]}\inNoose_{\tilde{G}}(N_ij)$.  
Consider the set $J=\{q-\ell+d+1,\ldots,q+\ell-d\}$, where the smallest index is strictly less than $q$ 
(because $d<\ell-\alpha\leq \ell-1$, by assumption) and the largest index is strictly more than $q$ (because $\ell-d>\alpha\geq1$, 
by assumption).  
 Also, since $P$  is connected, 
we get that $(b)$ either $V(P)\subseteq \bigcup_{i\in [2q],j<q}\inNoose_{\tilde{G}}(N_ij)$ or  
$V(P)\subseteq \bigcup_{i\in [2q], j>q}\inNoose_{\tilde{G}}(N_ij)$.
Now the observation follows from statements $(a)$ and $(b)$. 
\end{proof}

In order to define  solutions behaving nicely on some frames we want to have the number of segments that cross some frames to be as few as possible.

\begin{definition}[{\bf Crossing Segment}]\label{def:crossingPath}
Let $(G,\delta,t,w',s')$ be an instance of \fFindFoliostar, 
and let ${\cal S}$ be a solution. Let $(M,{\cal N})$ be a $2q$-workspace of $\tilde{G}$, $\ell\in[q-1]_0$ and $d_1,d_2\in[\ell]$ where $\ell+d_1< q$ and $d_2<d_1$. An $(\ell,d_1)$-exterior $y$ segment (resp.~$(\ell,d_1)$-interior $y$ segment) $P$ of ${\cal S}$, 
where $y\in\{$up-left, up-right, down-left, down-right$\}$,   
 is called an {\em $(\ell,d_1,d_2)$-exterior $y$ crossing} (resp.~{\em $(\ell,d_1,d_2)$-interior $y$ crossing}),  if $P$ contains at least one vertex from $\inNoose_{\tilde{G}}(N)$ for some $N\in \fr[\ell+d_2]$ (resp.~$N\in \fr[\ell-d_2]$). For all $x\in\{$interior, exterior$\}$ and $y\in\{$up-left, up-right, down-left, down-right$\}$, an $(\ell,d_1,d_2)$-$x$ $y$ crossing $P$ is also called an {\em $(\ell,d_1,d_2)$-crossing}. We also call $P$ an $(\ell,d_1,d_2)$-crossing of $(H,\tilde{\phi},\tilde{\varphi})$ (where $(H,\tilde{\phi},\tilde{\varphi})$ is a tuple in the representation of  ${\cal S}$) if $P$ is a subpath of $\tilde{\varphi}(e)$ for some $e\in E(H)$.
\end{definition}


\begin{definition}[{\bf Few-Crossings Frame}]\label{def:fewCrossFrame}
Let $(G,\delta,t,w',s')$ be an instance of \fFindFoliostar\ and
let ${\cal S}$
 be a solution.
Let $(M,{\cal N})$ be a $2q$-workspace of $\tilde{G}$, $\ell\in[q-1]_0$, $d_1,d_2\in[\ell]$ where $\ell+d_1< q$ and $d_2<d_1$ and $\beta\in\mathbb{N}$. The solution ${\cal S}$ has {\em $(\ell,d_1,d_2,\beta)$-few crossings} if 
for any $(H,\tilde{\phi},\tilde{\varphi})$ in the representation of ${\cal S}$, 
it does not have more than $\beta$ $(\ell,d_1,d_2)$-$x$ $y$ crossings of $(H,\tilde{\phi},\tilde{\varphi})$ 
for every $x\in\{$exterior,interior$\}$ and $y\in\{$up-left, up-right, down-left, down-right$\}$ (individually). 
\end{definition}

Now we define oriented rows and columns in a workspace, which we use shortly in a proof; this definition will
be useful later throughout the paper.

\begin{definition}[{\bf Oriented Column/Row}]\label{def:orientedColumn}
Let $(M,{\cal N})$ be a $[2q]$-workspace of a plane graph $\tilde{G}$. For all $j\in[2q]$, a path $P$ in $\tilde{G}$ is a {\em $j$-column} of $(M,{\cal N})$ if the following conditions are satisfied.
\vspace{-0.5em}
\begin{itemize}
\itemsep0em 
\item For every vertex $v\in V(P)$, there exists $N_{i,j}\in{\cal N}$ for some $i\in[2q]$ such that $v\in\inNoose_{\tilde{G}}(N_{i,j})$.
\item For every $i\in[2q-1]$, there exists exactly one edge $\{u,v\}\in E(P)$ such that $u\in\inNoose_{\tilde{G}}(N_{i,j})$ and $v\in\inNoose_{\tilde{G}}(N_{i+1,j})$.
\end{itemize}
Furthermore, we call $P$ a {\em $j$-oriented column} 
to imply that the start-vertex of $P$ is in  $\inNoose_{\tilde{G}}(N_{1,j})$ and the end-vertex is in $\inNoose_{\tilde{G}}(N_{2q,j})$. 

A {\em $j$-row} and a {\em $j$-oriented row}, for any $j\in [2q]$, are defined analogously. 
\end{definition}

\begin{lemma}
\label{lemma:pathintersect}
Let $(G,\delta,t,w',s')$ be an instance of \fFindFoliostar.
Let $(M,{\cal N})$ be a $2q$-workspace of $\tilde{G}$, $\ell\in[q-1]_0$, $d_1,d_2,\alpha\in[\ell]$ where $\ell+d_1< q$, $d_2<d_1$ 
and $d_1+\alpha < \ell$. 
Let ${\cal S}$ 
 be a $(\ell,d_1,\alpha)$-vacant solution and let  
$P$ be an {\em $(\ell,d_1,d_2)$-$x$  $y$ crossing} of ${\cal S}$, where  
$x\in\{$exterior,interior$\}$ and $y\in\{$up-left, up-right, down-left, down-right$\}$. 
For all $j\in [2q]$, let $P_{j}$ be a $j$-oriented column.  Then, 
\begin{itemize}
\item[(a)] if $x=$exterior and $y\in \{$up-left,down-left$\}$, then $P$ intersects with $P_{q-\ell-d_1+1}, \ldots, P_{q-\ell-d_2-1}$, 
\item[(b)] if $x=$interior and $y\in \{$up-left,down-left$\}$, then $P$ intersects with $P_{q-\ell+d_2+1}, \ldots, P_{q-\ell+d_1-1}$,
\item[(c)]  if $x=$exterior and $y\in \{$up-right,down-right$\}$, then $P$ intersects with $P_{q+\ell+d_2+2}, \ldots, P_{q+\ell+d_1}$, and 
\item[(d)] if $x=$interior and $y\in \{$up-right,down-right$\}$, then $P$ intersects with $P_{q+\ell-d_1+2}, \ldots, P_{q+\ell-d_2}$. 
\end{itemize}
\end{lemma}

\begin{proof}
Here, we prove condition $(a)$. The proofs of all other cases are using arguments similar to 
the proof of condition $(a)$ and hence omitted. 
Let ${\cal A}=\{N_{i,j}\in {\cal N}:i\in \{q-\ell+d_1+1, \ldots,q-\ell+d_1+\alpha\}, j \in \{q-\ell-d_1,\ldots,q-\ell+d_2\}\}$ 
and $U=\bigcup_{N\in {\cal A}} \inNoose_{\tilde{G}}(N) \cap V(\tilde{G})$.  For any $i\in \{q-\ell-d_1+1,\ldots, q-\ell-d_2-1\}$, 
$V(P_i)$ is an $(s,t)$-separator for any $s\in \bigcup_{i\in [2q],j=q-\ell-d_1}\inNoose_{\tilde{G}}(N)$ 
and $t\in \bigcup_{i\in [2q],j=q-\ell-d_2}\inNoose_{\tilde{G}}(N)$. 
Any {\em $(\ell,d_1,d_2)$-exterior  $y$ crossing} $P$ of ${\cal S}$, where $y\in \{$up-left,down-left$\}$, is a path 
(which is connected subgraph) with at least one vertex in $\bigcup_{i\in [2q],j=q-\ell-d_1}\inNoose_{\tilde{G}}(N)$ 
and at least one vertex in $\bigcup_{i\in [2q],j=q-\ell-d_2}\inNoose_{\tilde{G}}(N)$. 
Hence, condition $(a)$ follows. 
\end{proof}


We require the following auxiliary result to prove the main lemma (\autoref{lem:fewCrossingsFrame}) of the subsection.  



\begin{lemma}
\label{lem:flowgrid}
Let $G$ be a graph and $t\in {\mathbb N}$. 
Let $\{P_1,\ldots,P_{6t}\}$ and $\{P'_1,\ldots,P'_{6t}\}$ be two sets of vertex disjoint paths 
such that for any $i,j\in [6t]$, $V(P_i)\cap V(P_j')\neq \emptyset$ and $H=P_1\cup\ldots P_{6t}\cup P_1'\cup\ldots\cup P_{6t}'$
be a planar graph.  Then there is a $t\times t$ grid minor in $H$.   
\end{lemma}

\begin{proof} 
The proof is based on the relation between the treewidth and the bramble number 
of a graph. A bramble in a graph $G'$ is a set of connected subgraphs
of $G'$ such that any two of these subgraphs have a nonempty
intersection or are joined by an edge. The order of a bramble is
the least number of vertices required to hit all the subgraphs 
in the bramble. The bramble number of a graph $G'$ is the maximum of the orders of all the brambles of $G'$. 
Seymour and Thomas~\cite{SEYMOUR199322} proved that $(a)$ the bramble number of a graph $G'$ is equal to $\tw(G')+1$.

Consider the graph $H$.
and the set of 
$(6t)^2$ subgraphs  ${\cal H}=\{P_i\cup P'_j  : i,j\in [6t] \}$. 
Since  $V(P_i)\cap V(P_j')\neq \emptyset$ for any $i,j\in [6t]$, we have that  
 all the graphs in ${\cal H}$ are connected subgraphs and they are pairwise intersecting. This implies 
that ${\cal H}$ forms a bramble. Now we claim that at least $6t$ vertices are required 
to hit all the graphs  in ${\cal H}$. Suppose not. Let $S$ be a set of vertices  such that 
$\vert S\vert < 6t$ and $S$ is a hitting set for ${\cal H}$. Since $\vert S\vert <6t$, there exist $i,j\in [6t]$ such that $V(P_i)\cap S =\emptyset$ 
and $V(P'_i)\cap S=\emptyset$. This contradicts the assumption that $S$ is a hitting set 
for  ${\cal H}$. Hence, by statement $(a)$, we conclude that $\tw(H)\geq 6t-1$. 
Thus, by \autoref{lem:planargridtw}, $H$ has a $t\times t$ grid minor.  
\end{proof}

Now we are ready to prove the main lemma of the subsection.  
We prove that if there is an $(\ell,d)$-terminal free solution, then there is an $(\ell,d_1)$-terminal free solution that is also $(\ell,d_1,\alpha)$-vacant and has $(\ell,d_1,d_2,\beta)$-few crossings, for appropriate values $d,d_1,d_2,\alpha$ and $\beta$.

\begin{lemma}\label{lem:fewCrossingsFrame}
Let $(G,\delta,t,w',s')$ be an instance of \fFindFoliostar.
Let $r$ be the constant mentioned in \autoref{cor:wrapped}.
Let $(M,{\cal N})$ be a $2q$-workspace, $\ell\in[q-1]_0$, $d,d_1,d_2,\alpha\in[\ell]$ where $\ell+d< q$, $d_2<d_1<d_1+\alpha< d<\ell$,  
$\alpha>r,d-d_1>r$ and $\beta\in\mathbb{N}$, such that  $d_1-d_2-1>2\beta>24(r+1)$. 
If there is an $(\ell,d)$-terminal free solution ${\cal S}$ to $(G,\delta,t,w',s')$, then $(G,\delta,t,w',s')$ has an $(\ell,d_1)$-terminal free solution that is also $(\ell,d_1,\alpha)$-vacant and has $(\ell,d_1,d_2,\beta)$-few crossings.
\end{lemma}

\begin{proof}
The basic idea of the proof is to define a sequence of noose-element pairs ${\cal F}$ which 
is fully contained in the terminal free region 
$\bigcup_{\ell-d\leq i \leq \ell+d} \fr[i]$ and then using  
\autoref{cor:wrapped}, we prove the required conclusion of the lemma. 
First of all, we fix a $j$-oriented column $P_j$ for all $j\in [2q]$. 
Next we define some sets of indices. 
\begin{eqnarray*}
UI&=&\{q-\ell-d,\ldots,q-\ell+d\} \qquad\qquad\qquad\quad(\emptyset \neq UI\subset {\mathbb N}, \mbox{ because }\ell+d<q) \\
UI_1=LJ&=&\{q-\ell-d_1,\ldots,q-\ell+d_1\}  \qquad\qquad\qquad(\emptyset \neq LJ\subset {\mathbb N}, \mbox{ because }\ell+d_1<q) \\
BI&=&\{q+1+\ell-d,\ldots,q+1+\ell+d\} \qquad\qquad(\emptyset \neq BI\subset {\mathbb N}, \mbox{ because }d<\ell)\\
BI_1=RJ&=&\{q+1+\ell-d_1,\ldots,q+1+\ell+d_1\} \qquad\quad(\emptyset \neq RJ\subset {\mathbb N}, \mbox{ because }d_1<\ell)\\
MI&=&\{q-\ell+d_1+1,\ldots,q+\ell-d_1\}\qquad\quad\qquad(\emptyset \neq MI\subset {\mathbb N}, \mbox{ because }d_1<\ell-1)\\
I_{\alpha}&=&\{q-\ell+d_1+\alpha+1,\ldots,q+\ell-d_1-\alpha\}\qquad(\emptyset \neq I_{\alpha}\subset {\mathbb N}, \mbox{ because }d_1+\alpha<\ell)\\
J&=&\{q-\ell-d,\ldots,q+1+\ell+d\} \qquad\qquad\qquad(\emptyset \neq J\subset {\mathbb N}, \mbox{ because }\ell+d<q)\\
J_1&=&\{q-\ell-d_1,\ldots,q+1+\ell+d_1\} \qquad\qquad\quad(\emptyset \neq J_1\subset {\mathbb N}, \mbox{ because }\ell+d_1<q)
\end{eqnarray*} 

Notice that, since $d_1<d_1+\alpha<d$, we have that $(a)$  $UI_1\subseteq UI, BI_1\subseteq BI, I_{\alpha}\subseteq MI$ and $LJ,RJ\subseteq J_1\subseteq J$. 
Since $d_1<\ell$, we have that $(b)$ $LJ\cap RJ=UI_1\cap BI_1=\emptyset$ , $UI_1\cap I_{\alpha} = \emptyset$, 
$BI_1\cap I_{\alpha} = \emptyset$ and $UI_1\cap BI_1 = \emptyset$. 
Now, we define eight sets of nooses in ${\cal N}$ as follows. 
\begin{eqnarray*}
{\cal A}_1&=& \{N_{i,j}\in {\cal N}:i\in UI_1, j \in J_1\}\\
{\cal A}_2&=& \{N_{i,j}\in {\cal N}:i\in BI_1, j \in J_1\} \\
{\cal A}_3&=& \{N_{i,j}\in {\cal N}:i\in I_{\alpha}, j \in LJ\} \\
{\cal A}_4&=& \{N_{i,j}\in {\cal N}:i\in I_{\alpha}, j \in RJ\}\\
{\cal B}_1&=& \{N_{i,j}\in {\cal N}:i\in UI, j \in J\}\\
{\cal B}_2&=& \{N_{i,j}\in {\cal N}:i\in BI, j \in J\} \\
{\cal B}_3&=& \{N_{i,j}\in {\cal N}:i\in MI, j \in LJ\} \\
{\cal B}_4&=& \{N_{i,j}\in {\cal N}:i\in MI, j \in RJ\}
\end{eqnarray*}

Statement $(a)$ implies that ${\cal A}_i\subseteq {\cal B}_i$ for all $i\in [4]$. 
Statement $(b)$ implies that 
$\{{\cal A}_i : i\in [4]\}$ is pairwise disjoint.  
By \autoref{obs:allpresent}, we have that for any $(\ell,d)$-$x$ segment $P$ of ${\cal S}$, $x\in \{$interior,exterior$\}$, 
$E(P)\subseteq E(\tilde{G})$. 
For any $i\in [4]$, let ${U}_i= \bigcup_{N\in {\cal A}_i} (\inNoose_{\tilde{G}}(N)\cap V(\tilde{G}))$.  Since $\{{\cal A}_i : i\in [4]\}$ is a family of pairwise disjoint set of nooses, we have that the family of sets $\{U_i~:~i\in [4]\}$ is pairwise disjoint. Moreover, by the properties 
$(c)$ and $(d)$ of \autoref{def:nooseGrid}, we have that $\tilde{G}[U_i]$ is a connected subgraph for all $i\in [4]$. 
Hence, by   \autoref{obs:connnoose}, we have that for any $i\in [4]$, there is a noose $F_i$ such that 
$U_i\cup E(\tilde{G}[U_i])= \inNoose^{\star}_{\tilde{G}}(F_i)$. 

Now, for each $F_i$, we define its wrap such that $\{(F_1,U_i),\ldots, (F_4,U_4)\}$ forms a  wrapped sequence  of noose-element pairs. 
For any $i\in [4]$, consider the graph $\tilde{G}_i'$ induced by the union of vertices in the nooses from ${\cal B}_i\setminus {\cal A}_i$. 
Since $d-d_1>r$ and $\alpha > r$, there is a sequence of concentric cycles ${\cal C}_i=C^i_0,\ldots, C^i_{r}$ in $\tilde{G}_i'$, 
such that $F_i$ is in the inner face of $C^i_0$ and the inner face of $C^{i}_{r}$.   
Since $d_1+\alpha>d$ and $d<\ell$ (by assumption), we have that for any $i,j\in[4],i\neq j$, ${\cal B}_i$ is disjoint from ${\cal A}_j$. 
This implies that for any $i,j\in [4]$, $i\neq j$, $V(\tilde{G}_i')$ is disjoint from $U_j$. 
Hence, we have that   $\{(F_1,U_1),\ldots, (F_4,U_4)\}$ is a  wrapped sequence  of noose-element pairs. 
Now we would like to extend the sequence $\{(F_1,U_1),\ldots, (F_4,U_4)\}$ 
by appending as many noose-element pairs as possible from $\tilde{G}^{\star}$, 
a  graph induced by the vertices from the nooses in ${\cal B}_3\cup {\cal B}_4$. 
%
For any $j\in [2q]$, let $P_j$ be a $j$-oriented column. 
Let $((N_1,Z_1),\ldots,(N_{q'},Z_{q'}))$ be a   maximal   
sequence of wrapped noose-element pairs in $\tilde{G}^{\star}$ 
with the property  that 
$${\cal F}=((F_1,U_1),(F_2,U_2),(F_3,U_3),(F_4,U_4),(N_1,Z_1),\ldots,(N_{q'},Z_{q'}))$$ is 
a sequence of  noose-element pairs and $Z_i\subseteq V(\tilde{G})\setminus \bigcup_{j\in X} V(P_j)$  
for any $i\in [q']$, where $X=\{q-\ell-d_1+1, q-\ell-d_1+3,\ldots, q-\ell-d_2-1\}$ (assume that both $d_1$ and $d_2$ are even numbers). 
By \autoref{cor:wrapped}, we know that there is a solution  ${\cal S}^{\star}=\{(\phi'_H,\varphi'_H) : H \mbox{ in $\delta^{\star}$-folio of } G\}$ such that 
\begin{itemize}
\item[$(i)$] No element from $\bigcup_{i\in [4]} U_i \cup\; \bigcup_{i\in [q']} Z_i$ is in the image of $\varphi'_H$ for 
any $H \mbox{ in $\delta^{\star}$-folio of } G$, and 
\item[$(ii)$] the sets of terminals with respect to ${\cal S}$ and ${\cal S}^{\star}$  are same.
\end{itemize}


Because of property $(ii)$, $d_1<d$, and the fact that ${\cal S}$ is $(\ell,d)$-terminal free,   
we conclude that ${\cal S}^{\star}$ is  an  $(\ell,d_1)$-terminal free solution. 
Now we show that ${\cal S}^{\star}$ is an $(\ell,d_1,\alpha)$-vacant solution. 
From the definition of $\bigcup_{i\in [4]}U_i$, we have that any vertex $v$ in $\inNoose_{\tilde{G}}(N_{i,j})$ for some $N_{i,j} \in \bigcup_{\ell-d_1\leq t \leq \ell+d_1} \fr[t]$, with $i\notin \{q-\ell+d_1+1,\ldots,q-\ell+d_1+\alpha\}\cup\{q+1+\ell-d_1-\alpha,\ldots,q+\ell-d_1\}$, also belongs to $\bigcup_{i\in [4]}U_i$. Hence, by property $(i)$, we conclude that ${\cal S}^{\star}$ is an $(\ell,d_1,\alpha)$-vacant solution.

Now we prove that ${\cal S}^{\star}$ has $(\ell,d_1,d_2,\beta)$-few crossings. 
Since ${\cal S}^{\star}$ is an $(\ell,d_1,\alpha)$-vacant solution, by \autoref{obs:ulrb}, we know  that 
for any  $(\ell,d)$-$x$ segment $P$, where $x\in \{$interior, exterior$\}$,  
$P$ is either up-left or up-right or down-left or down-right segment. 
Now, suppose ${\cal S}^{\star}$ is not a solution with $(\ell,d_1,d_2,\beta)$-few crossings. 
Then, there exist a tuple $(H,\phi,\varphi)$ in the representation of ${\cal S}^{\star}$, 
$x\in\{$exterior,interior$\}$ and $y\in\{$up-left, up-right, down-left, down-right$\}$, such that 
$(H,\phi,\varphi)$ has more than $\beta$ $(\ell,d_1,d_2)$-$x$ $y$ crossing. Here, we assume that 
$x=$exterior and $y=$up-left.  For all other cases, arguments in the proof are similar to those of the 
case when $x=$exterior and $y=$up-left and hence omitted. Let $H_3$ be the subgraph of $\tilde{G}$ induced on the vertices of the nooses in ${\cal B}_3$.  Let $\{Q_1,\ldots,Q_{\beta+1}\}$ be a set of $(\ell,d_1,d_2)$-exterior up-left crossings of $(H,\phi,\varphi)$. Notice that $\{Q_1,\ldots,Q_{\beta+1}\}$ are vertex disjoint paths in $H_3$. For any $j\in X$, let $P_j'$ be the maximal subpath of $P_j$ which is a path in $H_3$. Notice that $V(H_3)\cap V(P_j)=V(P_j')$ for all $j\in X$. By \autoref{lemma:pathintersect}, we know that each $Q_i$, $i\in [\beta+1]$, intersects with $P_{q-\ell-d_1+1},P_{q-\ell-d_1+3},\ldots, P_{q-\ell-d_2-1}$, and hence intersect with $P'_{q-\ell-d_1+1},P'_{q-\ell-d_1+3},\ldots, P'_{q-\ell-d_2-1}$. Recall the definition of $\tilde{G}^{\star}$. Since $d_1-d_2-1>2\beta$, by \autoref{lem:flowgrid}, we know that there is a $\lfloor\frac{\beta}{6}\rfloor\times\lfloor\frac{\beta}{6}\rfloor$ grid in the subgraph $H^{\star}=Q_1\cup\ldots \cup Q_{\beta}\cup P'_{q-\ell-d_1+1}\cup P'_{q-\ell-d_1+3}\ldots\cup P'_{q-\ell-d_2-1}$ of $\tilde{G}^{\star}$.   This implies that, since $\beta>12(r+1)$, there is a $(2r+2)\times (2r+2)$ grid in $H^{\star}$.  
Intuitively, this implies that there 
is a sequence of concentric cycles $C_0,\ldots, C_{r}$ such that the strict interior face of $C_0$ 
contains a cycle $O$. 
Also, since $P'_{q-\ell-d_1+1},P'_{q-\ell-d_1+3},\ldots, P'_{q-\ell-d_2-1}$ are pairwise disjoint paths, 
and any two vertices in two distinct paths are at distance at least $2$ in $\tilde{G}$ (because these paths are from ``alternate'' oriented columns) 
at least one vertex $v$ from $O$ is not a vertex in the paths $P'_{q-\ell-d_1+1},P'_{q-\ell-d_1+3},\ldots, P'_{q-\ell-d_2-1}$  and hence $v$ 
is vertex in $V(Q_1)\cup \ldots \cup V(Q_{\beta+1})$. This implies that there is a  wrapped noose-element 
pair $(F,\{v\})$ in $\tilde{G}^{\star}$ such that $v\in V(\tilde{G})\setminus \bigcup_{j\in X}V(P_j)$. Moreover,  no vertex of $H^{\star}$ belongs to $U_1\cup \ldots \cup U_4\cup Z_1,\ldots Z_{q'}$.  This contradicts the maximality of ${\cal F}$. 
This completes the proof of the lemma. 
\end{proof}

Next, by using \autoref{lem:fewCrossingsFrame}, we can prove a crucial lemma, namely, \autoref{lem:mainexrelestar}, which asserts the existence of an irrelevant vertex in the case of a large flat wall. 

\begin{proof}[Proof of \autoref{lem:mainexrelestar}]
Let $k'=\delta^{\star}+t$ and $r=h(k')$ be the constant mentioned in \autoref{cor:wrapped}. That is, $r$ depends on $\delta$ and $t$. Fix $d=27(r+2)$, $\beta=12(r+1)+1$, $d_1=25(r+2)$ and $\alpha=r+1$.
Let $q$ be an integer such that $q>2^{c'(\delta^{\star})^2}t d$, for some constant $c'$ (to be fixed later in the proof). 


Let ${\cal S}$ be a solution. 
By \autoref{prop:no.ofgraphsinfolios},  the number of distinct graphs (up to isomorphism) in the $\delta^{\star}$-folio of $G$ is upper bounded by  
$2^{\OO((\delta^{\star})^2)}\cdot  \vert R(G)\vert^{\OO(\delta^{\star})}=2^{\OO((\delta^{\star})^2)}$. 
This implies that the number of terminal vertices in ${\cal S}$ is upper bounded by 
$ 2^{c''(\delta^{\star})^2}$, where $c''$ is a constant. 
Moreover, for each $(H,\phi,\varphi)\in {\cal S}$, the number of edges with one end point in $A$ and other in $G_0\cup G_1 \cup \ldots \cup G_k$ is at most $\delta^{\star}t$. 
Therefore, there is a constant $c'$ such that 
if $q>2^{c'(\delta^{\star})^2}t d$ and  
there is $\ell \in [q]$ such that no terminal from ${\cal S}$ belongs $U$ and no edge in $\varphi_H$ (for any $(H,\phi_H,\varphi_H)\in {\cal S}$) with one endpoint in $U$ and other in $A$, where $U$ is the union of vertices in the nooses in $\bigcup_{\ell-d \leq j < \ell+d} \fr[j]$ and the union of vertices of $V(G_i)$, with $V(G_i)\cap V(G_0)$ contains a vertex in a noose in  $\bigcup_{\ell-d \leq j < \ell+d} \fr[j]$. 
This implies ${\cal S}$ is $(\ell,d)$-terminal free. Therefore, by \autoref{lem:fewCrossingsFrame}, 
there is an $(\ell,d_1,\alpha)$-vacant solution. Then, by \autoref{obs:irr:vacant}, there is an irrelevant vertex 
in $G$.

Let $\widehat{c}$ be the constant mentioned in \autoref{lem:workspace}. Let $p=2\widehat{c}q+1$ and $w=2\widehat{c}p$. We  are a given a $w\times w$ flat fall in $G\setminus A$.  Thus by \autoref{obs:wlalltogrid}, there is a ${w}\times {w}$-grid as a minor in $\tilde{G}$. This implies that the treewidth of $\tilde{G}$ strictly more than $\widehat{c}p$.  So by \autoref{lem:workspace}, there is a $(p,q)$-workspace of $\tilde{G}$.
This completes the proof of the lemma. 
\end{proof}

\subsection{Regret-Free Frames}

Even though we proved the existence of solutions with few $(\ell,d_1,d_2)$-crossings, these crossing segments may 
not be both interior and exterior segments (for example, see segment $P_3$ in \autoref{fig:frametypes}). Intuitively, if $d_1-d_2$ is large, then in a minimal solution we may not have such crossing segments. Towards formalizing it, we define when does such crossing segments are called  {\em regret}, regret cost associated with a solution and then solutions satisfying ``$(\ell,d_1,d_2,d_3)$-regret free''.

\begin{definition}[{\bf Regret}]\label{def:regret}
Let $(G,\delta,t,w',s')$ be an instance of \fFindFoliostar\ and let ${\cal S}$ be a solution.  
%
Let $(M,{\cal N})$ be a $2q$-workspace of $\tilde{G}$, $\ell\in[q-1]_0$ and $d_1,d_2\in[\ell]$ where $\ell+d_1< q$ and $d_2<d_1$. An 
$(\ell,d_1)$-$x$ $y$ segment $P$ of ${\cal S}$, where $x\in\{$exterior,interior$\}$ and $y\in\{$up-left, up-right, down-left, down-right$\}$, is an {\em $(\ell,d_1,d_2)$-$x$ $y$ regret} if $P$ is an $(\ell,d_1,d_2)$-$x$ $y$ crossing and  it does not have one endpoint in $\inNoose_{\tilde{G}}(N)$ for some $N\in \fr[\ell+d_1]$ and the other endpoint in $\inNoose_{\tilde{G}}(N')$ for some $N'\in \fr[\ell-d_1]$.
\end{definition}

\begin{definition}[{\bf Regret Cost}]\label{def:regretCost}
Let $(G,\delta,t,w',s')$ be an instance of \fFindFoliostar, 
and ${\cal S}$ be a solution. 
%
Let $(M,{\cal N})$ be a $2q$-workspace of $\tilde{G}$, $\ell\in[q-1]_0$ and $d_1,d_2\in[\ell]$ where $\ell+d_1< q$ and $d_2<d_1$. The {\em regret cost} of an $(\ell,d_1,d_2)$-$x$ $y$ regret $P$ of ${\cal S}$, where $x\in\{$exterior,interior$\}$ and $y\in\{$up-left, up-right, down-left, down-right$\}$, is a pair $(h,t)$ defined as follows. If $x=$exterior (resp.~$x=$interior), then $h$ is the smallest (resp.~largest) $h'\in\{\ell-d_2,\ldots,\ell+d_2\}$ for which there exist $v\in V(P)$ and $N\in\fr[h']$ such that $v\in\inNoose_{\tilde{G}}(N)$, and $t$ is the number of vertices $v\in V(P)$ for which there exists $N\in\fr[h']$ such that $v\in\inNoose_{\tilde{G}}(N)$.

Let $(H,\phi,\varphi)$ be a tuple in the representation of ${\cal S}$. The {\em $(\ell,d_1,d_2)$-exterior $y$ regret cost} ({\em $(\ell,d_1,d_2)$-interior $y$ regret cost}) of $(H,\phi,\varphi)$, 
for $y\in\{$up-left, up-right, down-left, down-right$\}$, is the tuple $(t_{\ell-d_2}, t_{\ell-d_2+1},\ldots,t_{\ell+d_2})$ (resp.~$(t_{\ell+d_2},t_{\ell+d_2-1},\ldots,t_{\ell-d_2})$) where for every $i\in\{\ell-d_2,\ldots,\ell+d_2\}$, $t_i$ is the sum of the second argument $t$ in the regret cost of every $(\ell,d_1,d_2)$-exterior $y$ regret  ($(\ell,d_1,d_2)$-interior $y$ regret) $P$ of $(H,\phi,\varphi)$,
 whose first argument is $h=i$.
 The {\em $(\ell,d_1,d_2)$-regret cost}  of $(H,\phi,\varphi)$ is  the entry-wise sum of the {\em $(\ell,d_1,d_2)$-$x$ $y$ regret cost}  of $(H,\phi,\varphi)$,  where the summation varies over  $x\in \{$exterior, interior$\}$ and $y\in\{$up-left, up-right, down-left, down-right$\}$. 
\end{definition}

The following definition explains how to compare regret costs of two solutions. 

\begin{definition}
Let $(G,\delta,t,w',s')$ be an instance of \fFindFoliostar\ and let $H$ be a graph in $\delta^{\star}$-folio of $G$. 
Let $(\phi_1,\varphi_1)$ and $(\phi_2,\varphi_2)$ be two witnesses for $H$ being a topological minor in $G$. 
Let $(H,\phi'_1,\varphi'_1)$ and $(H,\phi'_2,\varphi'_2)$ be the representations of $(H,\phi_1,\varphi_1)$ and $(H,\phi_2,\varphi_2)$.  
%
Let $(M,{\cal N})$ be a $2q$-workspace of $\tilde{G}$, $\ell\in[q-1]_0$ and $d_1,d_2\in[\ell]$ where $\ell+d_1< q$ and $d_2<d_1$. 
We say that the $(\ell,d_1,d_2)$-regret cost $\eta_1$ of $(H,\phi'_1,\varphi'_1)$ is smaller than $(\ell,d_1,d_2)$-regret cost $\eta_2$
of $(H,\phi'_2,\varphi'_2)$, if $\eta_1$ is lexicographically smaller than $\eta_2$.
\end{definition}

Next we define $(\ell,d_1,d_2,d_3)$-regret free solutions and prove that such a solution exists if there is an $(\ell,d)$-terminal free solution for appropriate values of $d,d_1,d_2$ and $d_3$. 

\begin{definition}[{\bf Regret-Free Frame}]\label{def:regretFreeFrame}
Let $(G,\delta,t,w',s')$ be an instance of \fFindFoliostar.
%
Let $(M,{\cal N})$ be a $2q$-workspace of $\tilde{G}$,
$\ell\in[q-1]_0$, $d_1,d_2,d_3\in[\ell]$ where $\ell+d_1< q$ and $d_3<d_2<d_1$. A solution ${\cal S}$ is {\em $(\ell,d_1,d_2,d_3)$-regret free} if it does not have an $(\ell,d_1,d_2)$-$x$ $y$ regret $P$ for any $x\in\{$exterior,interior$\}$ and $y\in\{$up-left, up-right, down-left, down-right$\}$
such that there is a vertex $v\in V(P)\cap \inNoose_{\tilde{G}}(N)$ for some $N\in\bigcup_{\ell-d_3 \leq t \leq \ell+d_3}\fr[t]$. 
\end{definition}

\begin{lemma}\label{lem:regretFreeFrame}
Let $(G,\delta,t,w',s')$ be an instance of \fFindFoliostar\ and $r$ be the constant mentioned in \autoref{cor:wrapped}.
%
Let $(M,{\cal N})$ be a $2q$-workspace of $\tilde{G}$, $\ell\in[q-1]_0$, $d,d_1,d_2,d_3,\alpha\in[\ell]$ where 
$\ell+d_1< q$ and $d_3<d_2<d_1<d$, $\alpha,d-d_1>r$ and $\beta\in\mathbb{N}$, such that  $d_1-d_2-1>2\beta>24(r+1)$ and $d_2-d_3-1>\beta$. Let ${\cal S}$ be an $(\ell,d)$-terminal free solution. Then, $(G,\delta,t,w',s')$ has an $(\ell,d_1)$-terminal free solution that is also $(\ell,d_1,\alpha)$-vacant, has $(\ell,d_1,d_2,\beta)$-few crossings and is $(\ell,d_1,d_2,d_3)$-regret free.
\end{lemma}

\begin{proof}
By \autoref{lem:fewCrossingsFrame}, we know that  $(G,\delta,t,w',s')$ has an $(\ell,d_1)$-terminal free solution which is $(\ell,d_1,\alpha)$-vacant, and has $(\ell,d_1,d_2,\beta)$-few crossings. Among all such solutions let ${\cal S}''$ be a solution such that $(a)$ the maximum $(\ell,d_1,d_2)$-regret cost $\eta$ of any $(H,\phi,\varphi)$ in the representation of ${\cal S}'$  is minimized. 
Among the solutions which are $(\ell,d_1)$-terminal free, $(\ell,d_1,\alpha)$-vacant,  with $(\ell,d_1,d_2,\beta)$-few crossings and satisfying statement $(a)$, choose a solution ${\cal S}'$ such that the number of tuples $(H,\phi,\varphi)$ in the representation of ${\cal S}'$ with $(\ell,d_1,d_2)$-regret cost $\eta$, is minimized. 
Now fix a tuple $(H,\phi,\varphi)$ in the representation of ${\cal S}'$ with 
 $(\ell,d_1,d_2)$-regret cost $\eta$. 

If  ${\cal S}'$ is $(\ell,d_1,d_2,d_3)$-regret free, then we are done. 
Otherwise, there is an $(\ell,d_1,d_2)$-$x$ $y$ regret $P$ (where $P$ is an $(\ell,d_1,d_2)$-$x$ $y$ crossing of $(H,\phi,\varphi)$) for some $x\in\{$exterior,interior$\}$ and $y\in\{$up-left, up-right, down-left, down-right$\}$
such that there is a vertex $v\in V(P)\cap \inNoose_{\tilde{G}}(N)$ for some $N\in\bigcup_{\ell-d_3 \leq t \leq \ell+d_3}\fr[t]$.
Here, we assume that $x=$ exterior and $y=$ up-left.  The proofs for all other cases are similar in arguments that of 
the case when $x=$ exterior and $y=$ up-left, and hence omitted. 
Notice that the first argument in the regret cost of $P$ is less than or equal to $\ell+d_3$ because $P$ is an $(\ell,d_1,d_2)$-exterior crossing.



Let us fix $j$-oriented columns 
$P_j$ for all $j\in \{q-\ell-d_1,\ldots,q-\ell+d_2\}$.   
Let ${\cal Q}=\{Q_1,\ldots, Q_{\beta'}\}$ be the set of  $(\ell,d_1,d_2)$-exterior up-left  crossings 
which are  $(\ell,d_1,d_2)$-exterior up-left regrets of $(H,\phi,\varphi)$. 
Since ${\cal S}'$  has $(\ell,d_1,d_2,\beta)$-few crossings, we have that $\beta' \leq \beta$.  
Notice that $P\in \{Q_1,\ldots,Q_{\beta'}\}$ and hence $q\geq 1$. 
Since there is a vertex $v\in V(P)\cap \inNoose_{\tilde{G}}{N}$ for some $N\in\bigcup_{\ell-d_3 \leq t \leq \ell+d_3}\fr[t]$, 
$P$ is an $(\ell,d_1,d_3)$-crossing as well. Hence, \autoref{lemma:pathintersect}, 
$P$ intersects with $P_j$ for any $j\in \{q-\ell-d_1+1,\ldots,q-\ell+d_3\}$. 
For each path $Q_i\in {\cal Q}$, let $u_i$ and $u_i'$ the end points of $Q_i$. Since each $Q_i$ is an  $(\ell,d_1,d_2)$-exterior up-left regret, both $u_i$ and $u_i'$ belong to $\bigcup_{i\in [2q]} \inNoose(N_{i,q-\ell-d_1})$ and no internal vertex of $Q_i$ is in $\bigcup_{i\in [2q]} \inNoose(N_{i,q-\ell-d_1})$. By properties $(c)$ and $(d)$ of \autoref{def:nooseGrid}, we know that the subgraph $J$ of $\tilde{G}$ induced on $\bigcup_{i\in [2q]} \inNoose_{\tilde{G}}( N_{i,q-\ell-d_1})\cap V(\tilde{G})$ is a connected graph. Let $J'$ be a spanning tree of $J$.  
Since each $Q_i\in {\cal Q}$ is an $(\ell,d_1)$-exterior up-left segment, 
$Q_i$ hits on $J'$ only on $u_i$ and $u_i'$. So there is a unique cycle formed by 
$J'$ and $Q_i$ and let $C_i$ be that cycle. 

\begin{claim}
\label{claim:laminar}
The inner faces of $\{C_1,\ldots,C_{\beta'}\}$ form a laminar family.
\end{claim}
\begin{proof}
Since for any $i\in [\beta']$, $Q_i$ does not contain a terminal 
with respect to ${\cal S}'$, ${\cal Q}$ is a set of vertex disjoint paths. This implies that 
for $i\neq j$, either $C_i$ is contained in the interior face of $C_j$ or  $C_j$ is contained in the interior face of $C_j$ or their inner faces are disjoint. 
Thus, we have that the inner faces of $\{C_1,\ldots,C_{\beta'}\}$ form a laminar family.
\end{proof}

\begin{claim}
\label{claim:nonhitcrossing}
Any $(\ell,d_1,d_2)$-exterior up-left crossing $Y$, which is not in ${\cal Q}$, does not hit the inner face of $C$ for any $C\in \{C_1,\ldots,C_{\beta'}\}$. 
\end{claim}

\begin{claim}
\label{claim:nocycles}
For any $C\in \{C_1,\ldots,C_{\beta'}\}$, if $C$ hits $P_{q-\ell-d_2+i}$ for any $i\in {\mathbb N}$, then  
the inner face of $C$ contains at least $i$ cycles from $\{C_1,\ldots, C_{\beta'}\}$.  
\end{claim}

\begin{proof}
Recall that $Q_1,\ldots,Q_{\beta'}$ are the $(\ell,d_1,d_2)$-exterior up-left crossings of $(H,\phi,\varphi)$, which are  $(\ell,d_1,d_2)$-exterior up-left regrets of $(H,\phi,\varphi)$.
We prove the claim using induction on $i$. The base case is when $i=1$. The inner face of $C$ 
contains $C$ itself and hence the statement follows. Now consider the induction step, where $i>1$. 
The cycle $C$ hits the path $P_{q-\ell-d_2+i}$. This implies that there is a subpath $R$ 
of $P_{q-\ell-d_2+i-1}$, of length at least one, such that $R$ is contained in the interior face of $C$, with end vertices
being (say $u$ and $v$) on $C$, 
and the path $P_C$  from $u$ to $v$ in $C$ intersects $P_{q-\ell-d_2+i}$. 
Suppose the path $R$ does not intersects with any $C'\in \{C_1,\ldots,C_r\}\setminus \{ C\}$. 
Then by \autoref{claim:nonhitcrossing}, we conclude that $R$ does not have any $(\ell,d_1,d_2)$-exterior up-left crossing. 
This implies that the internal vertices of $R$ are not in the image of $\varphi$. 
 Then by replacing the path $P_C$ with $R$, we get a new solution with smaller $(\ell,d_1,d_2)$-regret cost, which is a contradiction to the choice of ${\cal S}'$. 
Suppose $C'\in \{C_1,\ldots,C_{\beta'}\}\setminus \{ C\}$ be a cycle such that $C'$ is contained in the interior face 
of $C$, and intersects  with $R$. 
Since $C'$ intersects  with $R$, a subpath of  $P_{q-\ell-d_2+i-1}$, by induction hypothesis, 
we get that the inner face of $C'$ contains $i-1$ cycles ${\cal C}$ from $\{C_1,\ldots, C_{\beta'}\}$. These cycles 
are contained in the interior face of $C$, and by \autoref{claim:laminar}, we have that $C\notin {\cal C}$. 
Thus, ${\cal C}\cup \{C\}$ is the required set of $i$ cycles in the interior face of $C$. 
This completes the proof of the claim.  
\end{proof}

Since $P\in \{Q_1,\ldots,Q_{\beta'}\}$ and $P$ intersects with $P_{q-\ell-d_3-1}$, there is a cycle $C\in \{C_1,\ldots,C_q\}$ 
such that $C$ intersects with $P_{q-\ell-d_3-1}=P_{q-\ell-d_2+(d_2-d_3-1)}$. Thus, by \autoref{claim:nocycles}, 
we have that the interior face of $C$ contains at least $d_2-d_3-1$ cycles from $\{C_1,\ldots,C_{\beta'}\}$. 
This contradicts the assumption that $\beta'\leq \beta$, because $d_2-d_3-1>\beta$. 
 This completes the proof of the lemma. 
\end{proof}


The proof of the following lemma  is similar in arguments to that of  \autoref{lemma:pathintersect}. 

\begin{lemma}\label{lem:intersection}
Let $(G,\delta,t,w',s')$ be an instance of \fFindFoliostar.
Let $(M,{\cal N})$ be a $2q$-workspace of $\tilde{G}$, $\ell\in[q-1]_0$ and $d_1,d_2\in[\ell]$ where $\ell+d_1\leq q-1$ and $d_3<d_2<d_1$. Let ${\cal S}$ be a solution that is $(\ell,d_1,d_2,d_3)$-regret free. Let $P$ be an $(\ell,d_1,d_3)$-$y$  crossing of ${\cal S}$. 
Then, 
\begin{itemize}
\item[(a)] if $y\in \{$up-left,down-left$\}$, then $P$ intersects with $P_{q-\ell-d_1+1}, \ldots, P_{q-\ell+d_1-1}$, and 
\item[(b)] if $y\in \{$up-right,down-right$\}$, then $P$ intersects with $P_{q+\ell-d_1+1}, \ldots, P_{q+\ell+d_1-1}$.
\end{itemize}
\end{lemma}

\subsection{Untangled Frames}

In this subsection we define $(\ell,\eta)$-untangled solutions; such a solution behaves nicely on frame $\ell$, and its restriction to the frames numbered $\{0,\ldots, \ell\}$ has small size. These properties will make the computation of partial solution feasible in the later section.  In this subsection we prove that if there is an $(\ell,d)$-terminal free solution, then there is an $(\ell,\eta)$-untangled solution. Let us start with two definitions of special subgraphs in $\tilde{G}$ and $G$. 

\begin{definition}
\label{def:gltilde}
Let $(G,\delta,t,w',s')$ be an instance of \fFindFoliostar. Let $(M,{\cal N})$ be a $2q$-workspace of $\tilde{G}$ and $\ell\in[q-1]_0$. Let $U_{\ell}$ be the union of the vertices from the nooses in $\bigcup_{i \in [\ell]_0}\fr[i]$. Let ${\cal G} \subseteq \{G_1,\ldots,G_k\}$ be the graphs such that for any $i\in [k]$, $V(G_i)\cap V(G_0)\subseteq U_{\ell}$ if and only if $G_i\in {\cal G}$. Then, the graph $\tilde{G}_{\ell}$ and $G^{\star}_{\ell}$ are defined as the induced subgraphs of  $\tilde{G}\cup G_1\cup\ldots \cup G_k$ and ${G}\cup G_1\cup\ldots \cup G_k$, respectively, induced by the vertices $U_{\ell}\cup \; \bigcup_{J\in {\cal G}} V(J)$.  
\end{definition}

Next, we define disk segments and disk dangling segments, which are required to define what is an $(\ell,\eta)$-untangled solution.

\begin{definition}[{\bf Disk Segment}]\label{def:diskSegment}
Let $(G,\delta,t,w',s')$ be an instance of \fFindFoliostar.
Let ${\cal S}$ be a solution 
and ${\cal S}'$
be the representation of it.  Let $(H,{\phi}'_H,{\varphi}'_H)\in {\cal S}'$. 
Let $(M,{\cal N})$ be a $2q$-workspace of $\tilde{G}$ and $\ell\in[q-1]_0$. A path $P$ (with endpoints $u$ and $v$) in $\tilde{G}_{\ell}$ is an {\em $\ell$-interior segment} (or simply an~{\em $\ell$-segment}) of $(H,{\phi}'_H,{\varphi}'_H)$ if  there exists $e\in E(H)$  such that $P$ is a subpath of $P'=\varphi'_H(e)$ with the following properties.

\vspace{-0.5em}
\begin{enumerate}
\item There exist {\em distinct} nooses $N_u,N_v\in \fr[\ell]$ such that $u\in\inNoose_{\tilde{G}}(N_u)$ and $v\in\inNoose_{\tilde{G}}(N_v)$. 
\item $P$ can be partitioned into three subpaths, $P_u$, $P^{m}$ and $P_v$, such that every vertex on $P_u$ belongs to $\inNoose_{\tilde{G}}(N_u)$, every vertex in $P_v$ belongs to $\inNoose_{\tilde{G}}(N_v)$, and $P^m$ contains at least one vertex and every vertex on $P^m$ belongs to $V(\tilde{G}_{\ell-1})$.
\item $P'$ contains a vertex that is {\bf (i)} adjacent to $u$, {\bf (ii)} does not belong to $P$, and {\bf (iii)} does not belong to 
$V(\tilde{G}_{\ell})$. 
\item $P'$ contains a vertex that is {\bf (i)} adjacent to $v$, {\bf (ii)} does not belong to $P$, and {\bf (iii)} does not belong to 
$V(\tilde{G}_{\ell})$. 
\end{enumerate}
\end{definition}

\begin{definition}[{\bf Dangling Disk Segment}]\label{def:diskDanglingSegment}
Let $(G,\delta,t,w',s')$ be an instance of \fFindFoliostar.
Let ${\cal S}$
 be a solution 
and ${\cal S}'$
be the representation of it.  Let $(H,{\phi}'_H,{\varphi}'_H)\in {\cal S}'$. 
Let $(M,{\cal N})$ be a $2q$-workspace of $\tilde{G}$ and $\ell\in[q-1]_0$. A path $P$ (with endpoints $u$ and $v$) in $\tilde{G}_{\ell}$
is an {\em $\ell$-dangling segment} of $(H,{\phi}'_H,{\varphi}'_H)$ if there exists 
$e\in E(H)$ 
 such that $P$ is a subpath of $P'=\varphi'_H(e)$ with the following properties.  

\vspace{-0.5em}
\begin{enumerate}
\item There exists a noose $N_u\in \fr[\ell]$ such that $u\in\inNoose_{\tilde{G}}(N_u)$.
\item $P$ can be partitioned into two subpaths, $P_u$ and $P^{r}$, such that every vertex in $P_u$ belongs to $\inNoose_{\tilde{G}}(N_u)$, $P^{r}$ contains at least one vertex and every vertex on $P^{r}$ belongs to $V(\tilde{G}_{\ell-1})$. 
\item $P'$ contains a vertex that is {\bf (i)} adjacent to $u$, {\bf (ii)} does not belong to $P$, and {\bf (iii)} does not belong to 
$V(\tilde{G}_{\ell})$. 
\end{enumerate}
\end{definition}


\begin{definition}[{\bf Used Noose}]\label{def:usedNoose}
Let $(G,\delta,t,w',s')$ be an instance of \fFindFoliostar.
Let $(M,{\cal N})$ be a $2q$-workspace in $\tilde{G}$. 
Let ${\cal S}$
 be a solution 
and ${\cal S}'=\{(H,{\phi}'_H,{\varphi}'_H) : H \mbox{ in the $\delta^{\star}$-folio of } G\}$  
be the representation of it.  
 A noose $N\in{\cal N}$ is {\em used} by $(H,{\phi}'_H,{\varphi}'_H)\in {\cal S}'$ if there exists a vertex in the image of $\varphi'_H$ 
 that belongs to $\inNoose_{\tilde{G}}(N)$.
\end{definition}


Now we are ready to define what is an $(\ell,\eta)$-untangled solution and a lemma about its existence if there is an $(\ell,d)$-terminal free solution.

\begin{definition}[{\bf Untangled Frame}]\label{def:untangledFrame}
Let $(G,\delta,t,w',s')$ be an instance of \fFindFoliostar.
Let $(M,{\cal N})$ be a $2q$-workspace in $\tilde{G}$, $\ell\in[q-1]_0$ and $\eta\in\mathbb{N}$.
Let ${\cal S}$
 be a solution 
and ${\cal S}'=\{(H,{\phi}'_H,{\varphi}'_H) : H \mbox{ in the $\delta^{\star}$-folio of } G\}$  
be the representation of it.  
The solution ${\cal S}$ is {\em $(\ell,\eta)$-untangled} if it is $(\ell,3)$-terminal free and the three following conditions are satisfied for any  $(H,{\phi}'_H,{\varphi}'_H)\in {\cal S}'$. 
\vspace{-0.5em}
\begin{enumerate}
\itemsep0em 
\item At most $\eta$ nooses in $\fr[\ell]$ are used by $(H,{\phi}'_H,{\varphi}'_H)$. Each noose in $\fr[\ell]$ used by  $(H,{\phi}'_H,{\varphi}'_H)$ is either a right-noose or a left-noose.  
\item Every vertex $v$ in the image of $\varphi'_H$, belonging to a noose $N\in \fr[\ell]$, is also a vertex of an $\ell$-segment or an $\ell$-dangling segment of $(H,{\phi}'_H,{\varphi}'_H)$.
\item There do not exist a noose $N\in \fr[\ell]$ and two distinct $\ell$-(dangling) segments of $(H,{\phi}'_H,{\varphi}'_H)$, $P$ and $P'$, such that both $P$ and $P'$ intersect $\inNoose_{\tilde{G}}(N)$.
\end{enumerate}
\end{definition}

\begin{lemma}\label{lem:untangledFrame}
Let $(G,\delta,t,w',s')$ be an instance of \fFindFoliostar.
%
Let $r$ be the constant mentioned in \autoref{cor:wrapped}. Let $(M,{\cal N})$ be a $2q$-workspace in $\tilde{G}$, $\ell\in[q-1]_0$, $d,d_1,\beta\in[\ell]$  such that $4\beta+6<d_1<d<\ell$, $\ell+d_1+r< q$,  $\beta>12(r+1)$, and $d-d_1>r$. If $(G,\delta,t,w',s')$ has  an $(\ell,d)$-terminal free solution,  then it has an $(\ell,d_1)$-terminal free solution that is also $(\ell,\eta)$-untangled, where $\eta= 4\beta$. 
\end{lemma}
 
\begin{proof}
First, we fix three integers $d_2,d_3,\alpha \in {\mathbb N}$ as follows: $d_3=\beta+3,\alpha=r+1$ and $d_2=2\beta+5$.  
Notice that $d_3<d_2<d_1<d$, $\alpha>r$, $d_2-d_3-1>\beta>12(r+1)$, $d_1-d_2-1>2\beta>24(r+1)$ and $\ell+d_1<q$.
Therefore,  
since $(G,\delta,t,w',s')$ has  an $(\ell,d)$-terminal free solution, by \autoref{lem:regretFreeFrame}, we know that there is a solution  ${\cal S}$ which is $(\ell,d_1)$-terminal free, $(\ell,d_1,\alpha)$-vacant, $(\ell,d_1,d_2,\beta)$-few crossings and $(\ell,d_1,d_2,d_3)$-regret free.
Let ${\cal S}'=\{(H,\phi_H,\varphi_H) : H \mbox{ in the $\delta^{\star}$-folio of } G\}$ be a representation of ${\cal S}$. Let 
 $Q=\bigcup_{\ell-d_1\leq i \leq \ell+d_1}\fr[i]$ and $U=\bigcup_{N\in Q}\inNoose_{\tilde{G}}(N)\cap V(\tilde{G})$. 
Let $G'=\tilde{G}\cup G$ and $U'=  U \cup \{V(G_i)~:~i\in [k], V(G_i)\cap U\neq \emptyset\}$. 
By \autoref{obs:allpresent}, we know that there is no edge in $E(G[U'])\setminus E(\tilde{G})$ and $E_G(U',A)$ which is used by ${\cal S'}$. This implies that the edges used by ${\cal S}'$ in $G'[U]$  are only from the ``planar portion'' of $\tilde{G}[U]$.  
Now we modify ${\cal S}'$, to get a representation of a solution which is  $(\ell,\eta)$-untangled.

Since 
${\cal S}$ is $(\ell,d_1,\alpha)$-vacant, for any vertex $v$ in the image of $\varphi_H$ belonging to $\inNoose_{\tilde{G}}(N_{i,j})$ 
for some  $N_{i,j}\in \bigcup_{\ell-d_1\leq i' \leq \ell+d_1} \fr[i']$,  it holds that $i\in\{q-\ell+d_1+1,\ldots,q-\ell+d_1+\alpha\}\cup\{q+1+\ell-d_1-\alpha,\ldots,q+\ell-d_1\}$. Moreover, since ${\cal S}$ is $(\ell,d_1,d_2,d_3)$-regret free, we have that any $(\ell,d_1,d_3)$-crossing segment has one end point in a noose in $\fr[\ell+d_1]$ (we call it starting vertex) and other in a noose in $\fr[\ell-d_1]$ (we call it ending 
vertex) and these are the only  $(\ell,d_1)$-segments that intersect with $\fr[\ell]$.  Since  ${\cal S}$ is $(\ell,d_1,d_2,\beta)$-few crossings and $(\ell,d_1,\alpha)$-vacant, for any $(H,\phi_H,\varphi_H)\in {\cal S}'$, number of $(\ell, d_1,d_2)$-crossing segments of $(H,\phi_H,\varphi_H)$, and hence the number of $(\ell, d_1,d_3)$-crossing segments of $(H,\phi_H,\varphi_H)$ are upper bounded by $4\beta=\eta$. Even though the number 
of  $(\ell, d_1,d_3)$-crossing segments are bounded, they may intersect with $\fr[\ell]$ many times and as a result, the solution 
may not be  $(\ell,\eta)$-untangled.  

In what follows we show how to construct a representation ${\cal Z}'$ of a solution ${\cal Z}$   from ${\cal S}'$, such that 
${\cal Z}$ is $(\ell,\eta)$-untangled. Towards 
that we fix a tuple $(H,\phi_H,\varphi_H)\in {\cal S'}$ and reroute the $(\ell,d_1,d_3)$-crossing segments $(H,\phi_H,\varphi_H)$ such that we can get the required properties. Notice that $(\ell,d_1,d_3)$-crossing segments can be partitioned into $(\ell,d_1,d_3)$-up-left crossing segments, $(\ell,d_1,d_3)$-down-left crossing segments, $(\ell,d_1,d_3)$-up-right crossing segments and $(\ell,d_1,d_3)$-down-right crossing segments.
Let 

\begin{eqnarray*}
UL&=&\bigcup_{\substack{q-\ell+d_1+1\leq i\leq q-\ell+\alpha\\ q-\ell-d_1\leq j \leq q-\ell+d_1}}\inNoose_{\tilde{G}}(N_{i,j})\\
DL&=&\bigcup_{\substack{q+1+\ell-d_1-\alpha\leq i\leq q+1+\ell-d_1\\ q-\ell-d_1\leq j \leq q-\ell+d_1}}\inNoose_{\tilde{G}}(N_{i,j})\\
UR&=&\bigcup_{\substack{q-\ell+d_1+1\leq i\leq q-\ell+d_1+\alpha\\ q+1+\ell-d_1\leq j \leq q+1+\ell+d_1}}\inNoose_{\tilde{G}}(N_{i,j})\\
DR&=&\bigcup_{\substack{q+1+\ell-d_1-\alpha\leq i\leq q+1+\ell-d_1\\ q+1+\ell-d_1\leq j \leq q+1+\ell+d_1}}\inNoose_{\tilde{G}}(N_{i,j})
\end{eqnarray*}

 While re-routing $(\ell,d_1,d_3)$-up-left crossing segments, $(\ell,d_1,d_3)$-down-left crossing segments, 
$(\ell,d_1,d_3)$-up-right crossing segments and $(\ell,d_1,d_3)$-down-right crossing segments,  the new paths will use only vertices 
from  $UL,DL,UR$ and $DR$, respectively. Since the sets  $UL,DL,UR$ and $DR$ are pairwise disjoint, these re-routing will not intersect each other. 
Here, we explain how to re-route  $(\ell,d_1,d_3)$-up-left crossing segments towards satisfying the properties of $(\ell,\eta)$-untangled. Other cases are symmetric to the case  of $(\ell,d_1,d_3)$-up-left crossing segments and hence omitted.

Let us fix four integers $w=q-\ell-d_1$, $e=q-\ell+d_1$, $n=q-\ell+d_1+1$ and $s=q-\ell+d_1+\alpha$. We know that the number of $(\ell,d_1,d_3)$-up-left crossing segments is at most $\beta$ and any such segment $S$ will have its starting vertex in a noose $N_{i,w}$ and ending vertex in a noose $N_{i',e}$, where  $i,i'\in \{n,\ldots,s\}$.  Moreover, each internal vertex of $V(S)$ will be in a noose  $N_{i,j}$ where $i \in \{n,\ldots,s\}$ and $j\in  \{w+1,\ldots,e-1\}$. Let $S_1,\ldots,S_{\beta'}$ be all the  $(\ell,d_1,d_3)$-up-left crossing segments, where $\beta'\leq \beta$. For any $i\in [\beta']$, let $u_i$ be the start-vertex of $S_i$ and $u_i'$ be the end-vertex of $S_i$.  Now we fix a $j$-oriented column $P_j$ for all $j\in \{w,\ldots,e\}$. Now fix four vertices $a \in \inNoose(N_{n-1,w-1})$, $a' \in \inNoose(N_{n-1,e+1})$, $b\in \inNoose(N_{s+1,w-1})$   and $b'\in \inNoose(N_{s+1,e+1})$. Let $Q_A$ be a $a$-$a'$ path in $\bigcup_{ w-1\leq j \leq e+1} \inNoose(N_{n-1,j})$ 
such that for any $j\in \{w,\ldots,e\}$, $V(P_j)\cap V(Q_A)$ appears consecutively in 
the path $Q_A$ and forms a subpath of $P_j$.   
Let  $Q_B$ be a $b$-$b'$ path  
in $\bigcup_{ w-1\leq j \leq e+1} \inNoose(N_{s+1,j})$ 
such that for any $j\in \{w,\ldots,e\}$, $V(P_j)\cap V(Q_B)$ appears consecutively in 
the path $Q_B$ and forms a subpath of $P_j$.  
Let $C$ be the unique cycle formed $P_w$, $Q_B$, $P_{e}$ 
and $Q_A$. Let $p_j$ be the last vertex of $P_j$ which intersects $Q_A$ for any $j\in \{w,\ldots,e\}$. 

Now we claim that for any $i\in [\beta']$, $V(S_i)$ is contained in the inner face of $C$. 
Notice that all the internal vertices of $V(S_i)$ is strictly contained in the interior face of $C$. 
Hence for any $i\in [\beta']$, $V(S_i)$ is contained in the interior face of $C$. 
Also, for any $i\in [\beta']$, $S_i$ intersects neither with $Q_A$ nor with $Q_B$.
This implies that $u_i$ is either in the  interior face of $C$ or intersects with $P_w$ and $u_i'$ is either in the  interior face of $C$ or intersects with $P_e$.   
Since ${\cal S}$ is an $(\ell,d_1)$-terminal free solution, the set of paths $\{S_1,\ldots,S_{\beta'}\}$ 
is pairwise vertex disjoint.
For each $i\in [\beta']$, let $R_i$, $R_i'$ be  shortest paths from $u_i$ to $P_w$ in $\tilde{G}[\bigcup_{i}\inNoose(N_{i,w})]$ and $u_i'$ to $P_e$ in $\tilde{G}[\bigcup_{i}\inNoose(N_{i,e})]$, respectively, 
such that $\vert \bigcup_{i\in [\beta']} E(R_i)\cup E(R_i')\vert$ is minimized among all such potential choices.  
Let $X_i=R_iS_iR_i'$.  Let $x_i$ and $x_i'$ be the start-vertex and end-vertex of 
$X_i$, respectively.  

For any $i\in [\beta']$, let $C_i$ be the cycle formed by $X_i,P_e,Q_A$ and $P_w$.
Notice that $(a)$ any point {\em close} to $X_i$ and at the {\em left side} of the curve corresponding 
to the path $X_i$ {\em from} $x_i$ to $x_i'$, and $(b)$   
any point {\em close} to $Q_A$ and at the {\em right side} of the curve corresponding 
to the path $Q_A$ {\em from} $a$ to $a'$ are in the strict interior face of $C_i$.

\begin{observation}
\label{obs:intcontain}
For any $i,j\in [\beta']$, $i\neq j$, interior face of $C_i$ is contained in the interior face of 
$C_j$ or vice versa. 
\end{observation}

\begin{proof}

We first prove that either $R_i$ and $R_j$ are disjoint or $R_i=RA_i$ and $R_j=RA_j$  
such that $A_i$ and $A_j$ intersects only at the starting vertex.  Suppose not, then 
there exist a cycle containing $R_i$ and $R_j$. This will contradict the minimality assumption of 
$\vert \bigcup_{i\in [\beta']} E(R_i)\cup E(R_i')\vert$. 
By similar arguments we can show that  either $R_i'$ and $R_j'$ are disjoint or $R_i'=R'A_i'$ and $R_j'=R'A_j'$  
such that $A'_i$ and $A'_j$ intersects only at the starting vertex.
Moreover, we know that $S_i$ and $S_j$ are vertex disjoint. As a result the paths 
$X_i=R_iS_iR_i'$ and $X_j=R_jS_jR_j'$ will have one of the following forms. 
\begin{itemize}
\item $X_i$ and $X_j$ are vertex disjoint, or 
\item $X_i=W_iZ$, $X_j=W_jZ$ such that $W_i$ and $W_j$ intersects only at the end-vertex, or 
\item $X_i=ZW_i$, $X_j=ZW_j$ such that $W_i$ and $W_j$ intersects only at the start-vertex, or 
\item $X_i=Z_1W_iZ_2$, $X_j=Z_1W_jZ_2$ such that $W_i$ and $W_j$ intersects only at the start-vertex and end-vertex. 
\end{itemize}
Notice that any point to the right side of the curve $Q_A$ and close to $Q_A$ is part of both $C_i$ and $C_j$. 
Also, since $X_i$ and $X_j$ have the above mentioned properties, the claim follows. 
\end{proof}

Due to \autoref{obs:intcontain}, there is a linear  order $\subset$ over $\{C_1,\ldots,C_{\beta'}\}$. 
Without loss of generality assume that 
$$C_1\subset C_2\subset\ldots \subset C_{\beta'}.$$

Let $P$ be a path. Recall that $P_j$ is  a $j$-oriented column for some $j\in \{w,\ldots,e\}$. Then, 
if $P$ intersects with $P_j$, 
 then $\topm(P,P_j)$
 is defined as the first vertex on $P_j$ that belongs to $P$.

\begin{claim}
\label{claim:top}
Let $j\in \{w,\ldots,e\}$.
For any $i,i'\in [\beta']$ and $i<i'$, $\topm(S_i,P_j)$ precedes $\topm(S_{i'},P_j)$. 
\end{claim}
\begin{proof}
Notice that $C_i\subset C_{i'}$ and point close to $p_j$, but after $p_j$ in the curve 
corresponding to $P_j$ is contained in $C_i$. 
Now follow the curve  corresponding to the oriented column $P_j$ from $p_j$ and since $C_i\subset C_{i'}$ and $S_i$ is vertex disjoint from $S_j$, this curve first hits the boundary of $C_i$ strictly before the boundary of $C_{i'}$. 
\end{proof}

Since $X_i=R_iS_iR_i'$ and $R_i$ and $R_i'$ are paths $\tilde{G}[\bigcup_{j}\inNoose(N_{j,w})]$ and  $\tilde{G}[\bigcup_{j}\inNoose(N_{j,e})]$, respectively, we have that 
if $v$ is a vertex in $X_i$ which intersect with $P_{q-\ell-d_3+i}$, then $v\in V(S_i)$.  
Let $t_i=\topm(S_i,P_{q-\ell-d_3+i})=\topm(X_i,P_{q-\ell-d_3+i})$ and let $P'_i$ be the subpath of $P_{q-\ell-d_3+i}$ from $p_{q-\ell-d_3+i}$ to $t_i$ for all $i\in [\beta']$. (Recall that $p_{q-\ell-d_3+i}$ is the last vertex of $P_{q-\ell-d_3+i}$ which intersects $Q_A$.)  Let $S_i'$ be the subpath of $S_i$ from $u_i$ to $t_i$.

\begin{claim}
\label{claim:ul1}
There is a set of  vertex disjoint paths $\{W_1,\ldots,W_{\beta'}\}$ with the following conditions.
\begin{enumerate}
\item For all $i\in [\beta']$, $W_i$ is a path from $u_i$ to $t_i$, $\topm(W_i,P_{q-\ell-d_3+i})=t_i$ and $E(W_i)\subseteq E(S_i)\cup \left( \bigcup_{w\leq j\leq e} E(P_{j})\right)$. 
\item For all $i\in [\beta']\setminus \{1\}$, $W_i$ does not intersect with $P'_j$ for any $j<i$.
\item For all $i\in [\beta']\setminus \{1\}$, $W_i$ is in the outer face of the cycle formed by $P_w,R_{i-1}, W_{i-1},P'_{i-1}$ and $Q_A$.  
\item  For all $i\in [\beta']$, $W_i$ is a path in the graph  $\tilde{G}[\bigcup_{\substack {n\leq i\leq s \\ w \leq j\leq e}} \inNoose_{\tilde{G}}(N_{i,j}) ]$. 
 \item For all $i\in [\beta']$, $W_i$ does not intersects with $P_{q-\ell-d_3+j}$ for any $j>i$. 
\end{enumerate}
\end{claim}

\begin{proof}
First we show that there exists a set of paths satisfying conditions $(1)-(4)$. 
Consider the set of paths $\{S'_1,\ldots, S'_{\beta'}\}$.  
From the definition of $t_i$ and $S_i'$ for all $i\in [\beta']$, condition $(1)$ follows. 
Condition $(2)$ follows from \autoref{claim:top} and the definition of $P_j'$ for all $j\in [\beta']$.   
Condition $(3)$ follows from \autoref{obs:intcontain}. 
Since $\{S_1,\ldots,S_{\beta'}\}$ is the set of $(\ell,d_1,d_3)$-crossing segments of a $(\ell,d_1,\alpha)$-vacant solution, the 
condition $(4)$ follows for the  paths $S_1',\ldots,S'_{\beta'}$.

Among the set of paths satisfying properties $(1)-(4)$, let ${\cal W}=\{W_1,\ldots,W_{\beta'}\}$ be a set of paths which 
minimizes the edges on these paths from $E(\tilde{G})\setminus  \left(\bigcup_{w\leq j\leq e} E(P_{j})\right)$. 
We claim that condition $(5)$ holds for the paths $W_1,\ldots,W_{\beta'}$. 
For the sake of contradiction, suppose condition $(5)$ is false for the paths $W_1,\ldots, W_{\beta'}$. 
Let $i$ be the least integer such that $W_i$ intersects with $P_{q-\ell-d_3+j}$ for some $j>i$.
Let 
$C_R$ be the cycle formed by $P_{q-\ell-d_3+i},P_e,Q_A$ and $Q_B$ and $C'$ be the cycle formed by $R_i,W_i,P_i',Q_A$ and $P_w$. Since $t_i=\topm(W_i,P_{q-\ell-d_3+i})$ and $W_i$ intersects with $P_{q-\ell-d_3+j}$ for some $j>i$,  there is a vertex $x$ of $W_i$ which belongs to $P_{q-\ell-d_3+i}$, and the second edge $e$ incident to $x$ on the {\em oriented} path $W_i$ from $u_i$ to $t_i$ is in the interior face of the cycle $C_R$ (notice that $e\notin E(P_{q-\ell-d_3+i})$). Let $x^{\star}$ be the first such vertex on $W_i$ and $e^{\star}_2$ be the second edge incident on $x^{\star}$ in $W_i$. Let $e^{\star}$ be the first edge incident on $x^{\star}$ in the oriented path $P_{q-\ell-d_3+i}$. 

\smallskip
\noindent
{\bf Case 1: $e^{\star} \notin E(W_i)$.}  
Then $e^{\star}$ is in the left side of the curve corresponding to the path $W_i$ from $u_i$ to $t_i$ and hence, 
$e^{\star}$ is in the inner face of $C'$. Since $t_i=\topm(W_i,P_{q-\ell-d_3+i})$ appears before $x^{\star}$ on $P_{q-\ell-d_3+i}$ and 
$t_i$ is the last vertex of $W_i$, there is a subapth $P^{\star}$ of $P_{q-\ell-d_3+i}$ with one endpoint $x^{\star}$, other endpoint on $V(W_i)\cap V(P_{q-\ell-d_3+i})$,  which is either $t_i$ or after $t_i$ on oriented path $P_{q-\ell-d_3+i}$ (call it $y^{\star}$) such that 
$V(P^{\star})$ is contained in the interior face of $C'$. Moreover the subpath $W^{\star}$ of $W_i$ from $x^{\star}$ to $y^{\star}$ contains the edge $e_2^{\star}$. That is, $(a)$ $E(W^{\star})\cap (E(\tilde{G})\setminus  \left(\bigcup_{w\leq j\leq e} E(P_{j})\right))\neq \emptyset$. Let $W^f$ be the subpath of $W_i$ from $u_i$ to $x^{\star}$ and let $W^{l}$ is the subpath 
of $W_i$ from $y^{\star}$ to $t_i$. Now consider the path $W'_i=W^fP^{\star}W^l$.   For any $i'<i$, $W_i$ does not intersect with $W'_i$, because condition (4) holds for all $i'<i$ (by the choice of $i$). Since $P^{\star}$ is in the interior face of $C'$, by condition $(3)$, $W_j$ is internally vertex disjoint from $W_i'$ for any $j>i$. This implies that ${\cal W}'=\{W_1,\ldots,W_{i-1},W_i',W_{i+1},W_{\beta'}\}$ is a set of vertex disjoint paths. Moreover, since $V(P^{\star})$ are after the vertex $t_i$ on the oriented path $P_{q-\ell-d_3+i}$ and $E(P^{\star})\subseteq E(P_{q-\ell-d_3+i})$, we have that  conditions $(1)$ and $(2)$ hold for ${\cal W}'$. Since $P^{\star}$ is in the interior face of $C'$, ${\cal W}'$ satisfies condition $(3)$.  It is easy to see that $W_i'$ is a path in $\tilde{G}[\bigcup_{\substack {n\leq i\leq s \\ w \leq j\leq e}} \inNoose(N_{i,j}) ]$. This implies that ${\cal W}'$ satisfies condition $(4)$. Finally, because of statement $(a)$, the number edges in $\left(E(\tilde{G})\setminus  \left(\bigcup_{w\leq j\leq e} E(P_{j})\right) \right)  \cap E({\cal W}')$ 
is strictly less than the number of edges in $\left(E(\tilde{G})\setminus  \left(\bigcup_{w\leq j\leq e} E(P_{j})\right) \right) \cap E({\cal W})$. 
This is a contradiction to our assumption. 

\smallskip
\noindent
{\bf Case 2: $e^{\star} \in E(W_i)$.} Let $P_e^{\star}$ be the maximal subpath of $P_{q-\ell-d_3+i}$ such that $P_{e^{\star}}$ ends at $e^{\star}$ and $P_e^{\star}$ is a subpath of $W_i$. Let $z^{\star}$ be the first vertex in the path $P_e^{\star}$. Notice that $P_e^{\star}$ is a path from $z^{\star}$ to $x^{\star}$ and $z^{\star}$ appears before $x^{\star}$ in the oriented path $P_{q-\ell-d_3+i}$. In this case, the first edge $e_z^{\star}$  in  $P_{q-\ell-d_3+i}$, incident on $z^{\star}$, is in the inner face of $C'$ and $e_z^{\star}\notin E(W_i)$. Similar to Case 1, there is a subapth $P^{\star}$ of $P_{q-\ell-d_3+i}$ with one endpoint $z^{\star}$, other endpoint on $V(W_i)\cap V(P_{q-\ell-d_3+i})$,  which is either $t_i$ or after $t_i$ on oriented path$P_{q-\ell-d_3+i}$ (call it $y^{\star}$) such that 
$V(P^{\star})$ is contained in the interior face of $C'$. Moreover the subpath of $W_i$ from $z^{\star}$ to $y^{\star}$ contains the edge $e_2^{\star}$. Here we replace this subpath with $P^{\star}$. The correctness proof in this case is similar to that of Case~$1$. 
\end{proof}

Let $t_i'=\topm(S_i,P_{q-\ell+d_3-i})=\topm(X_i,P_{q-\ell+d_3-i})$ and let $P''_i$ be the subpath of $P_{q-\ell+d_3-i}$ from $p_{q-\ell+d_3-i}$ to $t_i'$ for all $i\in [\beta']$.  Let $S_i''$ be the subpath of $\overleftarrow{S_i}$ from $u_i'$ to $t_i'$.  
The proof of the following claim is by using arguments similar to that of \autoref{claim:ul1} and hence omitted. 

\begin{claim}
\label{claim:ul2}
There is a set of  vertex disjoint paths $\{Z_1,\ldots,Z_{\beta'}\}$ with the following properties.
\begin{enumerate}
\item For all $i\in [\beta']$, $Z_i$ is a path from $u_i'$ to $t_i'$, $\topm(Z_i,P_{q-\ell+d_3-i})=t_i'$ and $E(Z_i)\subseteq E(S_i)\cup \left( \bigcup_{w\leq j\leq e} E(P_{j})\right)$. 
\item For all $i\in [\beta']\setminus \{1\}$, $Z_i$ does not intersect with $P''_j$ for any $j<i$.
\item For all $i\in [\beta']\setminus \{1\}$, $Z_i$ is in the outer face of the cycle formed by $P_e,R_{i-1}', Z_{i-1},P''_{i-1}$ and $Q_A$.  
\item  For all $i\in [\beta']$, $Z_i$ is path in the graph induced on $\tilde{G}[\bigcup_{\substack {n\leq i\leq s \\ w \leq j\leq e}} \inNoose_{\tilde{G}}(N_{i,j}) ]$
 \item For all $i\in [\beta']$, $Z_i$ does not intersects with $P_{q-\ell+d_3-j}$ for any $j>i$. 
\end{enumerate}
\end{claim}

Using \autoref{claim:ul1} and \autoref{claim:ul2}, we reroute the segments $S_i$ as follows.
Let $\{W_1,\ldots,W_{\beta}'\}$ and $\{Z_1,\ldots, Z_{\beta}'\}$ 
be the set of paths specified by Claims~\ref{claim:ul1} and \ref{claim:ul2}, respectively. These paths are disjoint because $d_3>\beta+2$. 
 Let $J_i$ be an $i$-oriented row  for $i \in \{q-\ell+1,\ldots,q-\ell+\beta'\}$ and these paths will not intersect with 
$W_i$s and $Z_i$s because $\beta'< d_1$. 
For any $i\in [\beta']$, let $I_i$ be the unique path from $t_i$ to $t_i'$ in the tree $P'_i \cup P''_i \cup J_{q-\ell+i}$. 
The set of paths $\{I_1,\ldots,I_{{\beta}'}\}$ are pairwise vertex disjoint. 
Let $Y_i$ be the path $W_iI_i\overleftarrow{Z}_i$. 
That is, $Y_i$ is a path from $u_i$ to $u_i'$. The set of paths $\{Y_1,\ldots,Y_{\beta'}\}$ are pairwise 
vertex disjoint because of 
Claims~\ref{claim:ul1} and \ref{claim:ul2}. 
Thus, we replace each $(\ell,d_1,d_3)$-up-left crossing 
segment $S_i$ with a path $Y_i$. 
In a similar manner we replace each of the $(\ell,d_1,d_3)$-crossings in ${\cal S}$. 
This implies that for each $H \mbox{ in the $\delta^{\star}$-folio of } G$, 
we get a new witness for $H$ being a topological minor in $\tilde{G}$. 
Let ${\cal R}'=\{(H,\phi'_H,\varphi'_H) : H \mbox{ in the $\delta^{\star}$-folio of } G\}$ be the 
set constructed from ${\cal S}'$.

Now we prove that indeed ${\cal R}'$ is a representation of a solution to $(G,\delta,t,w',s')$, 
which is $(\ell,d_1)$-terminal free and $(\ell,\eta)$-untangled. Towards that we first 
prove that ${\cal R}'$ is a representation of a solution ${\cal R}$. Notice that ${\cal R}'$ 
is obtained from a representation ${\cal S}'$ of a $(\ell,d_1)$-terminal free solution ${\cal S}$ 
by replacing some subpaths of paths specified by ${\cal S}'$.  
Moreover these subpaths as well as its 
replacement fully lie inside the graph induced by the nooses in frames $\bigcup_{i\leq d_1}\fr[\ell-i]\cup \fr[\ell+i]$.
This implies that if  an edge $e'=\{u,v\}\in E(\tilde{G})\setminus E(G)$ is used in a replacement subpath, then there is a 
graph $G_i,i\in[k]$ such that no vertex from $V(G_i)\setminus V(G_0)$ is used by ${\cal S}'$ and hence by ${\cal R}'$ (see \autoref{obs:allpresent}). 
This implies that we can reroute the  path $u-v$ with a  path from $u$ to $v$ in $G_i$ with internal vertices 
being in $V(G_i)\setminus V(G_0)$  (see Condition~\ref{conditionsix} of \autoref{obs:cflatmore}).  
Hence, we have that indeed  ${\cal R}'$ is the representation of a solution ${\cal R}$ to $(G,\delta,t,w',s')$.  
Moreover in the subpath replacement process no terminal vertex is modified and since 
${\cal S}$ is $(\ell,d_1)$-terminal free solution, ${\cal R}$ is also a $(\ell,d_1)$-terminal free solution.  

Now we show that ${\cal R}$ is $(\ell,4\beta)$-untangled. Since ${\cal R}$ is $(\ell,d_1)$-terminal free and $d_1\geq 3$, it is also 
$(\ell,3)$-terminal free.  Notice that,   by the above construction for $(H,\phi_H',\varphi'_H)\in {\cal R}'$, 
the set of $(\ell,d_1,d_3)$-up-left crossing segments are $\{Y_1,\ldots,Y_{\beta'}\}$. 
Each of these segment intersect on one noose in $\fr[\ell]$, and moreover, these segments intersect different 
nooses from $\fr[\ell]$, because the vertices from any noose in $\fr[\ell]$, which belongs to $Y_i$ is from 
the path $J_{q-\ell+i}$. This implies that nooses in $\fr[\ell]$ used by ${\cal R}'$ are from the right-nooses and/or from the left-nooses of $\fr[\ell]$.  
Also, notice that the number of  $(\ell,d_1,d_3)$-up-left crossing segments of $(H,\phi_H',\varphi'_H)\in {\cal R}'$ is  $\beta'\leq \beta$, 
and ${\cal R}$ is a $(\ell,d_1,d_2,d_3)$-regret free solution (this follows from the fact that ${\cal S}$ is $(\ell,d_1,d_2,d_3)$-regret free).  
Since ${\cal R}$ is an $(\ell,d_1,d_2,d_3)$-regret free, the only $(\ell,d_1)$-segments 
which contains a vertex from a noose in $\fr[\ell]$ are $(\ell,d_1,d_3)$-crossing segments and the total number of such segments for $(H,\phi_H',\varphi'_H)\in {\cal R}'$ is upper bounded by $4\beta$ (from all the four corners).
Moreover, any $(\ell,d_1, d_3)$-crossing segment in ${\cal R}$ uses only one noose from $\fr[\ell]$.  
 As a result, condition $(1)$ of \autoref{def:untangledFrame} is satisfied.

Now we prove that condition $(2)$ of \autoref{def:untangledFrame} holds for ${\cal R}$. Let $v$ be a vertex in the image of $\varphi'_H$, 
such that  $v\in \inNoose_{\tilde{G}}(N)$ for some $\fr[\ell]$. Since ${\cal R}$ is $(\ell,3)$-terminal free solution $v\neq \phi'_H(h)$ for any $h\in V(H)$. Thus  $v$ is an internal vertex in $P=\varphi_H'(e)$ for some $e=\{h_1,h_2\}\in E(H)$. Since ${\cal R}$ is $(\ell,d_1,d_2,d_3)$-regret free, there is a subpath $P'$ of  $P$ containing $v$, whose endpoints are in $N$ and $N'$, where $N\in \fr[\ell+d_1]$ and $N'\in \fr[\ell-d_1]$. Let $u\in \inNoose(N')$ be one of the endpoint of $P'$. Let $P_1$ be the subpath of $P'$ with one endpoint $v$, the other endpoint in $\{\phi'_H(h_1),\phi'_H(h_2)\}$ and $u$ is an internal vertex in $P_1$.  Either an internal vertex of $P_1$, is a noose in $\fr[\ell]$,  or no internal vertex of $P_1$ is in a noose $N''$ for any $N''\in \fr[\ell]$. In the former case we have that $v$ is a vertex of an $\ell$-segement and in the later case we have that $v$ is a vertex of 
an $\ell$-dangling segment. This implies that ${\cal R}$ satisfies condition $(2)$ of \autoref{def:untangledFrame}. 


Now we show that ${\cal R}$ satisfies condition $(3)$ of \autoref{def:untangledFrame}. From the construction of ${\cal R}$, we know that any noose $N\in \fr[\ell]$ is hit by at most one $(\ell,d_1,d_3)$-crossing segment $Y$ of $(H,\phi_H',\varphi'_H)$ and we also know that ${\cal R}$ is a $(\ell,3)$-terminal free. The segment $Y$ is a path from a vertex from $\fr[\ell+d_1]$ to $\fr[\ell-d_1]$.  Moreover, the vertices from $\fr[\ell]$ in $Y$ form a subpath $Z$ of $Y$.  This path $Z$ will be part of a $\ell$-(dangling) segment. The arguments about its proof is similar to the one in the above paragraph. This completes the proof of the lemma. 
%
%
\end{proof}


\section{Taking Snapshots}\label{sec:snapshot}

We have proved the existence of an $(\ell,\eta)$-untangled solution. However, to find an irrelevant vertex, we have to consider the computational aspect of our arguments. In the next two sections, we exhibit an integer $\ell^{\star}<\ell$ that can be used to prove that {\em partial solutions realizable} in $G^{\star}_{\ell^{\star}}$ are the same as the those {\em realizable} in $G^{\star}_{\ell}$. Moreover, a partial solution in $G^{\star}_{\ell^{\star}}$ can be extended to a solution without using vertices from the up-nooses of $\fr[\ell^{\star}]$ (and some outermore frames). As a result, any vertex in any up-noose of $\fr[\ell^{\star}]$ is irrelevant.

Without loss of generality we assume that in an instance $(G,\delta,t,w',s')$  of \fFindFoliostar, $\rho_G(R(G))=[\vert R(G)\vert]$. 
Recall that $\delta^{\star}$ and $\boundary$ are fixed constants depending only on $\delta$. The set $A$ is a set of at most $t$ vertices 
which may have neighbours anywhere, including $G_0,\ldots,G_k$. Because of this we consider $3(\delta^{\star}+t)$-folios of $G$ instead of $\delta^{\star}$-folios while constructing partial solutions, whenever we use a vertex from $A$, all its neighbors will be considered as terminals. Below we define the notion of a snapshot, which is the sufficient  information needed for restriction of a tuple in a solution to $G^{\star}_{\ell}$.  

\begin{definition}[{\bf Snapshot}]\label{def:snapshot}
A {\em snapshot} is a tuple $(C,T,T',\rho,{\cal P}, f,{\sf out})$, 
where $T$ is a set of size at most $6(\delta^{\star}+t)$,  $T'\subseteq T$,  $\rho \colon T'\rightarrow [\boundary]$, ${\sf out} \colon T\setminus T' \rightarrow 2^{A}$, $C$ is a cycle (i.e, a graph which is a cycle), ${\cal P}$ is a subset of pairs of $T$ and $f$ is function from $V(C)$ to $T\cup [V(C)]$ such that for all $w\in[V(C)]$, $\vert f^{-1}(w)\vert \in \{0, 1\}$ and $\vert T\vert +\vert {\cal P}\vert \leq 3(\delta^{\star}+t)$.  
Furthermore, given $\eta\in\mathbb{N}$, an {\em $\eta$-snapshot} is a snapshot $(C,T,T',\rho,{\cal P}, f, {\sf out})$ where $|V(C)|\leq \eta$. 
Finally, $\SNP[\eta]$ denotes the set of all $\eta$-snapshots.
\end{definition}

\begin{definition}[{\bf Snapshot Equivalence}]\label{def:snapshotEquiv}
Let $(C,T,T',\rho,{\cal P}, f,{\sf out})$ and $(\widehat{C},\widehat{T},\widehat{T}',\widehat{\rho},\widehat{\cal P}, \widehat{f}, \widehat{\sf out})$ be two snapshots. These snapshots are {\em equivalent}, denoted by $(C,T,T',\rho,{\cal P}, f, {\sf out})\equiv (\widehat{C},\widehat{T},\widehat{T}',$ $\widehat{\rho},\widehat{\cal P}, \widehat{f},\widehat{\sf out})$, if $T=\widehat{T}$, ${\cal P}=\widehat{\cal P}$, $\rho(t)=\widehat{\rho}(t)$ for any $t\in T'=\widehat{T}'$, ${\sf out}(t)=\widehat{{\sf out}}(t)$ for any $t\in T\setminus T'$, 
and there exists an isomorphism $g: V(C)\rightarrow V(\widehat{C})$ such that for all $v\in V(C)$, $f(v)=\widehat{f}(g(v))$. 
\end{definition}

\begin{observation}\label{obs:numSnapshots}
The number of distinct $\eta$-snapshots is upper bounded by $(t+\delta^{\star})^{\OO(t^2+t \delta^{\star})}\cdot\eta^{\OO(\eta)}$. 
\end{observation}


\begin{proof}
The number of choices for $T,T'$, and $\rho$ in snapshots is bounded by $(\boundary+2)^{\OO(t+ \delta^{\star})}$.
 The number of choices of $C$ in $\eta$-snapshots is at most $\eta$. The number of choices for $f$ in $\eta$-snapshots is bounded by $(6t+ 6\delta^{\star}+\eta)^{\eta}$.  The quantity $(6t+ 6\delta^{\star}+\eta)^{\eta}$ is upper bounded by $\eta^{\OO(\eta)}$, when $t + \delta^{\star}\leq \eta$, and otherwise it is upper bounded $(t+ \delta^{\star})^{\OO(t+ \delta^{\star})}$. The number of choices for ${\sf out}$ is at most $2^{t (6t+6\delta^{\star})}$. Therefore  the number of distinct $\eta$-snapshots is upper bounded by  $(t+\delta^{\star})^{\OO(t^2+ t\delta^{\star})}\cdot \eta^{\OO(\cdot \eta)}$
\end{proof}



The following definition is required to define partial solution. 

\begin{definition}
Let $H$ be a rooted graph. 
First, $\sd(H)$ denotes the set of rooted graphs that can be obtained from a subgraphs  of $H$ by adding at most $d_H(v)$ pendant vertices on each vertex $v$ of $H$. For simplicity, the vertices in a graph in $\sd(H)$ originating from vertices in $H$ are assumed to have the identities of the corresponding vertices in $H$.  Second, $\sd^\eta(H)$ denotes the set of rooted graphs that can be obtained from a graph in $\sd(H)$ by adding to it at most $\eta$ new connected components that are each a path on two vertices. Moreover, the roots of the graphs in $\sd(H)$ and $\sd^\eta(H)$ are the roots of $H$ present in these graphs.
\end{definition}



Recall the definition of $G^{\star}_{\ell}$ and $\tilde{G}_{\ell}$ from \autoref{def:gltilde}.

\begin{definition}[{\bf Partial Solution}]\label{def:partial}
Let $(G,\delta,t,w',s')$ be an instance of \fFindFoliostar.
Let $(M,{\cal N})$ be a $2q$-workspace in $\tilde{G}$, $\ell\in[q-1]_0$ and $\eta\in\mathbb{N}$.
An {\em $(\ell,\eta)$-partial solution} is a tuple $(H, H',\phi,\varphi, {\sf out})$ with the following properties. 
\vspace{-0.5em}
\begin{enumerate}
\itemsep0em
\item $H$ is a rooted graph with $|E(H)|+\isolated(H)\leq 3(\delta^{\star}+t)$ and $H'$ is a graph in $\sd^\eta(H)$. 
\item  $H'$ is a topological  minor in $G^{\star}_{\ell}$ witnessed by some $(\phi_{H'},\varphi_{H'})$,  $(H',\phi,\varphi)$ is the representation of $(H',\phi_{H'},\varphi_{H'})$ in $\tilde{G}_{\ell}$ and ${\sf out} \colon V(H') \setminus R(H')\rightarrow 2^A$ with the following conditions. 
	\begin{enumerate}
	\item Every vertex in $\image(\varphi)\setminus (\phi(V(H')\setminus V(H))$ belongs to $V(\tilde{G}_{\ell-1})$. 
	\item Every vertex in $\phi(V(H')\setminus V(H))$ belongs to $\inNoose_{\tilde{G}}(N)$ for some $N\in\fr[\ell]$ and $N$ is either a left-noose or a right noose in $\fr[\ell]$. 
\item For any $N\in\fr[\ell]$, $|\inNoose_{\tilde{G}}(N)\cap \phi(V(H'))|\leq 1$. 
\item For any $v\in (V(H')\cap V(H)) \setminus R(H')$, ${\sf out}(v)\subseteq A\cap (N_G(\phi(v)))$.
	\end{enumerate}
\end{enumerate}
\end{definition}


The validity of the definition below implicitly relies on the properties of an $(\ell,\eta)$-partial solution.

\begin{definition}[{\bf Camera}]\label{def:camera}
Let $(G,\delta,t,w',s')$ be an instance of \fFindFoliostar. Let $(M,{\cal N})$ be a $2q$-workspace in $\tilde{G}$, $\ell\in[q-1]_0$ and $\eta\in\mathbb{N}$. The {\em $(\ell,\eta)$-camera} is the function $\camera_{\ell,\eta}$ whose domain is the set of all $(\ell,\eta)$-partial solutions, whose codomain is $\SNP[\eta]$, and which maps an $(\ell,\eta)$-partial solution $(H,H',\phi,\varphi,{\sf out}')$ to an $\eta$-snapshot $(C,T,T',\rho,{\cal P}, f, {\sf out})$ as follows.
\vspace{-0.5em}
\begin{itemize}
\itemsep0em
\item $T=V(H')\cap V(H)$ and $T'=T\cap R(H)$.
\item ${\sf out}={\sf out}'$. 
\item ${\cal P}=E(H')\cap E(H)$.
\item $C$ is a cycle on $s=|V(H')\setminus V(H)|$ vertices.
\item Denote $C=v_1-v_2-\cdots-v_s-v_1$. Moreover, let $N_1,N_2,\ldots,N_s$ denote that nooses in $\fr[\ell]$ that enclose vertices in $\phi(V(H')\setminus V(H))$, ordered according to $<$. For all $i\in[s]$, $u^G_i$ denotes the (unique) vertex in $\phi(V(H')\setminus V(H))$ enclosed by $N_i$, $u^H_i=\phi^{-1}(u^G_i)$, and $w^H_i$ denotes the other endpoint of the (unique) edge in $H'$ incident to $u^H_i$. Then, for all $i\in[s]$, $f$ is defined as follows.
  \vspace{-0.5em}
	\begin{itemize}
	\itemsep0em
	\item If $w^H_i\in T$, then $f(v_i)=w^H_i$
	\item Else, let $j\in[s]$ be the (unique) index such that $\phi(w^H_i)=u^G_j$.  Then 
$f(v_i)=v_j$.
	\end{itemize}
\end{itemize}
\end{definition}

\begin{observation}
\label{obs:snapshotiso}
Let $(G,\delta,t,w',s')$ be an instance of \fFindFoliostar. Let $(M,{\cal N})$ be a $2q$-workspace in $\tilde{G}$, $\ell\in[q-1]_0$ and $\eta\in\mathbb{N}$. 
Let  $(H,H',\phi,\varphi, {\sf out})$  be an $(\ell,\eta)$-partial solution  and $(C,T,T',\rho,{\cal P}, f, {\sf out})=\camera_{\ell,\eta}(H,H',\phi,\varphi,{\sf out})$. Let $J$ be a rooted graph on the vertex set $T\cup V(C)$, with roots $R(J)=T'$ and the edges set 
$$E(J)={\cal P}\cup \{\{w_1,w_2,\}\colon w_1,w_2\in V(C), f(w_1)=w_2\}\cup \{\{w,z\}\colon w\in V(C),z\in T, f(w)=z\}.$$ 
Then there is an isomorphism ${\sf iso}$ from $J$ to $H'$ such that ${\sf iso}$ restricted on $T$ is the identity map.  
\end{observation}

\begin{proof}
As $V(H')\cap V(H)=T$, we set ${\sf iso}$ restricted on $T$ to be the identity map. Now we explain how to map vertices in  $V(C)$. 
Let $C_1=v_1-v_2-\ldots-v_{s}-v_1$ and let $N_1,N_2,\ldots,N_{s}$ denote the nooses in $\fr[\ell]$ that enclose vertices in $\phi'(V(H')\setminus V(H))$, ordered according to $<$ and the function $f$ is defined as mentioned in \autoref{def:camera}.   Let $u_i'$ be the (unique) vertex in $V(H')\setminus V(H)$,  such that $\phi(u'_i)\in N_i$ for all $i\in [s]$. Then ${\sf iso}(v_i)=u_i'$. 

Now we prove that ${\sf iso}$ is indeed an isomorphism from $J$ to  $H'$.  Let $x,y \in V(J)$, $x'={\sf iso}(x)$, and $y'={\sf iso}(y)$. Suppose $x,y \in T$. Then  $\{x,y\}\in E(J)$ if and only if $\{x',y'\}\in E(H')$, because ${\cal P}=E(H')\cap E(H)$.  Suppose $x,y \in V(C)$.  Let $N_i$ and $N_j$ be the two nooses such that $\phi(x)\in N_i$ and  $\phi(y)\in N_j$, respectively. There is an edge between $x$ and $y$ in $J$ if and only if $f(v_i)=v_j$. From the last condition in \autoref{def:camera}, we conclude that there is an edge between $\phi(x)$ and $\phi(y)$ in $H'$ if and only if $f(v_i)=v_j$. Suppose $x\in V(C)$ and $y\in T$ and let $N_i$ be the noose in $\fr[\ell]$ such that $\phi(x)\in N_i$. There is an edge between $x$ and $y$ in $J$ if and only if $f(x)=y$. From the last condition in \autoref{def:camera}, we conclude that there is an edge between $\phi(x)$ and $\phi(y)=y$ in $H'$ if and only if $f(v_i)=y$. This completes the proof of the observation. 
%
%
\end{proof}

As mentioned before, whenever a tuple in a solution uses a vertex in $A$, we consider its closed neighborhood are terminal. This fact is formalized in the below definition.  

\begin{definition}[{\bf Apex-Addition}]\label{def:apexaddition}
Let $(G,\delta,t,w',s')$ be an instance of \fFindFoliostar. Let $(M,{\cal N})$ be a $2q$-workspace in $\tilde{G}$, $\ell\in[q-1]_0$ and $\eta\in\mathbb{N}$. Let ${\cal S}'$ be the representation of  a solution of $(G,\delta,t,w',s')$. Let $(H,\phi,\varphi)\in {\cal S}'$. 
The {\em apex-addition} of $(H,\phi,\varphi)$, denoted by $\apex(H,\phi,\varphi)$, is a tuple $(H^{\star},\phi^{\star},\varphi^{\star})$ defined as follows. 
	\begin{itemize}
	\itemsep0em
         \item Let $G'$ be the realization of topological minor $H$ in $G\cup \tilde{G}$, witnessed by $(\phi,\varphi)$. Let $N_{G'}[A\cap V(G')]\setminus \phi(V(H))=Z=\{z_1,\ldots,z_b\}$.  
         \item $H^{\star}$ is a graph on vertex set $V(H)\cup \{z'_1,\ldots,z'_b\}$ (edge set will be defined later), which is a topological minor in $G$, witnessed by  $(\phi^{\star},\varphi^{\star})$ and realized by $G'$. 
\item $\phi^{\star}(v)=\phi(v)$ for all $v\in V(H)$ and $\phi^{\star}(z_i')=z_i$ for all $i\in [b]$. 
\item  For each $e\in E(H)$, if $\phi(e)=P_1\ldots P_s$ with end-vertices of $P_j$s in $\phi^{\star}(H^{\star})$ and internal vertices not in $\phi^{\star}(H^{\star})$, we have $\{u,v\}\in E(H^{\star})$ and $\varphi^{\star}(\{u,v\})=P_j$ where $P_j$ is a path between $u$ and $v$.    
         \end{itemize}
\end{definition}

In an $(\ell,\eta)$-partial solution, vertices in $V(H')\setminus V(H)$ are mapped to a vertices in nooses in $\fr[\ell]$, and these vertices are not actual terminal vertices, but internal vertices of paths mapped by the edges of $H$. To identify these vertices for each tuple in a solution, we have the following definition.

\begin{definition}[{\bf Extension}]\label{def:extend}
Let $(G,\delta,t,w',s')$ be an instance of \fFindFoliostar. Let $(M,{\cal N})$ be a $2q$-workspace in $\tilde{G}$, $\ell\in[q-1]_0$ and $\eta\in\mathbb{N}$. Let ${\cal S}'$ be the representation of  an $(\ell,\eta)$-untangled solution of $(G,\delta,t,w',s')$. Let $(H,\phi,\varphi)\in {\cal S}'$. 
The {\em $(\ell,\eta)$-extension} of $(H,\phi,\varphi)$, denoted by $\extend_{\ell,\eta}(H,\phi,\varphi)$, is defined as follows. 
\begin{itemize}
\itemsep0em
\item Let  $(H^{\star},\phi^{\star}, \varphi^{\star})=\apex(H,\phi,\varphi)$.
\item Initialize $\widehat{H}=H^{\star}$, $\widehat{\phi}=\phi^{\star}$ and $\widehat{\varphi}=\varphi^{\star}$. For every edge $e=\{u,v\}\in E(H)$:
	\begin{itemize}
	\itemsep0em
	\item Denote $\varphi^{\star}(e)=P=w_1-w_2-\cdots-w_r$. Let $I$ denote the set of all indices $i\in[r]$ such that there exists $N\in\fr[\ell]$ that encloses $w_i$, and there exists $N\in\fr[\ell']$ for some $\ell'< \ell$, that encloses  either $w_{i-1}$ or $w_{i+1}$.
	\item Subdivide $e$ $|I|$ times, update $\widehat{H}$ accordingly, and denote the new path in $\widehat{H}$ by $x_0=u-x_1-x_2-\ldots-x_{|I|}-v=x_{|I|+1}$.  Let $I=\{j_1,\ldots,j_s\}$. 
	\item Extend $\widehat{\phi}$ so that for all $i\in [|I|]=[s]$, $\widehat{\phi}(x_i)=w_{j_i}$. Moreover, remove $\{u,v\}$ from $\domain(\widehat{\varphi})$, and then extend $\widehat{\varphi}$ so that for all $i\in[|I|+1]$, $\widehat{\varphi}$ equals the subpath of $P$ between $\phi^{\star}(x_{i-1})$ and $\phi^{\star}(x_i)$.
	\end{itemize}
\item The resulting $(\widehat{H},\widehat{\phi},\widehat{\varphi})$ is 
$\extend_{\ell,\eta}(H,\phi,\varphi)$. 
\end{itemize}
\end{definition}

Next we define restriction of a tuple in an $(\ell,\eta)$-untangled solution to $G^{\star}_{\ell}$. 

\begin{definition}[{\bf Restriction}]\label{def:restrict}
Let $(G,\delta,t,w',s')$ be an instance of \fFindFoliostar. Let $(M,{\cal N})$ be a $2q$-workspace in $\tilde{G}$, $\ell\in[q-1]_0$ and $\eta\in\mathbb{N}$. 
Let ${\cal S}'$ be the representation of  an $(\ell,\eta)$-untangled solution of $(G,\delta,t,w',s')$. Let $(H,\phi,\varphi)\in {\cal S}'$. The {\em $(\ell,\eta)$-restriction} of $(H,\phi,\varphi)$, denoted by $\restrict_{\ell,\eta}(H,\phi,\varphi)$, is defined as follows. 
Let $(H^{\star},\phi^{\star}, \varphi^{\star})=\apex(H,\phi,\varphi)$ and $\extend_{\ell,\eta}(H,\phi,\varphi)=(\widehat{H},\widehat{\phi},\widehat{\varphi})$, and define $\restrict_{\ell,\eta}(H,\phi,\varphi)=(H^{\star}, H',\phi',\varphi',{\sf out}')$ as follows. Remove from $\widehat{H}$  the vertices $v\in V(\widehat{H})$, if there is a noose $N$ in  $\fr[\ell']$ for some $\ell'>\ell$ that either encloses 
 $\widehat{\phi}(v)$ (when $\widehat{\phi}(v)\in V(G_0)$)    or encloses a vertex of $G_i$ (when $\widehat{\phi}(v)\in V(G_i)$ for some $i\in [k]$). 
Discard all of the vertices and edges removed from the domains of $\widehat{\phi}$ and $\widehat{\varphi}$, respectively. Let $H'$ be the resulting graph $\phi', \varphi'$ be the resulting functions. 
Let $G'$ be the realization obtained through $(\phi^{\star}, \varphi^{\star})$. For any $v\in (V(H')\cap V(H^{\star}))\setminus R(H^{\star})$, ${\sf out}'(v)=A\cap N_{G'}(\phi^{\star}(N_{H^{\star}}(v)))$. This completes the definition of  $(H^{\star},H',\phi',\varphi', {\sf out}')$. 
\end{definition}

%

Now we prove that our notion of $(\ell,\eta)$-partial solution is indeed a correct notion of partial solution.

\begin{lemma}\label{def:restrictIsPartial}
Let $(G,\delta,t,w',s')$ be an instance of \fFindFoliostar. Let $(M,{\cal N})$ be a $2q$-workspace in $\tilde{G}$, $\ell\in[q-1]_0$ and $\eta\in\mathbb{N}$. 
Let ${\cal S}$ be an $(\ell,\eta)$-untangled solution and let  ${\cal S}'$ be the representation of ${\cal S}$.  
Then, for any $(H,\phi,\varphi)\in {\cal S}'$, $\restrict_{\ell,\eta}(H,\phi,\varphi)$ is an $(\ell,\eta)$-partial solution.
\end{lemma}

\begin{proof}
Fix an arbitrary tuple $(H,\phi,\varphi)\in {\cal S}'$. Let $(H^{\star},\phi^{\star}, \varphi^{\star})=\apex(H,\phi,\varphi)$, $(\widehat{H},\widehat{\phi},\widehat{\varphi})$ $=\extend_{\ell,\eta}(H,\phi,\varphi)$ and $\restrict_{\ell,\eta}(H,\phi,\varphi)=(H^{\star}, H',\phi',\varphi',{\sf out}')$. 
We want to show that $(H^{\star}, H',\phi',\varphi')$ is an $(\ell,\eta)$-partial solution. 
Since ${\cal S}$ is an $(\ell,\eta)$-untangled solution and $(H,\phi,\varphi)\in {\cal S}'$, we have the following. 
\begin{enumerate}[label=(\alph*)]
\itemsep0em 
\item\label{statementa} ${\cal S}$ is $(\ell,3)$-terminal free. 
\item\label{statementb} At most $\eta$ nooses in $\fr[\ell]$ are used by $(H,\phi,\varphi)$. Moreover, each noose in $\fr[\ell]$ used by  $(H,\phi,\varphi)$ is either a left-noose or a right-noose in $\fr[\ell]$. 
\item\label{statementc} Every vertex $v$ in the image of $\varphi$, belonging to a noose $N\in \fr[\ell]$, is also a vertex of an $\ell$-segment or an $\ell$-dangling segment of $(H,\phi,\varphi)$.
\item\label{statementd} There do not exist a noose $N\in \fr[\ell]$ and two distinct $\ell$-(dangling) segments $P$ and $P'$ of $(H,\phi,\varphi)$, such that both $P$ and $P'$ intersect $\inNoose_{\tilde{G}}(N)$.
\end{enumerate}

Let $Q=\bigcup_{\ell-3\leq i \leq {\ell}+3}\fr[i]$ and $U= \bigcup_{N\in Q}\inNoose_{\tilde{G}}(N)\cap V(\tilde{G})$.  Since ${\cal S}$ is $(\ell,3)$-terminal free, the following conditions hold. 
\begin{itemize}
\item There does not exist a terminal (with respect to ${\cal S}$) in $U$.
\item  There is no terminal (with respect to ${\cal S}$) in $V(G_i)$, where $V(G_i)\cap U\neq \emptyset$.
\item There is no edge $\{u,v\}$  in the image of $\varphi$ such that $u\in A$ and $v\in U \cup \{V(G_i) : i\in [k], V(G_i)\cap U\neq \emptyset\}$. 
\end{itemize}

Consider the tuple $(H^{\star},\phi^{\star},\varphi^{\star})$. 
Let $P_e=\varphi^{\star}(e)$, for all $e\in E(H^{\star})$. 
Since ${\cal S}$ is $(\ell,3)$-terminal free (by statement~\ref{statementa}), any end-vertex of $P_e,e\in E(H^{\star})$, 
does not belong to $U$. That is, 

\begin{itemize}
\item[$(1)$] There does not exist a terminal with respect to $(H^{\star},\phi^{\star},\varphi^{\star})$ in $U$,
\item[$(2)$]  There is no terminal with respect to $(H^{\star},\phi^{\star},\varphi^{\star})$ in $V(G_i)$, where $V(G_i)\cap U\neq \emptyset$.
\item[$(3)$] There is no edge $\{u,v\}$  in the image of $\varphi^{\star}$ such that $u\in A$ and $v\in U \cup \{V(G_i) : i\in [k], V(G_i)\cap U\neq \emptyset\}$. 
\end{itemize}

Because of statement~\ref{statementb}, the total number of times all these paths {\em cross}  $\fr[\ell]$ is at most $\eta$ as explained below. For each $\phi^{\star}(e)=P_e=w^e_1,\ldots,w^e_r$, let $I_e$ denote the set of indices in $[r]$, such that there exists $N\in\fr[\ell]$ that encloses $w^e_i$, and there exists $N\in\fr[\ell']$ for some $\ell'< \ell$, that encloses  either $w^e_{i-1}$ or $w^e_{i+1}$.  Then, because of statement~\ref{statementb}, $\bigcup_{e}\vert I_e\vert \leq \eta$.  
Therefore, $\widehat{H}$ is obtained by subdividing each edge $e$ of $H^{\star}$, $\vert I_e\vert$ times and hence the total number of newly added vertices is at most $\eta$.  
Let $u=x^e_0-x^e_1-\ldots-x^e_{\vert I_e\vert}-x^e_{\vert I_e\vert+1}=v$ be the path in $\widehat{H}$ which is the result of subdividing $e=\{u,v\}$ in $H^{\star}$, $\vert I_e\vert$ times.   Let $I_e=\{j_1,\ldots,j_s\}$. Recall the definition of $(\widehat{\phi},\widehat{\varphi})$. For all $i\in [|I_e|]$, $\widehat{\phi}(x^e_i)=w^e_{j_i}$ and $\widehat{\varphi}(\{x^e_{i-1},x^e_{i}\})$ is mapped to the subpath of $P_e$ between $\phi^{\star}(x^e_{i-1})$ and $\phi^{\star}(x^e_i)$. Because of statement~$(1)$ above, $(i)$ for any $N\in \fr[\ell]$, $\inNoose_{\tilde{G}}(N)\cap \widehat{\phi}(V(H^{\star}))=\emptyset$. Because of statement~\ref{statementd}, $(ii)$ for any $N\in \fr[\ell]$, $|\inNoose_{\tilde{G}}(N)\cap \widehat{\phi}(V(\widehat{H}))|\leq 1$. By the definition of $\widehat{\phi}$, $(iii)$ any vertex in $\widehat{\phi}(V(\widehat{H})\setminus V(H^{\star}))$ belongs to $\inNoose_{\tilde{G}}(N)$ for some $N\in\fr[\ell]$.

\begin{claim}
\label{claim:wellbehavepath}
For any edge $\widehat{e}\in E(\widehat{H})$, either all the internal vertices of $\widehat{\varphi}(\widehat{e})$ belong to   $V(\tilde{G}_{\ell-1})$
 or at least one internal vertex in $\widehat{\varphi}(\widehat{e})$ belongs to $V(G)\setminus V(\tilde{G}_{\ell})$. 
\end{claim}

\begin{proof}
For the sake of contradiction assume that there is an edge $\widehat{e}\in E(\widehat{H})$ such that no internal vertex of $\widehat{\varphi}(\widehat{e})$ belongs to $V(G)\setminus V(\tilde{G}_{\ell})$ and there is an internal vertex $v$ in $V(\tilde{G}_{\ell})\setminus V(\tilde{G}_{\ell-1})$.   Then either there is a noose $N\in \fr[\ell]$ such that $v\in \inNoose_{\tilde{G}}(N)$ or $v\in V(G_i)\setminus V(G_0)$ for some $G_i$ such that $V(G_i)\cap U\neq \emptyset$. The latter case is not possible because of statements $(2)$ and $(3)$ as explained below. No terminal with respect to $(H^{\star},\phi^{\star},\varphi^{\star})$ is present in $V(G_i)$, for any $G_i$ with $V(G_i)\cap U \neq \emptyset$ and any edge between $V(G_i)$ and $A$ is not used by $\varphi^{\star}$.  
  This implies that any vertex used by $\varphi^{\star}$ are not from $V(G_i)\setminus V(G_0)$, where  $V(G_i)\cap U \neq \emptyset$ 
(because $(H^{\star},\phi^{\star},\varphi^{\star})$ is derived from a representation $(H,\phi,\varphi)$ of tuple in a solution). Therefore any vertex used by $\widehat{\varphi}$ are not from $V(G_i)\setminus V(G_0)$. Now suppose that $v\in \inNoose_{\tilde{G}}(N)$ for some $N\in \fr[\ell]$.

By our assumption, $v$ is an internal vertex in $\widehat{\varphi}(\widehat{e})$. Let $x$ and $y$ be the vertices adjacent to $v$ in 
$\widehat{\varphi}(\widehat{e})$. Since $v$ is an internal vertex of $\widehat{\varphi}(\widehat{e})$, by the construct of $\widehat{H}$, both the adjacent vertices $x$ and $y$ of $v$ in  $\widehat{\varphi}(e)$ belong to  $\inNoose_{\tilde{G}}(N_x)$ and $\inNoose_{\tilde{G}}(N_y)$ for some $N_x,N_y\in \fr[\ell]$ (otherwise $v\in V(\widehat{H})$ and it will lead to a contradiction to the assumption that $v$ is an internal vertex in $\widehat{\varphi}(\widehat{e})$). Let $e\in E(H)$ such that $v\in \varphi(e)$.
Because of statement $(c)$ above, $N_x=N_y=N$. Again because of statement $(c)$, there is subpath $P$ in $\varphi(e)$ which is of the form of either $x-v-y-w_1-\ldots-w_r-z$ or $y-v-x-w_1-\ldots-w_r-z$ where all vertices except $z$ are in $\inNoose_{\tilde{G}}(N)$ and $z$ does not belongs to $V(\tilde{G}_{\ell})$ (see last two condition in Definitions~\ref{def:diskSegment} and \ref{def:diskDanglingSegment}).
 Moreover $P$ is a subpath of $\widehat{\varphi}(\widehat{e})$ and $z$ is not a terminal vertex because of statement $(1)$. This implies that $z$ is an internal vertex in $\widehat{e}$ and belong to $V(G)\setminus V(\tilde{G}_{\ell})$, a contradiction to our assumption. 
\end{proof}

Recall that $(H^{\star},H',\phi',\varphi',{\sf out}')=\restrict_{\ell,\eta}(H,\phi,\varphi)$. Notice that $H'$ is obtained from $\widehat{H}$ by deleting all the vertices $v$ such that $\widehat{\varphi}(v)$ either belongs to a noose in $\fr[\ell']$ for some $\ell'>\ell$ or belongs to $V(G_i)$, where some noose in $\fr[\ell']$ for some $\ell'>\ell$ encloses a vertex in $V(G_i)$, and edges incident on these vertices. 
The functions $\phi'$ and $\varphi'$ are the restrictions of functions $\widehat{\phi}$ and $\widehat{\varphi}$ to the domains $V(H')$ and $E(H')$, respectively.  To prove that $(H^{\star},H',\phi',\varphi',{\sf out}')$ is an $(\ell,\eta)$-partial solution we need to prove the conditions in \autoref{def:partial}. Clearly, since $(H,\phi,\varphi)\in {\cal S}'$ (a representation of an $(\ell,\eta)$-solution), we have that $|E(H)|+\isolated(H)\leq \delta^{\star}$. Also, since $\vert A\vert \leq t$, we have that $|E(H^{\star})|+\isolated(H^{\star})\leq 3(\delta^{\star}+t)$. 

 Now we prove that $H'$ is a graph in $\sd^\eta(H^{\star})$. Towards that we first prove that if there is an edge $e'=\{u,v\}\in H'$ such that $u,v\in V(H')\setminus V(H^{\star})$, then $e'$ is a connected component in $H'$. For the sake of contradiction assume that $w-u-v$ is a path in $H'$ and $u,v\in V(H')\setminus V(H^{\star})$. Let $N_u$ and $N_v$ be the nooses ${\cal N}$ such that $\phi'(u)\in N_u$ and $\phi'(v)\in N_v$. 
By statements $(ii)$ and $(iii)$ above, we have that $N_u\neq N_v$ and $N_u,N_v\in \fr[\ell]$. 
Moreover, the paths $\varphi'(\{u,v\})=\widehat{\varphi}(\{u,v\})$ and $\varphi'(\{w,u\})=\widehat{\varphi}(\{w,u\})$ are subpaths of $\varphi^{\star}(e)$ for some $e\in E(H^{\star})$.  By \autoref{claim:wellbehavepath}, all the internal vertices of the paths  $\varphi'(\{u,v\})$ and $\varphi'(\{w,u\})$ are in  $V(\tilde{G}_{\ell-1})$. 

\begin{claim}
\label{clm:uvinternal}
$\varphi'(\{u,v\})$ contains at least one internal vertex. 
\end{claim}
\begin{proof}
Notice that since ${\cal S}$ is $(\ell,\eta)$-untangled, $\phi'(u)$ is a vertex of an $\ell$-(dangling) segment of $(H,\phi,\varphi)$. But since $\phi'(u)$ is an internal vertex in $\varphi^{\star}(e)$, if  $\varphi^{\star}(e)$ contains a subpath $\phi'(u)-\phi'(v)$ (i.e., a path on two vertices, one in $N_u$ and other in $N_v$), then $\phi'(u)$ is not part of any   $\ell$-(dangling) segment of $(H,\phi,\varphi)$ which is a contradiction. 
\end{proof}

The proof of the following claim is  similar in arguments to the proof of \autoref{clm:uvinternal}

\begin{claim}
\label{clm:wuinternal}
If $w \in V(H')\setminus V(H^{\star})$, then $\varphi'(\{u,v\})$ contains at least one internal vertex. 
\end{claim}

By \autoref{clm:uvinternal}, the subpath $P$ of $\varphi^{\star}(e)$ (which is a subpath of $\varphi(e_1)$ for some $e_1\in E(H)$), between $\phi'(u)$ and $\phi'(v)$ contains at least one internal vertex and by \autoref{claim:wellbehavepath}, all the internal vertices of $P$ are in $V(\tilde{G}_{\ell-1})$.  
If $w \in V(H^{\star})$ then $\varphi'(\{w,u\})$ is a subpath of $\varphi^{\star}(e)$ with all vertices except $\phi'(u)$ is in $V(\tilde{G}_{\ell-1})$.  That is, there is subpath $P^{\star}$ of $\varphi^{\star}(e)$ (which is a subpath of $\varphi(e_1)$ for some $e_1\in E(H)$),  
where all the vertices except $\phi'(u)$  are in $V(\tilde{G}_{\ell-1})$. This implies that $\phi'(u)$ is not part any $\ell$-(dangling) segment of $(H,\phi,\varphi)$, which is a contradiction. Now consider the case $w \in V(H')\setminus V(H^{\star})$. Let $N_w$ be the noose in ${\cal N}$ such that $\phi'(w)\in N_w$. By statements $(ii)$ and $(iii)$ above, we have that $N_u,N_v,N_w\in \fr[\ell]$ and they all are distinct.  
By \autoref{clm:wuinternal}, $\varphi'(\{w,u\})$ contains at least one internal vertex. Therefore, we have  that the subpaths $P$ and $P'$ of $\varphi(e)$, between $\phi'(u)$ and $\phi'(v)$, and  between $\phi'(w)$ and $\phi'(u)$, contain at least one internal vertices each, and by \autoref{claim:wellbehavepath}, all the internal vertices of $P$ and $P'$ are in $V(\tilde{G}_{\ell-1})$.  This implies that $\phi'(u)$ is not part any $\ell$-(dangling) segment of $(H,\phi,\varphi)$, which is a contradiction. That is, now  we proved that if there is an edge $e'=\{u,v\}\in H'$ such that $u,v\in V(H')\setminus V(H)$, then $e'$ is a connected component in $H'$.  
Since any vertex $x$ in the image $\varphi$, belonging to a noose $N\in \fr[\ell]$, is also a vertex of $\ell$-(dangling) segment, no vertex in $V(H')\setminus V(H)$ is an isolated vertex in $H'$. 
 Also, since $\vert V(H')\setminus V(H)\vert \leq \eta$, we conclude that $H'\in \sd^\eta(H)$. Therefore, condition $(1)$ of \autoref{def:partial} is satisfied. 
 
 Now we prove that $H'$ is a topological minor in $G^{\star}_{\ell}$ and $(\phi',\varphi')$ is a representation of it in $\tilde{G}_{\ell}$. Notice that the realization of $(\widehat{\phi},\widehat{\varphi})$ is the same as the realization of $(\phi,\varphi)$. Moreover, since  $(\phi,\varphi)$ is a representation of a pair of functions witnessing that $H$ is a topological minor in $G$, for all the edges used  by the realization of $(\phi,\varphi)$ (an hence in the $(\widehat{\phi},\widehat{\varphi})$), can be replaced with paths in $G$  whose vertices are not used by $(\widehat{\phi},\widehat{\varphi})$. Therefore $H'$ is a topological minor in $G^{\star}_{\ell}$ and $(\phi',\varphi')$ is a representation of it in $\tilde{G}_{\ell}$.  The condition $2(a)$ of $(\phi',\varphi')$ follows from \autoref{claim:wellbehavepath}, and from the construction of $(\phi',\varphi')$ from $(\widehat{\phi},\widehat{\varphi})$. Condition $2(b)$ of $(\phi',\varphi')$ follows from statements $(iii)$  and \ref{statementb}.   Condition $2(c)$ follows from statement $(ii)$ and condition $2(d)$ follows from the definition of ${\sf out}'$ (see \autoref{def:restrict}).  This completes the proof of the lemma. 
 \end{proof}


%
%

Next we define realizable snapshots and prove a lemma about their computation.

\begin{definition}[{\bf Realizable Snapshot}]\label{def:realizableSnapshot}
Let $(G,\delta,t,w',s')$ be an instance of \fFindFoliostar. Let $(M,{\cal N})$ be a $2q$-workspace in $\tilde{G}$, $\ell\in[q-1]_0$ and $\eta\in\mathbb{N}$. 
An $\eta$-snapshot $(C,T,T',\rho,{\cal P}, f,{\sf out})$ is {\em $(\ell,\eta)$-realizable} if there is an $(\ell,\eta)$-partial solution $(H,H',\phi,\varphi,{\sf out})$ such that $(C,T,T',\rho,{\cal P},$ $f,{\sf out})=\camera_{\ell,\eta}(H,H',\phi,\varphi,{\sf out})$.
\end{definition}

\begin{observation}
\label{obs:Hprime:partial}
Let $(G,\delta,t,w',s')$ be an instance of \fFindFoliostar. Let $(M,{\cal N})$ be a $2q$-workspace in $\tilde{G}$, $\ell\in[q-1]_0$ and $\eta\in\mathbb{N}$. Given an $\eta$-snapshot $(C,T,T',\rho,{\cal P}, f, {\sf out})$, there is a unique $H'$ such that if $(C,T,T',\rho,{\cal P}, f,{\sf out})$ is $(\ell,\eta)$-realizable, then there is an $(\ell,\eta)$-partial solution $(H,H',\phi,\varphi,{\sf out})$ such that $\camera_{\ell,\eta}(H,H',\phi,\varphi,{\sf out})=(C,T,T',\rho,{\cal P}, f,{\sf out})$.  Moreover, the graph $H'$ can be constructed in time $(\eta+\delta^{\star}+t)^{\OO(1)}$ and to test whether $(C,T,T',\rho,{\cal P}, f,{\sf out})$ is $(\ell,\eta)$-realizable, it is enough to test whether $H'$ satisfies condition $(2)$ of  \autoref{def:partial}. 
\end{observation}

\begin{lemma}\label{lem:testRealizable}
Let $(G,\delta,t,w',s')$ be an instance of \fFindFoliostar. Let $(M,{\cal N})$ be a $(p,2q)$-workspace in $\tilde{G}$, $\ell\in[q-1]_0$ and $\eta\in\mathbb{N}$. 
Given an $\eta$-snapshot $(C,T,T',\rho,{\cal P}, f, {\sf out})$, it can be determined whether $(C,T,T',\rho,{\cal P}, f, {\sf out})$ is $(\ell,\eta)$-realizable in time $\Delta^{\OO(\Delta)}n+(\eta+\delta^{\star}+t)^{\OO(1)}$, where $\Delta=\max \{p,s'+3\}$. 
\end{lemma}

\begin{proof}[Proof sketch]
Let $\tilde{G}^{\star}=\tilde{G}[[\inNoose_{\tilde{G}}(M)\cap V(\tilde{G})]]$. Since $(M,{\cal N})$ be a $(p,2q)$-workspace, we have that the treewidth of $\tilde{G}^{\star}$ is at most $p$.  
Let ${\cal G} \subseteq \{G_1,\ldots,G_k\}$ be the collection of graphs such that for any $i\in [k]$, $V(G_i)\cap V(G_0)\subseteq  V(\tilde{G}^{\star})$ if and only if $G_i\in {\cal G}$. Without loss of generality assume that ${\cal G}=\{G_1,\ldots, G_s\}$ for some $s\in [k]$. 
Let $G^{\star}=\tilde{G}^{\star}\cup \bigcup_{i\in [s]} G_i$. 
We claim that the treewidth of $G^{\star}$ is $\max \{p,s'+3\}$. 
Let $(T,\beta)$ be a tree decomposition of $\tilde{G}^{\star}$ of width at most $p$. 
For each $i\in [s]$, let $(T_i,\beta_i)$ be a tree decomposition of $G_i$ of width at most $s'$. 
For any $i\in [s]$, $V(G_i)\cap V(\tilde{G}^{\star})$ is a clique (see condition $(4)$ in \autoref{def:KWT}). Therefore for any $i\in [s]$, there is a node $t_{i}\in T$ such that $\beta(t_{i})$ contains $V(G_i)\cap V(\tilde{G}^{\star})$. So to get a tree decomposition of $G^{\star}$ we add an edge between the root of $T_i$ and $t_i$. Moreover for all $i\in [s]$, we add $V(G_i)\cap V(G\tilde{G}^{\star})$ 
to all the bags corresponding to the nodes in $T_i$. 
One can show the constructed pair is indeed a tree decomposition of $G^{\star}$ of width 
at most $\max \{p,s'+3\}$. Let $\Delta=\max \{p,s'+3\}$. Since $\tilde{G}_{\ell}$ is a subgarph of $G^{\star}$, the treewidth of $\tilde{G}_{\ell}$ is at most $\Delta$ for any $\ell\in[q-1]_0$.

Given an $\eta$-snapshot $(C,T,T',\rho,{\cal P}, f, {\sf out})$, we construct (the unique) $H'$ using \autoref{obs:Hprime:partial}. To test whether $(C,T,T',\rho,{\cal P}, f,{\sf out})$ is $(\ell,\eta)$-realizable, one can test whether $H'$ satisfies condition $(2)$ of  \autoref{def:partial} in time $\Delta^{\OO(\Delta)}n$ along the lines of \autoref{prop:twDPNoDec}. 
\end{proof}


\section{Creating Photo Albums}\label{sec:album}

In this section we will point to an irrelevant vertex and a  sufficient condition for being an irrelevant vertex. In the next section we combine several of the lemmas proved so far to find an irrelevant vertex. 

\subsection{Redundant Albums}\label{sec:redundant}

An $(\ell,\eta)$-snapshot collection is the set of all $\eta$-snapshots that are $(\ell,\eta)$-realizable, defined as follows.

\begin{definition}[{\bf Snapshot Collection}]
Let $(G,\delta,t,w',s')$ be an instance of \fFindFoliostar. Let $(M,{\cal N})$ be a $2q$-workspace in $\tilde{G}$, $\ell\in[q-1]_0$ and $\eta\in\mathbb{N}$. 
The {\em $(\ell,\eta)$-snapshot collection}, denoted by $\col_{\ell,\eta}$, is the set of all $\eta$-snapshots that are $(\ell,\eta)$-realizable.
\end{definition}

Clearly, there is a total ordering among $\col_{i,\eta}$, $i\in [q-1]_0$.

\begin{lemma}\label{lem:snapshotSetMonotone}
Let $(G,\delta,t,w',s')$ be an instance of \fFindFoliostar. Let $(M,{\cal N})$ be a $2q$-workspace in $\tilde{G}$ and $\eta\in\mathbb{N}$. For all $i\in[q-2]_0$, $\col_{i,\eta}\subseteq \col_{i+1,\eta}$. Moreover, for any $\eta$-snapshot $(C,T,T',\rho,{\cal P}, f, {\sf out})\in \col_{i,\eta}$  and an $(\ell,\eta)$-partial solution $(H,H',\phi,\varphi,{\sf out})$ such that $\camera_{i,\eta}(H,H',\phi,\varphi,{\sf out})=(C,T,T',\rho,{\cal P}, f,{\sf out})$, there is an $(i+1,\eta)$-partial solution $(H,H',\phi',$ $\varphi',{\sf out})$ such that $\camera_{i+1,\eta}(H,H',\phi',\varphi',{\sf out})=(C,T,T',\rho,{\cal P}, f,{\sf out})$ and $\phi'(V(H)\cap V(H'))\subseteq V(\tilde{G}_{i-1})$. 
\end{lemma}

\begin{proof}
Fix an arbitrary $i\in[q-2]_0$ and an $\eta$-snapshot $(C,T,T',\rho,{\cal P}, f, {\sf out})\in \col_{i,\eta}$. We need to prove that  $(C,T,T',\rho,{\cal P}, f, {\sf out})\in \col_{i+1,\eta}$. Let  $(H,H',\phi,\varphi,{\sf out})$ be an $(i,\eta)$-partial solution such that $\camera_{i,\eta}(H,H',\phi,\varphi,{\sf out})=(C,T,T',\rho,{\cal P}, f, {\sf out})$. Since $(H,H',\phi,\varphi,{\sf out})$ is an $(i,\eta)$-partial solution, we know that $H'$ is a topological minor  in $G^{\star}_{i}$ witnessed by $(\phi_{H'},\varphi_{H'})$ and $(H',\phi,\varphi)$ is a representation of $(H',\phi_{H'},\varphi_{H'})$ in $\tilde{G}_{i}$ with the following properties. 
	\begin{enumerate}[label=(\alph*)]
	\item\label{propertya} Every vertex in $\image(\varphi)\setminus (\phi(V(H')\setminus V(H))$ belongs to $V(\tilde{G}_{i-1})$. 
	\item\label{propertyb} Every vertex in $\phi(V(H')\setminus V(H))$ belongs to $\inNoose_{\tilde{G}}(N)$ for some $N\in\fr[\ell]$ and $N$ is either a left-noose or a right noose in $\fr[\ell]$.
\item\label{propertyc} For any $N\in\fr[\ell]$, $|\inNoose_{\tilde{G}}(N)\cap \phi(V(H'))|\leq 1$. 
\item For any $v\in (V(H')\cap V(H)) \setminus R(H')$, ${\sf out}(v)\subseteq A\cap (N_G(\phi(v)))$.
	\end{enumerate}

Now we construct $(\phi',\varphi')$ from $(\phi,\varphi)$ and then prove that $\camera_{i+1,\eta}(H,H',\phi',\varphi', {\sf out})=(C,T,T',\rho,{\cal P}, f, {\sf out})$. Towards that let $\{v_1,\ldots,v_{s}\}=V(H')\setminus V(H)$ and $u_j=\phi(v_j)$ for all $j\in [s]$. Let $N_j$ be the noose in $\fr[i]$  such that $u_j\in \inNoose_{\tilde{G}}(N_j)$. By property~\ref{propertyc}, we have that these nooses $N_1,\ldots,N_s$ are distinct and by property~\ref{propertyb}, each of these nooses is either a left noose or a right noose  in $\fr[i]$. There exist distinct nooses $N_1',\ldots, N_s'$ in $\fr[i+1]$ such that $N_j$ is adjacent to $N_j'$ for all $j\in [s]$, in the noose grid ${\cal N}$ and these nooses are from the union of the right-nooses and left-nooses of $\fr[i+1]$.  Let $P_j$ be a shortest path from $u_j$ to a vertex in the noose $N_j'$ in $\tilde{G}[\inNoose_{\tilde{G}}(N_j)\cup \inNoose_{\tilde{G}}(N_j')]$.  Notice that all the vertices in $P_j$, except the ending vertex are in $\inNoose_{\tilde{G}}(N_j)$ and the ending vertex $w_j$ of $P_j$, belong to $\inNoose_{\tilde{G}}(N_j')$.

For any $\ell\in [q-1]_0$, let $U_{\ell}$ be the union of the vertices from the nooses in $\bigcup_{i \in [\ell]_0}\fr[i]$. Let ${\cal G}_{\ell} \subseteq \{G_1,\ldots,G_k\}$ be the graphs such that for any $i\in [k]$, $V(G_i)\cap V(G_0)\subseteq U_{\ell}$ if and only if $G_i\in {\cal G}$. 
Let $F$ be the set of edges in the paths $P_1,\ldots,P_s$ such that no edge in $F$ is in $G^{\star}_{i+1}$. This implies that for any edge $\{w,w'\}\in F$, there is a graph $G_j$ such that $G_j\in {\cal G}_{i+1}\setminus {\cal G}_{i-1}$ and $w,w'\in V(G_j)$. Because of property~\ref{propertya}, no vertex in $V(G_j)\setminus V(G_0)$ is used by 
$\varphi$ and $\varphi_{H'}$.  Therefore using condition~\ref{conditionsix} of \autoref{obs:cflatmore}, we can reroute  the edges in $F$ with vertex disjoint paths from $G^{\star}_{i+1}$ and these new vertices are not used by $\varphi$ and $\varphi_{H'}$. Let $P_1^{\star},\ldots,P_s^{\star}$ be paths obtained from $P_1,\ldots,P_s$ as explained above. 


Now we define $(\phi',\varphi')$ and $(\phi'_{H'},\varphi'_{H'})$ such that $(\phi',\varphi')$ is a representation of 
$(\phi'_{H'},\varphi'_{H'})$ and $(H,H', \phi'_{H'},\varphi'_{H'}, {\sf out})\in \col_{i+1,\eta}$.  For each $v\in V(H')\cap V(H)$, $\phi'(v)=\phi(v)$. For any $v_j\in V(H')\setminus V(H)$, $\phi'(v_j)=w_j$. We set $\phi'_{H'}=\phi'$. Now we define $\varphi'$ and $\varphi'_{H'}$ as follows. Fix an edge $e=\{x,y\}\in E(H')$. If $x,y\in V(H')\cap V(H)$, then $\varphi'(e)=\varphi(e)$ and $\varphi'_{H'}(e)=\varphi_{H'}(e)$. If $x\in V(H')\cap V(H)$ and $y=v_j\in V(H')\setminus V(H)$, then $\varphi'(e)$ is the concatenation of paths $\varphi(e)$ and $P_j$ and $\varphi'_{H'}(e)$ is the concatenation of paths $\varphi_H(e)$ and $P^{\star}_j$. If $x,y \in V(H')\setminus V(H)$, then let $x=v_j$ and $y=v_{j'}$. Then $\varphi'(e)$ is a concatenation of paths $P_j$, $\varphi(e)$, and $P_{j'}$ and $\varphi'_{H'}(e)$ is a concatenation of paths $P^{\star}_j$, $\varphi_H(e)$, and $P^{\star}_{j'}$.  

From the construction, we have that $(H',\phi',\varphi')$ is a representation of $(H',\phi_{H'},\varphi_{H'})$  in $G^{\star}_{i+1}$. Now we prove that $(H',\phi',\varphi')$ satisfies condition $(2)$ of  \autoref{def:partial}. 
Condition $2(a)$ follows from the fact that the only vertices in $\fr[i+1]$ used by $\varphi'$ is $\{w_1,\ldots,w_b\}$. Condition $2(b)$, follows from the definition $\phi'$ and the fact that $N_1',\ldots,N_s'$ are nooses from the union of right-nooses and left-nooses of $\fr[i+1]$.  Condition $3(b)$ follows from the fact that $w_1,\ldots,w_b$ belong to distinct nooses in $\fr[i+1]$. Since $\phi(v)=\phi'(v)$ for all $v\in V(H')\cap V(H)$, condition $2(d)$ is satisfied for the function ${\sf out}$. Therefore $(H,H',\phi',\varphi', {\sf out})$ is an $(i+1,\eta)$-partial solution. 

Now we prove that $\camera_{i+1,\eta}(H,H',\phi',\varphi',{\sf out})=(C,T,T',\rho,{\cal P}, f,{\sf out})$. By assumption we have that $\camera_{i,\eta}(H,H',\phi,\varphi, {\sf out})=(C,T,T',\rho,{\cal P}, f, {\sf out})$ and therefore the first four conditions in  \autoref{def:camera} clearly follows. By replacing $N_1,\ldots,N_s$ with $N_1,\ldots,N_s$ and $u_j^G$ with $w_j$, one can verify the fifth condition is also satisfied.  

From the construction of $\phi'$, we have that $\phi'(V(H)\cap V(H'))=\phi(V(H)\cap V(H'))\subseteq V(\tilde{G}_{i-1})$ (by property~\ref{propertya}). This completes the proof of the lemma. 
\end{proof}

Intuitively, an album is defined to be the series of snapshot collections at regular intervals.

\begin{definition}[{\bf Album}]\label{def:album}
Let $(G,\delta,t,w',s')$ be an instance of \fFindFoliostar. Let $(M,{\cal N})$ be a $2q$-workspace in $\tilde{G}$ 
and $\eta,\lambda\in\mathbb{N}$ where $\lambda$ divides $q$. The {\em $(\eta,\lambda)$-album} is the function $\album_{\eta,\lambda}: [q/\lambda]\rightarrow 2^{\SNP[\eta]}$ where for all $i\in[q/\lambda]$, $\album_{\eta,\lambda}(i)=\col_{(i\lambda-1),\eta}$. 
\end{definition}

As a corollary to \autoref{lem:snapshotSetMonotone}, we have the following statement.

\begin{corollary}\label{cor:snapshotSetMonotone}
Let $(G,\delta,t,w',s')$ be an instance of \fFindFoliostar, $(M,{\cal N})$ be a $2q$-workspace in $\tilde{G}$ 
and $\eta,\lambda\in\mathbb{N}$ where $\lambda$ divides $q$.
For all $i\in[q/\lambda-1]$, $\album_{\eta,\lambda}(i)\subseteq \album_{\eta,\lambda}(i+1)$.
\end{corollary}
%

If snapshot collections are the same for a large consecutive sequence in an album, then we we say that it is a redundant album.

\begin{definition}[{\bf Redundant Album}]\label{def:redundantAlbum}

Let $(G,\delta,t,w',s')$ be an instance of \fFindFoliostar. Let $(M,{\cal N})$ be a $2q$-workspace in $\tilde{G}$ 
and $\eta,\lambda,\mu \in\mathbb{N}$ where $\lambda$ divides $q$.
The $(\eta,\lambda)$-album is {\em $\mu$-redundant} if there exists $s\in[q/\lambda-\mu+1]$, called a {\em $\mu$-redundancy stamp}, such that $\album_{\eta,\lambda}(s)=\album_{\eta,\lambda}(s+1)=\cdots=\album_{\eta,\lambda}(s+\mu-1)$.
\end{definition}

In the following lemma, we prove that there is an efficient algorithm to compute a $\mu$-redundancy stamp.

\begin{lemma}\label{lem:redundantAlbumExistence}
Let $(G,\delta,t,w',s')$ be an instance of \fFindFoliostar. Let $(M,{\cal N})$ be a $(p,2q)$-workspace in $\tilde{G}$ 
and $\eta,\lambda,\mu \in\mathbb{N}$ where $\lambda$ divides $q$. 
There is a constant $c$ such that if $q/\lambda-\mu+ 1>(t+\delta^{\star})^{c(t^2+t\delta^{\star})}\cdot\eta^{c\eta} \cdot 2\mu$, 
then the $(\eta,\lambda)$-album is {\em $\mu$-redundant} and a $\mu$-redundancy stamp for the $(\eta,\lambda)$-album can be found in time 
 $((t+\delta^{\star})^{\OO(t^2+t\delta^{\star})}\cdot\eta^{\OO(\eta)}\cdot \mu) \Delta^{\OO(\Delta)} n$
 where $\Delta=\max \{p,s'+3\}$. 
\end{lemma}

\begin{proof}
By \autoref{obs:numSnapshots}, we know that the number of distinct $\eta$-snapshots is upper bounded by $(t+\delta^{\star})^{c(t^2+t\delta^{\star})}\cdot\eta^{c\eta}$, where $c$ is a constant.  
By \autoref{cor:snapshotSetMonotone}, $\album_{\eta,\lambda}(i)\subseteq \album_{\eta,\lambda}(i+1)$. Therefore if  $q/\lambda-\mu+ 1>(t+\delta^{\star})^{c(t^2+t\delta^{\star})}\cdot\eta^{c\eta} \cdot 2\mu$,
then there exists  $s\in [q/\lambda-\mu+ 1]$ such that $\album_{\eta,\lambda}(s)=\album_{\eta,\lambda}(s+1)=\cdots=\album_{\eta,\lambda}(s+\mu-1)$.

To compute a $\mu$-redundancy stamp $s\in [q/\lambda-\mu+ 1]$, it is enough to compute $\album_{\eta,\lambda}(i)$ for all  $i\in [q/\lambda-\mu+ 1]$. Towards that we need to use \autoref{lem:testRealizable} 
$(t+\delta^{\star})^{c(t^2+t\delta^{\star})}\cdot\eta^{c\eta} \cdot 2\mu+1$ times. 
This implies that the running time to compute a $\mu$-redundancy stamp  is 
$(t+\delta^{\star})^{c(t^2+t\delta^{\star})}\cdot\eta^{c\eta} \cdot 2\mu \cdot \Delta^{\OO(\Delta)} (\eta+\delta^{\star}+t)^{\OO(1)} n$, which is equal to 
$((t+\delta^{\star})^{\OO(t^2+t\delta^{\star})}\cdot\eta^{\OO(\eta)}\cdot \mu) \Delta^{\OO(\Delta)} n$, 
where $\Delta=\max \{p,s'+3\}$.
\end{proof}

%
%

If there is a $\mu$-redundancy stamp $s$ for a large $\mu$, then there is a $(\ell,\lambda-1)$-terminal free solution, where $\ell\geq (s+1)\lambda$.

\begin{observation}\label{obs:existsTermFreeFrame}
Let $(G,\delta,t,w',s')$ be an instance of \fFindFoliostar\ and ${\cal S}$ be a solution. Let $(M,{\cal N})$ be a $2q$-workspace in $\tilde{G}$ and $\eta,\lambda,\mu \in\mathbb{N}$ where $\lambda$ divides $q$ such that 
$q/\lambda-\mu+ 1>(t+\delta^{\star})^{c(t^2+t\delta^{\star})}\cdot\eta^{c\eta} \cdot 2\mu$, 
where 
$c$ is the constant mentioned in  \autoref{lem:redundantAlbumExistence}. Let $s$ be a $\mu$-redundancy stamp for the $(\eta,\lambda)$-album. Then, there is a constant $c'$ such that if 
$\mu> 2^{c'(\delta^{\star})^2}t$, 
then  there exists $j\in \{s+1,\ldots, s+\mu-1\} $ such that ${\cal S}$ is $(j\lambda,\lambda-1)$-terminal free.
\end{observation}

\begin{proof}
By \autoref{prop:no.ofgraphsinfolios},  the number of distinct graphs (upto isomorphism) in the $\delta^{\star}$-folio of $G$ is upper bounded by  
$2^{\OO((\delta^{\star})^2)}\cdot  \vert R(G)\vert^{\OO(\delta^{\star})}=2^{\OO((\delta^{\star})^2)}$. 
This implies that the number of terminal vertices in ${\cal S}$ is upper bounded by 
$ 2^{c''(\delta^{\star})^2}$, where $c''$ is a constant. 
Moreover, for each $(H,\phi,\varphi)\in {\cal S}$, the number of edges with one end point in $A$ and other in $G_0\cup G_1 \cup \ldots \cup G_k$ is at most $\delta^{\star}t$. 
Therefore, there is a constant $c'$ such that 
if $\mu>2^{c'(\delta^{\star})^2}t$ and  $s$ be a $\mu$-redundancy stamp for the $(\eta,\lambda)$-album, there is $i\geq s$ such that no terminal from ${\cal S}$ belongs $U$ and no edge in $\varphi_H$ (for any $(H,\phi_H,\varphi_H)\in {\cal S}$) with one endpoint in $U$ and other in $A$, where $U$ is the union of vertices in the nooses in $\bigcup_{i\lambda\leq j < (i+2)\lambda} \fr[j]$ and the union of vertices of $V(G_i)$, with $V(G_i)\cap V(G_0)$ contains a vertex in a noose in  $\bigcup_{i\lambda\leq j < (i+2)\lambda} \fr[j]$. 
This implies ${\cal S}$ is 
$((i+1)\lambda,\lambda-1)$-terminal free.
\end{proof}



\subsection{Patching Snapshots}\label{sec:patch}

In this subsection we prove that if there is a $\mu$-redundancy stamp $s$ for a large $\mu$, then any vertex in any up-noose of $\fr[s\lambda-1]$ is irrelevant. Recall that in the last subsection we proved that there is an $(\ell,\lambda-1)$-terminal free solution for some $\ell\geq (s+1)\lambda$. We prove that there is a ``patch'' between  an $(s\lambda-1,\eta)$-partial solution and the ``future part'' of the corresponding $(\ell,\eta)$-partial solution without using any up-noose of  $\fr[s\lambda-1]$. 

\begin{definition}[{\bf Zoom}]\label{def:solutionToSnap}
Let $(G,\delta,t,w',s')$ be an instance of \fFindFoliostar. Let $(M,{\cal N})$ be a $2q$-workspace in $\tilde{G}$, $\ell\in[q-1]_0$ and $\eta\in\mathbb{N}$. Let ${\cal S}'$ be the representation of  an $(\ell,\eta)$-untangled solution of $(G,\delta,t,w',s')$. Let $(H,\phi,\varphi)\in {\cal S}'$. The {\em $(\ell,\eta)$-zoom} of $(H,\phi,\varphi)$, denoted by $\zoom_{\ell,\eta}(H,\phi,\varphi)$, is  $\camera_{\ell,\eta}(\restrict_{\ell,\eta}(H,\phi,\varphi))$.
\end{definition}

\begin{definition}[{\bf Miniature}]\label{def:miniature}
Let $(G,\delta,t,w',s')$ be an instance of \fFindFoliostar.
Let $(M,{\cal N})$ be a $2q$-workspace in $\tilde{G}$ and $\eta,\lambda,\mu \in\mathbb{N}$ where $\lambda$ divides $q$. Let $s$ be a $\mu$-redundancy stamp for the $(\eta,\lambda)$-album and $\ell\geq (s+1)\lambda$. Let ${\cal S}'$ be a representation of an $(\ell,\eta)$-untangled solution and $(H,\phi,\varphi)\in {\cal S}'$. An $(s\lambda-1,\eta)$-partial solution $(H^{\star},H',\phi',\varphi',{\sf out})$ is an {\em $(\ell,s)$-miniature of $(H,\phi,\varphi)$} if $\zoom_{\ell,\eta}(H,\phi,\varphi)=\camera_{s\lambda-1,\eta}(H^{\star},H',\phi',\varphi',{\sf out})$. 
\end{definition}

The following observation follows from \autoref{obs:snapshotiso}. 

%
%

\begin{observation}
\label{obs:graphsequal}
Let $(G,\delta,t,w',s')$ be an instance of \fFindFoliostar.
Let $(M,{\cal N})$ be a $2q$-workspace in $\tilde{G}$ and $\eta,\lambda,\mu \in\mathbb{N}$ where $\lambda$ divides $q$. Let $s$ be a $\mu$-redundancy stamp for the $(\eta,\lambda)$-album and $\ell\geq (s+1)\lambda$. Let ${\cal S}'$ be the representation of an $(\ell,\eta)$-untangled solution, and $(H,\phi,\varphi)\in {\cal S}'$. Let $(H^{\star},H',\phi',\varphi',{\sf out})$ be  an $(\ell,s)$-miniature of $(H,\phi,\varphi)$. Let $(H^{\star}_1,H'_1,\phi'_1,\varphi'_1,{\sf out}_1)=\restrict_{\ell,\eta}(H,\phi,\varphi)$. Then,  ${\sf out}={\sf out}_1$, there exists an isomorphism ${\sf iso}'$ from  $H'$ to  $H_1'$ such that ${\sf iso}'$  maps  $V(H')\setminus V(H^{\star})$ to $V(H'_1)\setminus V(H_1^{\star})$  and ${\sf iso}'$ restricted to $V(H)\cap V(H')$ is an identity map (Here $V(H^{\star})\cap V(H')=V(H^{\star}_1)\cap V(H_1')$). 
\end{observation}

\begin{proof}[Proof sketch]
By assumption we have that   
$\camera_{s\lambda-1,\eta}(H^{\star},H',\phi',\varphi',{\sf out}')=\camera_{\ell,\eta}(H^{\star}_1,H'_1,\phi'_1,$ $\varphi'_1,{\sf out}_1')$. That is, there exist two equivalent $\eta$-snapshot $(C_1,T_1,T'_1,\rho_1,{\cal P}_1, f_1,{\sf out}_1)$ and $(C,T,T',$ $\rho,{\cal P}, f,{\sf out})$ such that 
\begin{eqnarray*}
\label{eqn:isomopartial}
(C,T,T',\rho,{\cal P}, f, {\sf out})&=& \camera_{s\lambda-1,\eta}(H^{\star},H',\phi',\varphi',{\sf out}')\\
(C_1,T_1,T'_1,\rho_1,{\cal P}_1, f_1, {\sf out}_1)&=&   \camera_{\ell,\eta}(H^{\star}_1,H'_1,\phi'_1,\varphi'_1,{\sf out}_1') 
\end{eqnarray*}

Since $(C_1,T_1,T'_1,\rho_1,{\cal P}_1, f_1,{\sf out}_1)$ and $(C,T,T',\rho,{\cal P}, f,{\sf out})$, 
we have that ${\sf out}={\sf out}_1={\sf out}'={\sf out}_1'$. 
Let $J$ and $J_1$ be the two graphs defined from 
$(C,T,T',\rho,{\cal P}, f,{\sf out})$ and $(C_1,T_1,T'_1,\rho_1,{\cal P}_1, f_1,$ ${\sf out}_1)$, 
respectively, as mentioned in \autoref{obs:snapshotiso}. Let ${\sf iso}$ and ${\sf iso}_1$ be the two isomorphisms from $J$ to $H'$ and $J_1$ to $H'_1$, respectively, as mentioned in \autoref{obs:snapshotiso}. 

Since $(C,T,T',\rho,{\cal P}, f,{\sf out})$ and $(C_1,T_1,T'_1,\rho_1,{\cal P}_1, f_1,{\sf out}_1)$ are equivalent, there is an isomorphism $g$ from $C$ to $C_1$ such that for all $v\in V(C)$, $f(v)=f_1(g(v))$. 
The function $g$ can be extended to an isomorphism from $J_1$ to $J_2$ as follows. For any $t\in T=T_1$, we set $g(t)=t$. To get the required isomorphism ${\sf iso}$ from $H'$ to $H'_1$, we apply ${\sf iso}^{-1}$, then $g$, and then ${\sf  iso}_2$. That is, for any $v\in V(H')$, ${\sf iso}'(v)={\sf iso}_1(g({\sf iso}^{-1} (v)))$. 
\end{proof}

Because of \autoref{obs:graphsequal}, whenever we say that $(H^{\star},H',\phi',\varphi',{\sf out})$ is an $(\ell,s)$-miniature of  $(H,\phi,\varphi)$, we assume that $(H^{\star}_1,H',\phi'_1,\varphi'_1,{\sf out})=\restrict_{\ell,\eta}(H,\phi,\varphi)$, for some witness $(\phi'_1,\varphi'_1)$ of $H'$. 
Next we formally define a patch and prove that indeed it exists. 

\begin{definition}[{\bf Patch}]\label{def:patch}
Let $(G,\delta,t,w',s')$ be an instance of \fFindFoliostar, ${\cal S}$ be a solution and ${\cal S}'$ be its representation. Let $(M,{\cal N})$ be a $2q$-workspace in $\tilde{G}$ and $\eta,\lambda,\mu \in\mathbb{N}$ where $\lambda$ divides $q$. Let $s$ be a $\mu$-redundancy stamp for the $(\eta,\lambda)$-album.  
Let $\ell\geq (s+1)\lambda$ be such that ${\cal S}$ is $(\ell,\eta)$-untangled. Let $(H,\phi,\varphi)\in {\cal S'}$ and $(H^{\star},H',\phi',\varphi', {\sf out})$ be an $(\ell,s)$-miniature of $(H,\phi,\varphi)$. 
Denote $(H^{\star}_1,H',\phi'_1,\varphi'_1,{\sf out})=\restrict_{\ell,\eta}(H,\phi,\varphi)$.  
A {\em patch} from  $(H^{\star},H',\phi',\varphi', {\sf out})$ to $(H,\phi,\varphi)$ is a set $\cal P$ of $|V(H')\setminus V(H^{\star})|$ vertex-disjoint paths where for all $v\in V(H')\setminus V(H^{\star})$, there exists $P\in{\cal P}$ between $\phi'(v)$ and $\phi'_1(v)$ such that every internal vertex on $P$ is enclosed either by the noose in $\cal N$ that encloses $\phi'(v)$ or by a noose $N\in\fr[i]$ for some $i\in\{s\lambda,s\lambda+1,\ldots,\ell-1\}$.

For any integer $f\in\mathbb{N}$, we say that a patch $\cal P$ is {\em $f$-empty} if for every vertex $v$ of every path in $\cal P$ and for every up-noose $N_{i',j'}\in\fr[i]$ for all $i\in\{\lambda s-1,\lambda s,\ldots, \lambda s+f-2\}$, $v$ does not belong to $\inNoose_{\tilde{G}}(N_{i',j'})$. 
\end{definition}

Observe that by the definition of a miniature and a partial solution, any patch in the definition above must be $1$-empty.


\begin{lemma}\label{lem:patchExists}
Let $(G,\delta,t,w',s')$ be an instance of \fFindFoliostar, ${\cal S}$ be a solution and ${\cal S}'$ be its representation. Let $(M,{\cal N})$ be a $2q$-workspace in $\tilde{G}$ and $\eta,\lambda,\mu \in\mathbb{N}$ where $\lambda$ divides $q$ and $\lambda \geq 2\eta+1$. Let $s$ be a $\mu$-redundancy stamp for the $(\eta,\lambda)$-album.  Let $\ell\geq (s+1)\lambda$ be such that ${\cal S}$ is $(\ell,\eta)$-untangled. Let $(H,\phi,\varphi)\in {\cal S'}$ and $(H^{\star},H',\phi',\varphi', {\sf out})$ be an $(\ell,s)$-miniature of $(H,\phi,\varphi)$. 
Then, there exists a patch from  $(H^{\star},H',\phi',\varphi', {\sf out})$ to $(H,\phi,\varphi)$ that is $f$-empty for $f=\lambda-2\eta$. 
\end{lemma}

The correctness of Lemma \ref{lem:patchExists} directly following from the following simple observation and the definitions of $(\ell,\eta)$-untangled and miniature.

\begin{observation}
Let $a,r,t\in\mathbb{N}$ such that $t\leq 2r< 2a$ and $2t<a-r$. Let $G$ be a $(2a)\times (2a)$ grid with $V(G)=\{v_{i,j}: i,j\in\{1,2,\ldots,2a\}\}$, and $H$ be a $2r\times 2r$ subgrid of $G$ where $V(H)=\{v_{i,j}: i,j\in\{a-(r-1),\ldots,a,a+1,\ldots,a+r\}\}$. Denote $C_G$ and $C_H$ denote the outermost cycles of $G$ and $H$, i.e., the cycles induced by their set of vertices of degree smaller than $4$. Let ${\cal M}=\{(p_\ell,q_\ell): \ell\in\{1,2,\ldots,t\}, p_\ell\in \{v_{i,j}\in V(C_G): i\neq 2a\},q_\ell\in \{v_{i,j}\in V(C_H): i\neq a+r\}\}$ be a collection of vertex pairs such that when we traverse $C_G$ and $C_H$ in cyclic order starting at $p_1$ and $q_1$, we encounter $p_1,p_2,\ldots,p_t$ and $q_1,q_2,\ldots,q_t$ in this order, respectively. Then, there exists $t$ vertex disjoint paths $P_1,P_2,\ldots,P_t\}$ in $G$ such that for every $i\in\{1,2,\ldots,t\}$, the endpoints of $P_i$ are $p_i$ and $q_i$, and its internal vertices belong to $V(G)\setminus (V(H)\cup V(C_G)\cup X)$ where $X=\{v_{i,j}: i\in\{a+r,a+r+1,\ldots,2a-2t\}, j\in\{2t+1,2t+2,\ldots,2a-2t+1\}\}$.   
\end{observation}

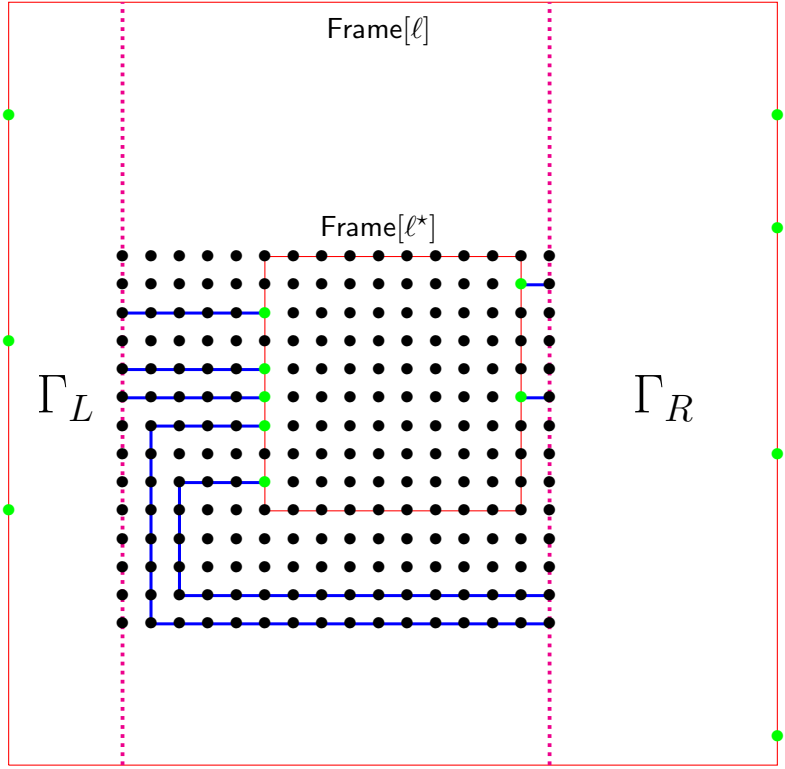
\begin{figure}
\begin{center}
\begin{tikzpicture}[scale=0.75]

\draw[red] (-4.5,-4.5)--(9,-4.5)--(9,9)--(-4.5,9)--(-4.5,-4.5); 

\node[green] (a1) at (9,7) {$\bullet$}; 
\node[green] (a1) at (9,5) {$\bullet$}; 
\node[green] (a1) at (9,1) {$\bullet$}; 
\node[green] (a1) at (9,-4) {$\bullet$};

\node[green] (a1) at (-4.5,0) {$\bullet$}; 
\node[green] (a1) at (-4.5,3) {$\bullet$}; 
\node[green] (a1) at (-4.5,7) {$\bullet$};

\draw[magenta, dotted, line width=0.5mm] (5,-4.5)--(5,9);

\draw[magenta, dotted, line width=0.5mm] (-2.5,-4.5)--(-2.5,9);

\draw[blue, line width=0.4mm](0,0.5)--(-1.5,0.5)--(-1.5,-1.5)--(5,-1.5);

\draw[blue, line width=0.4mm](0,1.5)--(-2,1.5);
\draw[blue, line width=0.4mm](-2,1.5)--(-2,-2)--(5,-2);

\draw[blue, line width=0.4mm](0,2)--(-2.5,2);

\draw[blue, line width=0.4mm](0,2.5)--(-2.5,2.5);
\draw[blue, line width=0.4mm](0,3.5)--(-2.5,3.5);

\node[] (a1) at (2,5) {$\fr[\ell^{\star}]$}; 
\node[] (a1) at (2,8.5) {$\fr[\ell]$}; 
\node[] (a1) at (7,2) {\huge{$\Gamma_R$}}; 

\node[] (a1) at (-3.5,2) {\huge{$\Gamma_L$}}; 

\node[] (a1) at (-2.5,-2) {$\bullet$}; 
\node[] (a1) at (-2,-2) {$\bullet$}; 
\node[] (a1) at (-1.5,-2) {$\bullet$}; 
\node[] (a1) at (-1,-2) {$\bullet$}; 
\node[] (a1) at (-0.5,-2) {$\bullet$}; 

\node[] (a1) at (0,-2) {$\bullet$}; 
\node[] (a1) at (0.5,-2) {$\bullet$}; 
\node[] (a1) at (1,-2) {$\bullet$}; 
\node[] (a1) at (1.5,-2) {$\bullet$}; 
\node[] (a1) at (2,-2) {$\bullet$}; 
\node[] (a1) at (2.5,-2) {$\bullet$}; 
\node[] (a1) at (3,-2) {$\bullet$};
\node[] (a1) at (3.5,-2) {$\bullet$};  
\node[] (a1) at (4,-2) {$\bullet$};
\node[] (a1) at (4.5,-2) {$\bullet$}; 
\node[] (a1) at (5,-2) {$\bullet$};

\node[] (a1) at (-2.5,-1.5) {$\bullet$}; 
\node[] (a1) at (-2,-1.5) {$\bullet$}; 
\node[] (a1) at (-1.5,-1.5) {$\bullet$}; 
\node[] (a1) at (-1,-1.5) {$\bullet$}; 
\node[] (a1) at (-0.5,-1.5) {$\bullet$}; 

\node[] (a1) at (0,-1.5) {$\bullet$}; 
\node[] (a1) at (0.5,-1.5) {$\bullet$}; 
\node[] (a1) at (1,-1.5) {$\bullet$}; 
\node[] (a1) at (1.5,-1.5) {$\bullet$}; 
\node[] (a1) at (2,-1.5) {$\bullet$}; 
\node[] (a1) at (2.5,-1.5) {$\bullet$}; 
\node[] (a1) at (3,-1.5) {$\bullet$};
\node[] (a1) at (3.5,-1.5) {$\bullet$};  
\node[] (a1) at (4,-1.5) {$\bullet$};
\node[] (a1) at (4.5,-1.5) {$\bullet$}; 
\node[] (a1) at (5,-1.5) {$\bullet$};

\node[] (a1) at (-2.5,-1) {$\bullet$}; 
\node[] (a1) at (-2,-1) {$\bullet$}; 
\node[] (a1) at (-1.5,-1) {$\bullet$}; 
\node[] (a1) at (-1,-1) {$\bullet$}; 
\node[] (a1) at (-0.5,-1) {$\bullet$}; 

\node[] (a1) at (0,-1) {$\bullet$}; 
\node[] (a1) at (0.5,-1) {$\bullet$}; 
\node[] (a1) at (1,-1) {$\bullet$}; 
\node[] (a1) at (1.5,-1) {$\bullet$}; 
\node[] (a1) at (2,-1) {$\bullet$}; 
\node[] (a1) at (2.5,-1) {$\bullet$}; 
\node[] (a1) at (3,-1) {$\bullet$};
\node[] (a1) at (3.5,-1) {$\bullet$};  
\node[] (a1) at (4,-1) {$\bullet$};
\node[] (a1) at (4.5,-1) {$\bullet$}; 
\node[] (a1) at (5,-1) {$\bullet$};

\node[] (a1) at (-2.5,-0.5) {$\bullet$}; 
\node[] (a1) at (-2,-0.5) {$\bullet$}; 
\node[] (a1) at (-1.5,-0.5) {$\bullet$}; 
\node[] (a1) at (-1,-0.5) {$\bullet$}; 
\node[] (a1) at (-0.5,-0.5) {$\bullet$}; 

\node[] (a1) at (0,-0.5) {$\bullet$}; 
\node[] (a1) at (0.5,-0.5) {$\bullet$}; 
\node[] (a1) at (1,-0.5) {$\bullet$}; 
\node[] (a1) at (1.5,-0.5) {$\bullet$}; 
\node[] (a1) at (2,-0.5) {$\bullet$}; 
\node[] (a1) at (2.5,-0.5) {$\bullet$}; 
\node[] (a1) at (3,-0.5) {$\bullet$};
\node[] (a1) at (3.5,-0.5) {$\bullet$};  
\node[] (a1) at (4,-0.5) {$\bullet$};
\node[] (a1) at (4.5,-0.5) {$\bullet$}; 
\node[] (a1) at (5,-0.5) {$\bullet$};

\draw[red] (4.5,0)--(4.5,4.5);
\draw[red] (0,4.5)--(0,0);
\draw[red] (4.5,0)--(0,0);

\node[] (a1) at (-2.5,0) {$\bullet$}; 
\node[] (a1) at (-2,0) {$\bullet$}; 
\node[] (a1) at (-1.5,0) {$\bullet$}; 
\node[] (a1) at (-1,0) {$\bullet$}; 
\node[] (a1) at (-0.5,0) {$\bullet$}; 

\node[] (a1) at (0,0) {$\bullet$}; 
\node[] (a1) at (0.5,0) {$\bullet$}; 
\node[] (a1) at (1,0) {$\bullet$}; 
\node[] (a1) at (1.5,0) {$\bullet$}; 
\node[] (a1) at (2,0) {$\bullet$}; 
\node[] (a1) at (2.5,0) {$\bullet$}; 
\node[] (a1) at (3,0) {$\bullet$};
\node[] (a1) at (3.5,0) {$\bullet$};  
\node[] (a1) at (4,0) {$\bullet$};
\node[] (a1) at (4.5,0) {$\bullet$}; 
\node[] (a1) at (5,0) {$\bullet$};

\node[] (a1) at (-2.5,0.5) {$\bullet$}; 
\node[] (a1) at (-2,0.5) {$\bullet$}; 
\node[] (a1) at (-1.5,0.5) {$\bullet$}; 
\node[] (a1) at (-1,0.5) {$\bullet$}; 
\node[] (a1) at (-0.5,0.5) {$\bullet$}; 

\node[green] (a1) at (0,0.5) {$\bullet$}; 
\node[] (a1) at (0.5,0.5) {$\bullet$}; 
\node[] (a1) at (1,0.5) {$\bullet$}; 
\node[] (a1) at (1.5,0.5) {$\bullet$}; 
\node[] (a1) at (2,0.5) {$\bullet$}; 
\node[] (a1) at (2.5,0.5) {$\bullet$}; 
\node[] (a1) at (3,0.5) {$\bullet$};
\node[] (a1) at (3.5,0.5) {$\bullet$};  
\node[] (a1) at (4,0.5) {$\bullet$};
\node[] (a1) at (4.5,0.5) {$\bullet$}; 
\node[] (a1) at (5,0.5) {$\bullet$};

\node[] (a1) at (-2.5,1) {$\bullet$}; 
\node[] (a1) at (-2,1) {$\bullet$}; 
\node[] (a1) at (-1.5,1) {$\bullet$}; 
\node[] (a1) at (-1,1) {$\bullet$}; 
\node[] (a1) at (-0.5,1) {$\bullet$}; 

\node[] (a1) at (0,1) {$\bullet$}; 
\node[] (a1) at (0.5,1) {$\bullet$}; 
\node[] (a1) at (1,1) {$\bullet$}; 
\node[] (a1) at (1.5,1) {$\bullet$}; 
\node[] (a1) at (2,1) {$\bullet$}; 
\node[] (a1) at (2.5,1) {$\bullet$}; 
\node[] (a1) at (3,1) {$\bullet$};
\node[] (a1) at (3.5,1) {$\bullet$};  
\node[] (a1) at (4,1) {$\bullet$};
\node[] (a1) at (4.5,1) {$\bullet$}; 
\node[] (a1) at (5,1) {$\bullet$};

\node[] (a1) at (-2.5,1.5) {$\bullet$}; 
\node[] (a1) at (-2,1.5) {$\bullet$}; 
\node[] (a1) at (-1.5,1.5) {$\bullet$}; 
\node[] (a1) at (-1,1.5) {$\bullet$}; 
\node[] (a1) at (-0.5,1.5) {$\bullet$}; 

\node[green] (a1) at (0,1.5) {$\bullet$}; 
\node[] (a1) at (0.5,1.5) {$\bullet$}; 
\node[] (a1) at (1,1.5) {$\bullet$}; 
\node[] (a1) at (1.5,1.5) {$\bullet$}; 
\node[] (a1) at (2,1.5) {$\bullet$}; 
\node[] (a1) at (2.5,1.5) {$\bullet$}; 
\node[] (a1) at (3,1.5) {$\bullet$};
\node[] (a1) at (3.5,1.5) {$\bullet$};  
\node[] (a1) at (4,1.5) {$\bullet$};
\node[] (a1) at (4.5,1.5) {$\bullet$}; 
\node[] (a1) at (5,1.5) {$\bullet$};

\node[] (a1) at (-2.5,2) {$\bullet$}; 
\node[] (a1) at (-2,2) {$\bullet$}; 

\draw[blue, line width=0.4mm](4.5,2)--(5,2);

\node[] (a1) at (-1.5,2) {$\bullet$}; 
\node[] (a1) at (-1,2) {$\bullet$}; 
\node[] (a1) at (-0.5,2) {$\bullet$}; 
\node[green] (a1) at (0,2) {$\bullet$}; 
\node[] (a1) at (0.5,2) {$\bullet$}; 
\node[] (a1) at (1,2) {$\bullet$}; 
\node[] (a1) at (1.5,2) {$\bullet$}; 
\node[] (a1) at (2,2) {$\bullet$}; 
\node[] (a1) at (2.5,2) {$\bullet$}; 
\node[] (a1) at (3,2) {$\bullet$};
\node[] (a1) at (3.5,2) {$\bullet$};  
\node[] (a1) at (4,2) {$\bullet$};
\node[green] (a1) at (4.5,2) {$\bullet$}; 
\node[] (a1) at (5,2) {$\bullet$};

\node[] (a1) at (-2.5,2.5) {$\bullet$}; 
\node[] (a1) at (-2,2.5) {$\bullet$}; 
\node[] (a1) at (-1.5,2.5) {$\bullet$}; 
\node[] (a1) at (-1,2.5) {$\bullet$}; 
\node[] (a1) at (-0.5,2.5) {$\bullet$}; 

\node[green] (a1) at (0,2.5) {$\bullet$}; 
\node[] (a1) at (0.5,2.5) {$\bullet$}; 
\node[] (a1) at (1,2.5) {$\bullet$}; 
\node[] (a1) at (1.5,2.5) {$\bullet$}; 
\node[] (a1) at (2,2.5) {$\bullet$}; 
\node[] (a1) at (2.5,2.5) {$\bullet$}; 
\node[] (a1) at (3,2.5) {$\bullet$};
\node[] (a1) at (3.5,2.5) {$\bullet$};  
\node[] (a1) at (4,2.5) {$\bullet$};
\node[] (a1) at (4.5,2.5) {$\bullet$}; 
\node[] (a1) at (5,2.5) {$\bullet$};

\node[] (a1) at (-2.5,3) {$\bullet$}; 
\node[] (a1) at (-2,3) {$\bullet$}; 
\node[] (a1) at (-1.5,3) {$\bullet$}; 
\node[] (a1) at (-1,3) {$\bullet$}; 
\node[] (a1) at (-0.5,3) {$\bullet$}; 

\node[] (a1) at (0,3) {$\bullet$}; 
\node[] (a1) at (0.5,3) {$\bullet$}; 
\node[] (a1) at (1,3) {$\bullet$}; 
\node[] (a1) at (1.5,3) {$\bullet$}; 
\node[] (a1) at (2,3) {$\bullet$}; 
\node[] (a1) at (2.5,3) {$\bullet$}; 
\node[] (a1) at (3,3) {$\bullet$};
\node[] (a1) at (3.5,3) {$\bullet$};  
\node[] (a1) at (4,3) {$\bullet$};
\node[] (a1) at (4.5,3) {$\bullet$}; 
\node[] (a1) at (5,3) {$\bullet$};

\node[] (a1) at (-2.5,3.5) {$\bullet$}; 
\node[] (a1) at (-2,3.5) {$\bullet$}; 
\node[] (a1) at (-1.5,3.5) {$\bullet$}; 
\node[] (a1) at (-1,3.5) {$\bullet$}; 
\node[] (a1) at (-0.5,3.5) {$\bullet$}; 

\node[green] (a1) at (0,3.5) {$\bullet$}; 
\node[] (a1) at (0.5,3.5) {$\bullet$}; 
\node[] (a1) at (1,3.5) {$\bullet$}; 
\node[] (a1) at (1.5,3.5) {$\bullet$}; 
\node[] (a1) at (2,3.5) {$\bullet$}; 
\node[] (a1) at (2.5,3.5) {$\bullet$}; 
\node[] (a1) at (3,3.5) {$\bullet$};
\node[] (a1) at (3.5,3.5) {$\bullet$};  
\node[] (a1) at (4,3.5) {$\bullet$};
\node[] (a1) at (4.5,3.5) {$\bullet$}; 
\node[] (a1) at (5,3.5) {$\bullet$};

\draw[blue, line width=0.4mm](4.5,4)--(5,4);

\node[] (a1) at (-2.5,4) {$\bullet$}; 
\node[] (a1) at (-2,4) {$\bullet$}; 
\node[] (a1) at (-1.5,4) {$\bullet$}; 
\node[] (a1) at (-1,4) {$\bullet$}; 
\node[] (a1) at (-0.5,4) {$\bullet$}; 

\node[] (a1) at (0,4) {$\bullet$}; 
\node[] (a1) at (0.5,4) {$\bullet$}; 
\node[] (a1) at (1,4) {$\bullet$}; 
\node[] (a1) at (1.5,4) {$\bullet$}; 
\node[] (a1) at (2,4) {$\bullet$}; 
\node[] (a1) at (2.5,4) {$\bullet$}; 
\node[] (a1) at (3,4) {$\bullet$};
\node[] (a1) at (3.5,4) {$\bullet$};  
\node[] (a1) at (4,4) {$\bullet$};
\node[green] (a1) at (4.5,4) {$\bullet$}; 
\node[] (a1) at (5,4) {$\bullet$};

\draw[red] (0,4.5)--(4.5,4.5);

\node[] (a1) at (-2.5,4.5) {$\bullet$}; 
\node[] (a1) at (-2,4.5) {$\bullet$}; 
\node[] (a1) at (-1.5,4.5) {$\bullet$}; 
\node[] (a1) at (-1,4.5) {$\bullet$}; 
\node[] (a1) at (-0.5,4.5) {$\bullet$}; 

\node[] (a1) at (0,4.5) {$\bullet$}; 
\node[] (a1) at (0.5,4.5) {$\bullet$}; 
\node[] (a1) at (1,4.5) {$\bullet$}; 
\node[] (a1) at (1.5,4.5) {$\bullet$}; 
\node[] (a1) at (2,4.5) {$\bullet$}; 
\node[] (a1) at (2.5,4.5) {$\bullet$}; 
\node[] (a1) at (3,4.5) {$\bullet$};
\node[] (a1) at (3.5,4.5) {$\bullet$};  
\node[] (a1) at (4,4.5) {$\bullet$};
\node[] (a1) at (4.5,4.5) {$\bullet$}; 
\node[] (a1) at (5,4.5) {$\bullet$}; 

\end{tikzpicture}
\end{center}
\caption{Illustration of a patch which avoid all up-nooses of intermediate frames.}
\label{fig:patch}
\end{figure}

\begin{lemma}\label{lem:internalLayerTermFree}
Let $(G,\delta,t,w',s')$ be an instance of \fFindFoliostar\ and $r=h(3\delta^{\star}+3t)$ be the constant mentioned in \autoref{cor:wrapped}.
Let $\mu>  2^{c'(\delta^{\star})^2}t$, where $c'$ is the constant mentioned in  \autoref{obs:existsTermFreeFrame}. 
Let $(M,{\cal N})$ be a $2q$-workspace in $\tilde{G}$ and $\eta,\lambda \in\mathbb{N}$, such that 
$q/\lambda-\mu+ 1>(t+\delta^{\star})^{c(t^2+t\delta^{\star})}\cdot\eta^{c\eta} \cdot 2\mu$, 
$\lambda$ divides $q$ and 
$\lambda\geq 50(r+5)$, 
where $c$ is the constant mentioned in \autoref{lem:redundantAlbumExistence}, $\eta=48(r+2)$.  Let $s$ be a $\mu$-redundancy stamp for the $(\eta,\lambda)$-album.  
Then, if ${\cal S}_1$ is a representation for a solution, then there also exists a representation $\cal S$ such that ${\cal S}_1$ and $\cal S$ are identical outside the given workspace, and ${\cal S}$ does not consist of any vertex that belongs to any up-noose of $\fr[i]$ for $i\in\{\lambda s-1, \lambda s,\ldots, \lambda s +(\lambda-2\eta)-2\}$. (In particular, the set of all vertices in these up-nooses is irrelevant.)
\end{lemma}

\begin{proof}
Let $\ell^{\star}=\lambda s-1$ and ${\cal S}_1$ be an arbitrary solution. 
By \autoref{obs:existsTermFreeFrame}, we know that there exists $j\in \{s+1,\ldots,s+\mu-1\}$ such that ${\cal S}_1$ is $(j\lambda,\lambda-1)$-terminal free. We fix $\ell\geq (s+1)\lambda$ such that ${\cal S}_1$ is $(\ell,\lambda-1)$-terminal free. Let $\beta=12(r+2)$ and $d_1=4\beta+7$. Let $d=\lambda-1> 50(r+5)-2$. 
 Then by \autoref{lem:untangledFrame}, there is an $(\ell,\eta)$-untangled solution ${\cal S}_2$.  Let ${\cal S}_2'$ be the representation of ${\cal S}_2$. From ${\cal S}_2'$ we construct a representation ${\cal S}'$ of a solution ${\cal S}$, which will not use any vertex from the up-nooses of $\fr[i]$ for all $i\in\{\lambda s-1, \lambda s,\ldots, \lambda s +(\lambda-2\eta)-2\}$ and that will conclude the proof. 
 Towards that for each $(H,\phi_2,\varphi_2)\in {\cal S}'_2$, we construct a tuple $(H,\phi,\varphi)$ which is a representation of 
a tuple $(H,\phi',\varphi')$, where $(\phi',\varphi')$ witnesses that $H$ is topological minor in $G$ and 
no vertex from any up-noose of $\fr[i]$ for $i\in\{\lambda s-1, \lambda s,\ldots, \lambda s +(\lambda-2\eta)-2\}$ is in the image of $\varphi'$.

Fix  a tuple $(H,\phi_2,\varphi_2)$ in ${\cal S}'_2$. Let $(H^{\star},H',\phi_2^{\star},\varphi_2^{\star},{\sf out})$ be an $(\ell,s)$-miniature of $(H,\phi_2,\varphi_2)$. Let $\extend_{\ell,\eta}(H,\phi_2,\varphi_2)=(\widehat{H},\widehat{\phi}_2,\widehat{\varphi}_2)$, and define $\restrict_{\ell,\eta}(H,\phi_2,\varphi_2)=(H^{\star}_2,H',\phi'_2,\varphi'_2, {\sf out})$ (Recall that because of \autoref{obs:graphsequal}, the second and fifth arguments of $\restrict_{\ell,\eta}(H,\phi_2,\varphi_2)$ and an $(\ell,s)$-miniature of $(H,\phi_2,\varphi_2)$ are same). Notice that $V(\widehat{H})=V(H_2^{\star})\cup V(H')$. Let $s=\vert V(H')\setminus V(H^{\star}_2)\vert=\vert V(H')\setminus V(H^{\star})\vert$ (see \autoref{obs:graphsequal}). 
 By \autoref{lem:patchExists}, there exists a patch ${\cal P}$ from  $(H^{\star},H',\phi_2^{\star},\varphi_2^{\star},{\sf out})$ to $(H,\phi_2,\varphi_2)$. Here,  ${\cal P}$ is a set of vertex-disjoint paths where for all $v\in V(H')\setminus V(H^{\star})$, there is a path $P_v$ between $\phi^\star_2(v)$ and $\phi'_2(v)=\widehat{\phi}_2(v)$ such that every internal vertex on $P_v$ is enclosed either by the noose in $\cal N$ that encloses $\phi^\star_2(v)$ or by a noose $N\in\fr[j]$ for some $j\in\{s\lambda,s\lambda+1,\ldots,\ell\}$. Further, no vertex from any up-noose of $\fr[i]$ for $i\in\{\lambda s-1, \lambda s,\ldots, \lambda s +(\lambda-2\eta)-2\}$ belongs to any path in the patch.


Now we define a pair $(\phi_{\widehat{H}},\varphi_{\widehat{H}})$ which witnesses that $\widehat{H}$ is a topological minor in $\tilde{G}\cup G$ and the set of vertices and edges used by $(\phi_{\widehat{H}},\varphi_{\widehat{H}})$ in $\tilde{G}^{\star}_{\ell-1}$ is a subset of the set of vertices and edges used by $(\phi_2^{\star},\varphi_2^{\star})$ and ${\cal P}$.  This will imply the correctness of the lemma. 

Since $\extend_{\ell,\eta}(H,\phi_2,\varphi_2)=(\widehat{H},\widehat{\phi}_2,\widehat{\varphi}_2)$  and $\restrict_{\ell,\eta}(H,\phi_2,\varphi_2)=(H^{\star}_2,H',\phi'_2,\varphi'_2, {\sf out})$ we know that $\widehat{H}$ is obtained by subdividing the edges in $H_2^{\star}$ and $H'$ is an induced subgraph of $\widehat{H}$. For any $v\in  V(H') \subseteq V(\widehat{H})$, we set $\phi_{\widehat{H}}(v)=\phi_2^{\star}(v)$. For any $v\in V(\widehat{H})\setminus V(H')$, $\phi_{\widehat{H}}(v)=\widehat{\phi}_2(v)$. 


Now we define $\varphi_{\widehat{H}}$. Fix an edge $\{u,v\}\in E(\widehat{H})$. Notice that either $\{u,v\}\in E(H')$ or $\{u,v\}\in E(\widehat{H})\setminus E(H')$.  We have the following cases.

\medskip
\noindent
{\bf Case 1: $\{\phi_{\widehat{H}}(u),\phi_{\widehat{H}}(v)\}\cap A\neq \emptyset$.}
We claim that $\{\phi_{\widehat{H}}(u),\phi_{\widehat{H}}(v)\}\in E(G)$. Recall that $(\widehat{H},\widehat{\phi}_2,\widehat{\varphi}_2)$ is constructed from  $\apex_{\ell,\eta}(H,\phi_2,\varphi_2)$ (see \autoref{def:extend}). This implies that $u$ and $v$ are terminal vertices in $\widehat{H}$ and $\{\widehat{\phi}_2(u),\widehat{\phi}_2(v)\}\in E(G)$. Since  $\{\phi_{\widehat{H}}(u),\phi_{\widehat{H}}(v)\}\cap A\neq \emptyset$, at least one of $\phi_{\widehat{H}}(u)$ or $\phi_{\widehat{H}}(v)$ is in $A$. Say $\phi_{\widehat{H}}(u)\in A$. Then $u\notin V(H')$. 
That  is, $\phi_{\widehat{H}}(u)=\widehat{\phi}_2(u)$. Suppose $v$ also does not belong to $V(H')$. Then $\phi_{\widehat{H}}(v)=\widehat{\phi}_2(v)$ and hence $\{\phi_{\widehat{H}}(u),\phi_{\widehat{H}}(v)\}=\{\widehat{\phi}_2(u),\widehat{\phi}_2(v)\}\in E(G)$. Suppose $v$ belongs to $V(H')$. 
Since ${\cal S}_2$ is $(\ell,\eta)$-untangled, $\{\widehat{\phi}_2(u),\widehat{\phi}_2(v)\}$ is an edge used by $\varphi_2$ and $\widehat{\phi}_2(u)\in A$, we have that $v\in V(H_2^{\star})\cap V(H')$. 
If $v\in R(H')$, then $\phi_{\widehat{H}}(v)=\phi_{2}^{\star}(v)=\phi'_2(v)$ and hence $\{\phi_{\widehat{H}}(u),\phi_{\widehat{H}}(v)\}\in E(G)$. Otherwise $\widehat{\phi}_2(u)\in {\sf out}(v)$. Also, since $(H^{\star},H',\phi_2^{\star},\varphi_2^{\star},{\sf out})$ is an  $(\ell,s)$-miniature of $(H,\phi_2,\varphi_2)$, we have that $\widehat{\phi}_2(u)=\phi^{\star}_2(u) \in N_G(\phi_2^{\star}(v))$. 
Hence $\{\phi_{\widehat{H}}(u),\phi_{\widehat{H}}(v)\}\in E(G)$. So we set $\varphi_{\widehat{H}}(\{u,v\})=\phi_{\widehat{H}}(u)-\phi_{\widehat{H}}(v)$. 

\medskip
\noindent
{\bf Case 2: $\{u,v\}\in E(H')$.}
In this case $\{\phi_{\widehat{H}}(u),\phi_{\widehat{H}}(v)\}\cap A= \emptyset$. Here we set $\varphi_{\widehat{H}}(\{u,v\})=\varphi^{\star}_2(\{u,v\})$.

\medskip
\noindent
{\bf Case 3: $u,v\in V(H')$ and $\{u,v\}\notin E(H')$.} 
In this case $\{\phi_{\widehat{H}}(u),\phi_{\widehat{H}}(v)\}\cap A= \emptyset$. 
Since $\{u,v\}\notin E(H')$ and $u,v\in V(H')$, $\widehat{\phi}_2(v)=\phi'_2(v)$ is  a vertex  present in  a noose $N_{v}$ in $\fr[\ell]$ and $\widehat{\phi}_2(u)$ is a vertex in a noose $N_u$ in $\fr[\ell]$. Moreover, all the vertices in the path $\widehat{\varphi}_2(\{u,v\})$ belong to nooses in $\bigcup_{\ell'\geq  \ell}\fr[\ell']$, because $(H,\phi_2,\varphi_2)$ is  in $S_2'$, the representation $(\ell,\eta)$-untangled solution ${\cal S}_2$.  Let $J_{u,v}$ the graph $P_u\cup P_v\cup \widehat{\varphi}_2(\{u,v\})$, which is connected graph. Choose a path $P_{u,v}$ from $J_{u,v}$ between $\phi^{\star}_2(u)$ and $\phi^{\star}_2(v)$. Then we set   $\varphi_{\widehat{H}}(\{u,v\})=P_{u,v}$.

\medskip
\noindent
{\bf Case 4: $u,v \in V(\widehat{H})\setminus V(H')$ and $\{\phi_{\widehat{H}}(u),\phi_{\widehat{H}}(v)\}\cap A= \emptyset$.}
In this case $\{u,v\}\in E(\widehat{H})\setminus E(H')$. 
The vertices $\widehat{\phi}_2(u)$ and $\widehat{\phi}_2(v)$ belong to  nooses  in $\bigcup_{\ell'>\ell}\fr[\ell']$. Moreover all the vertices in the path $\widehat{\varphi}_2(\{u,v\})$ belong to nooses in $\bigcup_{\ell'> \ell}\fr[\ell']$, because $(H,\phi_2,\varphi_2)$ is  
in $S_2'$, the representation $(\ell,\eta)$-untangled solution ${\cal S}_2$.  So we set $\varphi_{\widehat{H}}(\{u,v\})=\widehat{\varphi}_2(\{u,v\})$.

\medskip
\noindent
{\bf Case 5: $u\in V(\widehat{H})\setminus V(H')$, $v\in V(H')$, and $\{\phi_{\widehat{H}}(u),\phi_{\widehat{H}}(v)\}\cap A= \emptyset$.}
Since $u\in V(\widehat{H})\setminus V(H')$, $\extend_{\ell,\eta}(H,\phi_2,\varphi_2)=(\widehat{H},\widehat{\phi}_2,\widehat{\varphi}_2)$, $\restrict_{\ell,\eta}(H,\phi_2,\varphi_2)=(H^{\star}_2,H',\phi'_2,\varphi'_2, {\sf out})$, and $\{u,v\}\in E(\widehat{H})$, we have that  $v\notin H_2^{\star}$. That is $\widehat{\phi}_2(v)=\phi'_2(v)$ is  a vertex  present in  a noose $N_{v}$ in $\fr[\ell]$ and $\widehat{\phi}_2(u)$ is a vertex in a noose $N_u$ in $\bigcup_{\ell'>\ell}\fr[\ell']$. 
Moreover all the vertices in the path $\widehat{\varphi}_2(\{u,v\})$ belong to nooses in $\bigcup_{\ell'\geq  \ell}\fr[\ell']$, because ${\cal S}_2$ is an $(\ell,\eta)$-untangled solution. 
Let $J_{u,v}$ be the (connected) graph $\widehat{\varphi}_2(\{u,v\})\cup P_v$. Choose a path $P_{u,v}$ from $J_{u,v}$ between $\phi_{\widehat{H}}(u)=\widehat{\phi}_2(u)$ and $\phi_{\widehat{H}}(v)=\phi^{\star}_2(v)$. Then we set   $\varphi_{\widehat{H}}(\{u,v\})=P_{u,v}$. 


This completes the definition $\varphi_{\widehat{H}}$. Notice that in the above construction we used ``the portion of $\widehat{\varphi}_2$ which belong to $\bigcup_{\ell'\geq \ell}\fr[\ell']$'', the edges used by $\varphi'_2$ from 
$\bigcup_{\ell'\leq \ell^{\star}}\fr[\ell']$ and a subset of edges from the patch ${\cal P}$.  That is, any vertex in any up-noose in $\fr[\ell^{\star}]$ does not belong to the image of $\varphi_{\widehat{H}}$. Thus, if we prove that $(\widehat{H}, \phi_{\widehat{H}},\varphi_{\widehat{H}})$  is a representation of $(\widehat{H},f_{\widehat{H}},g_{\widehat{H}})$, where $(f_{\widehat{H}},g_{\widehat{H}})$ is a witness for $\widehat{H}$ being a topological minor in $G$, then we are done, because $\widehat{H}$ is subdivision of $H$.  Towards that, we first prove the following claim. 

\begin{claim}
\label{claim:Hhattopo}
$(\phi_{\widehat{H}},\varphi_{\widehat{H}})$ witnesses that $\widehat{H}$ is a topological minor in $\tilde{G}\cup G$.  
\end{claim}

\begin{proof}[Proof sketch]
Clearly for any $\{u,v\}\in E(\widehat{H})$,  the end-vertices of $\varphi_{\widehat{H}}(\{u,v\})$ are $\phi_{\widehat{H}}(u)$ and $\phi_{\widehat{H}}(v)$. Now we prove that the  paths in $\varphi_{\widehat{H}}(E(\widehat{H}))$ are internally vertex disjoint. This follows from the following facts. 
\begin{itemize}
\item The set of paths ${\cal P}$ are vertex disjoint and all the vertices in $V({\cal P})$  are from the  nooses in $\bigcup_{\ell^{\star}\leq j \leq \ell}\fr[j]$. 
\item $(H^{\star},H',\phi_2^{\star},\varphi_2^{\star}, {\sf out})$ is an $(\ell^{\star},\eta)$-partial solution and hence condition $(2)$ in \autoref{def:partial} is satisfied. That is, the paths in $\varphi^{\star}_2(E(H'))$ are internally vertex disjoint and the internal vertices are from $V(\tilde{G}_{\ell^{\star}-1})$. 
\item The set of paths $\widehat{\varphi}_2(E(\widehat{H}))\setminus \widehat{\varphi}_2(E(H'))$ are internally vertex disjoint and the internal vertices in these paths $\widehat{\varphi}_2(E(\widehat{H}))\setminus \widehat{\varphi}_2(E(H^{\star}))$ are  from  $V(G)\setminus (V(\tilde{G}_{\ell-1})\cup A)$. 
\item If an end-vertex of a path $P$ in $\varphi_{\widehat{H}}(E(\widehat{H}))$ is a vertex in $A$, then its length is one and hence $P$ has no internal vertex (see Case $1$ above). 
\end{itemize}
The paths in the first two items together are internally vertex disjoint. The vertices in paths in third item is disjoint from the the vertices of paths in the second item.  The vertices in paths in the first item may not be internally vertex disjoint from the paths in third item.  But, because ${\cal S}_2$ is $(\ell,\eta)$-untangled, if any two paths intersect, then one of the endpoints of those paths are same and this endpoint is an internal vertex only in one path in $\varphi_{\widehat{H}}(E(\widehat{H}))$. 
\end{proof}

From $(\phi_{\widehat{H}},\varphi_{\widehat{H}})$, we construct a pair $(f_{\widehat{H}},g_{\widehat{H}})$ such that it is a witness for $\widehat{H}$ being a topological minor in $G$ and that will complete the proof.  
%
%
%
We set $f_{\widehat{H}}=\phi_{\widehat{H}}$. 
By \autoref{claim:Hhattopo}, we know that $(\phi_{\widehat{H}},\varphi_{\widehat{H}})$ is witness for the topological minor  $\widehat{H}$ in $\tilde{G}\cup G$. 
If no edge of $E(\tilde{G})\setminus E(G)$ is used by  $\varphi_{\widehat{H}}$, then $g_{\widehat{H}}=\varphi_{\widehat{H}}$.  Otherwise we modify $\varphi_{\widehat{H}}$ to get $g_{\widehat{H}}$. Let 
$F\subseteq E(\tilde{G})\setminus E(G)$ be the set of edges outside $E(G)$, used by $\varphi_{\widehat{H}}$. 
Notice that the edge used by $\varphi_{\widehat{H}}$ can be partitioned into $E_1\uplus E_2\uplus E_3$ such that 
$(i)$ $E_3$ used by $\widehat{\varphi}_2$  and 
disjoint from $E(\tilde{G}_{\ell-1})$, 
 $(ii)$ $E_2$ is used by 
${\cal P}$ and 
disjoint from $E(\tilde{G}_{\ell^{\star}-1})$
and $(iii)$ $E_1$ is used by $\varphi^{\star}_2$.  

\begin{claim}
\label{claim:E3nice}
Let $Q_{\ell}=\bigcup_{ i \leq \ell}\fr[i]$ and $U_{\ell}= \bigcup_{N\in Q_{\ell}}\inNoose_{\tilde{G}}(N)\cap V(\tilde{G})$. There is no $G_i, i\in [k]$, such that $V(G_i)\cap U_{\ell}\neq \emptyset $ and $V(G_i)$ contain a terminal with respect to ${\cal S}_2$. 
\end{claim}
\begin{proof}
The proof follows from the fact that ${\cal S}_2$ is $(\ell,3)$-terminal free (because ${\cal S}$ is $(\ell,\eta)$-untangled).
\end{proof}

We partition the edges 
$F$ into $F_1\uplus F_2\uplus F_3$ such that $F_i=F\cap E_i$ for all $i\in [3]$.  For all  edges in $F_1$, $(a)$ there exist internally vertex disjoint paths where the internal vertices are from $V(G_i)\setminus V(G_0)$ some $G_i$ with $V(G_i)\cap V(G_0)\subseteq V(\tilde{G}_{\ell^{\star}-1})$ (see condition $(2)$ in \autoref{def:partial}). For 
all edges in $F_2$, $(b)$ there exist internal vertex disjoint paths where the internal vertices are from $V(G_i)\setminus V(G_0)$ for some $G_i$ with $V(G_i)$ contains the endpoints of the corresponding edges. Moreover the paths are internally vertex disjoint from the vertices used by the paths mentioned in statement $(a)$. 
For all edges in $F_3$, there exit internal vertex disjoint paths where the internal vertices are from $V(G_i)\setminus V(G_0)$ some $G_i$ with $V(G_i)$ contains the the endpoints of the corresponding edges (because $\widehat{\varphi}_2$ is derived from a representation of a solution). Moreover, because of \autoref{claim:E3nice}, the internal vertices in these paths are disjoint from the paths mentioned in statements $(a)$ and $(b)$. This completes the construction of $g_{\widehat{H}}$. Therefore $\widehat{H}$ is a topological minor in $G$, witnessed by $(f_{\widehat{H}},g_{\widehat{H}})$ and  $(\widehat{H},\phi_{\widehat{H}},\varphi_{\widehat{H}})$ is the  representation of $(\widehat{H},f_{\widehat{H}},g_{\widehat{H}})$. As  vertices in the up-nooses of $\fr[i]$ for all $i\in\{\lambda s-1, \lambda s,\ldots, \lambda s +(\lambda-2\eta)-2\}$ are not  used by $\varphi_{\widehat{H}}$, those vertices are also not used by $g_{\widehat{H}}$. This completes the proof of the lemma. 
\end{proof}

Notice that in Lemma \ref{lem:internalLayerTermFree}, the graph induced by the specified irrelevant set of vertices has a $(\lambda-2\eta)\times(\lambda-2\eta)$ grid as a minor.  Thus, we derive the following corollary to Lemma \ref{lem:internalLayerTermFree}.

\begin{corollary}\label{cor:internalLayerTermFree}
Let $(G,\delta,t,w',s')$ be an instance of \fFindFoliostar\ and $r=h(3\delta^{\star}+3t)$ be the constant mentioned in \autoref{cor:wrapped}.
Let $\mu>  2^{c'(\delta^{\star})^2}t$, where $c'$ is the constant mentioned in  \autoref{obs:existsTermFreeFrame}. 
Let $(M,{\cal N})$ be a $2q$-workspace in $\tilde{G}$ and $\eta,\lambda \in\mathbb{N}$, such that 
$q/\lambda-\mu+ 1>(t+\delta^{\star})^{c(t^2+t\delta^{\star})}\cdot\eta^{c\eta} \cdot 2\mu$, 
$\lambda$ divides $q$ and 
$\lambda\geq 50(r+5)$, 
where $c$ is the constant mentioned in \autoref{lem:redundantAlbumExistence}, $\eta=48(r+2)$.  Let $s$ be a $\mu$-redundancy stamp for the $(\eta,\lambda)$-album.  
Then, there exists a $w''\times w''$ flat wall within the input flat wall that is computable in polynomial time such that if ${\cal S}_1$ is a representation of a solution, then there also exists a representation $\cal S$ that is identical to ${\cal S}_1$ outside the input flat wall and which does not consist of any vertex of the output $w''\times w''$ flat wall, where $w''=\frac{\lambda}{2}-\eta$.
\end{corollary}

\section{Proof of Theorem~\ref{thm:main1Flatstar}}

%

We are now ready to prove \autoref{thm:main1Flatstar}. For the sake of clarity, let us recall its statement in more detail. 


\begin{theorem}
\label{thm:flatdetail}
There is a computable function $g$ and an algorithm that, given an instance $(G,\delta,t,w',s')$ of \fFindFoliostar\ such that  $|R(G)|\leq \boundary$ and $w'\geq g(\delta^{\star},t)$, and $w''\in\mathbb{N}$, finds a $w''\times w''$ flat wall within the input $w'\times w'$ flat wall such that if ${\cal S}'$ is a representation of the solution, then there also exists a representation $\cal S$ that is identical to ${\cal S}'$ outside the input flat wall and does not use any vertex of the output inner flat wall, which runs in time 
$ 2^{2^{\OO(((t+\delta^{\star})^2+r)\log (t+\delta^{\star}+r))}}(s')^{\OO(s')}(w'')^{\OO(w'')} n$
where  $g(\delta^{\star},t)=(t+\delta^{\star}+r)^{\OO((t+\delta^{\star})^2+r)}w''$  and $r=h(\delta^{\star}+t)$.  (In particular, the set of all vertices of the output inner flat wall is irrelevant.)
\end{theorem}

\begin{proof}
Let $k'=3(\delta^{\star}+t)$ and $r=h(k')$ be the constant mentioned in \autoref{cor:wrapped}.
 That is, $r$ depends on $\delta$ and $t$. Fix $\lambda=\max\{50(r+5),2w''+2\eta\}$, $\eta=48(r+2)$, and 
$\mu= 2^{c'(\delta^{\star})^2}t$, where $c'$ is the constant mentioned in  \autoref{obs:existsTermFreeFrame}. 
Let $q$ be the least integer such that 
$q/\lambda-\mu+ 1>(t+\delta^{\star})^{c(t^2+t\delta^{\star})}\cdot\eta^{c\eta} \cdot 2\mu$
 and $\lambda$ divides $q$, where 
$c$ is the constant mentioned in  \autoref{lem:redundantAlbumExistence}. 
 That is, $q\geq g(\delta^{\star},t)$ for some computable function $g$ such that $g(\delta^{\star},t)=(t+\delta^{\star})^{\OO((t+\delta^{\star})^2)}r^{\OO(r)}=(t+\delta^{\star}+r)^{\OO((t+\delta^{\star})^2+r)}$.  

Let $\widehat{c}$ be the constant mentioned in \autoref{lem:workspace}. Let $p=2\widehat{c}q+1$ and $w'=2\widehat{c}p$. 
We  are given a $w'\times w'$ flat wall in $G\setminus A$.  Thus by \autoref{obs:wlalltogrid}, there is a ${w'}\times {w'}$-grid as a minor in $\tilde{G}$. This implies that the treewidth of $\tilde{G}$ strictly more than $\widehat{c}p$.  So we apply \autoref{lem:workspace} and get a $(p,q)$-workspace of $\tilde{G}$ in time $p^{\OO(1)}n$.  Next we apply \autoref{lem:redundantAlbumExistence}, and compute a $\mu$-redundancy stamp $s$ in time 
$((t+\delta^{\star})^{\OO(t^2+t\delta^{\star})}\cdot\eta^{\OO(\eta)}\cdot \mu) \Delta^{\OO(\Delta)} n$
 where $\Delta=\max \{p,s'+3\}$.
Then because of \autoref{lem:internalLayerTermFree}, we output 
a vertex in an up-noose of $\fr[s\lambda -1]$. Therefore, the total running time of the algorithm is  
$((t+\delta^{\star})^{\OO(t^2+t\delta^{\star})}\cdot\eta^{\OO(\eta)}\cdot \mu)  
p^{\OO(p)} (s')^{\OO(s')} n= 2^{2^{\OO(((t+\delta^{\star})^2+r)\log (t+\delta^{\star}+r))}}(s')^{\OO(s')} (w'')^{\OO(w'')}n$. 
This completes the proof of the theorem. 
\end{proof}

\section{Final Argument for Flat Walls: Proof of \autoref{thm:finalFlatWall}}

\label{sec:finallyend}

Finally, based on 
\autoref{thm:flatdetail},  we are ready to prove \autoref{thm:finalFlatWall}. For the sake of clarity, let us recall its statement. 


\dellargewall*


\begin{proof}
Let $\widehat{g}$ be the computable function in \autoref{thm:flatdetail}
 where $w''=1$. We first describe the algorithm. Given $k\in\mathbb{N}$ and an instance $(G,\delta,t,w',s')$ of \fFindFolio\ such that  $|R(G)|\leq \boundary$ and $w'\geq (\widehat{g}(\delta^{\star},t))^{k+2}$, it works as follows. For every $i\in\{1,2,\ldots,k+2\}$, let $w_i=(\widehat{g}(\delta^{\star},t))^{k+3-i}$. Define the first flat wall as the input flat wall. Then, for every $i\in\{1,2,\ldots,k+1\}$, execute the following:
 Apply the algorithm in 
 \autoref{thm:flatdetail}
  with $w'$ (in that statement) being $w_i$, the $i$-th flat wall, and $w''=w_{i+1}$, and call the output the $(i+1)$-flat wall. (Here, the first call is valid because $w'\geq w_1$.) Eventually, output any vertex $v$ in the $(k+2)$-flat wall as the $(\delta,k)$-irrelevant vertex. Notice that $w_{k+2}\geq 1$, thus such a vertex exists.

The running time is upper bounded by the running time to perform $(k+1)$ calls to the algorithm in 
\autoref{thm:flatdetail}
with $w''$ being upper bounded by $w_1$. As the time to perform one such call is bounded by $2^{2^{\OO(k((t+\delta^{\star})^2+r)\log (t+\delta^{\star}+r))}}(s')^{\OO(s')}n$, the running time stated in the theorem follows (the factor $k+1$ is subsumed by the $\OO$ notation in the exponent).

We now argue that the algorithm is correct, that is,  $v$ is a $(\delta,k)$-irrelevant vertex. Targeting towards a contradiction, suppose that $v$ is not a $(\delta,k)$-irrelevant vertex. Then, there exists a set $S\subseteq V(G)$ of size at most $k$ such that $v$ is not an irrelevant vertex to the extended $\delta$-folios of $G\setminus S$.
Thus, by \autoref{obs:flatfolioonly}, the $\delta^\star$-folio of $G\setminus S$ and $G\setminus (S\cup\{v\})$ are different (clearly, the first is a superset of the second). Let $\cal S$ be the solution for the instance $(G\setminus S,\delta,t,w',s')$ of \fFindFoliostar, and let ${\cal S}'$ be a representation of it. 
 By the pigeon-hole principle, there exists $i\in\{1,2,\ldots,k\}$ such that $S$ does not contain any vertex in the difference between the $i$-th flat wall and the $(i+1)$-th flat wall. Then, by the correctness of the $i^{th}$ call to the algorithm in 
 \autoref{thm:flatdetail}, 
  there exists a representation ${\cal S}^\star$ of ${\cal S}$ with respect to $(G,\delta,t,w',s')$ that is the same as ${\cal S}'$ outside the $i$-th flat wall and does not use any vertex of the $(i+1)$-th flat wall. Because ${\cal S}^\star$ is the same as ${\cal S}'$ outside the $i$-th flat wall, it does not use any vertex of $S$ that is outside the $i$-flat wall. Further, because $S$ does not contain any vertex in the difference between the $i$-th flat wall and the $(i+1)$-th flat wall, it is trivial that ${\cal S}^\star$ does not use any vertex of $S$ that belongs to this difference (there is no such vertex), and because ${\cal S}^\star$ does not use use any vertex of the $(i+1)$-th flat wall, it is also trivial that is does not use any vertex of $S$, as well as $v$, that belongs to the $(i+1)$-th flat wall. Overall, ${\cal S}^\star$ does not use any vertex of $S\cup\{v\}$. Thus, ${\cal S}^\star$ is a representation of ${\cal S}$ with respect to $(G\setminus (S\cup\{v\},\delta,t,w',s')$, which is a contradiction.
\end{proof}

\section{Conclusion and Future Directions}\label{sec:conclusion}
In this paper we established fixed-parameter tractability of \TMH.
The {immersion} relation~\cite{RobertsonS04} is another well-studied graph containment relation with strong ties to graph minors. 
%
%
 %
For the immersion relation, an  \FPT{} algorithm for  
$\Pi$-{\sc Edge Deletion} (deleting $k$ edges instead of $k$ vertices) for all $\preceq_{im}$-closed properties follows from the fact that graphs are well-quasi ordered by immersion~\cite{RobertsonS04} and that there exists an $f(H)\cdot n^{\OO(1)}$ time algorithm for deciding whether a given graph $H$ is a immersion of $G$ ~\cite{DBLP:conf/stoc/GroheKMW11}.  
We believe that our methods could be useful in designing an \FPT\ algorithm for $\Pi$-{\sc Deletion} for all $\preceq_{im}$-closed properties. However, for now this is just speculation, and we leave this as an interesting open problem.

The running time dependence $f(h^\star,k)$ of our algorithm for \TMH\ is humongous. This should not come as a surprise---\TMH\ is a generalization of \textsc{Minor Deletion}, thus our algorithm also solves   $\Pi$-{\sc Deletion}  for every {minor-closed} property $\Pi$. On the other hand, for some special but still interesting  cases of  \textsc{Minor Deletion}, such as when the family of forbidden minors contains a planar graph,  \textsc{Minor Deletion} is solvable in time bounded by a single-exponential function of $k$~\cite{FominLMS12,KLPRRSS16}. It would be interesting if such type of results could be obtained for \TMH\ when some constraints are imposed on the graphs in the family of forbidden topological minors. 

%
\bibliographystyle{siam}
\bibliography{Refs,RefsMZ,ff,kernels,referencesPDeletion,RefFP}

\appendix

\end{document}